\documentclass[10pt,a4]{amsart}
\usepackage{verbatim}
\usepackage{enumerate}
\usepackage{hyperref}
\usepackage{amssymb}
\usepackage{mathbbol}

\makeatletter
\@namedef{subjclassname@2020}{%
  \textup{2020} Mathematics Subject Classification}
\makeatother

\newcommand{\ddd}{\mathrm{d}}
\newcommand{\EE}{\mathrm{e}}
\newcommand{\II}{\mathrm{i}}

\newcommand{\ii}{a}
\newcommand{\oo}{b}

\newcommand{\frg}{\mathfrak{g}}

\DeclareMathOperator{\Span}{span}
\DeclareMathOperator{\Ad}{Ad}

\DeclareMathOperator{\Aut}{Aut}
\DeclareMathOperator{\Id}{Id}

\mathchardef\ordinarycolon\mathcode`\:
\mathcode`\:=\string"8000
\begingroup \catcode`\:=\active
  \gdef:{\mathrel{\mathop\ordinarycolon}}
\endgroup

\usepackage[sep=3pt, offset=1.2em]{simpler-wick} 
\usepackage{float} 

\usepackage{subcaption}

\usepackage{tikz}
\usepackage{tkz-euclide}
\usepackage{animate}
\usetikzlibrary{arrows.meta}
\usetikzlibrary{quotes}
\tikzset{
  fermion/.style={draw=black, postaction={decorate},decoration={markings,mark=at position .55 with {\arrow{>}}}},
    bdry/.style={draw,shape=circle,fill=black,minimum size=5pt,inner sep=0pt},
  b1/.style={draw,shape=circle,minimum size=5pt,inner sep=0pt},
  b2/.style={draw,shape=circle,fill=gray!40,minimum size=5pt,inner sep=0pt},
  b3/.style={draw,shape=circle,fill=gray!80,minimum size=5pt,inner sep=0pt},
  bv/.style={draw,shape=circle,fill=black!80,minimum size=3pt,inner sep=0pt},
  iv/.style={draw,shape=circle,fill=gray!80,minimum size=3pt,inner sep=0pt}
  }
  
\newcommand{\dd}{\partial}
\newcommand{\ra}{\rightarrow}

\newcommand{\mr}{\mathrm}

\newcommand{\KK}{\mathbb{K}}
\newcommand{\CC}{\mathbb{C}}
\newcommand{\RR}{\mathbb{R}}
\newcommand{\g}{\mathfrak{g}}
\newcommand{\FF}{\mathcal{F}}

\newcommand{\BB}{\mathcal{B}}
\newcommand{\YY}{\mathcal{Y}}

\newcommand{\PP}{\mathcal{P}}
\newcommand{\LL}{\mathcal{L}}

\newcommand{\til}{\widetilde}

\newcommand{\ttt}{\mathbb{t}}

\newcommand{\As}{\mathsf{A}}

\newcommand{\AsI}{\mathsf{A}_I}

\newcommand{\ad}{\mathrm{ad}}
\newcommand{\nl}{\mathrm{nl}}
\newcommand{\ol}{\overline}
\renewcommand{\ggg}{\mathcal{G}}
\newcommand{\tr}{\mathrm{tr}}

\newcommand{\bl}{\textcolor{blue}}

\theoremstyle{remark}
\newtheorem{remark}{Remark}[section]
\newtheorem{digression}[remark]{Digression}

\theoremstyle{plain}
\newtheorem{lemma}[remark]{Lemma}
\newtheorem{proposition}[remark]{Proposition}

\newtheorem{theorem}[remark]{Theorem}

\theoremstyle{definition}
\newtheorem{definition}[remark]{Definition}
\newtheorem{example}[remark]{Example}
\newtheorem{assumption}[remark]{Assumption}

\begin{document}
\title[Constraints, Hamilton--Jacobi, quantization
]
{
Constrained Systems, Generalized Hamilton--Jacobi Actions, and Quantization
}

\begin{abstract}
Mechanical systems (i.e., one-dimensional field theories) with constraints are the focus of this paper.
In the classical theory, systems with infinite-dimensional targets are considered as well (this then encompasses also higher-dimensional field theories in the hamiltonian formalism). The properties of the Hamilton--Jacobi (HJ) action are described in details and several examples are explicitly computed (including nonabelian Chern--Simons theory, where the HJ action turns out to be the gauged Wess--Zumino--Witten action). Perturbative quantization, limited in this note  to finite-dimensional targets, is performed in the framework of the Batalin--Vilkovisky (BV) formalism in the bulk and of the Batalin--Fradkin--Vilkovisky (BFV) formalism at the endpoints. As a sanity check of the method, it is proved that the semiclassical contribution of the physical part of the evolution operator is still given by the HJ action. Several examples are computed explicitly. In particular, it is shown that the toy model for nonabelian Chern--Simons theory and the toy model for 7D Chern--Simons theory with nonlinear Hitchin polarization do not have quantum corrections in the physical part (the extension of these results to the actual cases is discussed in 
the companion paper \cite{CS_cyl}). 
Background material for both the classical part (symplectic geometry, generalized generating functions, HJ actions, and the extension of these concepts to infinite-dimensional manifolds) and the quantum part (BV-BFV formalism) is provided.
\end{abstract}

\author{Alberto S. Cattaneo}
\address{Institut f\"ur Mathematik, Universit\"at Z\"urich\\
Winterthurerstrasse 190, CH-8057 Z\"urich, Switzerland}  
\email{cattaneo@math.uzh.ch}

\author{Pavel Mnev}
\address{University of Notre Dame}
\address{St. Petersburg Department of V. A. Steklov Institute of Mathematics of the Russian Academy of Sciences}
\email{pmnev@nd.edu}

\author{Konstantin Wernli}
\address{University of Notre Dame}
\email{kwernli@nd.edu}

\thanks{This research was (partly) supported by the NCCR SwissMAP, funded by the Swiss National Science Foundation. A.S.C. and K.W. acknowledge partial support of SNF Grant No.\ 200020\_192080.  K. W. also acknowledges support from a BMS Dirichlet postdoctoral fellowship and the SNF Postdoc.Mobility grant P2ZHP2\_184083, and would like to thank the Humboldt-Universit\"at Berlin, in particular the group of Dirk Kreimer, and the university of Notre Dame for their hospitality.}

\keywords{
Hamilton--Jacobi, (generalized) generating functions,
Chern--Simons, Wess--Zumino--Witten,
 nonlinear (Hitchin) phase space polarization,  Kodaira--Spencer (BCOV) action,
 Batalin--Vilkovisky, Batalin--Fradkin--Vilkovisky}

\subjclass[2020]{81T70, 
53D22, 
70H20, 
53D55, 
53D50 (Primary); 
81T13, 
81S10, 
70H15, 
57R56, 
81T45 (Secondary). 
}

\dedicatory{In memory of Kirill Mackenzie}

\maketitle

\setcounter{tocdepth}{2}
\tableofcontents

\allowdisplaybreaks

\section{Introduction}
The Hamilton--Jacobi (HJ) action for a nondegenerate system is a generating function for the graph of the hamiltonian flow of the system and is obtained by substituting a solution of the Euler--Lagrange (EL) equations into the action functional. We extend this result to degenerate systems (i.e., systems with constraints). More precisely, we consider one-dimensional field theories with constraints in involution (a.k.a.\ first-class constraints), i.e., theories constrained to live in a coisotropic submanifold of the (possibly infinite-dimensional) target. In the finite-dimensional case, we also study their perturbative quantization.

This paper is mainly based on a discussion of examples. 
The two main general results, on the role of the HJ action, are 
presented in Theorem \ref{t-thm1} on page~\pageref{t-thm1}
at the classical level and in Theorem~\ref{t-thm2} on page~\pageref{t-thm2} at the quantum level. 

Among the many examples, we discuss in particular Chern--Simons theory and its finite-dimensional toy model versions,  get the gauged WZW action as 
the generalized generating function of the evolution relation via
the HJ action, and, in the finite-dimensional version, show that there are no quantum corrections in the physical sector.
We will prove the latter result in the infinite-dimensional case in the companion paper \cite{CS_cyl}.

We start with a recollection of the known results for a regular mechanical system (i.e., a one-dimensional field theory with regular lagrangian). In this case, the EL equations define a hamiltonian flow and therefore its graph, which we call the evolution relation, is a lagrangian submanifold of the product of the symplectic manifolds associated to the two endpoints (in the framework of this paper the symplectic manifold associated to an endpoint is the target symplectic manifold, up to the sign of the symplectic form which depends on the orientation of the endpoint). The HJ action, defined as the evaluation of the classical action on solutions of the EL equation for a given choice of endpoint polarizations\footnote{By polarization we mean, in this paper, an explicit realization of a symplectic manifold as a cotangent bundle. Generating functions, wave functions and integral kernels of operators will always be understood as functions on the base of the cotangent bundle.} {(boundary conditions)}, is then a generating function\footnote{\label{f:genfun}A generating function for a lagrangian submanifold $L$ of a cotangent bundle is a function on the base that yields $L$ as the graph of its differential.} for the evolution relation. At the quantum level, the HJ action appears as the dominant contribution in the semiclassical expansion of Feynman's path integral (simply as the evaluation of the action on a background classical solution).

In a constrained system, the following differences occur:
\begin{enumerate}
\item In the framework of this paper,\footnote{with the exception of the first part of Section~\ref{ss: 2dYM and EM} which portrays a more general situation} the EL equations split into two classes which we call the evolution equations (those that involve a time derivative) and the constraints {(those that do not involve time derivatives)}.  The fields actually split into two classes: maps to the target symplectic manifolds and Lagrange multipliers; the evolution equations arise from variations with respect to the former and the constraints from variations with respect to the latter.
\item Solutions of the EL equations still define a lagrangian submanifold---which we keep calling the evolution relation---of the product of the symplectic manifolds associated to the two endpoints. Yet, this submanifold is no longer a graph.\footnote{Geometrically, the constraints define a coisotropic submanifold $C$, and the evolution equations describe a motion along a leaf ot the characteristic distribution. The evolution relation then consists of pairs of points of $C$ that lie on the same leaf.}
\item The HJ action, defined as the evaluation of the classical action on solutions of the evolution equation for a given choice of endpoint polarizations, is then a generalized generating function\footnote{A generalized generating function is a generating function, see footnote~\ref{f:genfun}, that depends on additional parameters. It defines the lagrangian submanifold of the given cotangent bundle as the intersection of the graph of its differential with respect to the
base variables
and the vanishing locus of the differential with respect to the additional parameters.} for the evolution relation. Note that the constraints are not used in the definition of the HJ action but are recovered from it (as they appear in the definition of the evolution relation).
\item The HJ action depends parametrically on the Lagrange multipliers but is invariant under a (Lie algebra or, more generally, Lie algebroid) gauge transformation thereof that is trivial at the endpoints.
\item At the quantum level, the path integral requires a gauge-fixing procedure one has to keep track of. The formalism we use in this paper is the one due to Batalin and Vilkovisky, in particular in its version coupled to the boundary, referred to as the BV-BFV formalism. The HJ action then appears as the physical (i.e., ghost independent) part of
the dominant contribution of the perturbative expansion of the path integral for certain choices of endpoint polarizations and of residual fields (the latter being, in physical parlance, a choice of infrared components, not to be integrated out, of the fields). The HJ action now arises from the resummation of tree diagrams and not from the evaluation of the action on a background classical solution, which is not at hand in the BV-BFV formalism.
\end{enumerate}
Points (3) and (4) are the content of Theorem~\ref{t-thm1} and point (5) of Theorem~\ref{t-thm2}---see also equations \eqref{e:qHJgeneralLie} and \eqref{e:ZINL}.

One topic we are interested in is the change of (say, final) endpoint polarization. This is realized by a generating function of the corresponding symplectomorphism. At the classical level, the HJ action for the new polarization is simply the composition of the HJ action for the old polarization with the generating function for the change of polarization. At the quantum level, one has to find an appropriate quantization of the latter generating function. There are two, in general irreconcilable, wishes: unitarity and compatibility with the BV-BFV formalism. In this paper we focus on the latter and show that a BV-BFV quantization is possible if the constraints are linear in the momenta of both endpoint polarizations. This is a serious limitation, yet it applies to some very interesting examples.

\subsection{Structure of the paper}
This paper consists of three parts. 

In the first part (Sections~\ref{s:HJnondeg} to~\ref{r:other}), we discuss the HJ action for constrained systems (after briefly reviewing the regular, unconstrained case). We discuss several examples (in particular constraints that are linear or affine in both the positions and the momenta), prove {Theorem~\ref{t-thm1}}, i.e., point (2) of the above list, describe how to incorporate a nontrivial evolution in addition to the constraints, and extend the discussion of generating functions and HJ actions to some bad endpoint conditions.

In the second part (Section~\ref{s:inftarg}), we
compute the HJ action for
several examples with infinite-dimensional target (Chern--Simons theory, $BF$ theory, 2D Yang--Mills theory, abelian Yang--Mills theory in any dimension).

In the third part (Sections~\ref{s:BFVAKSZBBV} to~\ref{s:nonlinpols}), we discuss the BV-BFV quantization. We start with a review of BV and BFV, tailored to the theories we consider in this paper (which are ultimately one-dimensional AKSZ models). Next we discuss the quantization in the case of finite-dimensional targets, both with linear and nonlinear endpoint polarizations, proving in particular Theorem~\ref{t-thm2}, i.e., point (5) of the above list, and show that in the linear and affine examples (the toy models for Chern--Simons theory) no quantum corrections arise in the physical sector: see equations \eqref{e:ZIlin}, \eqref{e:ZIbiaff} and \eqref{e:ZINLlin}.
 
 In Appendix~\ref{a:generatingfunctions} we review several results on symplectic geometry and on (generalized) generating functions, both in the finite- and in the infinite-dimensional case.

In Appendix~\ref{s:mQME} we prove a technical result that ensures that the BV-BFV quantization works in the case at hand.

We defer the BV-BFV quantization with infinite-dimensional target (mainly, Chern--Simons theories) to the companion paper \cite{CS_cyl}.

\subsection{How to read this paper}
The first part is necessary for the whole paper, with the exception of Section~\ref{r:nontrivev} which is relevant only for
Sections~\ref{ss: 2dYM and EM}, \ref{s:QM} and~\ref{s:QRP}.

The second and the third part can be read independently from each other, as they focus on different developments of the first part.

Appendix~\ref{a:generatingfunctions}  contains additional material that completes the discussion in the first part and provides a rigorous foundation for the second part. Appendix~\ref{s:mQME} contains results for the foundation of the third part.



\subsection{Notations}
To simplify the reading of the paper, we collect here some of the main notations and some terminology.

Arguments (e.g., of an action, a HJ action, a partition function) are written in square brackets if they are fields (i.e., if they depend on the variable on the source interval) and in round brackets otherwise.

We will use the subscript HJ to denote a HJ action. We will have several versions thereof, with notations like $S_\text{HJ}$, $S^f_\text{HJ}$ or $\Hat S^f_\text{HJ}$ referring to different incarnations (change of polarization by $f$; dependency on gauge classes of the parameters). This pedantic notation is important to avoid confusion, but the only important point to focus on is that these are all HJ actions in their own way.
In the arguments of HJ actions, or more generaly of generalized generating functions, we will consistenly separate the generalized position variables from the parameters by a semicolon.

In the BFV formalism, the classical master equation (CME) is the equation $\{S,S\}=0$, where $\{\ ,\ \}$ is the given even Poisson bracket and 
$S$ is an odd functional. By $\Omega$ we will denote a coboundary operator that quantizes $S$.

In the BV formalism, the CME is the equation $(S,S)=0$, where $(\ ,\ )$ is the given odd Poisson bracket and 
$S$ is an even functional. By $\Delta$ we will denote the BV Laplacian (on functions or on half-densities, depending on the context). The quantum master equation (QME) is the equation $\Delta\psi=0$, where $\psi$ is an even function or half-density. For $\psi$ of the form $\EE^{\frac\II\hbar S}$, under the condition $\Delta S=0$ (which we usually assume), the QME for $\psi$ is equivalent to the CME for $S$.

In the BV-BFV formalism, where we have both $\Delta$ and $\Omega$, the modified quantum master equation (mQME) is the equation $(\Omega+\hbar^2\Delta)\psi=0$. We often use the abbreviation $\Delta^\Omega:=\Delta+\frac1{\hbar^2}\Omega$.

For a partition function (or a state), we will use the notation $Z$, possibly with some label denoting the case we are considering. If we define it as a half-density, we will use the notation $\Hat Z$.

\subsection{Teaser}
We illustrate one of the results of the paper (summarizing the content of Sections~\ref{s:fLiealgebra}, \ref{s:Iq}, and~\ref{s:nonlinpols}, with simplified, adapted notations for this introduction).

The data of the system are a quadratic Lie algebra $\frg$, a vector space $Z$ and two derivations: $v\colon\frg\to\frg\otimes Z$
and $\bar v\colon\frg\to\frg\otimes Z^*$.
The target symplectic manifold is $T^*(\frg\otimes Z)$ which we identify with
$\frg\otimes Z^*\oplus \frg\otimes Z\ni(\bar q,q)$ (using the nondegenerate pairing $\langle\ ,\rangle$ on $\frg$).\footnote{We write $\bar q$ for the image of the momentum $p$ under the isomorphism $\frg^*\otimes Z^*\to\frg\otimes Z^*$.}
The fields are a path $(\bar q,q)\colon [0,1]\to T^*(\frg\otimes Z)$ and a $\frg$-valued one-form $e$ (i.e., a connection one-form) on $[0,1]$. The action functional, a toy model for nonabelian Chern--Simons theory (see Section~\ref{s:naCS}),\footnote{In this case, we have $\frg=\Omega^0(\Sigma,\frg')$ and $Z=\Omega^{0,1}(\Sigma)$,  where $\frg'$ is a finite-dimensional quadratic Lie algebra, and $\Sigma$ is a closed, oriented surface with a choice of complex structure. The derivations $v,\bar v$ are now the Dolbeault operators $\dd,\bar\dd$.} is
\[
S[\bar q,q,e]=\int_{0}^{1} (\langle\bar q,\ddd_e q\rangle - \langle\bar q,v(e)\rangle -\langle\bar v(e),q\rangle),
\]
where $\ddd_e$ denotes the covariant derivative.

Under the assumption that the exponential map to the simply connected Lie group $G$ integrating $\frg$ is surjective,
we compute the HJ action---for initial value $q_\ii$ of $q$, final value $\bar q^\oo$ of $\bar q$, and a parameter $g=\EE^\xi\in G$---as
\[
\Hat S_\text{HJ}(q_\ii,\bar q^\oo;g)=-\langle \bar q^\oo,{g}^{-1}q_\ii g\rangle
-\langle \bar q^\oo,{g}^{-1}v(g)\rangle
-\langle \bar v(g)g^{-1},q_\ii\rangle
-\text{WZW}(g)
\]
with the Wess--Zumino--Witten term
\[
\text{WZW}(g)=\frac12\langle\bar v(g)g^{-1},v(g)g^{-1}\rangle+
\frac12\int_0^1 \langle\bar v(h)h^{-1},[v(h)h^{-1},\ddd h h^{-1}]\rangle,
\]
where $h=\EE^{t\xi}$.

The BV-BFV quantization then provides an integral kernel which, in adddition to the classical coordinates $(q_\ii,\bar q^\oo)$, also depends on ``ghost variables'' $c_\ii,c_\oo\in\frg[1]$ and is a half-density on $T^*[-1]G\ni (g,g^+)$:
\[
\Hat Z(q_\ii,c_\ii,\bar q^\oo,c_\oo;g,g^+) =
\EE^{\frac\II\hbar
\Hat S_\text{HJ}(q_\ii,\bar q^\oo;g)}
\;\EE^{-\frac\II\hbar \langle g^+,gc_\ii-c_\oo g
\rangle}\, \mu_G,
\]
where $\mu_G$ is the Haar measure. 

This integral kernel satisfies the modified quantum master equation (mQME)
$(\Omega+\hbar^2\Delta)\Hat Z=0$---and is defined up to the image of $\Omega+\hbar^2\Delta$---where $\Delta$ is the canonical BV Laplacian on half-densities on $T^*[-1]G$ and
\[
\begin{split}
 \Omega &= \left\langle\II\hbar\frac\dd{\dd q_\ii},v(c_\ii)\right\rangle +
 \left\langle\bar v(c_\ii),q_\ii\right\rangle
 -\frac{\II\hbar}2\left\langle[c_\ii,c_\ii],\frac\dd{\dd c_\ii}\right\rangle\\
&\phantom{=}- \left\langle\bar q^\oo,v(c_\oo)\right\rangle -
 \left\langle\bar v(c_\oo),\II\hbar\frac\dd{\dd\bar q^\oo}\right\rangle
 -\frac{\II\hbar}2\left\langle[c_\oo,c_\oo],\frac\dd{\dd c_\oo}\right\rangle
\end{split}
\]
is a quantization of the endpoint BFV action describing the coisotropic reduction by the constraints.

In the abelian case, where we use the multiplicative convention for the group $G$,
we write $g=\EE^T$, and the HJ action becomes
\[
\Hat S_\text{HJ}(q_\ii,\bar q^\oo;T)=-\langle \bar q^\oo,q_\ii\rangle
-\langle \bar q^\oo,v(T)\rangle
-\langle \bar v(T),q_\ii\rangle
-\frac12\langle \bar v(T),v(T)\rangle.
\]
In this case
it is also interesting, as a toy model for 7D Chern--Simons theory with endpoint Hitchin polarization (see Section~\ref{s:7d}), to change the final Darboux coordinates by some generating function $f(q,Q)$ with
$\langle\bar q,\ddd q\rangle=\langle\bar Q,\ddd Q\rangle + \ddd f$. The HJ action 
for initial value $q_\ii$ of $q$ and final value $Q_\oo$ of $Q$ is then
\[
\Hat S_\text{HJ}^f(q_\ii,Q_\oo;T)=-f(q_\ii+v(T),Q_\oo)-\langle \bar v(T),q_\ii\rangle
-\frac12\langle \bar v(T),v(T)\rangle.
\]

The BV-BFV quantization turns out to be possible when the transformed $\frg^*$-valued hamiltonian 
$\Tilde H(\bar Q,Q)=v^*(\bar q(\bar Q,Q))+\bar v^*(q(\bar Q,Q))$, with $*$ denoting  transposition,
is linear in $\bar Q$. In this case, the integral kernel, with $(T,T^+)\in T^*[-1]\frg$, is
\[
\Hat Z^f(q_\ii,c_\ii,Q_\oo,c_\oo;T,T^+)=
\EE^{\frac\II\hbar
\Hat S^f_\text{HJ}(q_\ii,Q_\oo;T)}
\EE^{-\frac\II\hbar T^+_\alpha\,(c_\ii^\alpha-c_\oo^\alpha)}\sqrt{\ddd^kT\ddd^k T^+},
\]
with $k=\dim\frg$. This integral kernel also satisfies the modified quantum master equation (mQME), with 
$\Delta$ the canonical BV Laplacian on half-densities on $T^*[-1]\frg$ and
\[
\Omega =  \left\langle\II\hbar\frac\dd{\dd q_\ii},v(c_\ii)\right\rangle +
 \left\langle\bar v(c_\ii),q_\ii\right\rangle
 -\left\langle c_\oo,\Tilde H\left(-\II\hbar\frac\dd{\dd Q_\oo},Q_\oo\right)\right\rangle.
\]

\begin{remark}[Parameters and quantum corrections]
In the examples just presented, there are no quantum corrections in the physical part of the partition function $\Hat Z$. This statement if of course correct only if we compare the classical and quantum theory in the presence of parameters. The point is that usually one can get rid of (some of) the parameters. Let us discuss for simplicity the case when all parameters can be eliminated. In the classical theory, the constraints are recovered by setting to zero the differential of the HJ action as a function of the parameters $g$ or $T$ (parametrized by the generalized position variables).
If the Hessian is nondegenerate, we can (locally) solve these equations and write the parameters as functions of the generalized position variables. Inserting this back into the HJ action yields a new, equivalent generating function, now without parameters. At the quantum level, we can instead integrate out the parameters  $g$ or $T$  in the partition function, after setting their odd momenta $g^+$ or $T^+ $ to zero. In the abelian case with linear polarization (generalized position variables $q_\ii$ and $\bar q_\oo$), this is equivalent to the classical procedure (apart from an irrelevant constant depending on the determinant of the quadratic form $\langle \bar v(T),v(T)\rangle$). In the other cases (abelian case with nonlinear polarization and nonabelian case), the quantum procedure yields nontrivial quantum corrections.
\end{remark}

%
%
%
%




\subsection*{Acknowledgements}
We thank Samson Shatashvili for suggesting the study of 7D abelian Chern--Simons theory in the quantum BV-BFV formalism, now in \cite[Section 6.3]{CS_cyl},
which was the original motivation out of which this paper grew. 
We also thank Francesco Bonechi, Ivan Contreras, {Philippe Mathieu}, Nicolai Reshetikhin, Pavel Safronov, Michele Schiavina,
{Stephan Stolz}, Alan Weinstein, Ping Xu and Donald Youmans for useful discussions. 

\newcommand{\EL}{\text{\sl EL}}

\section{Hamilton--Jacobi for nondegenerate actions}\label{s:HJnondeg}
In this section we collect several known facts on the Hamilton--Jacobi action, which will be used in the rest of the paper (for a general reference, see, e.g., \cite[Sections 47-48]{Arnold}).
We only focus on actions of the form
\[
S[p,q] =\int_{t_\ii}^{t_\oo} (p_i\ddd q^i - H(p,q)\ddd t),
\]
where a sum over repeated indices is understood. Here $(p,q)$ denotes a point (or a path) in the cotangent bundle
$T^*M$ of a fixed manifold $M$, and $H$ is a smooth function on $T^*M$. 
By abuse of notation, we denote by $\ddd q^i$ both the coordinate 1-form on $M$ and
its pullback by the map $q\colon[t_\ii,t_\oo]\to M$ (which then
actually means $\dot q^i\ddd t$). 
We will later be interested in some generalizations, like adding constraints and working with an infinite-dimensional target.

Recall that the Euler--Lagrange (EL) equations for the above action are Hamilton's equations
\[
\dot q^i = \frac{\dd H}{\dd p_i},\quad \dot p_i = -\frac{\dd H}{\dd q^i}.
\]
Next assume that for every $q_\ii,q_\oo\in M$
these equations have a unique solution, which we will denote by $(p^{q_a,q_\oo},q_{q_\ii,q_\oo})$, satisfying
$q_{q_\ii,q_\oo}(t_\ii)=q_\ii$ and $q_{q_\ii,q_\oo}(t_\oo)=q_\oo$. Then one defines the Hamilton--Jacobi (HJ) action (a.k.a.\ Hamilton's principal function)
\[
S_\text{HJ}(q_\ii,q_\oo):=S[p^{q_\ii,q_\oo},q_{q_\ii,q_\oo}].
\]
Jacobi's theorem then, in particular, asserts that
\begin{equation}\label{e:HJone}
p^{q_\ii,q_\oo}_i(t_\oo)=\frac{\dd S_\text{HJ}}{\dd q^i_\oo},\quad p^{q_\ii,q_\oo}_i(t_\ii)=-\frac{\dd S_\text{HJ}}{\dd q^i_\ii}.
\end{equation}
In other words, $S_\text{HJ}$ is a generating function for the graph of the hamiltonian flow of $H$ (from $t_\ii$ to $t_\oo$), see Appendix~\ref{a:generatingfunctions} (in particular, Remark~\ref{r:gfcanrel}).
Let us briefly review the proof because we will need some generalization thereof in the following.
The variation of the action reads
\[
\delta S = \EL + p_i(t_\oo)\delta q^i(t_\oo)-p_i(t_\ii)\delta q^i(t_\ii),
\]
where $\EL$ denotes the bulk term responsible for the EL equations. If we insert a solution of the EL equations, the  $\EL$ term vanishes. In particular, inserting the solution $(p^{q_\ii,q_\oo},q_{q_\ii,q_\oo})$ yields
\[
\delta S_\text{HJ} = p^{q_\ii,q_\oo}_i(t_\oo)\delta q^i_\oo-p^{q_\ii,q_\oo}_i(t_\ii)\delta q_\ii^i,
\]
which implies \eqref{e:HJone}.
\begin{remark}[Time dependency]\label{r:timedep}
Above we have considered the time endpoints $t_\ii$ and $t_\oo$ as fixed, but one can let them vary as well. 
We put the explicit dependence in the notation: 
\[
S[p,q](t_\ii,t_\oo) =\int_{t_\ii}^{t_\oo} (p_i\ddd q^i - H(p,q)\ddd t).
\]
The time-dependent HJ action is then
\begin{equation}\label{e:HJtimedep}
S_\text{HJ}(q_\ii,q_\oo;t_\ii,t_\oo):=S[p^{q_\ii,q_\oo},q_{q_\ii,q_\oo}](t_\ii,t_\oo),
\end{equation}
where the solution is still required to satisfy $q_{q_\ii,q_\oo}(t_\ii)=q_\ii$ and $q_{q_\ii,q_\oo}(t_\oo)=q_\oo$. As a result,
$(p^{q_\ii,q_\oo},q_{q_\ii,q_\oo})$ also depends on $t_\ii$ and $t_\oo$ (even though we do not include this in the notation). In addition to the formulae in \eqref{e:HJone}, we now also have\footnote{Equations \eqref{e:HJone} and \eqref{e:HJtime} are also important because they imply that $S_\text{HJ}$ is a solution of the Hamilton--Jacobi equation.}
\begin{equation}\label{e:HJtime}
\frac{\dd S_\text{HJ}}{\dd t_\oo}=-H(p^{q_\ii,q_\oo}(t_\oo),q_\oo),\quad
\frac{\dd S_\text{HJ}}{\dd t_\ii}=H(p^{q_\ii,q_\oo}(t_\ii),q_\ii).
\end{equation}
To get these formulae, we have to observe that, e.g., a variation $\delta t_\oo$ of $t_\oo$ produces not only a variation in the final endpoint of the integral, which yields $(p^{q_\ii,q_\oo}_i(t_\oo)\dot q_{q_\ii,q_\oo}^i(t_\oo)-H(p^{q_\ii,q_\oo}(t_\oo),q_\oo))\delta t_\oo$,
but also an induced variation $-\dot q_{q_\ii,q_\oo}(t_\oo)\delta t_\oo$ of $q$ at the final endpoint.
\end{remark}

\begin{remark}[Uniqueness]
For simplicity we have assumed that there is a unique solution 
$(p^{q_\ii,q_\oo},q_{q_\ii,q_\oo})$ satisfying
$q_{q_\ii,q_\oo}(t_\ii)=q_\ii$ and $q_{q_\ii,q_\oo}(t_\oo)=q_\oo$. This might not be true. Generically one may exclude focal points $q_\ii,q_\oo$ connected by a continuous family of solutions,\footnote{{For instance, if the target $M$ is a Riemannian manifold with metric $g$ and we take $H(p,q)=\frac{\|p\|_g^2}{2m}$, then the solutions to the EL equations are cotangent lifts of parametrized geodesics. If $M$ is the sphere with round metric and we take antipodal points $q_\ii$ and $q_\oo$, then we have a continuous family of solutions joining them. However, it is enough to perturb one of them, to get a discrete family.}} but typically we may expect a discrete, or finite, family to exist. In this case, the above procedure is actually meant to work in a neighborhood of a given nonfocal pair $(q_\ii,q_\oo)$
for a choice of one solution in the discrete set. When one varies $q_\ii$ or $q_\oo$ (or the time interval), 
it is understood that one follows the variation of the chosen solution. The resulting HJ action is then a local generating function. To get the whole graph of the hamiltonian flow of $H$ one needs the various, possibly infinitely many, HJ actions associated to the different solutions.
\end{remark}

\begin{example}\label{example-nonuniqueS}
A simple example where uniqueness of solutions with given endpoints is not satisfied is the free particle on $S^1$, the unit circle.\footnote{The nonuniqueness in this example is due to the fact that $S^1$ is not simply connected, but it is possible to find examples of nonuniqueness on simply connected configuration spaces as well.} The hamiltonian is simply $H(p,q)=\frac{\,p^2}{2m}$. The EL equations are $\dot p =0$ and $\dot q=\frac pm$. The graph of the hamiltonian flow of $H$ is then
\[
L:=\left\{(p^\ii,q_\ii,p^\oo,q_\oo)\in T^*S^1\times T^*S^1\ |\ p^\oo=p^\ii,\ q_\oo = q_\ii +\frac {p^\ii}m (t_\oo-t_\ii)\bmod2\pi
\right\}.
\]
On the other hand, if we fix $q_\ii$ and $q_\oo$, we have the solutions 
\[
\begin{split}
p(t) &= \frac{m(q_\oo-q_\ii+2k\pi)}{t_\oo-t_\ii},\\
q(t) &= q_\ii + (q_\oo-q_\ii+2k\pi)\frac{t-t_\ii}{t_\oo-t_\ii}\bmod 2\pi,
\end{split}
\]
for each $k\in\mathbb{Z}$. It follows that, for a given $k$, the HJ action reads
\[
S_{\text{HJ},k}(q_\ii,q_\oo) = \frac12 m \frac{(q_\oo-q_\ii+2k\pi)^2}{t_\oo-t_\ii},
\]
which is indeed a generating function for $L$ in a neighborhood of $(q_\ii^0,p_0,q_\oo^0,p_0)$ with
$p_0=\frac{m(q_\oo^0-q_\ii^0+2k\pi)}{t_\oo-t_\ii}$. To generate the whole $L$, one needs all the $S_{\text{HJ},k}$s.
\end{example}

\begin{remark}[Semiclassical approximation]\label{r:semiclassicalapproximation}
In the path integral approach to quantum mechanics (say, on the configuration space $M=\RR^n$), the integral kernel $K(q_\ii,q_\oo)$ for the evolution operator is, formally, given by the path integral \cite{Dirac,Feynman,FH}
\begin{equation}\label{e:Kpathintegral}
K(q_\ii,q_\oo) = \int_{\substack{q(t_\ii)=q_\ii\\q(t_\oo)=q_\oo}}Dp\,Dq\;\EE^{\frac\II\hbar S[p,q]}.
\end{equation}
Assuming for simplicity that there is a unique solution 
$(p^{q_\ii,q_\oo},q_{q_\ii,q_\oo})$ satisfying
$q_{q_\ii,q_\oo}(t_\ii)=q_\ii$ and $q_{q_\ii,q_\oo}(t_\oo)=q_\oo$, one usually makes the affine change of variables (translation)
\begin{equation}\label{e:affpq}
p\mapsto p^{q_\ii,q_\oo} + \Hat p,\quad q\mapsto q_{q_\ii,q_\oo} +\Hat q,
\end{equation}
getting
\[
K(q_\ii,q_\oo) = \EE^{\frac\II\hbar S_\text{HJ}(q_\ii,q_\oo)}
\int_{\Hat q(t_\ii)=\Hat q(t_\oo)=0}D\Hat p\,D\Hat q\;\EE^{\frac\II\hbar \Hat S_{(q_\ii,q_\oo)}[\Hat p,\Hat q]}
\]
with
\[
\Hat S_{(q_\ii,q_\oo)}[\Hat p,\Hat q]=S[p^{q_\ii,q_\oo} + \Hat p,q_{q_\ii,q_\oo} +\Hat q]-S_\text{HJ}(q_\ii,q_\oo).
\]
Since $(p^{q_\ii,q_\oo},q_{q_\ii,q_\oo})$ is a solution, $\Hat S_{(q_\ii,q_\oo)}$ starts with a quadratic term in $(\Hat p,\Hat q)$, so one can define the path integral perturbatively. One then sees that the HJ action yields the semiclassical approximation to $K$. The crucial point here is that the Hessian of $S$ at a solution is nondegenerate (one also says that the action or the lagrangian is regular). This does not happen for degenerate actions (e.g., in gauge theories). In this case one has to use some technique to gauge fix the action, but one also has to take into account that the symplectic space of initial conditions is defined via some symplectic reduction. A cohomological way to resolve these issues is to use the BV formalism \cite{BV} in the bulk and the BFV formalism \cite{BF,FV}
on the boundary. In \cite{CMR15} it was shown how to 
reconcile 
the two formalisms, into what was termed the quantum BV-BFV formalism,
 and it was in particular observed that, in general, extending boundary values to a solution of the EL equations is incompatible with the formalism (we will review this in Sections~\ref{sec:goodsplit} and~\ref{s:IIdiscosplit}).
The above argument to show that the semiclassical contribution to the integral kernel of the evolution operator is the HJ action is therefore no longer valid. 
We will see in Theorem~\ref{t-thm2} 
that the result still holds, yet the HJ action is recovered not by evaluating the classical action on a solution but via resumming trees in the Feynman diagrams expansion.
\end{remark}


\begin{remark}[Changing the endpoint conditions]\label{r:changing}
Fixing $q$ at the endpoints yields the EL equations as critical points for $S$. One may however be interested in choosing different endpoint conditions. In general, they have to correspond to a lagrangian submanifold of $T^*M$, not necessarily the zero section, but one has to adapt the one-form $p_i\ddd q^i$. For simplicity, we now assume $M=\RR^n$. Let $(P_i,Q^i)$ be another set of Darboux coordinates on $T^*\RR^n$ and assume (see also Example~\ref{r:gfchangepol} for a more general viewpoint) that we have a generating function $f(q,Q)$ for the transfomation: namely,
\[
p_i\ddd q^i=P_i\ddd Q^i + \ddd f,
\]
i.e.,
\[
p_i = \frac{\dd f}{\dd q^i},\quad P_i = -\frac{\dd f}{\dd Q^i}.
\]
Assume we want to fix $Q$ at the final endpoint and $q$ at the initial endpoint. We then change the action by a boundary term:
\[
S^f[p,q] :=- f(q(t_\oo),Q(p(t_\oo),q(t_\oo))) + \int_{t_\ii}^{t_\oo} (p_i\ddd q^i - H(p,q)\ddd t).
\]
When we take a variation, we now get
 \[
\delta S^f = \EL + P_i(t_\oo)\delta Q^i(t_\oo)-p_i(t_\ii)\delta q^i(t_\ii),
\]
with the same $\EL$ term as before. Therefore, the solutions of the EL equations are the critical points of $S^f$
with fixed $q$ at the initial endpoint and $Q$ at the final endpoint. We define the corresponding HJ action as
\[
S^f_\text{HJ}(q_\ii,Q_\oo):=S^f[p^{q_\ii,Q_\oo},q_{q_\ii,Q_\oo}],
\]
where $(p^{q_\ii,Q_\oo},q_{q_\ii,Q_\oo})$ is the (assumedly unique) solution ot the EL equations with
$q_{q_\ii,Q_\oo}(t_\ii)=q_\ii$ and $Q(p^{q_\ii,Q_\oo}(t_\oo),q_{q_\ii,Q_\oo}(t_\oo))=Q_\oo$. 
Exactly as in the case discussed at the very beginning of this section, inserting a solution into the variation of $S^f$
kills the $EL$ term, so we get
\[
P_i(p^{q_\ii,Q_\oo},q_{q_\ii,Q_\oo})(t_\oo)=\frac{\dd S^f_\text{HJ}}{\dd Q^i_\oo},\quad p^{q_\ii,Q_\oo}_i(t_\ii)=-\frac{\dd S^f_\text{HJ}}{\dd q^i_\ii}.
\]
In other words, $S^f_\text{HJ}$ is a generating function for the graph of the hamiltonian flow of $H$ (from $t_\ii$ to $t_\oo$) with respect to the new set of coordinates (see again Remark~\ref{r:gfcanrel}). 
The dependency on time can be studied repeating the steps of Remark~\ref{r:timedep} verbatim.
\end{remark}

\begin{example}\label{exa:HJonlyf}
Consider the case when $H=0$, i.e., $S[p,q] =\int_{t_\ii}^{t_\oo} p_i\ddd q^i$. Since the EL equations have constant solutions, we have $S[p^{q_\ii,Q_\oo},q_{q_\ii,Q_\oo}]=0$ and $q(t_\oo)=q_\ii$.
Therefore
\begin{equation}\label{e:HJonlyf}
S^f_\text{HJ}(q_\ii,Q_\oo) = -f(q_\ii,Q_\oo).
\end{equation}
\end{example}

\begin{remark}
For simplicity, 
we only discuss the change of the final endpoint conditions, but of course the same discussion may be repeated for the initial endpoint and for both endpoints at the same time. 
\end{remark}

\begin{example}\label{exa:pfinal}
A very common case is when one wants to fix $q$ at the initial point and $p$ at the final point. This fits in the above construction by choosing $Q=p$ and $P=-q$ and $f(q,Q)=Q_iq^i$. In this case, we simply have
\[
S^f[p,q] :=- p_i(t_\oo)q^i(t_\oo) + \int_{t_\ii}^{t_\oo} (p_i\ddd q^i - H(p,q)\ddd t),
\]
\[
\delta S^f = \EL - q^i(t_\oo)\delta p_i(t_\oo) 
-p_i(t_\ii)\delta q^i(t_\ii),
\]
and
\begin{equation}\label{e:qpinSf}
q_{q_\ii,p^\oo}^i(t_\oo)=-\frac{\dd S^f_\text{HJ}}{\dd p_i^\oo},\quad p^{q_\ii,p^\oo}_i(t_\ii)=-\frac{\dd S^f_\text{HJ}}{\dd q^i_\ii}.
\end{equation}
If we apply this to the Example~\ref{example-nonuniqueS} of the free particle on the circle, we get a unique solution for a given $(q_\ii,p^\oo)$ and the HJ action
\[
S^f_\text{HJ}(q_\ii,p^\oo)=-p^\oo q_\ii-\frac{(p^\oo)^2}{2m}(t_\oo-t_\ii).
\]
If we instead consider the case $H=0$ of Example~\ref{exa:HJonlyf}, we get, as a particular case of \eqref{e:HJonlyf},
\begin{equation}\label{e:HJonlyfpq}
S^f_\text{HJ}(q_\ii,p^\oo)=-p^\oo q_\ii.
\end{equation}
\end{example}

\begin{remark}[Legendre transform]\label{r:Legendretransf}
One can pass from $S^f_\text{HJ}(q_\ii,p^\oo)$ to $S_\text{HJ}(q_\ii,q_\oo)$, and back, by the Legendre transform:
\[
S_\text{HJ}(q_\ii,q_\oo)= \lambda^\text{crit}_iq_\oo^i+S^f_\text{HJ}(q_\ii,\lambda^\text{crit}),
\]
where $\lambda^\text{crit}$ is the critical point (assumed to be unique) of the variable $p^\oo$ for the function $f(p^\oo):=p^\oo_iq_\oo^i+S^f_\text{HJ}(q_\ii,p^\oo)$:
\[
\frac{\dd f}{\dd p^\oo} = 0 \text{\ at\ } p^\oo=\lambda^\text{crit}.
\]
We will return to this in Section~\ref{s:partLegtransf}.
\end{remark}

\begin{remark}\label{r:qexafinal}
In terms of the path-integral computation of the evolution operator, see Remark~\ref{r:semiclassicalapproximation},
changing the endpoint conditions, as discussed in Remark~\ref{r:changing}, corresponds to assigning different Hilbert space representations, corresponding to different choices of polarization, to the endpoints. For example, fixing $p$ at the final endpoint, as in Example~\ref{exa:pfinal}, corresponds to computing the integral kernel $\Hat K(q_\ii,p^\oo)$ for the position representation at the initial endpoint and the momentum representation at the final one. This is the Fourier transform of $K(q_\ii,q_\oo)$ with respect to the second argument. One can see a classical track of this in the
$p_i(t_\oo)q^i(t_\oo)$ term in $S^f$ (and in the Legendre transform of the previous remark).
\end{remark}

\begin{example}\label{exa:qpfinal}
Consider 
\[
\Hat K(q_\ii,p^\oo)= \int_{\substack{q(t_\ii)=q_\ii\\p(t_\oo)=p^\oo}}Dp\,Dq\;\EE^{\frac\II\hbar S^f[p,q]}.
\]
with $S^f[p,q]=-p(t_\oo)q(t_\oo)+S[p,q]$ for the simple action $S[p,q]=\int_{t_\ii}^{t_\oo} p\,\ddd q$. (Nothing changes if we have many $p$ and $q$ variables: just think of a hidden scalar product $p\ddd q=p_i\ddd q^i$.) As the EL equations imply that the solutions are constant, the change of variables \eqref{e:affpq} is now
\[
p\mapsto p^\oo + \Hat p,\quad q\mapsto q_\ii +\Hat q,
\]
with $\Hat p(t_\oo)=0$ and $\Hat q(t_\ii)=0$.
Since $S^f[p,q]=-p^\oo(q_\ii +\Hat q(t_\oo))+\int_{t_\ii}^{t_\oo} (p^\oo + \Hat p)\ddd \Hat q$, we have
\[
\Hat K(q_\ii,p^\oo)= \EE^{-\frac\II\hbar p^\oo q_\ii} 
\int_{\substack{\Hat q(t_\ii)=0\\\Hat p(t_\oo)=0}}D\Hat p\,D\Hat q\;\EE^{\frac\II\hbar \left(-p^\oo \Hat q(t_\oo)
+ \int_{t_\ii}^{t_\oo} p^\oo \ddd \Hat q 
+ \int_{t_\ii}^{t_\oo} \Hat p\ddd \Hat q
\right)}.
\]
This is a Gaussian integral with quadratic term $\int_{t_\ii}^{t_\oo} \Hat p\ddd \Hat q$ and no $\Hat p$-sources, so it simply yields a constant (which we normalize to one), and we get
\begin{equation}\label{e:HatKqp}
\Hat K(q_\ii,p^\oo)= \EE^{-\frac\II\hbar p^\oo q_\ii} ,
\end{equation}
the integral kernel for the Fourier transform, whose exponent is given by the HJ action \eqref{e:HJonlyfpq}.
\end{example}

\begin{example}[Formal quantum change of polarization]\label{exa:qQfinal}
The computation of Example~\ref{exa:qpfinal} may be generalized to the general case discussed in Remark~\ref{r:changing}. Namely, we have
\[
\Hat K^f(q_\ii,Q_\oo)= \int_{\substack{q(t_\ii)=q_\ii\\
Q(p(t_\oo),q(t_\oo))=Q_\oo}}Dp\,Dq\;\EE^{\frac\II\hbar S^f[p,q]},
\]
with $S^f[p,q]=- f(q(t_\oo),Q(p(t_\oo),q(t_\oo))) + S[p,q]$. Again we consider $S[p,q]=\int_{t_\ii}^{t_\oo} p\,\ddd q$.
 The change of variables \eqref{e:affpq} now becomes
 \[
p\mapsto p^0 + \Hat p,\quad q\mapsto q_\ii +\Hat q,
\]
with the constant solution $p^0$ given by 
\[
p^0 = \frac{\dd f}{\dd q}{(q_\ii +\Hat q(t_\oo),Q_\oo)},
\]
and endpoint conditions $\Hat p(t_\oo)=0$ and $\Hat q(t_\ii)=0$.\footnote{Note that the endpoint fluctuations of $p$ and $q$ are constrained by the condition $Q(p(t_\oo),q(t_\oo))=Q_\oo$. We have arranged the change of variables so that, while $p^0$ may fluctuate at the endpoint, $\Hat p$ does not.} In general this change of variables is no longer affine, since $p^0$ depends on $\Hat q$, possibly in a nonlinear way. On the other hand, the Jacobian determinant of the transformation is (formally) $1$ (the Jacobian is lower {triangular} 
with ``$1$''s on the diagonal).
We now have
\[
\Hat K^f(q_\ii,Q_\oo)= 
\int_{\substack{\Hat q(t_\ii)=0\\\Hat p(t_\oo)=0}}D\Hat p\,D\Hat q\;\EE^{\frac\II\hbar \left(-f(q_\ii +\Hat q(t_\oo),Q_\oo)
+ \int_{t_\ii}^{t_\oo} p^0 \ddd \Hat q 
+ \int_{t_\ii}^{t_\oo} \Hat p\ddd \Hat q
\right)},
\]
which is again a Gaussian integral with quadratic term $\int_{t_\ii}^{t_\oo} \Hat p\ddd \Hat q$ and no $\Hat p$-sources, Therefore,
\begin{equation}\label{e:HatKonlyf}
\Hat K^f(q_\ii,Q_\oo)= \EE^{-\frac\II\hbar f(q_\ii,Q_\oo)},
\end{equation}
the exponential of the HJ action \eqref{e:HJonlyf}.
\end{example}

\begin{digression}[Unitarity]\label{d:uni}
The result \eqref{e:HatKonlyf}
of the previous Example~\ref{exa:qQfinal} is correct insofar as the starting point is the quantization of the action $S^f$
as a partition ``function.'' This has to be corrected if we understand the integral kernel as a half-density in the boundary variables, which is more appropriate {for composition}.
The idea is that the functional measure is the product over points of the measure $\ddd p\,\ddd q/(2\pi\hbar)$ (for simplicity we discuss this for the target $T^*\RR$). At the end points we then retain only the square root of the fixed boundary variable, as the other corresponding square root will come from another integral kernel we want to compose with. For example, in passing from the $q$- to the $p$-representation, the integral kernel half-density will be
\[
\KK(q_\ii,p^\oo)= \sqrt{\ddd q_\ii}\;\EE^{-\frac\II\hbar p^\oo q_\ii}\;\sqrt{\frac{\ddd p^\oo}{2\pi\hbar}}
\]
instead of \eqref{e:HatKqp}. Note that $\KK$ is unitary:
\[
\int_{\{p^\oo\}}  \KK(q,p^\oo)\,\overline\KK(p^\oo,q') =
\sqrt{\ddd q}\;\left(\int_{\{p^\oo\}}  \EE^{\frac\II\hbar p^\oo(q'-q)}\,\frac{\ddd p^\oo}{2\pi\hbar}\right)
\;\sqrt{\ddd q'}=
\sqrt{\ddd q}\,\delta(q'-q)\,\sqrt{\ddd q'}.
\]
Unitarity is somehow expected, since reversing the orientation of the interval on which the theory is defined produces
$\overline\KK(p^\oo,q')$, and the composition of the quantum evolution on an interval with the reversed one should be the identity. 
In Example~\ref{exa:qQfinal}, we could follow the same reasoning and expect unitarity (indeed, at the classical level, composition of the generating functions $f(q',Q_\oo)$ and $-f(q,Q_\oo)$, see Remark~\ref{r:compgenfunII}, yields a generating function for the identity); however, in general, this is not the case at the quantum level as we will see in a moment. In this example as well
we split the fields in terms of the initial $q$ and final $p$ variables, but then we expressed the final values in terms of $Q_\oo$. The change of variables we perfomed there entails in particular
\begin{equation}\label{e:ptoQviaf}
p(t_\oo) = \frac{\dd f}{\dd q}{(q_\ii +\Hat q(t_\oo),Q_\oo)}
\end{equation}
and therefore
\[
\ddd p(t_\oo) = \frac{\dd^2 f}{\dd q\,\dd Q}{(q_\ii +\Hat q(t_\oo),Q_\oo)}\,\ddd Q_\oo +
\frac{\dd^2f}{\dd q^2}{(q_\ii +\Hat q(t_\oo),Q_\oo)}\,(\ddd q_\ii +\ddd\Hat q(t_\oo)).
\]
The terms in $\ddd q_\ii$ and $\ddd\Hat q(t_\oo)$ can be neglected, since they already appear in the functional measure. Therefore, we only retain the first term, of which we eventually have to take the square root, getting 
the integral kernel half-density
\[
\begin{split}
\KK^f(q_\ii,Q_\oo) = \sqrt{\ddd q_\ii}\;
\Bigg(\int_{\substack{\Hat q(t_\ii)=0\\\Hat p(t_\oo)=0}}D\Hat p\,D\Hat q\;\EE^{\frac\II\hbar \left(-f(q_\ii +\Hat q(t_\oo),Q_\oo)
+ \int_{t_\ii}^{t_\oo} p^0 \ddd \Hat q 
+ \int_{t_\ii}^{t_\oo} \Hat p\ddd \Hat q
\right)}\\
\sqrt{\frac{\dd^2 f}{\dd q\,\dd Q}{(q_\ii +\Hat q(t_\oo),Q_\oo)}}
\Bigg)
\;\sqrt{\frac{\ddd Q_\oo}{2\pi\hbar}}
,
\end{split}
\]
where we have assumed the nonzero $\frac{\dd^2 f}{\dd q\,\dd Q}$ term to be actually positive.
Again, since this is a Gaussian integral with quadratic term $\int_{t_\ii}^{t_\oo} \Hat p\ddd \Hat q$ and no $\Hat p$-sources, we get
\begin{equation}\label{e:KKfhalf}
\KK^f(q_\ii,Q_\oo) = \sqrt{\ddd q_\ii}\;
\EE^{-\frac\II\hbar f(q_\ii,Q_\oo)}\sqrt{\frac{\dd^2 f}{\dd q\,\dd Q}{(q_\ii,Q_\oo)}}
\;\sqrt{\frac{\ddd Q_\oo}{2\pi\hbar}}
\end{equation}
instead of \eqref{e:HatKonlyf}.
In a neighborhood of the diagonal $q=q'$, $\KK^f$ is almost unitary: i.e.,
\[
\int_{\{Q_\oo\}}  \KK^f(q,Q_\oo)\,\overline\KK^f(Q_\oo,q') =
\sqrt{\ddd q}\,(1+O(\hbar))\delta(q'-q)\,\sqrt{\ddd q'},
\]
where the $O(\hbar)$ term is a differential operator. This may be seen by the 
($q$ and $q'$ dependent)
change of variable
\[
\lambda(Q_\oo) = \frac{f(q',Q_\oo)-f(q,Q_\oo)}{q'-q}
\]
in the $\ddd Q_\oo$ integral.\footnote{In some cases, $\KK^f$ is exactly unitary. This happens, e.g., when $f$ has product form: $f(q,Q)=\alpha(q)\beta(Q)$. (In particular, this happens when we have a linear symplectomorphism, in which case, moreover, the correction term $\frac{\dd^2 f}{\dd q\,\dd Q}$ is constant.) In this case, we have
\[
\int_{\{Q_\oo\}}  \KK^f(q,Q_\oo)\,\overline\KK^f(Q_\oo,q') =
\sqrt{\ddd q}\;
\left(\int_{\{Q_\oo\}} \EE^{\frac\II\hbar \beta(Q_\oo)(\alpha(q')-\alpha(q))} 
\sqrt{\alpha'(q)\alpha'(q')}\beta'(Q)
\frac{\ddd Q_\oo}{2\pi\hbar}
\right)
\;\sqrt{\ddd q'},
\]
where we have assumed the nonzero $\alpha'$ and $\beta'$ terms to be actually positive. By the change of variable
$\mu=\beta(Q_\oo)$, we get
\[
\begin{split}
\int_{\{Q_\oo\}}  \KK^f(q,Q_\oo)\,\overline\KK^f(Q_\oo,q') &= 
\sqrt{\ddd q}\;
\left(\int_{\{\mu\}}
\EE^{\frac\II\hbar \mu(\alpha(q')-\alpha(q))} 
\sqrt{\alpha'(q)\alpha'(q')}
\frac{\ddd\mu}{2\pi\hbar}
\right)
\;\sqrt{\ddd q'}\\ &=
\sqrt{\ddd q}\;
\delta(\alpha(q')-\alpha(q)) 
\sqrt{\alpha'(q)\alpha'(q')}
\;\sqrt{\ddd q'}=
\sqrt{\ddd q}\,\delta(q'-q)\,\sqrt{\ddd q'}.
\end{split}
\]}
Note that \eqref{e:KKfhalf} resembles the starting point of the WKB approximation (see, e.g.,
\cite[Appendix]{Resh}, \cite[Section 15]{Der20}, and references therein). Indeed, one possible way to get corrections to it to make it unitary is to regard, if possible, $f$ as the Hamilton--Jacobi action for some hamiltonian system (on the other hand, the corrections are not uniquely determined, since
the choice of this hamiltonian system is not unique).
\end{digression}

\section{Systems with one constraint}\label{s:oneconstr}
Our first generalization is an action of the form\footnote{If one requires $e\not=0$ everywhere,
this action may also be regarded as the first-order version of one-dimensional gravity, described by the coframe $e$, coupled to matter, described by $q$. Its BV-BFV quantization is studied in \cite{CS17}.}
\[
S[p,q,e] =\int_{t_\ii}^{t_\oo} (p_i\ddd q^i - e H(p,q)),
\]
where the new field $e$, the Lagrange multiplier, is a one-form on the interval $[t_\ii,t_\oo]$. 

We will first consider the case when
$e$ is different from zero everywhere. This restriction will allow us to deal with endpoint conditions on the $q$s.
Next we will consider the general case, which, however, will force us to pick different endpoint conditions.

\subsection{Nondegenerate Lagrange multiplier}\label{s:ngLagm}
 We assume $e$ to be different from zero everywhere.  A variation of the action yields
\[
\delta S = \EL + p_i(t_\oo)\delta q^i(t_\oo)-p_i(t_\ii)\delta q^i(t_\ii)
\]
with the EL equations 
\[
\ddd q^i = e\frac{\dd H}{\dd p_i},\quad \ddd p_i = -e\frac{\dd H}{\dd q^i},\quad H(p,q)=0.
\]
The last equation, the variation with respect to $e$, is imposed at all times. Note that the endpoint symplectic structures are the same as in the nondegenerate case and depend only on $p$ and $q$. This in particular means that no endpoint condition on $e$ has to be imposed.

To solve the EL equations, we first define
\[
\tau(t):=\int_{t_\ii}^t e.
\]
The first two sets of EL equations now read
\[
\frac{\ddd q^i}{\ddd\tau} = \frac{\dd H}{\dd p_i},\quad \frac{\ddd p_i}{\ddd\tau} = -\frac{\dd H}{\dd q^i},
\]
so they are the usual Hamilton equations for the hamiltonian $H$ in the time variable $\tau\in[0,T]$, with
$T:=\int_{t_\ii}^{t_\oo} e$. The last EL equation, $H=0$, is automatically satisfied at any point if we impose it at the initial point (or at the final point). We then have that the endpoint values for $(p,q)$ corresponding to solutions form the set
\begin{multline*}
L=\{(p^\ii,q_\ii,p^\oo,q_\oo)\in T^*M\times T^*M\ |\ H(p^\oo,q_\oo)=0\\
\text{ and } \exists T\in\RR_{\not=0}:(p^\oo,q_\oo)=\Phi^H_T(p^\ii,q_\ii)\},
\end{multline*}
where $\Phi^H_\tau$ denotes the hamiltonian flow of $H$ at time $\tau$. Note that in this case $L$ is not the graph
of a map $T^*M\to T^*M$, but just a subset of $T^*M\times T^*M$.
We will call it the evolution relation of the system. (See Remark~\ref{r:evorel} for a more general perspective.)

For a fixed $e$, let $(p^{q_\ii,q_\oo},q_{q_\ii,q_\oo})$ denote the (assumedly unique) solution to the first two sets of EL equations (i.e., we do \textbf{not} impose the constraint $H=0$) satisfying
$q_{q_\ii,q_\oo}(t_\ii)=q_\ii$ and $q_{q_\ii,q_\oo}(t_\oo)=q_\oo$.
Since we may rewrite the action as $\int_0^{T} (p_i\ddd q^i - H(p,q)\ddd\tau)$, we get that the HJ action is
\[
S_\text{HJ}(q_\ii,q_\oo)[e]:=S[p^{q_\ii,q_\oo},q_{q_\ii,q_\oo},e]=
S^H_\text{HJ}(q_\ii,q_\oo;0,T),
\]
where the last term denotes the HJ action for the system with hamiltonian $H$ over the time interval $[0,T]$. 

The first remark is that the HJ action depends on $e$ only via its integral $T$; we will then simply write
\begin{equation}\label{e:HatST}
\Hat S_\text{HJ}(q_\ii,q_\oo;T):=S_\text{HJ}(q_\ii,q_\oo)[e],\quad \text{with }T=\int_{t_\ii}^{t_\oo} e,
\end{equation}
so
\[
\Hat S_\text{HJ}(q_\ii,q_\oo;T)=S^H_\text{HJ}(q_\ii,q_\oo;0,T).
\]

The second remark, which follows immediately from \eqref{e:HJone} and \eqref{e:HJtime}, is that the evolution relation
 $L$
is determined by the equations
 \[
 p^\oo_i=\frac{\dd \Hat S_\text{HJ}}{\dd q^i_\oo},\quad p^\ii_i=-\frac{\dd \Hat S_\text{HJ}}{\dd q^i_\ii},\quad \frac{\dd \Hat S_\text{HJ}}{\dd T}=0.
 \]
In other words, $\Hat S_\text{HJ}$ is a generalized generating function\footnote{In the terminology of \cite[Section 4.3]{BW97}, the triple $(M\times \RR_{\not=0},\pi_1,\Hat S_\text{HJ})$, where $\pi_1$ is the projection to the first factor, is a Morse family generating $L$.}
 for the lagrangian submanifold $L$, see Appendix~\ref{a:generatingfunctions} (in particular, Remark~\ref{r:gflinear} for the linear case and Remark~\ref{r:gfnonlinear} for the general case).

\subsection{The general case} 
The assumption $e\not=0$ everywhere is only needed to 
ensure 
that the map $t\mapsto\tau$ is a diffeomorphism and hence to relate $\Hat S_\text{HJ}$ to the HJ function for the system with hamiltonian $H$.  
The existence of a HJ action, however, holds also in the general case where we make no assumptions on $e$, provided we choose the endpoint conditions more carefully.

First observe that the first two sets of EL equations, which we will call the evolution equations,
\[
\ddd q^i = e\frac{\dd H}{\dd p_i},\quad \ddd p_i = -e\frac{\dd H}{\dd q^i},
\]
can be uniquely solved, for fixed $e$, in terms of initial conditions $p^\ii,q_\ii$. In fact, in the regions (closed intervals) where $e=0$, the solution is constant and on the other regions we can make a change of variables $t\mapsto\tau$ as above. Moreover, the function $H$ is constant when evaluated on a solution to the evolution equations,
which means that it is enough to impose this constraint at the initial (or final) endpoint to have it satisfied at every time. 

Let $L$ still denote the evolution relation of the system, i.e., the
 set of endpoint values for solutions to the EL equations. 
 Note that, for a given $e=\underline e\,\ddd t$, we have to compute the ``flow'' $\Phi^{\Tilde H}_T$ of 
the time-dependent hamiltonian $\Tilde H=\underline e H$ up to time $T=\int_{t_\ii}^{t_\oo}e$. 
 We then have
\begin{multline*}
L=\{(p^\ii,q_\ii,p^\oo,q_\oo)\in T^*M\times T^*M\ |\ H(p^\oo,q_\oo)=0\\
\text{ and } \exists e:(p^\oo,q_\oo)=\Phi^{\Tilde H}_T(p^\ii,q_\ii)\}.
\end{multline*}

\begin{remark}\label{e:coisoone}
We can have a better geometric description if we assume that the differential of $H$ is different from zero at every point of the zero locus $C$ of $H$. This ensures that $C$ is a submanifold of $T^*M$ and that the hamiltonian vector field $X^H$ of $H$ restricted to $C$ has no zeros, so it defines a distribution of lines. Note that $C$ is an example of what is called a coisotropic submanifold and $\Span X^H|_C$ is called its characteristic distribution (see Definition~\ref{d:specialsub} and Remark~\ref{r:sympred} for more on this).
In this geometric framework, the evolution relation $L\subset T^*M\times T^*M$ consists of pairs of points on $C$ that can be connected by a path along the characteristic distribution. (It is therefore lagrangian by Lemma~\ref{ex:LinCCman}.)
\end{remark}

Now observe that fixing $q_\ii$ and $q_\oo$ is problematic, as for $e=0$ we have no solution unless $q_\ii=q_\oo$. For this reason we have to choose different endpoint conditions, as in Remark~\ref{r:changing}; for instance, we can
fix $p$ at one endpoint as in Example~\ref{exa:pfinal}. 

We denote by $(p^{q_\ii,Q_\oo},q_{q_\ii,Q_\oo},e)$ 
a solution $(p^{q_\ii,Q_\oo},q_{q_\ii,Q_\oo})$  to the evolution equations for the given $e$ (but we do \textbf{not} impose the constraint $H=0$) satisfying
$q_{q_\ii,Q_\oo}(t_\ii)=q_\ii$ and $Q(p^{q_\ii,Q_\oo}(t_\oo),q_{q_\ii,Q_\oo}(t_\oo))=Q_\oo$. 
We then set
\[
S^f_\text{HJ}(q_\ii,Q_\oo)[e]:=S^f[p^{q_\ii,Q_\oo},q_{q_\ii,Q_\oo},e]
\]
with
\[
S^f[p,q,e]=- f(q(t_\oo),Q(p(t_\oo),q(t_\oo))) + \int_{t_\ii}^{t_\oo} (p_i\ddd q^i - e H(p,q)).
\]
We then have
\begin{equation}\label{e:deltaSfeone}
\begin{split}
\delta S^f_\text{HJ} &= P_i(t_\oo)\delta Q^i_\oo
-p_i(t_\ii)\delta q^i_\ii - \int_{t_\ii}^{t_\oo} \delta e\, H(p^{q_\ii,Q_\oo},q_{q_\ii,Q_\oo}) \\
&= P_i(t_\oo)\delta Q^i_\oo
-p_i(t_\ii)\delta q^i_\ii - \left(\int_{t_\ii}^{t_\oo} \delta e\right) H(p^{q_\ii,Q_\oo}(t_\oo),q_{q_\ii,Q_\oo}(t_\oo)) \\
&= P_i(t_\oo)\delta Q^i_\oo
-p_i(t_\ii)\delta q^i_\ii -\delta T \, H(p^{q_\ii,Q_\oo}(t_\oo),q_{q_\ii,Q_\oo}(t_\oo)),
\end{split}
\end{equation}
where, as above, we set $T=\int_{t_\ii}^{t_\oo} e$. Note that $S^f_\text{HJ}$ is invariant under variations of $e$ with vanishing integral, so we may define $\Hat S^f_\text{HJ}$ like in \eqref{e:HatST}:
\[
\Hat S^f_\text{HJ}(q_\ii,Q_\oo;T):=S^f_\text{HJ}(q_\ii,Q_\oo)[e],\quad \text{with }T=\int_{t_\ii}^{t_\oo} e.
\]
 Then the above equation reads
\[
\delta \Hat S^f_\text{HJ} = P_i(t_\oo)\delta Q^i_\oo
-p_i(t_\ii)\delta q^i_\ii  -\delta T \, H(p^{q_\ii,Q_\oo}(t_\oo),q_{q_\ii,Q_\oo}(t_\oo)),
\]
and we see that the equations
\[
P_i(p^{q_\ii,Q_\oo},q_{q_\ii,Q_\oo})(t_\oo)=\frac{\dd \Hat S^f_\text{HJ}}{\dd Q^i_\oo},\quad p^\ii_i=-\frac{\dd \Hat S^f_\text{HJ}}{\dd q^i_\ii},\quad \frac{\dd \Hat S^f_\text{HJ}}{\dd T}=0
 \]
define $L$.

\begin{quote}
\emph{The HJ action, obtained inserting a solution of the evolution equations but ignoring the constraint, is a generalized generating function for the evolution relation determined by all EL equations.}
\end{quote}


\begin{example}[Linear case]\label{exa:linearone}
Consider $M=\RR^n$ and suppose $H$ is linear on $T^*\RR^n$:
\[
H(p,q)=p_i v^i +w_i q^i,
\]
where the $(v^i,w_i)$s are given constants. In this case the evolution equations simply read
\[
\ddd q^i = e v^i,\quad \ddd p_i = - e w_i.
\]
They can be easily solved in terms of the initial conditions, so we get the evolution relation
\begin{align*}
L=& \{(p^\ii,q_\ii,p^\oo,q_\oo)\in T^*\RR^n\times T^*\RR^n\ | \\
& | \ \exists T\in\RR: p^\oo_iv^i+w_i(q^i_\ii+Tv^i)=0,
\ p^\ii=p^\oo+Tw,\ q_\oo=q_\ii+Tv
\},
\end{align*}
where $T$ corresponds to $\int_{t_\ii}^{t_\oo}e$, and we have evaluated the constraint $H=0$ at the final endpoint.
As in Example~\ref{exa:pfinal} we now fix $q_\ii$ and $p^\oo$. The evolution equations can then be easily integrated for the given endpoint conditions as
\[
\begin{split}
q_{q_\ii,p^\oo}^i(t)&=q_\ii^i + \int_{t_\ii}^t e v^i=q_\ii^i+\tau(t) v^i,\\
p^{q_\ii,p^\oo}_i(t)&=p^\oo_i + \int_t^{t_\oo}e w_i=p^\oo_i +(T-\tau(t)) w_i,
\end{split}
\]
with $\tau(t)=\int_{t_\ii}^te$. For this choice of endpoint conditions we consider
\[
S^f[p,q,e] :=- p_i(t_\oo)q^i(t_\oo) + \int_{t_\ii}^{t_\oo} (p_i\ddd q^i - e H(p,q)).
\]
Now observe that
\begin{align*}
p^{q_\ii,p^\oo}_i\ddd q_{q_\ii,p^\oo}^i(t)-eH(p^{q_\ii,p^\oo},q_{q_\ii,p^\oo})&=
p^{q_\ii,p^\oo}_ie v^i-ep^{q_\ii,p^\oo}_i v^i-ew_iq_{q_\ii,p^\oo}^i \\
&=
-\ddd\tau w_i(q_\ii^i+\tau v^i),
\end{align*}
which implies
\[
\int_{t_\ii}^{t_\oo}(p^{q_\ii,p^\oo}_i\ddd q_{q_\ii,p^\oo}^i(t)-eH(p^{q_\ii,p^\oo},q_{q_\ii,p^\oo}))=
-Tw_iq_\ii^i-\frac{T^2}2w_iv^i.
\]
On the other hand,
\[
p^{q_\ii,p^\oo}_i(t_\oo)q_{q_\ii,p^\oo}^i(t_\oo)=p^\oo_i(q_\ii^i+T v^i).
\]
In conclusion, we get
\begin{equation}\label{e:Sfablinone}
\Hat S^f_\text{HJ}(q_\ii,p^\oo;T)=-p^\oo_iq_\ii^i
-T(p^\oo_iv^i + w_iq_\ii^i)-
\frac{T^2}2w_iv^i,
\end{equation}
which is readily seen to be a generalized generating function for $L$.
\end{example}

\begin{remark}\label{r:solveT}
Sometimes it is possible to solve for $T$ the constraints and then obtain a standard (i.e., without parameters) generating function for the evolution relation $L$. For instance, in the linear case of Example~\ref{exa:linearone},
the constraint is
\[
p^\oo_i v^i + w_iq_\ii^i+
Tw_iv^i=0,
\]
which can be solved if $w_i v^i\not=0$. In this case we get
\[
T=-\frac{p^\oo_i v^i + w_i q_\ii^i}{w_iv^i},
\]
which can be inserted back into $\Hat S^f$ to get the standard generating function
\[
\Tilde S^f(q_\ii,p^\oo)=-p^\oo_iq_\ii^i+\frac{(p^\oo_i v^i + w_i q_\ii^i)^2}{2w_i v^i}.
\]
In the one-dimensional case, $n=1$, we get
\[
\Tilde S^f(q_\ii,p^\oo)=-p^\oo q_\ii+\frac{(p^\oo v + w q_\ii)^2}{2wv}=
\frac12\frac vw (p^\oo)^2 + \frac12\frac wv (q_\ii)^2,
\]
which correctly generates the evolution relation
\[
L = \{(p^\ii,q_\ii,p^\oo,q_\oo)\in T^*\RR\times T^*\RR\ |\ p^\ii v+w q_\ii=0\text{ and } p^\oo v+w q_\oo=0\}.
\]
Note that in this very particular example the evolution plays no role and  $L$ has the product form
$L=L'\times L'$, with $L'=\{(p,q)\in T^*\RR  \ |\ p v+w q=0\}$.
\end{remark}

{
\begin{remark}[Integrable systems I]\label{rem:intsysI}
The case $n=1$
can also be regarded as a one-dimensional integrable system with vanishing hamiltonian. In this case, another expression for the HJ action can be derived as follows. Around a point $(p_0,q_0)$ where $\ddd H \neq 0$, by the Carath\'eodory--Lie--Jacobi theorem\footnote{\label{f:CLJ}The Carath\'eodory--Lie--Jacobi theorem says that if $H_1,\ldots,H_k$ are functions on a symplectic manifold $(M,\omega)$ such that $\{H_i,H_j\}=0$ and $\ddd H_1, \ldots, \ddd H_k$ are linearly independent at some $p\in M$, then there exists a Darboux chart 
$(q_i,p^i)$ containing $p$ with $q_1 = H_1,\ldots,q_k=H_k$. See, for instance, \cite[Theorem 13.4.1]{LM}.} we may extend $H$ to a Darboux chart $(P,Q)$ with $P = H$ ($P$ is what is called the action variable and $Q$ the corresponding angle variable).  Let $f(q,Q)$ be the generating function satisfying $p\,\ddd q = P\,\ddd Q + \ddd f$. We may then write 
\[
S^f[p,q,e] := -f(q(t_\oo),Q(t_\oo)) + \int_{t_\ii}^{t_\oo} (p\,\ddd q-eH)=
- f(q(t_\ii),Q(t_\ii)) + \int_{t_\ii}^{t_\oo}(P\,\ddd Q - e P), 
\]
where $P = P(p,q), Q=Q(p,q)$. The evolution equations now simply read $\ddd P=0, \ddd Q=e$. Evaluating on a solution of the evolution equations, the action vanishes and we get   
\begin{equation}
S^f_{\text{HJ}}(q_\ii,Q_\oo;T) = - f(q_\ii,Q_\oo - T) \label{eq:HJintsysI_1}
\end{equation}
for $(q_\ii,Q_\oo)$ in a neighborhood of $(q_0,Q(p_0,q_0))$. Alternatively, if $g(q,P)$ is the generating function satisfying $p\,\ddd q = {-}Q\,\ddd P + \ddd g$, we can express the HJ action with final endpoint condition on $P$ by  
\begin{equation}
S^g_{\text{HJ}}(q_\ii,P^\oo;T) = -T P^{{\oo}} - g(q_\ii,P^{{\oo}}). \label{eq:HJintsysI_2}
\end{equation}
Notice also that \eqref{eq:HJintsysI_2} is the Legendre transform of \eqref{eq:HJintsysI_1}. 
 We will revisit the case of integrable systems in Remarks \ref{rem:intsysII} and \ref{rem:intsysIII}.
\end{remark}
}

\section{Systems with several constraints}\label{s:syssevcon}
{}From now on we consider an action of the form\footnote{Such an action functional is the classical part of an AKSZ theory \cite{AKSZ} on the interval $[t_\ii,t_\oo]$; see Section~\ref{s:BFVAKSZBBV} for a review. See also \cite{BDMGV} for the study of such a theory in the Dirac formalism.}

\begin{equation}\label{e:Ssevconstr}
S[p,q,e] =\int_{t_\ii}^{t_\oo} (p_i\ddd q^i - e^\alpha H_\alpha(p,q)),
\end{equation}
where the $e^\alpha$s (the Lagrange multipliers) are one-forms on the interval $[t_\ii,t_\oo]$ and the $H_\alpha$s are given functions (the constraints). (The sum over $\alpha$, like that over $i$, is now understood.)
The EL equations are
\[
\ddd q^i = e^\alpha\frac{\dd H_\alpha}{\dd p_i},\quad \ddd p_i = -e^\alpha\frac{\dd H_\alpha}{\dd q^i},\quad H_\alpha(p,q)=0.
\]
Again we will call the first two sets of equations the evolution equations and the last set the constraints.
Evaluating the constraints on a solution of the evolution equations yields
\[
\ddd H_\alpha = \frac{\dd H_\alpha}{\dd p_i}\ddd p_i + \frac{\dd H_\alpha}{\dd q^i}\ddd q^i=
-\{H_\alpha,H_\beta\}e^\beta,
\]
where $\{\ ,\ \}$ denotes the Poisson bracket.\footnote{We define the Poisson bracket as
$\{f,g\}=\frac{\dd f}{\dd p_i}\frac{\dd g}{\dd q^i}-\frac{\dd f}{\dd q^i}\frac{\dd g}{\dd p_i}$.}
If we assume, as we will do from now on, that the constraints are in involution,\footnote{Using Dirac's terminology, one also says that they are first-class constraints.}  i.e., that there are structure functions
$f_{\alpha\beta}^\gamma$ such that
\begin{equation}\label{e:involutivity}
\{H_\alpha,H_\beta\}=f_{\alpha\beta}^\gamma H_\gamma,
\end{equation}
then the previous equations become
\begin{equation}\label{e:dddHalpha}
\ddd H_\alpha =- f_{\alpha\beta}^\gamma e^\beta H_\gamma.
\end{equation}
This implies that, if the constraints are imposed at the initial (or final) endpoint, they will be satisfied at every time.
\begin{remark}\label{r:coisomany}
As in Remark~\ref{e:coisoone}, it may be convenient to assume that the differentials of the $H_\alpha$s
are linearly independent at every point of the common zero locus $C$ of the $H_\alpha$s. This ensures that $C$ is a submanifold of $T^*M$ and that the hamiltonian vector fields $X^{H_\alpha}$ restricted to $C$ are linearly independent, so they define a regular distribution. It follows that $C$ is also an example of what is called a coisotropic submanifold and $\Span\{X^H_\alpha|_C\}$ is called its characteristic distribution (see Definition~\ref{d:specialsub} and Remark~\ref{r:sympred} for more on this).\footnote{The characteristic distribution is also the same as the kernel of the restriction of the symplectic form to $C$.} 
As a consequence of \eqref{e:involutivity}, the characteristic distribution is involutive (and hence, by Frobenius' theorem, integrable).
In this geometric framework, the evolution relation $L\subset T^*M\times T^*M$ consists of pairs of points on $C$ that can be connected by a path along the characteristic distribution and is therefore lagrangian in $\overline{T^*M}\times T^*M$, where the bar denotes that one takes the opposite symplectic form.
\end{remark}

Like in the case with a single constraint, fixing $q_\ii$ and $q_\oo$ is problematic, as for $e^\alpha=0$ for all $\alpha$ we have no solution unless $q_\ii=q_\oo$. Therefore, we choose different endpoint conditions, as in Remark~\ref{r:changing}; for instance, we can
fix $p$ at one endpoint as in Example~\ref{exa:pfinal}. 

We will denote by $(p^{q_\ii,Q_\oo},q_{q_\ii,Q_\oo},e)$ 
a solution $(p^{q_\ii,Q_\oo},q_{q_\ii,Q_\oo})$  to the evolution equations with the given $e$ (but we do \textbf{not} impose the constraints $H_\alpha=0$) satisfying
$q_{q_\ii,Q_\oo}(t_\ii)=q_\ii$ and $Q(p^{q_\ii,Q_\oo}(t_\oo),q_{q_\ii,Q_\oo}(t_\oo))=Q_\oo$. 
We then set
\[
S^f_\text{HJ}(q_\ii,Q_\oo)[e]:=S^f[p^{q_\ii,Q_\oo},q_{q_\ii,Q_\oo},e]
\]
with
\[
S^f[p,q,e]=- f(q(t_\oo),Q(p(t_\oo),q(t_\oo))) + \int_{t_\ii}^{t_\oo} (p_i\ddd q^i - e^\alpha H_\alpha(p,q)).
\]
By taking a variation, we get
\[
\delta S^f_\text{HJ} = P_i(t_\oo)\delta Q^i_\oo
-p_i(t_\ii)\delta q^i_\ii - \int_{t_\ii}^{t_\oo} \delta e^\alpha\, H_\alpha(p^{q_\ii,Q_\oo},q_{q_\ii,Q_\oo}).
\]
This shows that $S^f_\text{HJ} $ is indeed a generalized generating function for the evolution relation $L$. However, there is a lot of redundancy because the variations with respect to the $e^\alpha(t)$s yield the constraints for every $t\in[t_\ii,t_\oo]$, whereas we have observed that it is enough to impose them at the initial (or final) endpoint. This is related to the fact that, thank to \eqref{e:dddHalpha} and upon integration by parts, $S^f_\text{HJ}$ is invariant under a variation of $e$ of the form
\begin{equation}\label{e:gaugetransf}
\delta e^\alpha = \ddd \gamma^\alpha + f_{\alpha\beta}^\gamma e^\alpha\gamma^\beta,
\end{equation}
with the $\gamma$s arbitrary functions on $[t_\ii,t_\oo]$ that vanish at both endpoints. 
We call a transformation as in \eqref{e:gaugetransf} a gauge transformation.
We therefore see that $S^f_\text{HJ}$ only depends on the $e^\alpha(t)$s via their gauge classes. Analogously to the case of equation \eqref{e:HatST}, we define
\begin{equation}\label{e:HatSTa}
\Hat S^f_\text{HJ}(q_\ii,Q_\oo;[e]):=S^f_\text{HJ}(q_\ii,Q_\oo,e),\quad \text{with }e\in[e].
\end{equation}
In summary, we have the following
\begin{theorem}\label{t-thm1}
The generalized HJ action $\Hat S^f_\text{HJ}$, obtained by inserting into the classical action $S^f$ the solution of the evolution equations, is a generalized generating functions for the evolution relation, which is defined in terms of all the EL equations (evolution and constrains).
\end{theorem}

In the next subsections we will discuss more explicitly the gauge classes and how the HJ action depends on them.
We will proceed in increasing order of difficulty. In Section~\ref{s:fstrictly}, we will consider the strictly involutive case where all the structure functions $f_{\alpha\beta}^\gamma$ vanish. This case, which is a straightforward generalization of the single-constraint case, 
comprises the case where the constraints $H_\alpha$ are linear, see Example~\ref{e:flinear}, and will be important for abelian Chern--Simons theory.
Next, in Section~\ref{s:fLiealgebra}, we will consider the case when the structure functions are actually constant, so that there is an underlying Lie algebra structure (the $H_\alpha$s are in this case the components of an equivariant momentum map). This case will be important for nonabelian Chern--Simons theory. Finally, for completeness, we will treat the general case in Section~\ref{s:fgeneralcase}: this case requires knowledge of Lie algebroids and Lie groupoids, see \cite{Mac} for a general introduction, but will not be used in the rest of the paper, so it may be safely skipped.

\subsection{The strictly involutive case}\label{s:fstrictly}
Suppose that all the structure functions $f_{\alpha\beta}^\gamma$ vanish; i.e., $\{H_\alpha,H_\beta\}=0$ for all $\alpha$ and $\beta$. This implies that, evaluating on a solution to the evolution equations, we have
$\ddd H_\alpha = 0$ for all $\alpha$.
If we define
\[
\tau^\alpha(t):=\int_{t_\ii}^t e^\alpha,
\]
we then have $e^\alpha=\ddd\tau^\alpha$ and $\tau^\alpha(t_\ii)=0$. Observe that the computation
in \eqref{e:deltaSfeone} can be repeated verbatim yielding
\[
\delta S^f_\text{HJ} = P_i(t_\oo)\delta Q^i_\oo
-p_i(t_\ii)\delta q^i_\ii -\delta T^\alpha \, H_\alpha(p^{q_\ii,Q_\oo}(t_\oo),q_{q_\ii,Q_\oo}(t_\oo)),
\]
with $T^\alpha=\tau^\alpha(t_\oo)=\int_{t_\ii}^{t_\oo} e^\alpha$. Again, this first implies that $S^f_\text{HJ}$ is invariant under variations of the $e^\alpha$s with vanishing integral, so we may define $\Hat S^f_\text{HJ}$ as in \eqref{e:HatSTa} as
\[
\Hat S^f_\text{HJ}(q_\ii,Q_\oo;T):=S^f_\text{HJ}(q_\ii,Q_\oo)[e],\quad \text{with }T^\alpha=\int_{t_\ii}^{t_\oo} e^\alpha.
\]
Second, the above computation shows that $\Hat S^f_\text{HJ}$ is a generalized generating function for $L$.
\begin{example}[Linear case]\label{e:flinear}
Assume
\[
H_\alpha=p_i v^i_\alpha +w_{i\alpha} q^i=pv_\alpha+w_\alpha q,
\]
where the $(v^i_\alpha,w_{i\alpha})$s are given constants and we use the product of row and column vectors to make the notation a bit lighter. 
The Poisson bracket of any two constraints is a constant function. The involutivity conditions \eqref{e:involutivity} then imply that all these constant functions vanish, so we actually have
\[
\{H_\alpha,H_\beta\}=0.
\]
This is equivalent to the conditions 
$w_\alpha v_\beta=w_\beta v_\alpha$
for all $\alpha$ and $\beta$.
For future convenience we define the (symmetric!) matrix
\begin{equation}\label{e:Aalphabeta}
A_{\alpha\beta}:=w_\alpha v_\beta=w_\beta v_\alpha.
\end{equation}
We now want to compute $\Hat S^f_\text{HJ}$  explicitly as a function of $q_\ii$ and $p^\oo$. First observe that the evolution 
 equations read
\[
\ddd q = e^\alpha v_\alpha,\quad \ddd p = - e^\alpha w_{\alpha},
\]
and can be easily integrated for the given endpoint conditions:
\[
\begin{split}
q_{q_\ii,p^\oo}(t)&=q_\ii +\tau^\alpha(t) v_\alpha,\\
p^{q_\ii,p^\oo}(t)&=p^\oo +(T^\alpha-\tau^\alpha(t)) w_{\alpha}.
\end{split}
\]
Moreover, we have
\[
\begin{split}
p^{q_\ii,p^\oo}\ddd q_{q_\ii,p^\oo}(t)-e^\alpha H_\alpha(p^{q_\ii,p^\oo},q_{q_\ii,p^\oo})&=
e^\alpha p^{q_\ii,p^\oo} v_\alpha-e^\alpha p^{q_\ii,p^\oo} v_\alpha-e^\alpha w_{\alpha}q_{q_\ii,p^\oo}\\
&=
-\ddd\tau^\alpha w_{\alpha}(q_\ii+\tau^\beta v_\beta)
=-\ddd\tau^\alpha w_{\alpha}q_\ii-\ddd\tau^\alpha\tau^\beta A_{\alpha\beta}\\
&=-\ddd\tau^\alpha w_{\alpha}q_\ii-\frac12 \ddd(\tau^\alpha\tau^\beta)A_{\alpha\beta}
\end{split}
\]
with $A_{\alpha\beta}$ the symmetric matrix defined in \eqref{e:Aalphabeta}. This implies
\[
\int_{t_\ii}^{t_\oo}(p^{q_\ii,p^\oo}\ddd q_{q_\ii,p^\oo}(t)-e^\alpha H_\alpha(p^{q_\ii,p^\oo},q_{q_\ii,p^\oo}))=
-T^\alpha w_{\alpha}q_\ii-\frac12 T^\alpha T^\beta A_{\alpha\beta}.
\]
On the other hand,
\[
p^{q_\ii,p^\oo}(t_\oo)q_{q_\ii,p^\oo}(t_\oo)=p^\oo(q_\ii+T^\alpha v_\alpha).
\]
In conclusion, we get
\begin{equation}\label{e:Sfablin}
\Hat S^f_\text{HJ}(q_\ii,p^\oo,T)=-p^\oo q_\ii
-T^\alpha(p^\oo v_\alpha + w_{\alpha}q_\ii)-
\frac12 T^\alpha T^\beta A_{\alpha\beta},
\end{equation}
which generalizes \eqref{e:Sfablinone} to the the case of several linear constraints.
\end{example}

\begin{remark}[General change of endpoint conditions]\label{r:HJgenchange}
The above example is simple enough for the answer to be computable for a general change of endpoint conditions as in Remark~\ref{r:changing}. Let $f(q,Q)$ be a generating function for the change and fix $q_\ii$ and $Q_\oo$. By inspection in the previous computation, we see that to compute $\Hat S^f_\text{HJ}$ we actually only have to solve $q(t)$ for the given initial condition $q_\ii$ (this is independent of $Q_\oo$), which has already been done above. We then get in general
\begin{equation}\label{e:Sfablingen}
\Hat S^f_\text{HJ}(q_\ii,Q_\oo;T)= - f(q_\ii+T^\alpha v_\alpha,Q_\oo)
-T^\alpha w_{\alpha}q_\ii-
\frac12 T^\alpha T^\beta A_{\alpha\beta}.
\end{equation}
\end{remark}
Note that this may also be obtained as the composition of the generating functions \eqref{e:Sfablin}
and \eqref{e:HJonlyf}, following Remark~\ref{r:compgenfun}. Namely, let us write
\[
S_1(q_\ii,p^\oo;T)=-p^\oo q_\ii
-T^\alpha(p^\oo v_\alpha + w_{\alpha}q_\ii)-
\frac12 T^\alpha T^\beta A_{\alpha\beta},
\]
and
\[
S_2(q_\ii,Q_\oo) = -f(q_\ii,Q_\oo).
\]
Let
\[
S_3(q_\ii,p^\oo,\bar q_\ii,Q_\oo;T):=S_1(q_\ii,p^\oo,T)+S_2(\bar q_\ii,Q_\oo)+p^\oo\bar q_\ii.
\]
The composition of $S_1$ and $S_2$ is the evaluation of $S_3$ at its critical point in $(p^\oo,\bar q_\ii)$:
\[
\bar q_\ii +\frac{\dd S_1}{\dd p^\oo}=0,\quad p^\oo +\frac{\dd S_2}{\dd\bar q_\ii}=0.
\]
One easily sees that the result is \eqref{e:Sfablingen}.

{
\begin{remark}[Integrable systems II]\label{rem:intsysII}
As a generalization of Remark \ref{rem:intsysI}, we may consider the case where we have exactly $n$ constraints $H_1,\ldots,H_n$ in strict involution - an $n$-dimensional integrable system with vanishing hamiltonians. Around a point $(p_0,q_0)$ where $\ddd H_1, \ldots, \ddd H_n$ are linearly independent, again by the Carath\'eodory--Jacobi--Lie theorem, see footnote~\ref{f:CLJ},
we may find a Darboux chart $(P,Q)$ with $P_i = H_i$ (the action-angle coordinates). By the same computation as in Remark \ref{rem:intsysI}, in this chart the HJ action is expressed as 
\begin{equation}
S^f_\text{HJ}(q_\ii,Q_\oo;T) = - f(q_\ii,Q_\oo - T)
\end{equation}
and 
\begin{equation}
S^g_\text{HJ}(q_\ii,P^\oo;T) = -T^iP_i^{\oo} - {g}(q_\ii,P^\oo).
\end{equation}
\end{remark}
}

\subsection{The Lie algebra case}\label{s:fLiealgebra}
We now suppose that all the $f_{\alpha\beta}^\gamma$s are constant (but not necessarily zero). They can be viewed as structure constants of some Lie algebra $\frg$. In turn the $e^\alpha$s may be viewed as the components of a $\frg$ valued one-form $e$ on $[t_\ii,t_\oo]$ or, equivalently, as a connection one-form. We may also view the $H_\alpha$s as the components of a map $H\colon T^*M\to\frg^*$, where $T^*M$ is the target symplectic manifold.
The involutivity conditions \eqref{e:involutivity} say that $H$ is an equivariant momentum map, whereas
 \eqref{e:dddHalpha} may be read as the condition $\ddd_e H=0$ that
the composition of $H$ with a solution to the evolution equations is covariantly constant. We rewrite the action as
\[
S[p,q,e] =\int_{t_\ii}^{t_\oo} (p_i\ddd q^i - \langle H(p,q),e\rangle),
\]
where $\langle\ ,\ \rangle$ denotes the pairing of $\frg^*$ with $\frg$.
Next we define
\begin{equation}\label{e:defhofe}
h(t)=P\EE^{\int_{t_\ii}^t e} \in G,
\end{equation}
where $P$ denotes the path ordered product and $G$ is the simply connected Lie group with Lie algebra $\frg$.
Note that $h(t_\ii)=1$ and $e=h^{-1}\ddd h$.\footnote{\label{f:lefttrans}We use for simplicity the matrix notation. If $G$ is not a matrix Lie group, the correct formula is 
\[
e(t)=(\ddd_1l_{h})^{-1}\ddd_th\colon T_tI\to\frg,
\]
where $1$ denotes the unit of $G$ and $l_h\colon G\to G$ denotes the left translation $g\mapsto hg$.}
We write
\[
g:=h(t_\oo)=P\EE^{\int_{t_\ii}^{t_\oo} e}.
\]
Finally, we define $\Tilde H:=\Ad_h^*H$. Observe that, when evaluated on a solution to the evolution equations, $\Tilde H$ is constant. We now repeat the computation as in \eqref{e:deltaSfeone}:
\[
\begin{split}
\delta S^f_\text{HJ} &= P_i(t_\oo)\delta Q^i_\oo
-p_i(t_\ii)\delta q^i_\ii- \int_{t_\ii}^{t_\oo} \langle H(p^{q_\ii,Q_\oo},q_{q_\ii,Q_\oo}),\delta e\rangle
\\
&= P_i(t_\oo)\delta Q^i_\oo
-p_i(t_\ii)\delta q^i_\ii- \int_{t_\ii}^{t_\oo} \langle \Tilde H(p^{q_\ii,Q_\oo},q_{q_\ii,Q_\oo}),\Ad_{h}\delta e\rangle
\\
&= P_i(t_\oo)\delta Q^i_\oo
-p_i(t_\ii)\delta q^i_\ii-  \left\langle \Tilde H(p^{q_\ii,Q_\oo}(t_\oo),q_{q_\ii,Q_\oo}(t_\oo)),\int_{t_\ii}^{t_\oo}\ddd(\delta h h^{-1})\right\rangle
\\
&= P_i(t_\oo)\delta Q^i_\oo
-p_i(t_\ii)\delta q^i_\ii-  \left\langle \Tilde H(p^{q_\ii,Q_\oo}(t_\oo),q_{q_\ii,Q_\oo}(t_\oo)),\delta g g^{-1}\right\rangle
\\
&= P_i(t_\oo)\delta Q^i_\oo
-p_i(t_\ii)\delta q^i_\ii-  \left\langle H(p^{q_\ii,Q_\oo}(t_\oo),q_{q_\ii,Q_\oo}(t_\oo)),g^{-1}\delta g\right\rangle.
\end{split}
\]
This explicitly shows that $S^f_\text{HJ}$ is invariant under variations of $e$ that do not change $g$. As in \eqref{e:HatSTa}, we define
\begin{equation}\label{e:SHJLie}
\Hat S^f_\text{HJ}(q_\ii,Q_\oo;g):=S^f_\text{HJ}(q_\ii,Q_\oo)[e],\quad \text{with }g=P\EE^{\int_{t_\ii}^{t_\oo} e},
\end{equation}
and it immediately follows that $\Hat S^f_\text{HJ}$ is a generalized generating function for the evolution relation $L$.

\begin{example}[Biaffine case]\label{exa:biaffine}
Consider a target $T^*V=V^*\oplus V\ni(p,q)$ and constraints that are affine both on $V$ and on $V^*$; namely,
\[
H_\alpha(p,q)=-p_i(\rho_\alpha)^i_j q^j+
p_i v^i_\alpha +w_{i\alpha} q^i=
- p\rho_\alpha q+
pv_\alpha+w_\alpha q,
\]
where, for each $\alpha$, $\rho_\alpha$ is a given endomorphism of $V$, $v_\alpha$ is a given vector in $V$ and $w_\alpha$ is a given vector in $V^*$ (we denote the pairing between $V^*$ and $V$ simply by juxtaposition).
The involutivity conditions \eqref{e:involutivity} are equivalent to the relations
\begin{equation}\label{e:involrhovw}
\begin{split}
[\rho_\alpha,\rho_\beta] &= f_{\alpha\beta}^\gamma\rho_\gamma,\\
\rho_\alpha v_\beta - \rho_\beta v_\alpha &= f_{\alpha\beta}^\gamma v_\gamma,\\
w_\alpha\rho_\beta - w_\beta\rho_\alpha &= f_{\alpha\beta}^\gamma w_\gamma,\\
w_\beta v_\alpha - w_\alpha v_\beta &= 0.
\end{split}
\end{equation}
The first relation says that $\rho$ is a representation of the Lie algebra $\frg$ on $V$, the second and the third that $v\colon \frg\to V$ and $w\colon \frg\to V^*$ are Lie algebra $1$-cocycles,   
 and the last is what we already saw in Example~\ref{e:flinear}. The 
evolution 
 equations read
\[
\ddd_e q = e^\alpha v_\alpha=v(e),\quad \ddd_e p = - e^\alpha w_{\alpha}=-w(e),
\]
where we introduced the covariant derivatives
\[
\ddd_e q := \ddd q + e^\alpha\rho_\alpha q,\quad \ddd_e p := \ddd p - e^\alpha p\rho_\alpha. 
\]
Let $R\colon G\to\Aut(V)$ denote the representation that integrates $\rho$ to the simply connected Lie group $G$.
If we introduce $\Tilde q(t):=R_{h(t)}q(t)$ with $h$ defined in \eqref{e:defhofe}, then the first evolution relation becomes
\[
\ddd\Tilde q = R_hv(e),
\]
which has solution, for the the inital condition $\Tilde q(t_\ii)=q(t_\ii)=q_\ii$, given by $\Tilde q(t)=q_\ii+\int_{t_\ii}^t R_hv(e)$. Therefore,
\[
q_{q_\ii,p^\oo}(t)=R_{h(t)}^{-1}\left(q_\ii+\int_{t_\ii}^t R_hv(e)\right).
\]
Similarly, we get
\[
p^{q_\ii,p^\oo}(t)=\bar R_{h(t)}^{-1}\left(p^\oo+\int_{t}^{t_\oo} \bar R_hw(e)\right),
\]
where $\bar R_h:=(R_h^{-1})^*$ is the dual representation on $V^*$. For later convenience we introduce the $1$-forms on $G$ (valued in $V$ and $V^*$, respectively)
\[
\alpha(h):=R_hv(e),\quad \beta(h):=\bar R_hw(e).
\]
As a consequence of the second and third relations of \eqref{e:involrhovw} (and of the fact that the Maurer--Cartan $1$-form $e$ is flat), we get that $\alpha$ and $\beta$ are closed. This implies that $\int_\gamma\alpha$ and $\int_\gamma\beta$ are invariant under homotopies of the path $\gamma\colon[t_\ii,t_\oo]\to G$ with fixed endpoints.
In particular, if $\gamma$ is a path joining the identity in $G$ to some element $g$, then we define
\[
\Phi(g):=\int_\gamma\alpha,\quad \Psi(g):=\int_\gamma\beta.
\]
Note that, since $G$ is simply connected, we get well-defined maps $\Phi\colon G\to V$ and $\Psi\colon G\to V^*$.\footnote{One can easily show that $\Phi(g_1g_2)=\Phi(g_1)+R_{g_1}\Phi(g_2)$ and $\Psi(g_1g_2)=\Psi(g_1)+\bar R_{g_1}\Psi(g_2)$; i.e., $\Phi$ and $\Psi$ are $1$-cocycles on $G$ with values in $V$ and $V^*$, respectively.
Our formulae give an explicit integration of the Lie algebra $1$-cocycles to the corresponding Lie group $1$-cocycles, which is guaranteed by general arguments, see \cite[Lemma 2.13]{Lu}.}
We now want to compute the HJ action. First observe that
\[
\begin{split}
p^{q_\ii,p^\oo}\ddd q_{q_\ii,p^\oo}-e^\alpha H_\alpha(p^{q_\ii,p^\oo},q_{q_\ii,p^\oo})&= 
p^{q_\ii,p^\oo}\ddd_e q_{q_\ii,p^\oo} - p^{q_\ii,p^\oo}v(e)-w(e)q_{q_\ii,p^\oo}\\
&= -w(e)q_{q_\ii,p^\oo}=-w(e) R_{h}^{-1}\left(q_\ii+\int_{t_\ii}^\bullet R_hv(e)\right),
\end{split}
\]
so
\[
\int_{t_\ii}^{t_\oo}(p^{q_\ii,p^\oo}\ddd q_{q_\ii,p^\oo}-e^\alpha H_\alpha(p^{q_\ii,p^\oo},q_{q_\ii,p^\oo}))=
-\Psi(g)q_\ii-\text{WZW}(g)
\]
with $g=h(t_\oo)$. We introduced the following notation (which will become the WZW term in the case of Chern--Simons theory):
\begin{equation}\label{e:WZWalphabeta}
\text{WZW}(g):=\int_{t_1<t_2} \beta_2\,\alpha_1
\end{equation}
with $\alpha_i:=\pi_i^*\alpha$, $\beta_i:=\pi_i^*\beta$ and
$\pi_i\colon [t_\ii,t_\oo]^2\to G$,
$\pi_i(t_1,t_2)=\gamma(t_i)$, where $\gamma$ is a(ny) path joining the identity to $g$ (e.g., the path $h(t)$). The fact that 
{$\text{WZW}(g)$}
is invariant under homotopies of $\gamma$ with fixed endpoints follows from the fact that $\alpha$ and $\beta$ are closed and from the fourth relation in \eqref{e:involrhovw}. Finally, we have
\[
p^{q_\ii,p^\oo}(t_\oo)q_{q_\ii,p^\oo}(t_\oo)=p^\oo R_{g}^{-1}\left(q_\ii+\Phi(g)\right),
\]
so
\begin{equation}\label{e:Sbiaffine}
\Hat S^f_\text{HJ}(q_\ii,p^\oo;g)=-p^\oo R_{g}^{-1}q_\ii
-p^\oo R_{g}^{-1}\Phi(g)
-\Psi(g)q_\ii-\text{WZW}(g),
\end{equation}
which generalizes \eqref{e:Sfablin} to the case of biaffine constraints.
\end{example}

\begin{remark}[Exponential map]\label{r:exp}
We now want to rewrite the WZW term \eqref{e:WZWalphabeta} of the last example to facilitate the comparison with the case of Chern--Simons theory under the assumption that the exponential map $\frg\to G$ is surjective. To simplify the notation, we also assume $t_\ii=0$ and $t_\oo=1$. In this case, given $g=\EE^\xi$, we may choose
$h(t)=\EE^{t\xi}$. It follows that $e=\xi\ddd t$. 
We introduce
\begin{equation}\label{e:W2W3}
\begin{split}
W_2(g)&:=\Psi(g)\Phi(g)=\int_{t_1<t_2} \beta_2\,\alpha_1+\int_{t_2<t_1} \beta_2\,\alpha_1,\\
W_3(g)&:=\int_0^1 \Psi(h)\rho_e\Phi(h)=
\int_{t_1<t_2,t_3<t_2}\beta_3\rho_{e_2}\alpha_1.
\end{split}
\end{equation}
We have
\[
\begin{split}
W_3(g)&=\int_{t_1<t_2,t_3<t_2}w(e_3)R_{h_3}^{-1}\rho_{e_2}\alpha_1=
\int_{t_1<t_2,t_3<t_2} w(\xi) R_{\EE^{t_3\xi}}^{-1}\rho(\xi)\ddd t_3\ddd t_2\alpha_1\\&=
\int_{t_1<t_2,t_3<t_2} w(\xi) \rho(\xi)R_{\EE^{t_3\xi}}^{-1}\ddd t_3\ddd t_2\alpha_1=
-\int_{t_1<t_2,t_3<t_2} w(\xi) \frac\ddd{\ddd t_3}R_{\EE^{t_3\xi}}^{-1}\ddd t_3\ddd t_2\alpha_1\\&=
-\int_{t_1<t_2} w(\xi)R_{\EE^{t_2\xi}}^{-1}\ddd t_2\alpha_1+\int_{t_1<t_2} w(\xi)\ddd t_2\alpha_1\\&=
-\int_{t_1<t_2} \beta_2\,\alpha_1+\int_{t_1<t_2} w(\xi)\ddd t_2\alpha_1.
\end{split}
\]
Therefore,
\[
W_2(g)-W_3(g)=2\text{WZW}(g)+\left(\int_{t_2<t_1} \beta_2\,\alpha_1-\int_{t_1<t_2} w(\xi)\ddd t_2\alpha_1\right).
\]
We now want to show that the term in brackets vanishes. In fact,
\[
\int_{t_2<t_1} \beta_2\,\alpha_1=\int_{t_2<t_1}w(\xi)R_{\EE^{(t_1-t_2)\xi}}v(\xi)\ddd t_2\ddd t_1
\underset{r=t_1}{\overset{s=t_1-t_2}{=}} \int_{s<r}w(\xi)R_{\EE^{s\xi}}v(\xi)\ddd r\ddd s.
\]
On the other hand,
\[
\int_{t_1<t_2} w(\xi)\ddd t_2\alpha_1=\int_{t_1<t_2} w(\xi)R_{\EE^{t_1\xi}}v(\xi)\ddd t_2\ddd t_1
\underset{r=t_2}{\overset{s=t_1}=}\int_{s<r} w(\xi)R_{\EE^{s\xi}}v(\xi) \ddd r\ddd s.
\]
In conclusion,
\begin{equation}\label{e:WZWW2W3}
\text{WZW}(g)=\frac12W_2(g)-\frac12W_3(g).
\end{equation}
\end{remark}

\begin{example}[Adjoint representation]\label{exa:adjoint}
We now specialize the above example and remark to the case when $(V,\rho)=(\frg,\ad)$ is the adjoint representation.
Note that the second relation of \eqref{e:involrhovw} now says that $v$ is a derivation of $\frg$. As such it may be extended to the universal enveloping algebra as a derivation\footnote{This is well-defined because, if $v$ is a Lie algebra derivation and is extended to the tensor algebra as a derivation, then it preserves the ideal generated by elements of the form ${\xi\otimes\eta-\eta\otimes\xi-[\xi,\eta]}$.} and then to (the image of the exponential map in) the group. We have
\[
v(\EE^\xi)=\int_0^1 \EE^{\xi t}v(\xi)\EE^{\xi(1- t)}\ddd t.
\]
In this case, for $g=\EE^\xi$, we get
\[
\Phi(g) = v(g)g^{-1},
\]
where again for notational simplicity we have assumed to have a matrix Lie group. We further assume that $\frg$ is a quadratic Lie algebra, i.e., it is equipped with an invariant nondegenerate pairing $\langle\ ,\ \rangle$, which we can use to identify $\frg^*$ with $\frg$. We will then denote by $\bar q$ the image of $p$ under this isomorphism (and we will then take the target to be $\frg\oplus\frg$ with symplectic structure induced from the pairing) and by $\bar v$ the composition of $w$ with the isomorphism. We can then rewrite \eqref{e:Sbiaffine} as
\begin{equation}\label{e:Sbiaffinead}
\Hat S^f_\text{HJ}(q_\ii,p^\oo;g)=-\langle \bar q^\oo,{g}^{-1}q_\ii g\rangle
-\langle \bar q^\oo,{g}^{-1}v(g)\rangle
-\langle \bar v(g)g^{-1},q_\ii\rangle
-\text{WZW}(g)
\end{equation}
and, using \eqref{e:W2W3} and \eqref{e:WZWW2W3}, the WZW term as
\begin{equation}\label{e:WZWad}
\text{WZW}(g)=\frac12\langle\bar v(g)g^{-1},v(g)g^{-1}\rangle+
\frac12\int_0^1 \langle\bar v(h)h^{-1},[v(h)h^{-1},\ddd h h^{-1}]\rangle.
\end{equation}
\end{example}
\begin{remark}[A trivial extension]\label{r:tensorW}
The last example may be generalized to the case when $V=\frg\otimes Z$, with $Z$ a trivial representation of $\frg$. We keep denoting the representation by $\ad$: $\ad_x(y\otimes z)=[x,y]\otimes z$.
Now 
$v\colon\frg\to\frg\otimes Z$ is a derivation in the sense that $v([x,y])=\ad_x v(y)-\ad_yv(x)$ for all $x,y\in\frg$.
We identify $V^*$ with $\frg\otimes Z^*$ using the nondegenerate pairing on $\frg$ and denote by $\bar q$ and $\bar v$ the composition of $p$ and $w$ with this isomorphism. Note that $\bar v$ is also a derivation in the above sense.
The HJ action is again given by \eqref{e:Sbiaffinead} and \eqref{e:WZWad} with the obvious understanding of notations.
\end{remark}


\begin{remark}[General solution via the exponential map]\label{r:genexp}
In the general case, but under the assumption that the exponential map $\frg\to G$ is surjective,
the HJ action can also be characterized as follows. As in \eqref{e:SHJLie}, we set $g=P\EE^{\int_{t_\ii}^{t_\oo} e}$.
We then take $T\in\frg$ with $\EE^T=g$ and define
\[
h(t) = \EE^{-\frac{t-t_\ii}{t_\oo-t_\ii}T}P\EE^{\int_{t_\ii}^{t} e}.
\]
We have $e =h^{-1}\ddd h + h^{-1}\frac{T\,\ddd t}{t_\oo-t_\ii}h$, which shows that we can reduce $e$ to $e_0:=\frac{T\,\ddd t}{t_\oo-t_\ii}$ by a gauge transformation. Since \eqref{e:SHJLie} does not depend on which representative of $[e]$ we take, we may pick $e_0$. With the change of variable $s=\frac{t}{t_\oo-t_\ii}$,
the action then becomes
\[
S[p,q,e_0] =\int_{0}^{1} (p_i\ddd q^i - \Tilde H(p,q)\ddd s),\qquad
\Tilde H:={\langle H,T\rangle},
\]
As a result, 
\begin{equation}\label{e:HJgeneralLie}
\Hat S^f_\text{HJ}(q_\ii,Q_\oo;g) = S_\text{HJ}^{f,\Tilde H}(q_\ii,Q_\oo;0,1),
\end{equation}
where $S_\text{HJ}^{f,\Tilde H}$ denotes the HJ action for a system with hamiltonian $\Tilde H$. This formula is very general, but, unlike the previous ones for the biaffine case, not explicit.
\end{remark}

\subsection{The general case}\label{s:fgeneralcase}
We briefly discuss the general case, where we make no simplifying assumptions on the structure functions $f_{\alpha\beta}^\gamma$, only for completeness, as this will not be needed for the rest of the paper. As a consequence, this section, which requires knowledge of Lie algebroids and Lie groupoids, may be safely skipped by the not interested reader.

We will assume that the differentials of the $H_\alpha$s on 
 their common zero locus $C$ 
are linearly independent. This in particular implies that $C$ is a submanifold of $N:=T^*M$ and, as observed in Remark~\ref{r:coisomany}, that the hamiltonian vector fields $X_\alpha:=X^{H_\alpha}$ restricted to $C$ are linearly independent, so they define a regular distribution: the characteristic distribution of the coisotropic submanifold $C$. (See Definition~\ref{d:specialsub} and Remark~\ref{r:sympred}.)

The first remark is that the structure functions $f_{\alpha\beta}^\gamma$ and the hamiltonian vector fields $X_\alpha$ of the constraints $H_\alpha$ define a Lie algebroid $A$ over $C$. As a vector bundle, $A$ is the trivial product $C\times \RR^k$, where $k$ is the number of constraints. The anchor map $\rho\colon A\to TC$ is given by the vector fields $X_\alpha$:
\[
\rho(x,u_\alpha)=X_\alpha(x),
\]
where $(u_1,\dots,u_k)$ is the standard basis of $\RR^k$. We can expand a section $\sigma\in\Gamma(A)$ as
$\sigma^\alpha u_\alpha$, where now the $u_\alpha$s are regarded as constant sections of $A$. We then
have $\rho(\sigma)= \sigma^\alpha X_\alpha$. The Lie bracket is defined as
\[
[\sigma,\tau]^\gamma = f_{\alpha\beta}^\gamma \sigma^\alpha\tau^\beta + \sigma^\alpha X_\alpha(\tau^\gamma)-
\tau^\alpha X_\alpha(\sigma^\gamma).
\]


\begin{remark}\label{r:LielagebroidC}
The Lie algebroid $A$ is isomorphic to the characteristic distribution of $C$ viewed as a Lie subalgebroid of $TC$. Note that in general the conormal bundle $N^*C$ of a coisotropic submanifold of a symplectic manifold $N$ has a canonical Lie algebroid structure, isomorphic to the characteristic distribution, which makes it into a Lie subalgebroid of $T^*N$. In our particular case, since $C$ is defined by constraints, its normal bundle, and therefore its conormal bundle, can be trivialized, and this is what have done above.
\end{remark}

The Lie algebroid $A\to C$ can be extended to the vector bundle $\Hat A\to N=T^*M$ with $\Hat A=N\times \RR^k$. We can then view the $H_\alpha$s as a section $H$ of $\Hat A^*$. It is also possible to extend the anchor map and the bracket as defined above, but in general $\Hat A$ will not be a Lie algebroid 
(it is, though, when the structure functions are constant).

Next we denote the map $(p,q)\colon I\to T^*M$, with $I$ the interval $[t_\ii,t_\oo]$, as $x\colon I\to N$. The field $e$ may be regarded as a section of $T^*I\otimes x^*\Hat A$.
The constraints $H_\alpha$s composed with $x$ define a section of $x^*\Hat A^*$, which we keep denoting as $H$. As a consequence, 
we may rewrite the action as
\[
S[p,q,e] =\int_{t_\ii}^{t_\oo} (p_i\ddd q^i - \langle H(p,q),e\rangle),
\]
where $\langle\ ,\ \rangle$ denotes the pairing of $\Hat A^*$ with $\Hat A$.
A solution $(x,e)$ of the evolution equations is the same as an anchor-preserving morphism $TI\to \Hat A$. 

\newcommand{\calG}{\mathcal{G}}

The generalized HJ action $\Hat S^f_\text{HJ}$ is defined as in \eqref{e:HatST} and, by Theorem~\ref{t-thm1}, is a generalized generating function for the evolution relation. 

It is possible to give a more transparent expression for $\Hat S^f_\text{HJ}(q_\ii,Q_\oo;[e])$, analogous to the one we have in the Lie algebra case, when 
the $x$ component of
the assumedly
unique solution of the evolution equations specified by $q_\ii$, $Q_\oo$, and $e$ has image in $C$ (note that, if $x(t)$ is in $C$ for some $t$, then $x(t)$ is in $C$ for every $t$).
%
%
%
In fact, in this case we may view $(x,e)$ as
an anchor-preserving morphism $x\colon TI\to A$, which is the same as a Lie algebroid morphism $TI\to A$.
Denote by $\calG$ the source simply connected Lie groupoid of $A$ (i.e., up to isomorphism, the monodromy groupoid of the characteristic distribution).
The Lie algebroid morphism $(x,e)$, with initial condition $x(t_\ii)=x_a\in C$, may be then uniquely integrated to a Lie groupoid morphism
$E\colon I\times I\to\calG$. If we denote by $\alpha$ and $\beta$ the source and target maps of $\calG$, then we have, in particular, 
\[
x(s)=\alpha(E(s,t))\ \forall t\in I,\qquad x(t)=\beta(E(s,t))\ \forall s\in I.
\]
Define $h(t):=E(t_\ii,t)$. We then have
\[
h(t_\ii)=1_{x_a},\quad \alpha(h(t))=x_a\ \forall t\in I,
\]
where $1_x$ denotes the unit of $\calG$ at $x\in N$.
In particular, $h$ is a map $I\to\alpha^{-1}(x_a)$. 

We may now recover $e$ from $h$, generalizing footnote~\ref{f:lefttrans},
 as follows. First consider the linear map 
$\ddd_th\colon T_tI\to T_{h(t)}\alpha^{-1}(x_a)$. Next consider the left translation 
$l_h\colon\alpha^{-1}(\beta(h))\to \alpha^{-1}({\alpha(h) })
$, $l_hg=hg$. Since $h1_{\beta(h)}=h$, we get the linear map
\[
\ddd_{1_{\beta(h)}}l_h\colon A_{\beta(h)}\to T_h \alpha^{-1}(\alpha(h)).
\]
Using $\beta(h(t))=x(t)$ and $\alpha(h(t))=x_a$, we get the linear map
\[
e(t)=(\ddd_{1_{x(t)}}l_h)^{-1}\ddd_th\colon T_tI\to A_{x(t)}, 
\]
which may be shown to be the value at $t$ of our section $e$ of $T^*I\otimes x^*\Hat A$.

The rest of the computation is now like in the Lie algebra case, and we have
\[
\Hat S^f_\text{HJ}(q_\ii,Q_\oo;g)=S^f_\text{HJ}(q_\ii,Q_\oo)[e]\quad \text{with }g=h(t_\oo)\in\alpha^{-1}(x(t_\ii))\cap\beta^{-1}(x(t_\oo)).
\]
Finally, we can take variations with respect to $q_\ii$, $Q_\oo$, and $g$, getting
\[
\delta \Hat S^f_\text{HJ}(q_\ii,Q_\oo;g)= 
P_i(t_\oo)\delta Q^i_\oo
-p_i(t_\ii)\delta q^i_\ii -  \left\langle H(x(t_\oo)),(\ddd_{{1_{x(t_\oo)} }
}l_g)^{-1}\delta g\right\rangle,
\]
with $x(t_\oo)=(p^{q_\ii,Q_\oo}(t_\oo),q_{q_\ii,Q_\oo}(t_\oo))$. This is consistent with $\Hat S^f_\text{HJ}$ being a generalized generating function for the evolution relation $L$ (it is in general a weaker statement, as we are now only allowed to take variations in the class of variables for which the $x$ component of a solution lies in $C$).

\begin{remark}
The case of constant structure functions, hence corresponding to a Lie algebra $\frg$, fits into the general case as the action Lie algebroid $T^*M\times\frg$ 
with the action Lie groupoid $T^*M\times G$ as its integration. In this case, the Lie algebroid and the Lie groupoid are actually defined over the whole of $T^*M$ and not only over $C$, so the formula for $\Hat S^f_\text{HJ}(q_\ii,Q_\oo;g)$ has no restrictions.
\end{remark}


\section{Systems with nontrivial evolution and constraints}\label{r:nontrivev}
For completeness, 
we briefly discuss here the case of systems which, in addition to constraints $H_\alpha$, $\alpha=1\dots,k$, also have a nontrivial evolution with hamiltonian $H$. Namely, we consider an action of the form
\[
S[p,q,e] =\int_{t_\ii}^{t_\oo} (p_i\ddd q^i - \ddd t\, H- e^\alpha H_\alpha(p,q)).
\]
In addition to assuming involutivity of the constraints, as in \eqref{e:involutivity}, we also assume them to be constants of motion for $H$: namely, $\{H,H_\alpha\}=0$ for all $\alpha$. 

This system can actually be treated as above simply adding $H$ to the set of the, now $k+1$, constraints, say, as $H_0$. The only difference is that the one-form $e^0$ instead of being free will be set to be $\ddd t$. 
This is a possible choice, as the Lie group $\Tilde G$ (in the case of structure constants $f_{\alpha\beta}^\gamma$; otherwise, more generally, 
the Lie groupoid $\Tilde \calG$) 
corresponding to the rank-$(k+1)$ Lie algebra (Lie algebroid) factorizes as $\RR\times G$ ($\RR\times\calG$), so we can fix $t_\oo-t_\ii$ in the first factor. 

The HJ action, which serves as a generalized generating function, is then, in the Lie algebra case,
\[
\Hat S^f_\text{HJ}(q_\ii,Q_\oo;g):=S^f_\text{HJ}(q_\ii,Q_\oo)[\ddd t,e],\quad \text{with }g=P\EE^{\int_{t_\ii}^{t_\oo} e}.
\]
In the Lie algebroid case, we get, more generally,
\[
\Hat S^f_\text{HJ}(q_\ii,Q_\oo;g):=S^f_\text{HJ}(q_\ii,Q_\oo)[\ddd t,e], 
\]
with $g:=\pi_{\calG}(\Tilde h(t_\oo))\in\alpha^{-1}(x(t_\ii))\cap\beta^{-1}(x(t_\oo))$ and  $\pi_\calG$ the projection $\Tilde \calG\to\calG$.

\subsection{Classical mechanics}\label{r:SwithH}
A more conceptual way is to add two new canonically conjugated variables $(E,\ttt)$ and consider instead the action
\[
S[p,q,E,\ttt,e^0] =\int_{t_\ii}^{t_\oo} (p_i\ddd q^i -E\ddd\ttt-e^0(H-E)- e^\alpha H_\alpha(p,q)).
\]
Note that the constraints $H_\alpha$, $\alpha=0\dots,k$, with $H_0=H-E$, are still first class. This is then a system as the ones we have studied in the previous sections. One of the evolution equations is now $\ddd E=0$, which says that eventually $E$ will be a constant (the energy). Another important evolution equation is $\ddd\ttt=e^0$, which can be used to substitute $\ddd\ttt$ to $e^0$, so that we may regard $\ttt$ as time. Note that, if we set endpoint conditions $\ttt(t_\ii)=t_\ii$ and $\ttt(t_\oo)=t_\oo$, we actually get {the time evolution} from $t_\ii$ to $t_\oo$. 

In the following, we focus on the case $H_\alpha=0$ for all $\alpha>0$, i.e.,
\[
S[p,q,E,\ttt,e^0] =\int_{t_\ii}^{t_\oo} (p_i\ddd q^i  \textcolor{blue}{-} E\ddd\ttt-e^0(H-E)).
\]
The evolution equations are
\[
\ddd q^i=e^0\frac{\dd H}{\dd p_i},\quad \ddd p_i=-e^0\frac{\dd H}{\dd q^i},\quad \ddd E=0,\quad \ddd\ttt=e^0,
\]
and the constraint is $H(p,q)=E$. The evolution relation is then
\begin{equation}\label{e:LCM}
L=\{(p^\ii,E^\ii,q_\ii,\ttt_\ii,p^\oo,E^\oo,q_\oo,\ttt_\oo)\ |\ (p,q)_\oo = \phi^H_{\ttt_\oo-\ttt_\ii}((p,q)_\ii),\ 
E^\oo=E^\ii=H(p^\ii,q_\ii)
\}
\end{equation}
where $\phi^H_T$ denotes the hamiltonian flow of $H$ for time $T$. 

We now consider endpoint conditions $(q_\ii,\ttt_\ii,q_\oo,E^\oo)$. We thus evaluate the modified action
$S^f[p,q,E,\ttt,e^0] := E(t_\oo)\ttt(t_\oo) + S[p,q,E,\ttt,e^0]$
on the solution, with components $(\Tilde p,\Tilde q,\Tilde \ttt)$,
of the evolution equations for fixed $(q_\ii,\ttt_\ii,q_\oo,E^\oo)$ and fixed $e^0$. The first remark is that $-E\ddd\ttt+e^0E$ vanishes on a solution. Therefore,
\[
S^f_\text{HJ}(q_\ii,\ttt_\ii,q_\oo,E^\oo,e^0)= E^\oo\,\Tilde\ttt(t_\oo) +
\int_{t_\ii}^{t_\oo}(\Tilde p_i\ddd \Tilde q^i - \ddd \Tilde\ttt\, H(\Tilde p, \Tilde q)).
\]
We can actually pick $e^0$ constant or at least of constant sign (by a gauge transformation).
With this choice  we are sure that the map $t\mapsto\Tilde\ttt(t)=\ttt_\ii+\int_{t_\ii}^te^0$ is a diffeomorphism. We can then make this change of variable in the integral, getting
\[
\Hat S^f_\text{HJ}(q_\ii,\ttt_\ii,q_\oo,\ttt_\oo;T)=E^\oo\,(\ttt_\ii+T) +\int_{\ttt_\ii}^{\ttt_\ii+T}(\Check p_i\ddd \Check q^i - \ddd \Tilde\ttt\, H(\Check p, \Check q)),
\]
with $\Check p(\Tilde\ttt):=\Tilde p(t(\ttt))$, $\Check q(\Tilde\ttt):=\Tilde q(t(\ttt))$ and $T:=\int_{t_\ii}^{t_\oo}e^0$.
But now we recognize in the second term on the right hand side the time-dependent HJ action \eqref{e:HJtimedep} for the hamiltonian $H$ from time $\ttt_\ii$ to time $\ttt_\ii+T$. Therefore, we have
\begin{equation}\label{e:HJCMtE}
\begin{split}
\Hat S^f_\text{HJ}(q_\ii,\ttt_\ii,q_\oo,E^\oo;T)&=S_\text{HJ}^H(q_\ii,q_\oo;\ttt_\ii,\ttt_\ii+T) + E^\oo(\ttt_\ii+T)\\
&=S_\text{HJ}^H(q_\ii,q_\oo;0,T) + E^\oo(\ttt_\ii+T).
\end{split}
\end{equation}
One can easily verify that $\Hat S^f_\text{HJ}$ is a generating function for the evolution relation \eqref{e:LCM}.

The HJ action $\Hat S_\text{HJ}$ 
for endpoint conditions $(q_\ii,\ttt_\ii,q_\oo,\ttt_\oo)$
can be obtained from \eqref{e:HJCMtE} by composing, following Remark~\ref{r:compgenfun}, $\Hat S^f_\text{HJ}$ with the generating function $-E^\oo\ttt_\oo$ of the identity map, and reducing with respect to the intermediate variable $E^\oo$; namely, $\Hat S_\text{HJ}$ is the evaluation of
$\Hat S^f_\text{HJ}-E^\oo\ttt_\oo$ at its critical point in $E^\oo$. We get
\begin{equation}\label{e:HJCMtt}
\Hat S_\text{HJ}(q_\ii,\ttt_\ii,q_\oo,\ttt_\oo)=S_\text{HJ}^H(q_\ii,q_\oo;\ttt_\ii,\ttt_\oo).
\end{equation}
One can easily verify that $\Hat S_\text{HJ}$ is also a generating function for the evolution relation \eqref{e:LCM}.
Observe that the dependency on $T$ has disappeared.

Note that we could not have computed $\Hat S_\text{HJ}$ directly because fixing the endpoint conditions $(q_\ii,\ttt_\ii,q_\oo,\ttt_\oo)$ yields no solution in general.
Namely, choosing $e^0$ arbitrarily would lead to no solutions if $\ttt_\oo-\ttt_\ii\not=T:=\int_{t_\ii}^{t_\oo} e^0$. A way to deal with such a situation will be discussed in Section~\ref{r:other}, see Remark~\ref{r:CM}. As a result, $T$ is fixed by the endpoint conditions and therefore is no longer an independent variable for $\Hat S_\text{HJ}$.

The HJ action for the time-independent, constrained system (a parametrization invariant theory) we considered in this section turns out to be the same as the usual HJ action for a time dependent system describing a hamiltonian evolution. In this case, time is restored via the endpoint conditions. 

For simplicity we have obtained this result choosing $e^0$ constant, so that we could make the change of variable globally.
We might have also worked with a generic $e^0$. In this case we would have made changes of variables in all regions where it is different from zero and ignored the other regions. The ``time'' variable $\Tilde\ttt$ would have suffered slowing down, speeding up, freezing and even rewinding, but with no change in the final result described
in \eqref{e:HJCMtt}
 by $\Hat S_\text{HJ}$.\footnote{As the prince observed, ``Most people think time is like a river that flows swift and sure in one direction, but I have seen the face of time and I can tell you they are wrong. Time is an ocean in a storm,'' yet ``what is written in the timeline cannot be changed.''
}

{
\begin{remark}[Integrable systems III]\label{rem:intsysIII}
As a generalization of Remarks \ref{rem:intsysI} and \ref{rem:intsysII}, let us consider the case of an integrable system where we have exactly $n$ strictly involutive constraints $H_1,\ldots,H_n$ that are integrals of motion for a hamiltonian $H$, i.e., we have $\{H,H_i\} = \{H_i,H_j\} = 0$ for all $i,j$. Around a point $(p_0,q_0)$ where the differentials $\ddd H_1,\ldots,\ddd H_n$ are linearly independent, again by the Carath\'eodory--Jacobi--Lie theorem,
see fotonote~\ref{f:CLJ},
we may find a Darboux chart $(P_i,Q^i)$ with $P_i = H_i$. Changing coordinates from $(p_i,q^i)$ to $(P_i,Q^i)$ in the integral we get 
\[
S^g[p,q,E,\ttt,e^0,e] = - g(q(t_\ii),P(t_\ii)) + \int_{t_\ii}^{t_\oo} (Q^i\ddd P_i - E \ddd\ttt -e^0(H-E) - e^i {P_i}).
\]
Notice that, since the constraints are integrals of motion, in $(P_i,Q^i)$ coordinates the hamiltonian $H$ depends only on $P_i$: $0 = \{H,H_i\} = \{H,P_i\} = \partial H /\partial Q_i$.
On solutions, we have $\ddd P_i =0$, $\ddd Q^i = e^i + e^0 \frac{\dd H}{\dd P_i}$ and $\ddd H = 0$. The HJ action then reads 
\begin{equation}
\Hat{S}^g_\text{HJ}(q_\ii,P^\oo,\ttt_\ii,\ttt_\oo;T) = -g(q_\ii,P^\oo) - (\ttt_\oo - \ttt_\ii)H(P^\oo) - T^iP_i
\end{equation}
in a neighborhood of $(p_0,q_0)$.\footnote{\label{f:noT}Here $T$ stands for the collection $T^i:=\int_{t_\ii}^{t_\oo} e^i$, $i>0$.
The parameter $T^0:=\int_{t_\ii}^{t_\oo} e^0$ does not appear, as in \eqref{e:HJCMtt}, because we fix $\ttt_\ii$ and $\ttt_\oo$.}
We may also compute the HJ action with $Q$ endpoint conditions. 
We have 
\[
S^f[p,q,E,\ttt,e^0,e] = - f(q(t_\ii),Q(t_\ii)) + \int_{t_\ii}^{t_\oo} (P_i\ddd Q^i - E \ddd\ttt -e^0(H-E) - e^iP_i).
\]
Solving the evolution equation for $Q^i$ we obtain the HJ action\footnote{As in footnote~\ref{f:noT}, 
$T$ stands for the collection $T^i$, $i>0$.} 
\begin{equation}
\Hat{S}^f_\text{HJ}(q_\ii,Q_\oo,\ttt_\ii,\ttt_\oo;T) = - f\left(q_\ii, Q_\oo - T - (\ttt_\oo - \ttt_\ii)\frac{\dd H}{\dd P}\right) + (\ttt_\oo - \ttt_\ii)H(P(q_\ii,Q_\oo)).
\end{equation}
\end{remark}
}
\subsection{The free relativistic particle}\label{s:RP}
The above discussion generalizes easily to the case of a free, relativistic particle. In this case, there is no dynamics but only the constraint $E^2=m^2+p^2$. (For simplicity, we omit indices. If we are not in one dimension, $p^2$ means
$\sum_i (p_i)^2$. Similarly, in the action we will simply write $p\ddd q$ instead of $p_i\ddd q^i$). Therefore, we consider the action
\[
S[p,q,E,\ttt,e^0] =\int_{t_\ii}^{t_\oo} \left(p\ddd q-E\ddd\ttt-\frac{e^0}2(p^2+m^2-E^2)\right),
\]
where we have introduced the factor $\frac12$ just for convenience. The evolution equations are simply
\[
\ddd q=e^0p,\quad\ddd p = 0,\quad\ddd E = 0,\quad\ddd\ttt=e^0E,
\]
whereas the constraint is
\[
p^2+m^2-E^2=0.
\]
We first compute the HJ action, without fixing $e^0$, for endpoint conditions given by $(q_\ii,\ttt_\ii,p^\oo,E^\oo)$. This means the we have to insert a solution of the evolution equations into 
\[
S^f[p,q,E,\ttt,e^0]=E(t_\oo)\ttt(t_\oo)-p(t_\oo)q(t_\oo) + S[p,q,E,\ttt,e^0].
\]

The first evolution equation implies $p\ddd q=e^0p^2$. The second implies that $p(t)=p^\oo$ for all $t$, so we actually 
have $p\ddd q=e^0(p^\oo)^2$. The third equation implies $E(t)=E^\oo$ for all $t$, which together with the last equation implies 
$E\ddd\ttt=e^0(E^\oo)^2$.
Therefore,
\[
S|_\text{solution}=\int_{t_\ii}^{t_\oo} \frac{e^0}{\textcolor{blue}{2}}\,((p^\oo)^2-m^2-(E^\oo)^2)= \frac{T}{\textcolor{blue}{2}}\,((p^\oo)^2-m^2-(E^\oo)^2),
\]
where we have set $T:=\int_{t_\ii}^{t_\oo} e^0$. Solving the evolution equations explicitly yields
\[
q(t_\oo) = q_\ii + p^\oo T\quad\text{and}\quad
\ttt(t_\oo) = \ttt_\ii + E^\oo T,
\]
which implies
\[
p(t_\oo)q(t_\oo) = p^\oo q_\ii + (p^\oo)^2T\quad\text{and}\quad
E(t_\oo)\ttt(t_\oo) = E^\oo\ttt_\ii + (E^\oo)^2 T
\]
on solutions. We then have
\begin{equation}\label{e:RPpE}
\Hat S^f_\text{HJ}(q_\ii,\ttt_\ii,p^\oo,E^\oo;T)=E^\oo\ttt_\ii -p^\oo q_\ii + \frac{(E^\oo)^2-(p^\oo)^2-m^2}2 T,
\end{equation}
and one can easily verify that this is a generating funciton for the evolution relation.

We may now easily pass to the HJ action for endpoint conditions $(q_\ii,\ttt_\ii,q_\oo,\ttt_\oo)$ simply composing the generating function $\Hat S^f_\text{HJ}$ with the generating function $p^\oo q_\oo-E^\oo\ttt_\oo$ of the identity map, 
and reducing with respect to the intermediate variables $(p^\oo,E^\oo)$; namely, we evaluate 
$\Hat S^f_\text{HJ}+p^\oo q_\oo-E^\oo\ttt_\oo$ at its critical point in $(p^\oo,E^\oo)$. Setting the derivatives with respect to $p^\oo$ and $E^\oo$ to zero yields
\[
p^\oo =\frac{\Delta q}T\quad\text{and}\quad
E^\oo =\frac{\Delta\ttt}T,
\]
where we have set
\[
\Delta q:=q_\oo-q_\ii\quad\text{and}\quad\Delta\ttt:=t_\oo-t_\ii.
\]
Inserting, we get
\begin{equation}\label{e:RPqtT}
\Hat S_\text{HJ}(q_\ii,\ttt_\ii,q_\oo,\ttt_\oo;T)=\frac{(\Delta q)^2-(\Delta\ttt)^2}{2T}-\frac12 m^2T.
\end{equation}

We can finally try to reduce with respect to the variable $T$, i.e., to evaluate $\Hat S_\text{HJ}$ at its critical point in $T$.  Setting the derivative with respect to $T$ to zero yields
\[
T^2=\frac{(\Delta\ttt)^2-(\Delta q)^2}{m^2}.
\]
This shows that we can perform this reduction if and only if the endpoint conditions select a timelike trajectory: 
$(\Delta\ttt)^2>(\Delta q)^2$.\footnote{This is actually a necessary condition for a solution to exist. In fact, solving the evolution equations for $\ttt$ and $q$ yields $\Delta\ttt=ET$ and $\Delta q=pT$, so 
$(\Delta\ttt)^2-(\Delta q)^2=(E^2-p^2)T^2=m^2T^2\ge0$. The case $(\Delta\ttt)^2=(\Delta q)^2$ would however imply $T=0$, so $\Delta\ttt=0$ and $\Delta q=0$, which is excluded because the initial and final conditions must be distinct in the HJ setting.} There are two roots in $T$ which yield, not unexpectedly, 
\begin{equation}\label{e:RPqt}
S_\text{HJ}(q_\ii,\ttt_\ii,q_\oo,\ttt_\oo)=\mp m\sqrt{(\Delta\ttt)^2-(\Delta q)^2},
\end{equation}
i.e., the Minkowskian length up to the factor $\mp m$. Both choices of sign give a generating function for the evolution relation.\footnote{Note that taking derivatives with respect to $\ttt_\ii$ and $\ttt_\oo$ yields
$E^\ii=E^\oo=\pm m\frac{\Delta\ttt}{\sqrt{(\Delta\ttt)^2-(\Delta q)^2}}$. In particular, for $\Delta q=0$, we get
$E^\ii=E^\oo=\pm m \operatorname{sgn}{\Delta\ttt}$. For $\Delta\ttt>0$, it is the plus sign that yields the relation $E=m$ for the relativistic particle, the minus sign corresponding to its antiparticle.}


\begin{remark} In the massless case $m=0$, equation (\ref{e:RPqtT}) becomes 
\[
\Hat S_\text{HJ}(q_\ii,\ttt_\ii,q_\oo,\ttt_\oo,T)=\frac{(\Delta q)^2-(\Delta\ttt)^2}{2T},
\]
so the last step---reduction in the variable $T$---is not possible (the critical point in $T$ does not exist for general values of $q_\ii,\ttt_\ii,q_\oo,\ttt_\oo$). Still, the vanishing of the derivative in $T$ correctly yields the ``lightlike displacement'' constraint: $(\Delta q)^2 - (\Delta \ttt)^2=0$.
\end{remark}

\section{Generalized generating functions for ``bad'' endpoint conditions}\label{r:other}
The presence of constraints is not the only source for the appearance of extra parameters in the generating function for the evolution relation. Another instance occurs when we choose ``bad'' endpoint conditions that do not ensure existence of a solution. 


\subsection{No evolution and no constraints}
We discuss a simple example here (more examples, also coupled to the presence of constraints, will appear in the rest of the paper). Consider the action
\[
S[p,q]=\int_{t_\ii}^{t_\oo}p\dot q\,\ddd t.
\]
The EL equations are simply $\dot p=\dot q =0$, so the evolution relation is just the diagonal in $T^*\RR\times T^*\RR$:
\[
L = \{(p,q,p,q)\in T^*\RR\times T^*\RR, (p,q)\in T^*\RR\}.
\]
If we now choose endpoint conditions $q_\ii$ and $q_\oo$, we will get no solutions unless $q_\ii=q_\oo$, in which case we get a whole family of solutions, parametrized by $p$ solving $\dot p =0$. 
Nevertheless, there is a generalized generating function for $L$:\footnote{This generalized generating function can also be obtained as the composition, in the sense of Remark~\ref{r:compgenfunII}, of the generating functions
$\psi_1(q_\ii,\lambda)=-\lambda q_\ii$ and $\psi_2(\lambda,q_\oo)=\lambda q_\oo$ that do and undo a $-\pi$-rotation in phase space.}
\[
S_\text{gen}(q_\ii,q_\oo;\lambda)=\lambda\,(q_\oo-q_\ii).
\]
In fact, $\frac{\dd S_\text{gen}}{\dd\lambda}=0$ yields the condition $q_\ii=q_\oo$, whereas
$p(t_\ii)=-\frac{\dd S_\text{gen}}{\dd q_\ii}=\lambda$ and $p(t_\oo)=\frac{\dd S_\text{gen}}{\dd q_\oo}=\lambda$
imply $p(t_\ii)=p(t_\oo)$.

This generating function may also be regarded as a sort of Hamilton--Jacobi action, where we impose only one EL equation: namely, $\dot p=0$. (The other equation, $\dot q=0$, cannot be imposed anyway because generically it has no solutions.) The parameter $\lambda$ arises here as the constant value of $p$, so it parametrizes the family of solutions.

We may also obtain the generating function $S_\text{gen}(q_\ii,q_\oo,\lambda)$ by a partial Legendre transform, i.e., as
\[
S_\text{gen}(q_\ii,q_\oo;\lambda)= \lambda q_\oo+S^f_\text{HJ}(q_\ii,\lambda),
\]
where
\[
S^f_\text{HJ}(q_\ii,p^\oo)=-p^\oo q_\ii
\]
is the HJ action for the good choice $(q_\ii,p^\oo)$ of endpoint conditions. Unlike the true Legendre transform  (see Remark~\ref{r:Legendretransf}), we do not evaluate at $\lambda^\text{crit}$, which in this case does not exist. We will put this observation more in context in Section~\ref{s:partLegtransf}.

\newcommand{\frc}{\mathfrak{c}}
\newcommand{\frd}{\mathfrak{d}}
The generating function $S_\text{gen}$ may appear even more naturally from the viewpoint of path integral quantization. In fact, since there is no evolution at all, the evolution operator is the identity operator, so, for the chosen polarizations, its integral kernel is a delta function:
\[
K(q_\ii,q_\oo)=\delta(q_\ii-q_\oo).
\]
By Fourier transform, we we can also write
\[
K(q_\ii,q_\oo)=\int \frac{\ddd\lambda}{2\pi\hbar}\;\EE^{\frac\II\hbar S_\text{gen}(q_\ii,q_\oo,\lambda)}.
\]

One formal way to get this result directly from the path integral is as follows. We start with \eqref{e:Kpathintegral} and write $p=\lambda+\Hat p$, where $\lambda$ is a constant (a solution to $\dot p=0$). Note that $\lambda$ here is still a variable to be integrated out, so we have to make sure that $\Hat p\colon [t_\ii,t_\oo]\to \RR$ is in a complement to the space of  constant maps. For example, we may impose $\int_{t_\ii}^{t_\oo}\Hat p\,\ddd t = 0$.\footnote{Let $\frc$ be the space of constant maps and $\frd$ the space of maps $\Hat p\colon[t_\ii,t_\oo]\to\RR$ with vanishing integral. We have $C^\infty([t_\ii,t_\oo])=\frc\oplus \frd$.
In fact, if $\lambda$ is constant and $(t_\ii-t_\oo)\lambda=\int_{t_\ii}^{t_\oo}\lambda\;\ddd t=0$, then $\lambda$ vanishes, and we have shown that $\frc\cap \frd=0$.
On the other hand, every map $p$ can be written as $\lambda+\Hat p$ with 
$\lambda=\frac1{t_\oo-t_\ii}\int_{t_\ii}^{t_\oo} p(t)\ddd t$. It follows immediately that $\int_{t_\ii}^{t_\oo}\Hat p\,\ddd t = 0$.}
Then we have
\[
S[\lambda+\Hat p,q]=\lambda\,(q_\oo-q_\ii)+\int_{t_\ii}^{t_\oo}\Hat p\dot q\, \ddd t.
\]
Inserting into \eqref{e:Kpathintegral}, we get
\[
K(q_\ii,q_\oo) = \int\ddd\lambda\; \EE^{\frac\II\hbar\, \lambda\,(q_\oo-q_\ii)}
\int_{\substack{q(t_\ii)=q_\ii\\q(t_\oo)=q_\oo\\ \int_{t_\ii}^{t_\oo}\Hat p\,\ddd t = 0}}D\Hat p\,Dq\;\EE^{\frac\II\hbar \int_{t_\ii}^{t_\oo}\Hat p\dot q\, \ddd t}.
\]
The second integral formally does not depend on $q_\ii$ and $q_\oo$: in fact, we can make the affine change of variables
$q\mapsto \Hat q$ with 
\[
\Hat q(t)=q(t)-\frac{t-t_\oo}{t_\ii-t_\oo}q_\ii-\frac{t-t_\ii}{t_\oo-t_\ii}q_\oo.
\]
Observe that $\int_{t_\ii}^{t_\oo}\Hat p\dot q\, \ddd t=\int_{t_\ii}^{t_\oo}\Hat p\dot{\Hat q}\, \ddd t$, since $\dot q-\dot{\Hat q}$ is constant,
and that $\Hat q(t_\ii)=\Hat q(t_\oo)=0$. Therefore,
\[
\begin{split}
K(q_\ii,q_\oo) &= \int\ddd\lambda\; \EE^{\frac\II\hbar\, \lambda\,(q_\oo-q_\ii)}
\int_{\substack{\Hat q(t_\ii)=\Hat q(t_\oo)=0\\  \int_{t_\ii}^{t_\oo}\Hat p\,\ddd t = 0}}D\Hat p\,D\Hat q\;\EE^{\frac\II\hbar \int_{t_\ii}^{t_\oo}\Hat p\dot{\Hat q}\, \ddd t}\\ 
&\propto \int \frac{\ddd\lambda}{2\pi\hbar}\;\EE^{\frac\II\hbar S_\text{gen}(q_\ii,q_\oo,\lambda)}.
\end{split}
\]


In this example, one might also think of $\lambda$ as parametrizing the vacua of the theory. Instead of integrating over vacua, one may decide to select just one, labeled by $\lambda$, and in this case interpret $S_\text{gen}(q_\ii,q_\oo,\lambda)$
as the corresponding semiclassical contribution, more in line with Hamilton--Jacobi.

\begin{remark}
For an analysis of this example in the BV-BFV setting, see \cite[Sect.\ 4.4]{CMR15}
\end{remark}

\begin{remark}[Classical mechanics]\label{r:CM}
We may apply the above considerations to the system with one constraint described in Section~\ref{r:SwithH}.
We fix $(q_\ii,\ttt_\ii,q_\oo,\ttt_\oo)$ and let $T=\int_{t_\ii}^{t_\oo} e^0$. {}From the evolution equation $\ddd\ttt=e^0$, we get that there is no solution unless $T=\ttt_\oo-\ttt_\ii$. Therefore, we have to introduce a new parameter $\lambda$ to fix this condition, and we get
\[
\Check S_\text{HJ}(q_\ii,\ttt_\ii,q_\oo,\ttt_\oo;T,\lambda)=S_\text{HJ}^H(q_\ii,q_\oo;\ttt_\ii,\ttt_\ii+T)+
\lambda\,(T-\ttt_\oo+\ttt_\ii).
\]
Of course this generating function may be reduced. Setting the derivative with respect to $\lambda$ to zero, we get back $T=\ttt_\oo-\ttt_\ii$, which may be inserted into $\Check S_\text{HJ}$ yielding $\Hat S_\text{HJ}$ as in \eqref{e:HJCMtt}.
\end{remark}

\subsection{The partial Legendre transform}\label{s:partLegtransf}
Suppose our system has a unique solution for endpoint conditions $(q_\ii,p^\oo)$, so that we have the HJ action
$S^f_\text{HJ}(q_\ii,p^\oo)$ as in Example~\ref{exa:pfinal}. We define the partial Legendre transform of $S^f_\text{HJ}$ as
\[
S_\text{gen}(q_\ii,q_\oo,\lambda):= \lambda_i q_\oo^i+S^f_\text{HJ}(q_\ii,\lambda).
\]
It follows from \eqref{e:qpinSf} that $S_\text{gen}$ is a generalized generating function for the same evolution relation as $S^f_\text{HJ}$ is. In fact, $L_{S_\text{gen}}$ is determined by $S_\text{gen}$ via the equations
\[
p^\oo_i=\frac{\dd S_\text{gen}}{\dd q^i_\oo},\quad
p^\ii_i=-\frac{\dd S_\text{gen}}{\dd q^i_\ii}, \quad
\frac{\dd S_\text{gen}}{\dd \lambda}=0.
\]
Since $\frac{\dd S_\text{gen}}{\dd q^i_\oo}=\lambda_i$, we get $p^\oo=\lambda$. Moreover, by using  \eqref{e:qpinSf}, the last equation reads $q_\oo-q_{q_\ii,\lambda}(t_\oo)=0$ and the middle equation yields
$p^\ii=p^{q_\ii,\lambda}(t_\ii)$. Therefore, $L_{S_\text{gen}}$ consists of endpoint values $(p^\ii,q_\ii,p^\oo,q_\oo)$
of a solution and is thus the evolution relation.\footnote{It may of course happen that the last equation, $q_\oo-q_{q_\ii,\lambda}(t_\oo)=0$, can be solved for a unique $\lambda$, which we then denote as $\lambda^\text{crit}$. We may then insert this value into $ S_\text{gen}(q_\ii,q_\oo,\lambda)$ and get a generating function without extra parameters, the true Legendre transform $S_\text{HJ}(q_\ii,q_\oo)$ as in Remark~\ref{r:Legendretransf}.}

The above argument extends immediately to the case with constraints. Namely, suppose we have the HJ action
$S^f_\text{HJ}(q_\ii,p^\oo)[e]$. 
Then we define its partial Legendre transform as
\[
S_\text{gen}(q_\ii,q_\oo,\lambda)[e]:= \lambda_i q_\oo^i+S^f_\text{HJ}(q_\ii,\lambda)[e],
\]
and one immediately verifies, as above, that this is also a generalized generating function for the evolution relation.

\begin{example}[Partial Legendre transform for linear constraints]\label{exa:pLegtrlin}
Consider the case of a system with linear constraints in strict involution as described in Example~\ref{e:flinear}. We will suppress the target-space index $i$. 
If we apply the partial Legendre transform to \eqref{e:Sfablin}, we get
\[
\Hat S_\text{gen}(q_\ii,q_\oo;\lambda,T):=\lambda q_\oo+\Hat S^f_\text{HJ}(q_\ii,\lambda;T)=
\lambda (q_\oo-q_\ii)
-T^\alpha(\lambda v_\alpha + w_{\alpha}q_\ii)-
\frac12 T^\alpha T^\beta A_{\alpha\beta},
\]
which is also a generating function for the evolution relation. We may change this expression a bit by observing that the derivative with respect to $\lambda$ (i.e., $q_\oo-q_\ii-T^\alpha v_\alpha$) will have to be set to zero to define the evolution relation. We can insert this relation into $\Hat S_\text{gen}$ (without having to solve with respect to $\lambda$). Actually, multiplying the relation with $T^\beta w_{\beta}$ and using \eqref{e:Aalphabeta}, we get
\[
T^\alpha T^\beta A_{\alpha\beta} = T^\beta w_{\beta}q_\oo-T^\beta w_{\beta}q_\ii.
\]
We can use this to get rid of the quadratic term in the $T$s getting the generalized generating function
\begin{equation}\label{e:newgengenfun}
\begin{split}
\Tilde S_\text{gen}(q_\ii,q_\oo,\lambda,T)&=\lambda(q_\oo-q_\ii)-T^\alpha\lambda v_\alpha
-\frac12T^\alpha w_{\alpha}q_\ii-\frac12T^\alpha w_{\alpha}q_\oo\\
&=
\left(\lambda-\frac12T^\alpha w_{\alpha}\right)q_\oo
-\left(\lambda+\frac12T^\alpha w_{\alpha}\right)q_\ii-T^\alpha\lambda v_\alpha.
\end{split}
\end{equation}

\end{example}


\section{Infinite-dimensional targets}\label{s:inftarg}
For simplicity, {up to now} 
we have only discussed examples where the target symplectic manifold is finite-dimensional. The whole discussion can be generalized, almost verbatim, to the case when the target is infinite-dimensional and there are possibly infinitely many constraints. We will touch upon more analytical issues in 
Appendix~\ref{a:generatingfunctions}, and in particular in~\ref{s:HJinf}. Note that the main problem is that, in principle, the evolution relation might fail to be lagrangian.\footnote{This is due to intrinsically analytical problems. See \cite[Sect.\ 5.9]{CMwave} for an example.}
In {in this section} we will only
be considering infinite-dimensional examples where this problem does not actually arise (mainly thanks to the Hodge decomposition of differential forms).

\subsection{Three-dimensional abelian Chern--Simons theory}\label{s:3dabCS}
As a warm up, we discuss here the case of abelian Chern--Simons theory in three dimensions. The action for a cylinder $I\times\Sigma$, where $I=[t_\ii,t_\oo]$ is an interval and $\Sigma$ is a closed, oriented surface,
 is
\[
S[
 A]=\frac12\int_{I\times\Sigma} 
 {A}\ddd
 {A}
\]
with $
A\in\Omega^1(I\times\Sigma)$. A variation of the action yields
\[
\delta S = \int_{I\times\Sigma} \delta
{A}\ddd
{A}
- \frac12\int_{\{t_\oo\}\times\Sigma}
{A}\delta
{A}
+\frac12\int_{\{t_\ii\}\times\Sigma} 
{A}
\delta
{A},
\]
from which we read off the EL equations $\ddd{ A}=0$ and the boundary symplectic structure:
\[
\mathcal F^\dd=\Omega^1(\Sigma),\quad \omega=- \frac12\int_{\Sigma} \delta A \delta A.
\]
We may split $
A$ into its ``horizontal'' and ``vertical'' part
\[
A = A_\Sigma + A_I
\]
and regard $A_\Sigma$ as a map $I\to \Omega^1(\Sigma)=\mathcal F^\dd$, with target symplectic form 
$-\frac12\int_\Sigma \delta A_\Sigma\delta A_\Sigma$, 
and $A_I$ as a one-form on $I$ with values in
$\Omega^0(\Sigma)$. Splitting also $\ddd$ as $\ddd_\Sigma+\ddd_I$, we then have
\[
S'[A_\Sigma, A_I]:=S[A_\Sigma + A_I]=\int_{I\times\Sigma}\left(\frac12 A_\Sigma\ddd_IA_\Sigma +A_I\ddd_\Sigma A_\Sigma
\right),
\]
from which we see that $A_\Sigma$ represents the dynamical variables, whereas $A_I$ represents the Lagrange multipliers. Moreover, we now split the EL equations into the evolution equations
\[
\ddd_IA_\Sigma=-\ddd_\Sigma A_I
\]
and the constraints
\[
\ddd_\Sigma A_\Sigma=0.
\]

This is a system with linear constraints (in strict involution, since $\ddd^2=0$) that, almost, fits into Example~\ref{e:flinear}. What is missing is a splitting of the symplectic space of fields into $p$ and $q$ variables (or, more precisely, 
{a splitting in the sense of}
Remark~\ref{r:splitss}). We do it in the complexified version 
$\mathcal F^\dd_\CC=\Omega^1(\Sigma)\otimes\CC$, which splits as $\Omega^{1,0}(\Sigma)\oplus\Omega^{0,1}(\Sigma)$ in terms of a complex structure on $\Sigma$.\footnote{
{This actually corresponds to choosing a complex polarization for the real Chern--Simons theory.}
}
 We write accordingly $A_\Sigma=A^{1,0}+A^{0,1}$ and $\ddd_\Sigma=\dd+\bar\dd$, so
\[
S'[A_\Sigma, A_I]=\int_{I\times\Sigma}\left(
\frac12 A^{1,0}\ddd_IA^{0,1}+ 
\frac12 A^{0,1}\ddd_IA^{1,0}+ 
A_I(\bar\dd A^{1,0}+\dd A^{0,1})
\right).
\]
We can integrate by parts the second term, so that the action reduces, up to a boundary term, to
\[
S''
(A^{1,0},A^{0,1},A_I)=\int_{I\times\Sigma}\left(
A^{1,0}\ddd_IA^{0,1}+ 
A_I(\bar\dd A^{1,0}+\dd A^{0,1})
\right),
\]
where we recognize $A^{1,0}$ as the $p$ variables and $A^{0,1}$ as the $q$ variables. The variation now, correctly, produces the boundary one-form $\int_\Sigma A^{1,0}\delta A^{0,1}$.

We can finally apply the results of Example~\ref{e:flinear} to this case where, up to signs, the matrix $v^i_\alpha$ is now replaced by the operator $\bar\dd$ and the matrix $w_{i\beta}$ is replaced by the operator $\dd$. We write $\sigma:=\int_I A_I\in\Omega^0(\Sigma)$ instead of $T$, $\As^{0,1}_\ii$ instead of $q_\ii$ and $\As^{1,0}_\oo$ instead of $p^\oo$, so equation \eqref{e:Sfablin} becomes
\begin{equation}\label{S_HJ ahol-hol}
\Hat S^f_\text{HJ}=\int_\Sigma \left(\As^{1,0}_\oo \As^{0,1}_\ii + \sigma (\bar\dd  \As^{1,0}_\oo+\dd \As^{0,1}_\ii)+\frac12 \sigma \dd \bar\dd \sigma\right).
\end{equation}

Following Example~\ref{exa:pLegtrlin} we may also get the generating function for the same choice of both endpoint conditions. Namely, extending \eqref{e:newgengenfun} to this case, we get 
\begin{equation}\label{S_HJ hol-hol}
\Tilde S_\text{gen}=\int_\Sigma \left(\As^{1,0}_\oo \left(\lambda+\frac12 \bar\dd \sigma\right) - \As^{1,0}_\ii \left(\lambda-\frac12 \bar\dd \sigma\right) +\lambda \dd \sigma\right)
\end{equation}
{with $\lambda \in \Omega^{0,1}(\Sigma)$.} 

\subsection{Nonabelian Chern--Simons theory}\label{s:naCS}
In this case the action is
\[
S[
A]=\int_{I\times\Sigma}\left( \frac12 \langle 
{A},\ddd
{A}\rangle + 
\frac16  \langle 
{A},[
{A},
{A}]\rangle
\right),
\]
where now $
A$ takes values in a quadratic Lie algebra $(\frg,\langle\ ,\ \rangle)$.

Proceeding as in the abelian case, we get 
\[
S''
(A^{1,0},A^{0,1},A_I)=\int_{I\times\Sigma}\left(
\langle A^{1,0},\ddd_IA^{0,1}\rangle+ 
\langle A_I,\bar\dd A^{1,0}+\dd A^{0,1}+[A^{0,1},A^{1,0}]\rangle
\right),
\]
where we recognize the data of Example~\ref{exa:adjoint} and Remark~\ref{r:tensorW}
with 
$-A^{1,0}$ as the $\bar q$, $A^{0,1}$ as the $q$, $A_I$ as the $e$ variables,
$\dd$ as the derivation $\bar v$, and $\bar\dd$ as the derivation $v$.\footnote{
To be more precise, the Lie algebra of 
Remark~\ref{r:tensorW} is now $\Omega^0(\Sigma,\frg)$ with pointwise Lie bracket induced from that of $\frg$, the trivial representation space $Z$ is $\Omega^{0,1}(\Sigma)$, and their tensor product, over $\Omega^0(\Sigma)$, is
$V=\Omega^{0,1}(\Sigma,\frg)$.
The  dual space $V^*$ is now replaced by $\Omega^{1,0}(\Sigma,\frg)$ with pairing to $V$ induced by the paring $\langle\ ,\ \rangle$ on $\frg$ and integration on $\Sigma$.}
We write $\As^{0,1}_\ii$ instead of $q_\ii$ and $-\As^{1,0}_\oo$ instead of $\bar q^\oo$,
so \eqref{e:Sbiaffinead} becomes 
\[
\Hat S^f_\text{HJ} = \int_\Sigma \Big( \langle \As^{1,0}_\oo, g^{-1}\, \As^{0,1}_\ii g \rangle +
\langle \As^{1,0}_\oo,g^{-1}\bar\dd g \rangle 
+ 
\langle  \As^{0,1}_\ii, \dd g\cdot g^{-1}  \rangle \Big)
+\mr{WZW}(g)
\]
with $g=P\EE^{\int_I A_I}\in\Omega^0(\Sigma,G)$. For $g=\EE^\xi$, $\xi\in \Omega^0(\Sigma,\frg)$,
from \eqref{e:WZWad} we get 
\[
\mr{WZW}(g)=- \frac12 \int_\Sigma  \langle \dd g\cdot g^{-1}, \bar\dd g\cdot  g^{-1} \rangle -\frac{1}{12}  \int_{\Sigma \times I}\langle \ddd h\cdot h^{-1},[\ddd h\cdot h^{-1},\ddd h\cdot h^{-1}] \rangle,
\]
with $h=\EE^{t\xi}$. 
{Thus, the HJ action of Chern--Simons theory can be identified with a ``gauged WZW action'' (see for instance \cite{Gawedzki CFT 2}). This points at a deep relationship between these two theories. We will revisit this correspondence in more detail in \cite{CS_cyl}.
}

\subsection{Nonabelian $BF$ theory}
Consider
 $BF$ theory, defined by the action 
\begin{equation} \label{e-non-abe-BF}
S_{BF}=\int_M \langle B,F_A \rangle.
\end{equation}
Here $M$ is an $n$-manifold and the fields are a $\g$-valued $1$-form  $A$  (viewed as a connection with curvature $2$-form $F_A=\ddd A+\frac12[A,A]$) and a $\g^*$-valued $(n-2)$-form $B$;\footnote{
More generally, $A$ is a connection in a principal $G$-bundle $\mathcal{P}$ over $M$ and $B$ is an $(n-2)$-form valued in the coadjoint bundle of $\mathcal{P}$. In the context of a cylinder $M=I\times \Sigma$, we would take $\mathcal{P}=\pi^* \mathcal{P}_\Sigma$ where $\mathcal{P}_\Sigma$ is a $G$-bundle over $\Sigma$ and $\pi\colon I\times \Sigma\ra\Sigma$ is the projection. The results of this subsection and of subsection \ref{ss: 2dYM and EM} generalize immediately to this setup.
} $\g$ is the Lie algebra of some fixed Lie group $G$. 

On an $n$-dimensional cylinder $M=I\times \Sigma$, the action becomes
\begin{equation}\label{S_BF cyl}
 S_{BF}=\int_{I\times \Sigma} \langle B_\Sigma, \ddd_I A_\Sigma \rangle + \langle A_I, 
\ddd_{A_\Sigma} B_\Sigma
\rangle + \langle B_I, 
F_{A_\Sigma}
 \rangle, 
\end{equation}
where $\ddd_{A_\Sigma} B_\Sigma = \ddd B_\Sigma + [A_\Sigma,B_\Sigma]$ is the covariant exterior derivative. Here we expanded the fields according to form degree on $I$:  $A=A_\Sigma+A_I$, $B=B_\Sigma+B_I$. 
Note that $A_I$ and $B_I$ in (\ref{S_BF cyl}) are Lagrange multipliers corresponding to the constraints $\ddd_{A_\Sigma}B_\Sigma=0$ and $F_{A_\Sigma}=0$.

Using $A_\ii,B_\oo$ boundary conditions, the generalized HJ action reads
$$ \hat{S}_\mr{HJ}(A_\ii,B_\oo;g,\beta) = 
\int_\Sigma \langle \beta, F_{A_\ii} \rangle-(-1)^n \langle B_\oo,g^{-1}A_\ii g+g^{-1}d_\Sigma g \rangle. $$
Here the auxiliary fields are $g=P\EE^{\int_{t_\ii}^{t_\oo} A_I} \in \mr{Map}(\Sigma,G)$  and $\beta = \int_{t_\ii}^{t_\oo} \mr{Ad}^*_{P\EE^{\int_{t_\ii}^t A_I}} (B_I) \in \Omega^{n-3}(\Sigma,\g^*)$. The case of $(A_\ii,A_\oo)$ boundary conditions is treated  by the partial Legendre transform: 
$$ S_\mr{gen}(A_\ii,A_\oo;g,\beta,\lambda)=\int_\Sigma \langle \beta, F_{A_\ii} \rangle-(-1)^n \langle \lambda,g^{-1}A_\ii g+g^{-1}d_\Sigma g \rangle + (-1)^n\langle \lambda,A_\oo \rangle. $$

A similar expression appeared in \cite{Hol} as a holographic dual of $BF$ theory (in the three-dimensional case and with a different---holomorphic---polarization on the boundary).

\begin{remark}[Relation to Group Quantum Mechanics]
Specializing to $n = 2$, and assuming $\Sigma$ to be connected, we have $\Sigma = S^1$. On $\Omega^\bullet(S^1,\g)$ we can consider a set of canonical coordinates given by $Q = B$, $P = A - *_{S^1}B$, where $*_{S^1}$ denotes the Hogde star for some choice of Riemannian metric on $S^1$. We choose the coordinates $q = A, p=B$ at $t_\ii$ and $Q,P$ at $t_\oo$. The corresponding generating function is $f(A,B) = \int_{S^1} \langle B,A\rangle  - \frac12 \langle B , *_{S^1}B \rangle$ and the HJ action is given by 
\begin{equation}
\Hat{S}^f_{\mr{HJ}}(A_a,B_b,g) = - \int_{S^1} \langle B_b, g^{-1}A_ag + g^{-1}d_{S^1}g \rangle - \frac12 \langle B_b, *_{S^1} B_b\rangle. 
\end{equation}
One may further specialize by choosing the boundary condition $A_a = 0$ and imposing the equation of motion $B_b = -*_{S^1}g^{-1}dg$ for $B_b$. The action then reduces to 
\begin{equation} 
S_{GQM}(g) = \frac12 \int_{ S^1} \langle g^{-1}dg, *_{S^1}g^{-1}dg\rangle, 
\end{equation}
the action for a free particle moving in the group $G$, studied for instance in \cite{Gawedzki CFT 1}.
We thus recover a relation between a 2D theory and group quantum mechanics
which could be seen as an instance of holographic correspondence. In the case $G=PSL(2,\RR)$ (corresponding to a model of 2D gravity whose classical solutions are constant negative curvature geometries) it was discussed
in the literature on $\mr{AdS}_2\mr{/CFT}_1$ correspondence. In \cite{VY} a very interesting subsequent reduction from $S_{GQM}$ was considered, in the case $G=PSL(2,\RR)$, leading to the Schwarzian 1D theory.
\end{remark}

\subsection{More examples: 2D Yang--Mills theory 
and electrodynamics in general dimension}
\label{ss: 2dYM and EM}
In this section we discuss two examples of systems with nontrivial evolution and constraints: two special cases of Yang--Mills theory on a cylinder. They provide (generalized) examples of the setup of Section~\ref{r:nontrivev}.

Consider the first-order formulation of Yang--Mills theory,
\begin{equation}\label{S_YM} 
S_\mr{YM}[A,B] = \int_M \langle B,  F_A \rangle - \frac12 \langle B, *B \rangle,  
\end{equation}
where $M$ is a pseudo-Riemannian $n$-manifold which we will take to be a cylinder $M=I\times \Sigma$ with product metric $g=-(\ddd t)^2+g_\Sigma$ and $g_\Sigma$ a Riemannian metric on $\Sigma$; here $*$ is the associated Hodge star operator. 
Note that the action (\ref{S_YM}) is an extension of the $BF$ action (\ref{e-non-abe-BF}) by a metric-dependent term.

Expanding the fields on the cylinder according to form degree along $I$ as $A=A_\Sigma+A_I$, $B=B_\Sigma+B_I$, we can write the action (\ref{S_YM}) as
\begin{equation} \label{S_YM cylinder}
S_\mr{YM}=\int_{I\times \Sigma} \langle  B_\Sigma,  \ddd_I A_\Sigma-\frac12 *B_\Sigma \rangle 
+ \langle B_I,F_{A_\Sigma}-\frac12 *B_I \rangle + \langle A_I,\ddd_{A_\Sigma} B_\Sigma \rangle.
\end{equation}
The EL equations read
\begin{subequations}
\begin{eqnarray}
\ddd_I A_\Sigma+\ddd_{A_\Sigma}  A_I-*B_\Sigma &=& 0 \label{YM eq1}, \\ 
\ddd_I B_\Sigma + [A_I,B_\Sigma]+\ddd_{A_\Sigma}B_I &=& 0 \label{YM eq2},\\
F_{A_\Sigma}-*B_I&=&0 \label{YM eq3} ,\\
\ddd_{A_\Sigma} B_\Sigma &=& 0 \label{YM eq4}.
\end{eqnarray}
\end{subequations}
Here $A_I$ is a Lagrange multiplier and (\ref{YM eq4}) is the corresponding constraint; $B_I$ is not a Lagrange multiplier (enters the action quadratically) and although the corresponding equation (\ref{YM eq3}) does not contain time derivatives, it does not impose constraints on the fields in the phase space (which $B_I$ is not a part of). Therefore, (\ref{YM eq3}) is treated as part of the evolution relations, alongside (\ref{YM eq1}) and (\ref{YM eq2}), for the purpose of computing the HJ action.

Introducing a group-valued variable $h(t)=P\EE^{\int_{t_\ii}^t A_I} \in \mr{Map}(I\times\Sigma,G)$ and performing a change of variables
$$ A_\Sigma=h^{-1} \til A_\Sigma h+h^{-1}\ddd_\Sigma h,\quad  B_\Sigma=h^{-1}\til B_\Sigma h ,\quad B_I=h^{-1} \til B_I h$$
we can rewrite the evolution equations as
\begin{gather}
\dd_t^2 \til A_\Sigma+(-1)^n *_\Sigma \ddd_{\til A_\Sigma}*_\Sigma F_{\til A_\Sigma} = 0, \label{YM eq (3')} \\
\til B_I =* F_{\til A_\Sigma}  ,\quad
\til B_\Sigma =*_\Sigma \dd_t \til A_\Sigma.   \nonumber
\end{gather} 
Thus, essentially, one has to solve the  (generally very complicated) nonlinear equation (\ref{YM eq (3')}) and recover the remaining fields using the other two equations. Equation (\ref{YM eq (3')})  becomes linear in two cases: the case $n=2$ (as the second term in (\ref{YM eq (3')}) vanishes by a degree reason) and the abelian case.

\subsubsection{2D Yang--Mills theory.} \label{sss: 2dYM} In the two-dimensional case  (with $\Sigma=S^1$), 
we have $B_I=0$ for a degree reason (and thus $B=B_\Sigma$), so we are exactly in the setting of Section~\ref{r:nontrivev}. The action (\ref{S_YM cylinder}) on $I\times S^1$ becomes
$$S_\mr{YM}=\int_{I\times S^1} \langle B,\ddd_IA_\Sigma-\frac12 *B \rangle + \langle A_I,
\ddd_{A_\Sigma} B
 \rangle. $$
In terms of the boundary conditions $(A_\ii, B_\oo)$, we obtain the following generalized HJ action:
$$ \hat{S}^f_\mr{HJ}(A_\ii,B_\oo;g)= \int_\Sigma -\langle B_\oo, g^{-1} A_\ii g+ g^{-1} d_\Sigma g \rangle - \frac{\tau}{2}\langle B_\oo,*_\Sigma B_\oo \rangle, $$
where $\tau=t_\oo-t_\ii$. 
Here $g=h(t_\oo)=P\EE^{\int_{t_\ii}^{t_\oo} A_I} \in \mr{Map}(\Sigma,G)$, the group-valued auxiliary variable for the generalized HJ action.  We can also consider $(A_\ii,A_\oo)$ boundary conditions, and then the answer is
\begin{equation*} \label{2dYM HJ A-A}
S_\mr{gen}(A_\ii,A_\oo;g,\lambda)= \int_\Sigma -\langle \lambda, g^{-1} A_\ii g+ g^{-1} d_\Sigma g \rangle - \frac{\tau}{2}\langle \lambda,*_\Sigma \lambda \rangle+\langle \lambda,A_\oo \rangle.
\end{equation*}
Here $\lambda\in \Omega^1(\Sigma,\g^*)$ is the new auxiliary field.

\subsubsection{Electrodynamics (general dimension).} Next consider Yang--Mills theory (\ref{S_YM})  on a cylinder of general dimension in the abelian case, $G=\RR$. 
 Solving the evolution equations and plugging the result into the action, we find, for boundary conditions $(A_\ii$, $B_\oo)$, the following generalized HJ action:\footnote{
In particular, equation (\ref{YM eq (3')}) reads $\dd_t^2 \til A_\Sigma+(-1)^n *_\Sigma \ddd_\Sigma *_\Sigma \ddd_\Sigma \til A_\Sigma=0$ (where the change of variables is $A_\Sigma= \til A_\Sigma+\ddd_\Sigma \int_{t_\ii}^t A_I$). It breaks into two equations for the closed and coexact parts of $\til A_\Sigma$: $\dd_t^2 \til A_\Sigma^\mr{cl}=0$ and $\dd_t^2 \til A_\Sigma^\mr{coex}=-\Delta_\Sigma A_\Sigma^\mr{coex}$. The solution of the boundary problem is then $\til A_\Sigma^\mr{cl}=A_\ii^\mr{cl}+(-1)^n  *_\Sigma B_\oo^\mr{cocl}\, t   
$,\;
 $ \til A_\Sigma^\mr{coex}=
\sum_\sigma \chi_\sigma \Big(\frac{\sin \omega_\sigma t}{\omega_\sigma \cos \omega_\sigma (t_\oo-t_\ii)}B_\oo^\sigma +\frac{\cos \omega_\sigma (t_\oo-t)}{\cos \omega_\sigma (t_\oo-t_\ii)} A_\ii^\sigma\Big)$. 
 }
\begin{multline}\label{EM S_HJ}
\hat{S}_\mr{HJ}(A_\ii,B_\oo;\gamma)=\int_\Sigma \left(-\frac{\tau}{2} B_\oo^\mr{cocl}\wedge *_\Sigma B_\oo^\mr{cocl}-A^\mr{cl}_\ii\wedge B^\mr{cocl}_\oo-\ddd_\Sigma \gamma\wedge B_\oo\right)+\\
+
\sum_\sigma \left(-\frac12 \frac{\tan \omega_\sigma \tau}{\omega_\sigma} (B_\oo^\sigma)^2-\frac12 \omega_\sigma \tan \omega_\sigma\tau\, (A_\ii^\sigma)^2-\frac{1}{\cos\omega_\sigma\tau} A_\ii^\sigma B_\oo^\sigma \right).
\end{multline}
Here:
\begin{itemize}
\item $\tau=t_\oo-t_\ii$. 
\item $\gamma = \int_{t_\ii}^{t_\oo} A_I\in C^\infty(\Sigma)$ is an auxiliary field.
\item We introduce $\{\chi_\sigma\}$, an orthonormal basis of coexact $1$-forms on $\Sigma$ which diagonalize the Laplacian: $\Delta_\Sigma\chi_\sigma=\omega_\sigma^2 \chi_\sigma$. We split the field $A_\ii$, according to the Hodge decomposition, into its closed and coexact parts:
$A_\ii=A_\ii^\mr{cl}+\sum_\sigma \chi_\sigma\, A_\ii^\sigma $.
 \item Similarly, $\{*_\Sigma\chi_\sigma\}$ is an orthonormal basis of exact $(n-2)$-forms on $\Sigma$ and we have a splitting $B_\oo=B_\oo^\mr{cocl}+\sum_\Sigma *_\Sigma \chi_\sigma\, B_\oo^\sigma$ into coclosed and exact parts.
\end{itemize}
Note that the term $\int_\Sigma d_\Sigma\gamma \wedge B_\oo$ in $\hat{S}_\mr{HJ}$ depends in fact only on the coexact part of $B_\oo$.

{
We also remark that the summand in the sum over $\sigma$ in (\ref{EM S_HJ}) is the HJ action of the harmonic oscillator with frequency $\omega_\sigma$. Here the oscillator corresponds to the field $A^\sigma(t)$ (the coefficient of the expansion of $A_\Sigma^\mr{coex}$ in the basis  $\{\chi_\sigma\}$) with corresponding momentum $B^\sigma(t)$ (the coefficient of the expansion of $B_\Sigma^\mr{ex}$ in the basis  $\{*_\Sigma\chi_\sigma\}$).
}

\begin{remark}[$p$-form electrodynamics]
The  case of abelian Yang--Mills theory generalizes straightforwardly to ``$p$-form electrodynamics''---the theory defined by the action 
$${S_{p\mr{-form}}=\int_M B\wedge \ddd A - \frac12 B\wedge *B},$$ where the fields are $A\in \Omega^p(M)$, ${B\in \Omega^{n-p-1}(M) }$.
The generalized HJ action on a cylinder is 
given by an expression identical to (\ref{EM S_HJ}) up to signs:
\begin{multline*}
\hat{S}_\mr{HJ}(A_\ii,B_\oo;\gamma)=\int_\Sigma \left(-\frac{\tau}{2} B_\oo^\mr{cocl}\wedge *_\Sigma B_\oo^\mr{cocl}+\epsilon A^\mr{cl}_\ii\wedge B^\mr{cocl}_\oo+\epsilon\, \ddd_\Sigma \gamma\wedge B_\oo\right)+\\
+
\sum_\sigma \left(-\frac12 \frac{\tan \omega_\sigma \tau}{\omega_\sigma} (B_\oo^\sigma)^2-\frac12 \omega_\sigma \tan \omega_\sigma\tau\, (A_\ii^\sigma)^2+\frac{\epsilon}{\cos\omega_\sigma\tau} A_\ii^\sigma B_\oo^\sigma \right).
\end{multline*}
Here $\epsilon=(-1)^{np+n+p}$; $A_\ii^\sigma$ are the coefficients of the expansion on $A_\ii^\mr{coex}$ in the orthonormal basis $\chi_\sigma$ of coexact $p$-forms on $\Sigma$; $B_\oo^\sigma$ are the coefficients of the expansion of $B_\oo^\mr{ex}$ in the orthonormal basis $*_\Sigma \chi_\sigma$ of exact $(n-p-1)$-forms on $\Sigma$; the auxiliary field $\gamma$ is a $(p-1)$-form on $\Sigma$.  Note that the case $p=0$ corresponds to the massless scalar theory which has no constraints (and the field $\gamma$ disappears for a degree reason).
\end{remark}

\subsection{Higher-dimensional Chern-Simons theory}
The results of Section \ref{s:3dabCS} are easily generalized to higher-dimensional abelian Chern--Simons theories. 
Consider again a cylinder $I\times M$, where $M$ is now a $4k+2$-dimensional manifold ($k=0,1,2,\ldots)$ with a complex structure and
\begin{equation}
S[A] =\frac12 \int_{I\times M}  A\ddd A,
\end{equation}
for $A\in \Omega^{2k+1}(M)$. Proceeding as above, we obtain $A = A_M + A_I$ with $A_M$ a map $I \to \Omega^{2k+1}(M)$ and $A_I$ a one-form on $I$ with values in $\Omega^{2k}(M)$, and we have
\begin{equation}
S'[A_M,A_I] = \int_{I \times M}\frac12 A_M \ddd_IA_M + A_I\ddd_M A_M.\label{eq:S'hdCS}
\end{equation}
To obtain a splitting of the symplectic space of fields $\Omega^{2k+1}(M)$, we pass to the complexification and employ a decomposition by complex bidegree, namely 
\begin{equation}
\mathcal F^\dd_\CC = \Omega^{2k+1}(M,\CC) = \underbrace{\Omega^{2k+1,0}(M) \oplus \cdots \oplus \Omega^{k+1,k}(M)}_{\Omega^{2k+1}_+(M)} \oplus \underbrace{\Omega^{k,k+1}(M)\oplus \cdots \oplus \Omega^{0,2k+1}(M)}_{\Omega^{2k+1}_-(M)}
\end{equation}
and accordingly we write $A_M = A_M^+ + A_M^-$. Plugging this decomposition into \eqref{eq:S'hdCS},  we obtain, up to a boundary term,
\begin{equation}S''(A_M^+,A_M^-,A_I) = \int_{I \times M} A_M^-d_IA_M^+ + A_I(d_MA_M^+ + d_MA_M^-)
\end{equation}
{with corresponding boundary 1-form $A_M^-\delta A_M^+$. We will choose the endpoint conditions $A_M^+$ on 
$M \times \{0\}$ and $A_M^-$ on $M \times \{1\}$.}\footnote{The attentive reader will have noticed that for the 3D case $k=0$ this gives the opposite boundary conditions than in Section~\ref{s:3dabCS}. This flips the sign of the last term in  \eqref{eq:S_HJ ahol-hol HD} below.}
Notice that for the component of $A_I$ of Hodge type $(k,k)$, the last term reads $A_I^{k,k}(\bar\dd A_M^{k+1,k} + \dd A_M^{k,k+1})$. This gives rise to the only nonvanishing ``component'' of the matrix $A_{\alpha\beta}$ of Example \ref{e:flinear}, {the other components vanishing because of $\ddd_M^2 = 0$}. Relabeling $\int A_I =:\sigma \in \Omega^{2k}(M)$ we obtain the HJ action 
\begin{equation}\label{eq:S_HJ ahol-hol HD}
\Hat S^f_\text{HJ}=\int_M \As^-_\oo \As^+_\ii + \sigma (d_M \As^+_\ii+ d_M\As^-_\oo) - \frac12 \sigma^{k,k} \dd \bar\dd \sigma^{k,k}.\end{equation}
In applications it is interesting to consider more general coordinates on the space $\Omega^{2k+1}(M)$. Assume that $(A^Q,A^P)$ is such a set of coordinates and that $f(A^+_M,A^Q)$ is the corresponding generating function. Then, according to Remark \ref{r:HJgenchange}, the HJ action is
\begin{equation}
\Hat S^f_\text{HJ} = f(\As^+_\ii + d^+_M\sigma,\As^Q_\oo) + \int_M \sigma d_M\As^+_\ii - \int_M\frac12 \sigma^{k,k} \dd \bar\dd \sigma^{k,k},
\end{equation}
where $\ddd_M^+$ is the composition of $\ddd_M$ with the projection to $\Omega^{2k+1}_+(M)$, and the sign in front of the last term is due to using the $+$-representation on the $\ii$-boundary.

\subsection{A nonlinear polarization in 7D Chern--Simons theory and the Kodaira--Spencer action}\label{s:7d}
Let us now focus on the 7D case ($k=1$ in the section above): i.e., $M$ is now a 6-dimensional manifold with a complex structure. We will now actually assume that it is K\"ahler, 
{and fix a reference nonvanishing holomorphic 3-form $\omega_0$.} Let us fix some notation. We denote by 
 $\Omega^{-p,q}(M)$ sections of the bundle $\bigwedge^p (T_\CC M)^{1,0} \wedge \bigwedge^q(T^*_\CC M)^{0,1}$, i.e., $(0,q)$-forms with values in $(p,0)$-vector fields. Contraction with the reference holomorphic 3-form provides a map 
\begin{align*}
\Omega^{-p,q}(M) &\to \Omega^{3-p,q}(M) \\ 
A &\mapsto A^\vee =: A\omega_0
\end{align*}
(we omit symbols for wedge products and contractions). For $A,B,C \in \Omega^{-1,1}(M)$, we further define the operations
\begin{subequations}
\begin{align}
A^\vee \circ B^\vee &= (AB)\omega_0 \in \Omega^{1,2}(M), \\
\langle A^\vee,B^\vee,C^\vee\rangle &= A^\vee(B^\vee\circ C^\vee)\in \Omega^{3,3}(M), \\
\langle A^3\rangle &= {\frac16} \frac{\langle A^\vee,A^\vee,A^\vee\rangle}{\omega_0\overline{\omega}_0}\in\Omega^0(M),\label{eq:mucube}
\end{align} 
\end{subequations}
and similarly for $A,B,C \in \Omega^{1,-1}(M)$ by replacing $\omega_0$ with its conjugate $\bar\omega_0$.  

The symplectic space of boundary fields is $\mathcal F^\dd = \Omega^3(M)_\CC$. The geometry of this space was studied by Hitchin in \cite{Hitchin}. The main point of interest for us is that any complex 3-form admits a splitting\footnote{{The superscript $\nl$  stands for ``nonlinear'' and indicates that the map $A_M \mapsto (A_M^{+,\nl}, A_M^{-,\nl}$) is nonlinear, as opposed to $A_M \mapsto (A_M^+,A_M^-)$}, which is linear.} $A_M = A_M^{+,\nl} + A_M^{-,\nl}$ where $A_M^{\pm,\nl}$ are decomposable three-forms.  Here, a 3-form $A$ is called decomposable if around every point there exists a local  coframe $\theta_i$ in which $A = \theta_1 \wedge \theta_2 \wedge \theta_3$. We will call $A$ nondegenerate if $A_{M}^{+,\nl} \wedge A_M^{-,\nl}$ is everywhere nonvanishing. All three-forms which are not themselves decomposable are nondegenerate and for nondegenerate 3-forms $A^{\pm,\nl}$ are uniquely determined by $A$ \cite{Hitchin}. {This splitting defines a polarization $\mathcal{P}^{\nl,-}$ on\footnote{Strictly speaking, the polarization is defined only on the open subset of nondegenerate forms.} $\Omega^3(M)_\CC$ by letting $\PP_{A_M} \subset T_{A_M}\Omega^3(M)_\CC \cong \Omega^3(M)_\CC$ be the subspace  spanned by the $A_M^{+,\nl}$, whose leaf space is parametrized by $A^{-,\nl}_M$.}  To write down a generating function $f(A_M^+,A_M^{-,\nl})$, 
we parametrize $A_M^{-,\nl}$ by 
\begin{equation}
A^{-,\nl} = \overline{\rho} e^{\overline{\mu}}\overline{\omega}_0 = \overline{\rho}\left(\overline{\omega_0} + \overline{\mu}\, \overline{\omega}_0 + \frac{\overline\mu^2}{2}\overline{\omega}_0 + \frac{\overline\mu^3}{6}\overline\omega_0 \right),
\end{equation}
where $\ol{\rho} \in \Omega^0_\CC(M)$ and $\ol{\mu}\in \Omega^{1,-1}(M)$, and $\mu\omega_0$ should be interpreted as extension of contraction to forms with values in vector fields. An expression for the generating function is then \cite{GS} 
\begin{multline}
f(A^{3,0},A^{2,1},\overline{\rho},\overline{\mu}) =\\
=\int_M \overline{\rho}( A^{3,0}\overline{\omega}_0 + A^{2,1}\overline{\mu}\,\overline{\omega}_0) + \overline{\rho}^2\langle\overline{\mu}^3\rangle\omega_0\overline{\omega}_0 -
\frac{\left\langle\left( (A^{2,1} - \frac12\overline{\rho}\,\overline{\mu}^2\overline{\omega}_0)^\vee\right)^3\right\rangle}{(A^{3,0})^\vee - \overline{\rho}\langle\overline{\mu}^3\rangle}\omega_0\overline{\omega}_0.
\end{multline}
The HJ action is then 
\begin{multline}
\Hat S^f_\text{HJ} =-\frac12\int_M\sigma^{1,1}\dd\bar\dd\sigma^{1,1} +\int_M \sigma^{1,1}\bar\dd \As_\ii^{2,1} + \sigma^{0,2}(\dd\As_\ii^{2,1} + \bar\dd\As^{3,0}_\ii) \\
+ f\left(\As^{3,0}_\ii + \dd \sigma^{2,0}, \As^{2,1}_\ii + \dd\sigma^{1,1} + \bar\dd \sigma^{2,0}, \ol{\rho},\ol{\mu}\right).
\end{multline}
To see the connection to the Kodaira--Spencer action functional introduced in \cite{BCOV}, take the following quantum state considered in \cite{GS}: 
\begin{equation}\label{e:KS}
\psi_\text{GS}(\ol\rho,\ol\mu) = \delta(\ol\mu)\exp\frac{i}{\hbar}\int_M\ol\rho\omega_0\ol\omega_0.
\end{equation}
$\psi_\text{GS}$ is the quantization of the lagrangian $L_\text{GS} \subset \Omega_\CC^3(M)$ given by $\ol\mu = 0, p_{\ol\rho} = \omega_0\ol\omega_0$ (here $p_{\ol\rho}$ is the canonical momentum coordinate associated to $\ol\rho$).   $L_\text{GS}$ has the following generalized generating function, with $\lambda \in \Omega^{-1,1}(M)$ an auxiliary parameter:
\begin{equation}
S_\text{GS}(\ol\mu,\ol\rho;\lambda) = \int_M (\lambda \ol \mu + \ol\rho) \omega_0\ol\omega_0 \label{eq:SGS}.
\end{equation}
Let $\mr{ev} \colon \Omega_\CC^3(M) \nrightarrow \Omega_\CC^3(M)$ denote\footnote{We use the notations of equation \eqref{e:nrightarrow} of Appendix~\ref{a:canrel} for canonical relations.}  the Chern--Simons evolution relation and interpret $L_\text{GS}$ as a canonical relation $\{pt\} \nrightarrow \Omega^3_\CC(M)$. We want to compute the composition\footnote{The relation $\mr{ev}$ goes from $A^{+,\mr l}$ to $A^{-,\nl}$, in the composition we require the transpose relation $\mr{ev}^T$. This changes the sign of the generating function: the generalized generating function of $\mr{ev}^T$ is $-\Hat S^f_\text{HJ}$.} $L := \mr{ev}^T \circ L_\text{GS} \colon \{pt\} \nrightarrow \Omega^3_\CC(M)$. This, following Remark~\ref{r:compgenfunII}, has the generalized generating function $
S = - \Hat S^f_\text{HJ} + S_\text{GS} $, where $\ol\mu,\ol\rho$ now become auxiliary fields. We can immediately solve the constraint $\ol\mu = 0$ to obtain a new generalized generating function for $L$,:
\begin{multline}
S' = \left[- \Hat S^f_\text{HJ} + S_\text{GS}\right]_{\ol\mu =0} = 
\frac12\int_M\sigma^{1,1}\dd\bar\dd\sigma^{1,1} - \int_M \sigma^{1,1}\bar\dd \As_\ii^{2,1} - \int_M \sigma^{0,2}(\dd\As_\ii^{2,1} + \bar\dd\As^{3,0}_\ii)   \\ - \int_M \ol\rho_\oo (\As^{3,0}_\ii + \dd \sigma^{2,0}-\omega_0) \ol\omega_0 + \int_M \frac{\langle ((\As^{2,1}_\ii + \dd \sigma^{1,1} + \bar\dd \sigma^{2,0})^\vee)^3\rangle}{(\As^{3,0}_\ii + \dd\sigma^{2,0})^\vee}\omega_0\ol\omega_0. 
\end{multline}

We now proceed with solving the constraint equations by setting derivatives with respect to auxiliary fields to zero. The equation $\delta S'/\delta \ol\rho = 0$ is
\begin{equation}
-\frac{\delta S' }{\delta\ol\rho} = \As^{3,0}_\ii + \dd\sigma^{2,0} -\omega_0 = 0\label{eq:A30eq}
\end{equation} 
which implies $\bar\dd\As^{3,0}_\ii = -\bar\dd\dd\sigma^{2,0}.$
The equation $\delta S'/\delta \sigma^{0,2} = 0$ then gives
\begin{equation}
-\frac{\delta S'}{\delta\sigma^{0,2}} = \dd\As^{2,1}_\ii +\bar\dd\As^{3,0}_\ii= \dd(\As^{2,1}_\ii + \bar\dd \sigma^{2,0}) = 0. \label{eq:ddA21equalzero}
\end{equation}   
Evaluating $S'$ on the set given by those constraints we obtain 
$$S'' = \frac12\int_M\sigma^{1,1}\dd\bar\dd\sigma^{1,1} - \int_M \sigma^{1,1}\bar\dd \As_\ii^{2,1} + \int_M \langle ((\As^{2,1}_\ii + \dd \sigma^{1,1} + \bar\dd \sigma^{2,0})^\vee)^3\rangle\omega_0\ol\omega_0.$$ 
By \eqref{eq:ddA21equalzero} and Hodge decomposition for $\dd$, there are a unique harmonic $(2,1)$-form $x$ and a $(1,1)$-form $\beta$ such that $x + \dd\beta = \As_\ii^{2,1} + \bar\dd \sigma^{2,0}$. Letting $b = \sigma^{1,1} + \beta$, we have $\bar\dd \As^{2,1}_\ii =  \bar\dd\dd\beta = \dd\bar\dd (\sigma^{1,1} -b)$. Plugging this into $S''$, we obtain
\begin{equation}
S''' = \frac12 \int_M \sigma^{1,1}\dd\bar\dd \sigma^{1,1} - \int_M \sigma^{1,1} \dd\bar\dd \sigma^{1,1} + \int_M\sigma^{1,1}\dd\bar\dd b + \frac16 \langle x + \dd b, x + \dd b, x + \dd b\rangle.
\end{equation}
After imposing the constraint $$\frac{\delta S''' }{\delta \sigma^{1,1}} = 0 = -\dd\bar\dd \sigma^{1,1} + \dd\bar\dd b,$$ the action becomes 
\begin{equation}
S_\text{KS}(x,b) = \frac12\int_M b \dd\bar\dd b + \frac16 \langle x + \dd b, x + \dd b, x + \dd b\rangle,
\end{equation}
the action functional for Kodaira--Spencer gravity 
defined in \cite{BCOV}. 
For any $A \in\Omega^{2,1}(M)$ we have 
$$\frac{\delta}{\delta A} \langle (A^\vee)^3\rangle\omega_0\ol\omega_0 = \frac12 A \circ A \in \Omega^{1,2}.$$
Thus, the equation for $b$ finally becomes  
\begin{equation}
\dd\bar\dd b + \frac12 \dd (x+\dd b)\circ(x+ \dd b) = 0,\label{eq:KS2}
\end{equation}
the Kodaira--Spencer equation for $A = x + \dd b$.

%

Finally, we can characterize the lagrangian $L$ as follows: 
\begin{proposition}\label{prop:lagrangian7dCS}
Let $\omega_0$ be a reference holomorphic 3-form on a compact Calabi--Yau manifold $M$. Then the composition $ L = \mr{ev}^T \circ L_\text{GS}\subset \Omega^3_\CC(M)$ of the Chern--Simons evolution relation  $\mr{ev}$ on $M \times I$ and the lagrangian $L_\text{GS}$ defined by  \eqref{eq:SGS} is given by complex 3-forms with decomposition $ A = A^{3,0}  + A^{2,1} + A^{1,2} + A^{0,3}$ satisfying 
\begin{subequations}
\begin{align}
A^{3,0} &= \omega_0 - \dd \sigma^{2,0} \label{eq:prop1}, \\
A^{2,1} &= x + \dd b - \bar\dd\sigma^{2,0} - \dd\sigma^{1,1} \label{eq:prop2}, \\
A^{1,2} &=  \frac12 (x + \dd b)\circ (x + \dd b) -\bar\dd\sigma^{1,1} - \dd\sigma^{0,2} \label{eq:prop3},\\
A^{0,3} &= c \ol\omega_0 -\bar\dd\sigma^{0,2} \label{eq:prop4},
\end{align}
\end{subequations}
where $x$ is a harmonic $(2,1)$ form and $b\in \Omega^{1,1}$ such that $x + \dd b $ satisfies the Kodaira--Spencer equation in the form \eqref{eq:KS2}, $c \in \CC$ is a constant and $\sigma = \sigma^{2,0} +  \sigma^{1,1} + \sigma^{0,2} \in \Omega^2_\CC(M)$ is any 2-form.
\end{proposition}
\begin{proof}
Equation \eqref{eq:prop1} is \eqref{eq:A30eq}. Equation \eqref{eq:prop2} follows from \eqref{eq:ddA21equalzero} and the definition of $b$. Equation \eqref{eq:prop3} follows from $A^{1,2} = \delta S' / \delta A^{2,1}$ together with the constraint equation \eqref{eq:A30eq}. Finally, we have 
$$A^{0,3} = \frac{\delta S'}{\delta A^{3,0}} = - \bar\dd \sigma^{0,2} -\ol\rho\, \ol\omega_0 - \frac{\langle ((\As^{2,1}_\ii + \dd \sigma^{1,1} - \bar\dd \sigma^{2,0})^\vee)^3\rangle}{((\As^{3,0}_\ii + \dd\sigma^{2,0})^\vee)^2}\ol\omega_0, $$
and again the denominator is 1 by \eqref{eq:A30eq}. Now, the equation 
\begin{multline*}0 = \frac{\delta S'}{\delta \sigma^{2,0}} = \dd\ol\rho\,\ol\omega_0 + \dd\frac{\langle ((\As^{2,1}_\ii + \dd \sigma^{1,1} - \bar\dd \sigma^{2,0})^\vee)^3\rangle}{((\As^{3,0}_\ii + \dd\sigma^{2,0})^\vee)^2}\ol\omega_0\\ +  \bar\dd\frac{ (\As^{2,1}_\ii + \dd \sigma^{1,1} - \bar\dd \sigma^{2,0})\circ (\As^{2,1}_\ii + \dd \sigma^{1,1} - \bar\dd \sigma^{2,0})}{((\As^{3,0}_\ii + \dd\sigma^{2,0}))^2}\ol\omega_0\end{multline*}
implies, again using \eqref{eq:A30eq},
$$\bar\dd\dd (\ol\rho + \langle ((\As^{2,1}_\ii + \dd \sigma^{1,1} - \bar\dd \sigma^{2,0})^\vee)^3\rangle) \ol\omega_0 = 0,$$
which means that the expression in brackets is constant. 
\end{proof}
\begin{remark}
It follows from the general discussion of Appendix \ref{a:generatingfunctions} that $L$ is isotropic, but one can also check it directly. To this end, compute the restriction of the symplectic form $\int_M \delta A^{3,0} \wedge \delta A^{0,3} + \delta A^{2,1} \wedge \delta A^{1,2}$ to $L$. Denoting $B = x + \dd b$, we get 
\begin{align*}
\delta A^{3,0} \wedge \delta A^{0,3} &= \int_M\dd\bar\dd\delta\sigma^{2,0}\delta\sigma^{0,2} \\
\delta A^{2,1} \wedge \delta A^{1,2} &= \int_M \delta\sigma^{1,1} \wedge \delta (\bar\dd B + \frac12 B \circ B ) + \int_M \bar\dd\dd\delta \sigma^{2,0}\delta\sigma^{0,2}
\end{align*}
(other terms cancel due to $\bar\dd^2 = \dd^2 =0$ or integration by parts). 
The terms containing $\sigma^{0,2}\sigma^{2,0}$ cancel, while the remaining term is zero by \eqref{eq:KS2}.
\end{remark}
\begin{remark}
Given a harmonic $(2,1)$-form $x$, solutions $b \in \Omega^{1,1}(M)$ of the Kodaira--Spencer equation are determined only up to a $\dd$-exact term. However, for fixed $x$ one can find a unique solution $b_x$ by demanding that $\bar\dd^* b_x = 0$, see \cite{Tian} or \cite{BCOV}, where $\dd^*$ denotes the formal adjoint of $\dd$ with respect to the K\"ahler metric.  Let $\underline{L} = \{(\omega_0, x + \dd b_x, \frac12 (x + \dd b_x) \circ (x + \dd b_x), c\ol\omega_0)\}$ with notation as in Propostion \ref{prop:lagrangian7dCS}.
We can then rewrite Proposition \ref{prop:lagrangian7dCS} as saying that 
\begin{equation} L = \underline{L} + d \Omega^2_\CC(M) \end{equation}
($\underline{L}$ is a submanifold of $\Omega^3_\CC(M)$ but not a subspace). 
At any point $(y, d\sigma) \in L$, the tangent space is $T_y\underline{L} \oplus d\Omega_\CC^2(M)$. For $y = (\omega_0,x + \dd b_x + \frac12(x + \dd b_x) \circ (x + \dd b_x), c\ol\omega_0)$, the tangent space is 
$$T_y\underline{L} =\{ (0, \dot{x}, x \circ \dot{x}, \dot{c}), \dot{x}\in H^{2,1}(M), \dot{c} \in \CC\}.$$ Thus, $L$ is actually  split lagrangian, with isotropic complement given by $L' = H^{3,0}(M) \oplus H^{1,2}(M) \oplus d^*\Omega^4_\CC(M).$ \end{remark}

\section{BFV, AKSZ and BV}\label{s:BFVAKSZBBV}
The degenerate action {(\ref{e:Ssevconstr}) we have been focusing on so far} 
can be put into BV form, which is suitable for the perturbative quantization (which we will consider in the next sections).

The first step consists in resolving the constraint
manifold $C=\{(p,q)\ |\ H_\alpha(p,q)=0\ \forall\alpha\}$
in terms of the BFV formalism \cite{BF,FV}. For this we introduce new odd variables $c^\alpha$, one for each constraint,
of ghost number $1$ and their momenta $b_\alpha$, odd and of ghost number $-1$;\footnote{Here odd means a Grassmann variable, whereas even refers to an ordinary variable; ghost number is an additional degree which is helpful for bookkeeping. The original variables $p$ and $q$ are even and are assigned ghost number zero.} i.e., we extend the symplectic structure to the BFV symplectic structure
\[
\omega_\text{BFV}=\ddd p_i\ddd q^i + \ddd b_\alpha\ddd c^\alpha.
\]
The BFV action is then defined as the odd function of ghost number $1$
\[
S_\text{BFV}(p,q,b,c) = c^\alpha H_\alpha(p,q) + \cdots,
\]
where the dots denote terms in the ghost momenta $b_\alpha$ such that the Poisson bracket of $S_\text{BFV}$ with itself vanishes (this is called the classical master equation). The BFV action exists and is uniquely defined up to symplectomorphisms \cite{BF,FV,Sta}. Note that the linear term in the $b$s depends on the structure functions,
\[
S_\text{BFV}(p,q,b,c) = c^\alpha H_\alpha(p,q) - \frac12 f_{\alpha\beta}^\gamma(p,q)\, b_\gamma c^\alpha c^\beta
+\cdots,
\]
and that  the higher-order terms are not needed in the Lie algebra case (in which the $f_{\alpha\beta}^\gamma$s are constant).

The next step is the AKSZ construction \cite{AKSZ}, which in this special case takes the following very simple form. We consider the replacements\footnote{The ``weird'' choice of signs is purely conventional. It is done in such a way that 
$i)$ the classical part of the BV action $S_\text{BV}$ introduced below agrees with \eqref{e:Ssevconstr}, and 
$ii)$ the BV form $\omega_\text{BV}$ introduced a bit later has the usual form.}
\[
p_i\to p_i+q^+_i,\quad q^i\to q^i - p_+^i,\quad c^\alpha \to c^\alpha - e^\alpha,\quad b_\alpha\to -e^+_\alpha+c^+_\alpha,
\]
where, in each assignment, the first term on the right hand side is a function on $[t_\ii,t_\oo]$ with the same Grassmann parity and ghost number as the variable on the left hand side, whereas the second term is a $1$-form on $[t_\ii,t_\oo]$ with opposite Grassmann parity and with ghost number reduced by $1$. Note, in particular, that we recover the even $1$-forms $e^\alpha$ of ghost number zero. With this notation we define the BV action as
\begin{multline*}
S_\text{BV}[p,q,e,c,p_+,q^+,e^+,c^+]:=\\ \int_{t_\ii}^{t_\oo} (
(p_i+q^+_i)\ddd(q^i-p_+^i) + (-e^+_\alpha+c^+_\alpha)\ddd (c^\alpha - e^\alpha)+
S_\text{BFV}(p+q^+,q-p^+,-e^++c^+,c-e))\\
= \int_{t_\ii}^{t_\oo} (p_i\ddd q^i - e^+_\alpha\ddd c^\alpha + S_\text{BFV}(p+q^+,q-p^+,-e^++c^+,c-e))
, 
\end{multline*}
where only the one-form components get integrated. 

\begin{example}\label{e:BVBFVLie}
In the main cases of interest for this paper, the structure functions $f_{\alpha\beta}^\gamma$ are constant. We then get
\begin{align*}
S_\text{BFV}(p,q,b,c) &= c^\alpha H_\alpha(p,q) - \frac12 f_{\alpha\beta}^\gamma\, b_\gamma c^\alpha c^\beta,\\
\begin{split}
S_\text{BV}[p,q,e,c,p_+,q^+,e^+,c^+]&= \int_{t_\ii}^{t_\oo} \bigg(
p_i\ddd q^i - e^+_\alpha\ddd c^\alpha -
e^\alpha H_\alpha(p,q) \\ 
&\phantom{=}+ c^\alpha\left(q^+_i\frac{\dd H_\alpha}{\dd p_i}(p,q)
-
p_+^i\frac{\dd H_\alpha}{\dd q^i}(p,q)
\right)\\
&\phantom{=}-
 \frac12 f_{\alpha\beta}^\gamma\, c^+_\gamma c^\alpha c^\beta +
 f_{\alpha\beta}^\gamma\, e^+_\gamma c^\alpha e^\beta
\bigg).
\end{split}
\end{align*}
These are the BFV and BV actions that we are going to consider in the next sections.
\end{example}

Let us return to the general case for the remaining part of this section.
The BV action satisfies the following properties:
\begin{enumerate}
\item It is even and of ghost number zero.
\item The terms independent of the antifields (i.e., the fields with a $+$ label) are the action \eqref{e:Ssevconstr}.
\item If we ignore boundary terms (or impose periodic boundary conditions for all fields), its Poisson bracket with itself, with respect to the odd symplectic\footnote{\label{f:omegaBVnd}The nondegeneracy of $\omega_\text{BV}$ is actually not needed for classical considerations. At the quantum level only a partial nondegeneracy is actually needed, namely the one along the directions for which one considers integration; see more on this in Section~\ref{sec:goodsplit}.} form of ghost number $-1$
\[
\omega_\text{BV} := \int_{t_\ii}^{t_\oo} (\delta q^+_i\delta q^i + \delta p_+^i\delta p_i + \delta e^+_\alpha\delta e^\alpha + \delta c^+_\alpha\delta c^\alpha
),
\]
vanishes.
\end{enumerate}
This is the starting point for the BV quantization.


\begin{remark}[Chern--Simons]
If we apply this construction to the Chern--Simons case decribed in the previous section, we get the BV formalism for Chern--Simons theory in $[t_\ii,t_\oo]\times\Sigma$. We will return to this in 
the 
{forthcoming paper \cite{CS_cyl}}.
\end{remark}

\begin{remark}\label{r:SwithHAKSZ}
The case where in addition to constraints there is also a nontrivial hamiltonian $H$ can be treated as at the beginning of 
{Section~\ref{r:nontrivev}},
as originally 
proposed in \cite{DG}, or, more conceptually, as in Section~\ref{r:SwithH}. 
\end{remark}


\section{An outline of elements of BV-BFV quantization}
The BV and BFV structures discussed in Section~\ref{s:BFVAKSZBBV} can be put together, and this may be used towards quantization \cite{CMR14,CMR15}. We briefly review this story in the case at hand assuming a finite-dimensional target. The case of infinite-dimensional targets discussed in Section~\ref{s:inftarg} will be addressed in the forthcoming paper \cite{CS_cyl}.

\subsection{The classical BV-BFV setting}
As described in Section~\ref{s:BFVAKSZBBV} we have a space $\FF$ of BV fields $(p,q,e,c,p_+,q^+,e^+,c^+)$ associated to the bulk $[t_\ii,t_\oo]$ and a space of BFV variables $(p,q,b,c)$ associated to each boundary point. As the boundary of $[t_\ii,t_\oo]$ actually consists of two endpoints (with opposite orientation), we define the space $\FF^\partial$ of BFV variables by doubling the $(p,q,b,c)$s: 
\[
\FF^\partial\ni(p^\ii,q_\ii,b^\ii,c_\ii,p^\oo,q_\oo,b^\oo,c_\oo).
\]
The BFV symplectic form and the BFV action are now
\[
\omega_\text{BFV}=\ddd p_i^\ii\ddd q^i_\ii + \ddd b_\alpha^\ii\ddd c^\alpha_\ii
-\ddd p_i^\oo\ddd q^i_\oo - \ddd b_\alpha^\oo\ddd c^\alpha_\oo
\]
and
\[
\begin{split}
S_\text{BFV} &= c^\alpha_\ii H_\alpha(p^\ii,q_\ii) - \frac12 f_{\alpha\beta}^\gamma(p^\ii,q_\ii)\, b_\gamma^\ii c^\alpha_\ii c^\beta_\ii
+\cdots\\
&\phantom{\ }
-c^\alpha_\oo H_\alpha(p^\oo,q_\oo) + \frac12 f_{\alpha\beta}^\gamma(p^\oo,q_\oo)\, b_\gamma^\oo c^\alpha_\oo c^\beta_\oo-\cdots.
\end{split}
\]

The bulk setting $(\FF,S_\text{BV},\omega_\text{BV})$ and the boundary BFV theory $(\FF^\partial,S_\text{BFV},\omega_\text{BFV})$ are related as follows:
\begin{enumerate}
\item We have a surjective submersion $\pi\colon \FF\to\FF^\partial$ given by
\begin{align*}
p^\oo &= p(t_\oo) &
q_\oo &= q(t_\oo) &
b^\oo &= -e^+(t_\oo) &
c_\oo &= c(t_\oo)\\
p^\ii &= p(t_\ii) &
q_\ii &= q(t_\ii) &
b^\ii &= -e^+(t_\ii) &
c_\ii &= c(t_\ii)
\end{align*}
\item If we denote by $Q_\text{BV}$ the hamiltonian vector field of $S_\text{BV}$ with respect to $\omega_\text{BV}$, obtained by ignoring boundary terms (or by imposing periodic boundary conditions for all fields),  and by $Q_\text{BFV}$ the hamiltonian vector field of $S_\text{BFV}$ with respect to $\omega_\text{BFV}$ we get
\begin{equation}\label{e:QQdpi}
[Q_\text{BV},Q_\text{BV}]=0,\quad [Q_\text{BFV},Q_\text{BFV}]=0, \quad \ddd\pi (Q_\text{BV}) = Q_\text{BFV},
\end{equation}
where $\ddd\pi$ denotes the differential of the map $\pi$ (the last equation may be equivalently formulated, more algebraically, as the property that $\pi^*(Q_\text{BFV}h)=Q_\text{BV}\pi^*h$ for every function $h$ on $\FF^\partial$).
\item If we introduce the potential
\[
\alpha_\text{BFV} = 
 p_i^\ii\ddd q^i_\ii +  b_\alpha^\ii\ddd c^\alpha_\ii
-p_i^\oo\ddd q^i_\oo - b_\alpha^\oo\ddd c^\alpha_\oo
\]
for the BFV form $\omega_\text{BFV}$, then 
the modified classical master equation (mCME)
\begin{equation}\label{e:mCMEwof}
\iota_{Q_\text{BV}}\omega_\text{BV} = \delta S_\text{BV} + \pi^*\alpha_\text{BFV}
\end{equation}
is satisfied. Here the variation $\delta$ of the BV action is computed by taking the boundary terms into account (they are precisely compensated by $ \pi^*\alpha_\text{BFV}$) and is regarded as producing a $1$-form on $\FF$.
\end{enumerate}
A simple consequence of the above, see \cite{CMR15}, is that we also have the equation
\begin{equation}\label{e:inducedCME}
\frac12\iota_{Q_\text{BV}}\iota_{Q_\text{BV}}\omega_\text{BV}=\pi^*S_\text{BFV},
\end{equation}
where the left hand side may be read as the correct replacement for the otherwise ill-defined Poisson bracket, induced by $\omega_\text{BV}$, of $S_\text{BV}$ with itself, whereas the right hand side yields the error boundary terms. This equation will play an important role in quantization.

\begin{remark}[Changing the endpoint conditions]\label{r:BVBFVchanging}
As in the purely classical setting, see Remark~\ref{r:changing}, we might want to use other variables, say $P$ and $Q$, at the final endpoint via a generating function $f(q,Q)$. We can arrange for this by subtracting 
$f$
 from the BV action, 
\[
S_\text{BV}^f[p,q,e,c,p_+,q^+,e^+,c^+] :=S_\text{BV}[p,q,e,c,p_+,q^+,e^+,c^+]-f(q(t_\oo),Q(p(t_\oo),q(t_\oo))),
\]
and at the same time adding $\ddd f$ to the BFV potential,
\[
\alpha_\text{BFV}^f := \alpha_\text{BFV} + \ddd f(q_\oo,Q(p^\oo,q_\oo))=
 p_i^\ii\ddd q^i_\ii +  b_\alpha^\ii\ddd c^\alpha_\ii
-P_i^\oo\ddd Q^i_\oo - b_\alpha^\oo\ddd c^\alpha_\oo
.
\]
This way we get the wished-for classical part of the symplectic form at the final endpoint, yet preserving the mCME \eqref{e:mCMEwof}, now in the form
\[
\iota_{Q_\text{BV}}\omega_\text{BV} = \delta S_\text{BV}^f + \pi^*\alpha_\text{BFV}^f.
\]
\end{remark}

\begin{remark}[Changing the endpoint field and ghost conditions]\label{r:BVBFVghostchanging}
We might also want to exchange the role of $b$ and $c$ at the final endpoint. To do so, we consider the extended generating function 
\[
\Tilde f(q,Q,c,C)=f(q,Q)+cC.
\]
We then get
\[
S_\text{BV}^{\Tilde f}[p,q,e,c,p_+,q^+,e^+,c^+] =S_\text{BV}^f[p,q,e,c,p_+,q^+,e^+,c^+]+c_\alpha(t_\oo)e^+(t_\oo)
\]
and
\[
\alpha_\text{BFV}^{\Tilde f} =  p_i^\ii\ddd q^i_\ii +  b_\alpha^\ii\ddd c^\alpha_\ii
-P_i^\oo\ddd Q^i_\oo - c^\alpha_\oo \ddd b_\alpha^\oo
.
\]
The mCME \eqref{e:mCMEwof} now reads
\begin{equation}\label{e:mCME}
\iota_{Q_\text{BV}}\omega_\text{BV} = \delta S_\text{BV}^{\Tilde f} + \pi^*\alpha_\text{BFV}^{\Tilde f}.
\end{equation}
\end{remark}

\begin{remark}\label{r:inducedCME}
The induced equation \eqref{e:inducedCME} holds as it is, since it is independent of the choice of the $1$-form.
\end{remark}


\subsection{The quantum BV-BFV setting}
In this section, we give a very short overview of the BV-BFV quantization, tailored for the examples of this paper.
For a more general exposition, see \cite{CMR15}.

\subsubsection{Boundary polarization}
For the quantum setting, we first have to choose a polarization in $\FF^\partial$. Concretely, in the case at hand, we 
write $\FF^\partial$ as $T^*\BB$, for a suitable lagrangian submanifold $\BB$, with canonical symplectic form. 
We require that the BFV $1$-form of \eqref{e:mCMEwof}  should be the canonical $1$-form of the cotangent bundle. This means that the choice of $\BB$ is related to the choice of boundary polarization. (Recall that,
on the initial endpoint, we choose $q_\ii$ and $c_\ii$ as coordinates on $\BB$; on the final endpoint, we choose $Q_\oo$ and $c_\oo$ or
$Q_\oo$ and $b^\oo$---i.e., we work in the setting of Remark~\ref{r:BVBFVchanging} or \ref{r:BVBFVghostchanging}, and we consider the corresponding $1$-form $\alpha^f_\text{BFV}$ or $\alpha^{\Tilde f}_\text{BFV}$.)

\subsubsection{Splitting}\label{s:splitting}
For simplicity, we assume the target (for the $(p,q)$ fields) to be $T^*\RR^n$. In this case $\BB$ is a (graded) vector space, and $\FF^\partial = \BB^*\oplus\BB$. The crucial point is to write the space of bulk fields as
\begin{equation}\label{e:splitting}
\Tilde \FF = \YY\oplus\BB, 
\end{equation}
where $\Tilde\FF$ is a suitable replacement (regularization) of the space of bulk fields, compatible with the boundary polarization by this splitting. Note that $\BB^*$ is now a subspace of $\YY$. We assume that equations
\eqref{e:QQdpi} and \eqref{e:mCME}, and therefore also \eqref{e:inducedCME}, are satisfied (now working on $\Tilde\FF$ instead of $\FF$ and with a suitable replacement $\Tilde S_\text{BV}^{\Tilde f} $ of $S_\text{BV}^{\Tilde f} $).

\begin{example}\label{e:2plittings}
The $q$-field component of $\BB$ is $q_\ii$. As $q$-component of $\YY$ we can then take the space 
$C^\infty_0([t_\ii,t_\oo],\RR^n)\ni \Hat q$ of maps that vanish at $t_\ii$. If we define
\[
q(t) = q_\ii + \Hat q(t),
\]
as we did in Remark~\ref{r:semiclassicalapproximation}, then we see that the $q$-component of $\Tilde\FF$ is the space of smooth maps $[t_\ii,t_\oo]\to\RR^n$, as in $\FF$. If on the other hand we set
\[
q(t)=\begin{cases}
q_\ii & t=t_\ii\\
\Hat q(t) & t> t_\ii
\end{cases},
\]
then we get, in $\Tilde\FF$, maps with a possible discontinuity at $t_\ii$, unlike in $\FF$. In this case, the quantum space of bulk fields $\Tilde\FF$ is slightly different from its classical counterpart $\FF$. It turns out that, in general, these discontinuous extensions are more suitable for the BV-BFV formalism, as we explain below (see Section~\ref{s:BVBFVqlin} and~\ref{s:nonlinpols} for examples and also for the correct definition of $\Tilde S_\text{BV}$ in these cases).
\end{example}

\subsubsection{Good splittings}\label{sec:goodsplit}
Given a splitting \eqref{e:splitting}, we decompose
\[
Q_\text{BV} = Q_\YY + Q_\BB
\]
and
\[
\omega_\text{BV} = \omega_{\YY\YY} + \omega_{\YY\BB} + \omega_{\BB\BB}
\]
according to which (co)tangent space(s) the vector field/$2$-form belongs. We assume $\omega_{\YY\YY}$ to be closed and (weakly) nondegenerate as a $2$-form on $\YY$. The mCME \eqref{e:mCME} now reads, in $\YY$- and $\BB$-components,\footnote{We write the formulae with generating function $\Tilde f$. Of course, they hold also with $f$ instead.}
\begin{align*}
\iota_{Q_\YY}\omega_{\YY\YY} + \iota_{Q_\BB} \omega_{\YY\BB} &= \delta_\YY \Tilde S_\text{BV}^{\Tilde f}\\
\iota_{Q_\YY}\omega_{\YY\BB} + \iota_{Q_\BB} \omega_{\BB\BB} &= \delta_\BB \Tilde S_\text{BV}^{\Tilde f} + \alpha_\text{BFV}^{\Tilde f},
\end{align*}
where we have omitted writing $\pi^*$. Note that, by assumption, $\alpha_\text{BFV}^{\Tilde f}$ only has $\BB$-components.

We say that we have a good splitting if $\iota_{Q_\BB} \omega_{\YY\BB}=0$, for in this case $Q_\YY$ is the hamiltonian vector field, with respect to $\omega_{\YY\YY}$,  of $\Tilde S_\text{BV}^{\Tilde f}$ on $\YY$ (with everything parametrized by $\BB$).  We assume from now on that we have a good splitting.
We denote by $(\ ,\ )_\YY$ the Poisson bracket induced by $\omega_{\YY\YY}$, so we have
\[
(\Tilde S_\text{BV}^{\Tilde f},\Tilde S_\text{BV}^{\Tilde f})_\YY = \iota_{Q_\YY}\iota_{Q_\YY}\omega_{\YY\YY}=
\iota_{Q_\text{BV}}\iota_{Q_\text{BV}}\omega_\text{BV}-\iota_{Q_\BB}\iota_{Q_\BB}\omega_{\BB\BB}.
\]
By the induced equation \eqref{e:inducedCME}, we then get
\[
(\Tilde S_\text{BV}^{\Tilde f},\Tilde S_\text{BV}^{\Tilde f})_\YY=2 S_{BFV}-\iota_{Q_\BB}\iota_{Q_\BB}\omega_{\BB\BB}.
\]
If the latter term depends only on the $T^*\BB$ variables, which we assume from now on, we may define a new function
\[
S^\partial := S_{BFV}-\frac12 \iota_{Q_\BB}\iota_{Q_\BB}\omega_{\BB\BB} 
\]on $T^*\BB$, and we have
\begin{equation}
(\Tilde S_\text{BV}^{\Tilde f},\Tilde S_\text{BV}^{\Tilde f})_\YY=2S^\partial. \label{eq:dCME}
\end{equation}
{
\begin{remark}\label{rem:Discontinuous splitti}
Usually we do not have a good splitting. A way out, as also suggested in Remark~2.33 of \cite{CMR15}, is to consider a family of splittings such that there is a limit in which the theory exists and corresponds to a good splitting (but of a different space of fields). For instance, one can define a section $s_\varepsilon$ of the short exact sequence $0 \to \YY\stackrel{\iota}{ \to} \FF \to \BB \to 0$ by multiplying fields in $\BB$ by a smooth function equal to 1 close to the boundary and with support contained in a collar of length $\varepsilon$. We thus obtain a family of BV-BFV theories $$\FF_\varepsilon = \left(\YY \oplus \BB, (\iota \oplus s_\varepsilon)^*\omega_\text{BV}, (\iota \oplus s_\varepsilon)^*S^{\Tilde{f}}_\text{BV}, (\iota \oplus s_\varepsilon)^*Q_\text{BV}, \pi \circ (\iota \oplus s_\varepsilon)\right)$$ 
over $(\FF^\dd,S_\text{BFV},\omega_{BFV})$
isomorphic to $(\FF, \omega_{\text{BV}}, S^{\Tilde{f}}_\text{BV}, Q_\text{BV},\pi)$
for $\varepsilon > 0$. The limit as $\varepsilon \to 0$ gives a well-defined theory $\FF_0$ on $\YY \oplus \BB$ (one just integrates by parts the terms containing $d s_\varepsilon$), but it does not arise from a well-defined section $s$ - rather, it corresponds to the discontinuous extension in Example \ref{e:2plittings}.\footnote{In this theory the BV form is nondegenerate along $\YY$, but degenerate otherwise. This is all what is needed to define BV quantization. 
}
This theory on $\YY \oplus \BB$ is then the replacement as in \eqref{e:splitting} which is suitable for quantization. The space $\Tilde{F}$ can be constructed by embedding $\FF$ in a larger space $\Tilde{\Tilde{F}}$ where the limit $s_0 \colon \BB\to \Tilde{\Tilde{F}}$ exists, and then we let $\Tilde{F} = (\iota \oplus s_0)(\YY \oplus \BB)$. For our purposes one can typically let $\Tilde{\Tilde{F}}$ be smooth functions on the interval with a discontinuity at the endpoints. 

\end{remark}
}
\subsubsection{The modified quantum master equation}
In the BV setting, one has a second-order differential operator $\Delta_\YY$, the BV Laplacian,  acting on functions on $\YY$
with the properties
\[
\Delta_\YY^2=0,\quad \Delta_\YY(fg)=\Delta_\YY(f)g \pm f\Delta_\YY g \pm (f,g)_\YY.
\]
In the finite-dimensional case, this operator is actually canonically defined on half-densisties \cite{Khud,Sev}. It is then transferred to functions upon the choice of a nondegenerate, $\Delta_\YY$-closed, reference half-density. In the infinite-dimensional setting, the definition
of $\Delta_\YY$ requires a regularization, which we presently assume to have. Moreover, for simplicity, we assume
$\Delta_\YY \Tilde S_\text{BV}^{\Tilde f}=0$. Therefore, we have
\[
\Delta_\YY\, \EE^{\frac\II\hbar \Tilde S_\text{BV}^{\Tilde f}}= -\frac1{2\hbar^2}(\Tilde S_\text{BV}^{\Tilde f},\Tilde S_\text{BV}^{\Tilde f})_\YY\,\EE^{\frac\II\hbar \Tilde S_\text{BV}^{\Tilde f}}=-\frac1{\hbar^2}S^\partial \, \EE^{\frac\II\hbar \Tilde S_\text{BV}^{\Tilde f}}.
\]
The final step consists in finding some operator $\Omega$ on the space of functions on $\BB$ satisfying
\begin{equation}\label{e:OmegaSde}
\Omega \, \EE^{\frac\II\hbar \Tilde S_\text{BV}^{\Tilde f}} = S^\partial \, \EE^{\frac\II\hbar \Tilde S_\text{BV}^{\Tilde f}}.
\end{equation}
If we find such an operator, then we get the modified Quantum Master Equation (mQME)
\begin{equation}\label{e:mQME}
(\Omega + \hbar^2\Delta_\YY)\, \EE^{\frac\II\hbar \Tilde S_\text{BV}^{\Tilde f}} = 0,
\end{equation}
which is a quantization of 
\eqref{e:mCME}.

Note that $\Omega$ commutes with $\Delta_\YY$ (the former acts on functions on $\BB$, the latter on functions on $\YY$), so the operator
\begin{equation}\label{e:DeltaOmega}
\Delta^\Omega_\YY := \Delta_\YY + \frac1{\hbar^2}\Omega
\end{equation}
squares to zero if and only if $\Omega$ does so, which we are going to assume. 
The reason for this normalization, in which we prefer to regard $\Delta^\Omega_\YY$ as a modification of $\Delta_\YY$ by a contribution due to the boundary,\footnote{One might also as well argue that the coboundary operator should be defined as $\Omega_\YY := \Omega + \hbar^2\Delta_\YY$---i.e., as a modification of $\Omega$ by a contribution due to the bulk---stressing that what matters more are the boundary fields, as the bulk fields will eventually be integrated out in the functional integration. One might also argue that, since, without boundary, one has 
$-\II\hbar\Delta\left(\EE^{\frac\II\hbar \Tilde S_\text{BV}}\mathcal{O}\right)=\EE^{\frac\II\hbar \Tilde S_\text{BV}}
(Q_\text{BV}\mathcal{O}-\II\hbar\Delta)$, where on the right hand side we see the quantum modification of the BV operator $Q_\text{BV}$, then $-\II\hbar\Delta$ is the correct coboundary operator in the bulk, so the correct coboundary operator in the presence of a boundary should be 
$-\II\hbar\Delta^\Omega_\YY=-\II\hbar\Delta_\YY-\frac\II\hbar\Omega$.

} will become clear in Section~\ref{s:BVpf}, where we will address gauge fixing and gauge-fixing independence. The mQME can now be rewritten as
\[
\Delta^\Omega_\YY\, \EE^{\frac\II\hbar \Tilde S_\text{BV}^{\Tilde f}} = 0.
\]

There are two main settings for the construction of $\Omega$, which we describe in the next two sections. (In principle, there may be more choices available, but we are not aware of them.)

\subsubsection{Case~I: No boundary symmetries}
Suppose $Q^\partial=0$. This implies $Q_\BB=0$ and $S^\partial=S_\text{BFV}=0$. 
In this case, every splitting is good and, moreover, we have 
$( \Tilde S_\text{BV}^{\Tilde f}, \Tilde S_\text{BV}^{\Tilde f})_\YY=0$.
Therefore, we obtain the mQME by simply setting $\Omega=0$. 

Note, in particular, that we might choose a splitting as in the first part of Example~\ref{e:2plittings} for all fields, so that $\Tilde\FF=\FF$ and $\Tilde S_\text{BV}^{\Tilde f}=S_\text{BV}^{\Tilde f}$. We may also choose, say for the $q$ part, to split $q(t)=q_0(t)+\Hat q(t)$, where $q_0$ is the solution to the EL equations with the given boundary conditions. We then recover the usual procedure for dealing with a path integral without gauge symmetries, as in equation \eqref{e:affpq} of Remark~\ref{r:semiclassicalapproximation}.

\subsubsection{Case~II: Discontinuous splitting}\label{s:IIdiscosplit}
A rather general procedure for dealing with the BV-BFV quantization, proposed in \cite{CMR15}, is based on a discontinuous splitting of the fields, as in the second part of Example~\ref{e:2plittings} and in Remark \ref{rem:Discontinuous splitti}. In this case, $\Tilde\FF\not=\FF$ and $\Tilde S_\text{BV}^{\Tilde f}\not=S_\text{BV}^{\Tilde f}$. Moreover, the regularization of $\Delta_\YY$ requires some care, so we will be rather sketchy here. The main advantage of such a choice is that $\omega_{\YY\BB} = \omega_{\BB\BB} = 0$. This is then a good splitting in which $S^\partial=S_\text{BFV}$ and
\[
\delta_\BB \Tilde S_\text{BV}^{\Tilde f} = - \alpha_\text{BFV}^{\Tilde f}.
\]
Since $\alpha_\text{BFV}^{\Tilde f}$ is by assumption the canonical $1$-form of the cotangent bundle 
$T^*\BB=\BB^*\oplus\BB\ni (z^*,z)$, we write $\alpha_\text{BFV}^{\Tilde f}=z^*_\mu\ddd z^\mu$. Therefore, the last equation becomes
\[
\frac{\dd \Tilde S_\text{BV}^{\Tilde f}}{\dd z^\mu} = - z^*_\mu.
\]
It then follows that, if we define $\Omega$ to be the Schr\"odinger quantization of $S_\text{BFV}$,
\begin{equation}\label{e:Omega}
\Omega := S_\text{BFV}\left(\II\hbar\frac\dd{\dd z},z\right)
\end{equation}
with all the derivatives to the right (standard ordering), then we get
\[
\Omega \EE^{\frac\II\hbar \Tilde S_\text{BV}^{\Tilde f}}=
S_\text{BFV}(z^*,z) \EE^{\frac\II\hbar \Tilde S_\text{BV}^{\Tilde f}},
\]
which is \eqref{e:OmegaSde} (recall that we now have $S^\partial=S_\text{BFV}$). 

\begin{remark}
A general observation at this point is that $\Omega$ is a quantization of $S_\text{BFV}$ and that the classical master equation 
$\{S_\text{BFV},S_\text{BFV}\}=0$ is the lowest-order condition for $\Omega^2=0$. Note that nothing guarantees that this condition is satisfied at all orders. This actually becomes a consistency condition for the theory. (See more, e.g., in Section~\ref{s:linpols}).
\end{remark}

In this presentation, we treated $\Delta_\YY$ rather formally. The correct argument requires a regularization and a careful analysis of the Feynman diagrams, which we present in Appendix~\ref{s:mQME}, and relies on certain assumptions.

The simplest case when
the above argument turns out to be correct is when the initial and final polarizations are related by a linear transformation (i.e., $f(q,Q)$ is linear in the $q$s and in the $Q$s). This is the setting of \cite{CMR15}, which we will use in Section~\ref{s:linpols} (in the particular case $f(q,Q)=Q_iq^i$ as, e.g., in Example~\ref{exa:pfinal}).

For a more general $f$, 
a compatibility conditions between the constraints and $f$ is required, see Assumption~\ref{a:} or~\ref{a:assu}.
In particular, 
 the above procedure actually works for a system as in Section~\ref{s:syssevcon} with the conditions that 
$i)$ all the hamiltonians $H_\alpha$ are linear in the $p$ variables and 
$ii)$ all the hamiltonians $\Tilde H_\alpha$ are linear in the $P$ variables, where
$\Tilde H_\alpha(P,Q):=H_\alpha(p(P,Q),q(P,Q))$, and we use the symplectomorphism induced by $f$ (no matter how nonlinear it is; what saves the game here is the linearity of the hamiltonians in the momenta). 
Note that this is not an artificial, quirky situation, but it actually occurs in important examples, the main one (although with infinite-dimensional target) being described in Section~\ref{s:7d}.


\subsubsection{The BV pushforward}\label{s:BVpf}
One of the properties of the BV formalism is the following. Assume we can write $\YY=\YY_1\oplus\YY_2$ as a sum of symplectic subspaces. If $\LL$ is a lagrangian submanifold of $\YY_2$ and $\psi$ is in the kernel of $\Delta^\Omega_\YY$, then the integral $\psi_1$ of $\psi$ over $\LL$ is in the kernel of $\Delta^\Omega_{\YY_1}:=\Delta_{\YY_1}+\frac1{\hbar^2}\Omega$; moreover, if we deform $\LL$, the result $\psi_1$ will change by a $\Delta^\Omega_{\YY_1}$-exact term.  

By the mQME \eqref{e:mQME}, $\EE^{\frac\II\hbar \Tilde S_\text{BV}^{\Tilde f}}$ is $\Delta^\Omega_\YY$-closed, so we may apply this procedure, 
which is known as the BV pushforward. The result of the pushforward can again be recast in the form  $\EE^{\frac\II\hbar S_1}$, where $S_1$ is a formal power series in $\hbar$.
The choice of $\LL$ is called a gauge fixing, and the last property in the previous paragraph is referred to as gauge-fixing independence. With the normalization in \eqref{e:DeltaOmega}, the argument of 
$\Delta^\Omega_{\YY_1}$ under a change of gauge fixing is a formal power series in $\hbar$ times $\EE^{\frac\II\hbar S_1}$, which is appropriate in perturbation theory.

We may, in particular, choose $\YY_2=\YY$ (assuming we can find a lagrangian $\LL$ on which the  integral of $\psi$ converges). In this case, $\psi_1$ will be a function on $\BB$ satisfying $\Omega\psi_1=0$. Note that in this
case $\Delta^\Omega$ is just $\Omega/\hbar^2$.


Two remarks are here in order. The first is that this procedure is more invariantly defined in terms of half-densitites instead of functions. The second is that there might be lagrangian submanifolds $\LL$ and $\LL'$ that are not deformations of one another (at least through a path of lagrangian submanifolds on which the  integral of $\psi$ converges); in this case, the two results of integration might not be $\Delta^\Omega_{\YY_1}$-cohomologous; see, e.g., Example~\ref{e:qBVBFVtoy} below.

\begin{remark}
The BV pushforward may of course be iterated. The setting is then the following. We have a fixed space $\BB$ with a fixed coboundary operator $\Omega$ acting on its functions. We then have, for some ordered indices $k$, an odd symplectic space $\YY_k$ and a $\Delta^\Omega_{\YY_k}$-closed function $\psi_k$ on $\YY_k\times\BB$. The $\Delta^\Omega_{\YY_{k+1}}$-closed function $\psi_{k+1}$ on $\YY_{k+1}\times\BB$ is obtained by BV-pushforward. 
On the new spaces $\YY_{k}\times\BB$ one can recover some derived structures reminiscent of the original one on $\YY\times\BB$; see \cite[Remark 2.17]{CMR15} and \cite[Proposition 9.1]{CMRcell}.
\end{remark}

\begin{example}\label{e:qBVBFVtoy}
As a toy example (actually a BV-BFV quantization of Example~\ref{exa:linearone} with $n=1$, $v=0$ and $w=1$ after a first BV pushforward),
we consider $\BB=\RR\oplus\RR[1]\oplus\RR\oplus\RR[1]\ni(q_\ii, c_\ii, p^\oo, c_\oo)$,
$\YY_1=\RR\oplus\RR[-1]\ni(T,T^+)$, 
$\Delta_{\YY_1}=\frac{\dd^2}{\dd T^+\dd T}$, and
\[
\Omega = c_\ii q_\ii - \II\hbar c_\oo\frac\dd{\dd p^\oo}.
\]
The state
\[
\psi_1(q_\ii, c_\ii, p^\oo, c_\oo,T,T^+)=
\EE^{-\frac \II\hbar(p^\oo q_\ii+Tq_\ii+T^+(c_\ii-c_\oo))}
\]
is $\Delta^\Omega_{\YY_1}$-closed. If we consider the lagrangian submanifold $\LL=\{T^+=0\}$ of $\YY_1$, we get\footnote{We conventionally normalize the measures $\ddd T$ and $\ddd T^+$ in order to get nicer looking formulae.}
\[
\psi_2^\LL(q_\ii, c_\ii, p^\oo, c_\oo)=\int_\LL \psi_1(q_\ii, c_\ii, p^\oo, c_\oo,T,T^+) =
\int  \frac{\ddd T}{2\pi\hbar} \psi_1(q_\ii, c_\ii, p^\oo, c_\oo,T,0) = \delta(q_\ii),
\]
which is clearly $\Omega$-closed. If, on the other hand, we consider the lagrangian submanifold $\Tilde\LL=\{T=0\}$,
which cannot be deformed to $\LL$, we get
\[
\begin{split}
\psi_2^{\Tilde\LL}(q_\ii, c_\ii, p^\oo, c_\oo)&=\int_{\Tilde\LL} \psi_1(q_\ii, c_\ii, p^\oo, c_\oo,T,T^+) =
\int  \II\hbar\, {\ddd T^+} \psi_1(q_\ii, c_\ii, p^\oo, c_\oo,0,T^+)\\ &= (c_\ii-c_\oo)\,
\EE^{-\frac \II\hbar p^\oo q_\ii}= \delta(c_\ii-c_\oo)\,
\EE^{-\frac \II\hbar p^\oo q_\ii}
,
\end{split}
\]
which is also $\Omega$-closed but certainly not $\Omega$-cohomologous to $\psi_2^\LL$ (the two have different ghost numbers).
\end{example}

\subsubsection{Composition}\label{s:compo}
The space $\BB$ we consider is actually a product $\BB_\ii\times\BB_\oo$, with $\BB_\ii$ having coordinates $q_\ii$ and $c_\ii$ and
with $\BB_\oo$ having coordinates $Q_\oo$ and $c_\oo$ or
$Q_\oo$ and $b^\oo$. We assume that $\Omega$ is local, i.e., $\Omega=\Omega_\ii-\Omega_\oo$, where
each summand is an operator on the space of functions of the corresponding space. In addition, we assume $\Omega_\ii$ and $\Omega_\oo$ to be self-adjoint operators on functions on $\BB_\ii$ and $\BB_\oo$, respectively.\footnote{\label{f:above}This is defined in terms of some measure on $\BB_\ii$. A more invariant formulation regards the states as half-densities, so this choice of measure is not necessary.}

For instance, see Example~\ref{e:qBVBFVtoy} where we have $\Omega_\ii=c_\ii q_\ii$ and $\Omega_\oo=\II\hbar c_\oo\frac\dd{\dd p^\oo}$. We will see more examples in
Sections~\ref{s:linpols} and~\ref{s:bnlin}.

Suppose we have a $\Delta^{\Omega_\ii-\Omega_{\oo}}_\YY$-closed state $\psi$ on $\BB_\ii\times\BB_\oo\times\YY$ and a
$\Delta^{\Omega_\oo}_{\Tilde\YY}$-closed state $\psi_\oo$ on $\BB_\oo\times\Tilde \YY$. We can then pair $\psi$ with $\psi_\oo$ integrating over $\BB_\oo$ (we assume this integration to converge). Then\footnote{In the composition integral we assume that there is a measure on $\BB_\oo$ or that the states are half-densities, see footnote~\ref{f:above}.}
\[
\psi_\ii := \int_{\BB_\oo} \psi\psi_\oo
\]
is a $\Delta^{\Omega_\ii}_{\YY\times\Tilde\YY}$-closed state on $\BB_\ii\times\YY\times\Tilde\YY$. A further reduction by BV pushforward, as in Section~\ref{s:BVpf}, is of course possible if we have $\YY\times\Tilde\YY=\YY_1\times\YY_2$.

\begin{example}
Take $\psi=\psi_1$ as in Example~\ref{e:qBVBFVtoy} and $\Tilde\YY$ to be a point. The ghost-number-one state
\[
\psi_\oo(p^\oo,c_\oo)=\phi(p^\oo)\delta(c_\oo)=\phi(p^\oo)c_\oo,
\]
where $\phi$ is an arbitary function, is obviously $\Omega_\oo$-closed. After pairing, it yields, up to a multiplicative   constant that depends on the choice of measure,
\[
\psi_\ii(q_\ii,c_\ii,T,T^+)=\EE^{-\frac\II\hbar(Tq_\ii+T^+c_\ii)}\Hat \phi(q_\ii),
\]
where $\Hat \phi$ is the Fourier transform of $\phi$. Note that
\[
-\hbar^2\Delta_{\YY_1}\psi_\ii=c_\ii q_\ii \psi_\ii = \Omega_\ii\psi_\ii,
\]
so $\psi_\ii$ is $\Delta^{\Omega_\ii}_{\YY_1}$-closed. We may now integrate over the lagrangian submanifolds $\LL=\{T^+=0\}$ or $\Tilde\LL=\{T=0\}$ getting
\[
\psi_\ii^\LL(q_\ii,c_\ii)=\delta(q_\ii)\Hat\phi(0)\quad\text{ and }\quad
\psi_\ii^{\Tilde \LL}(q_\ii,c_\ii)=\delta(c_\ii)\Hat \phi(q_\ii),
\]
which are clearly $\Omega_\ii$-closed. If, on the other hand, we consider the ghost-number-zero state
\[
\Tilde\psi_\oo(p^\oo,c_\oo)=1,
\]
which is obviously $\Omega_\oo$-closed, then we get, again up to a multiplicative constant,
\[
\Tilde\psi_\ii(q_\ii,c_\ii,T,T^+)=\delta(q_\ii)\delta(T^+)=\delta(q_\ii) T^+,
\]
which is obviously $\Delta_{\YY_1}$- and $\Omega_\ii$-closed, and therefore $\Delta^{\Omega_\ii}_{\YY_1}$-closed.
By a further BV pushforward as above, we get
\[
\Tilde\psi_\ii^\LL(q_\ii,c_\ii)=0\quad\text{ and }\quad
\Tilde\psi_\ii^{\Tilde \LL}(q_\ii,c_\ii)=\delta(q_\ii),
\]
which are clearly $\Omega_\ii$-closed.
\end{example}

\begin{remark}
A generalization of the above is when we have a $\Delta^{\Omega_\ii-\Omega_{\oo}}_\YY$-closed state $\psi$ on $\BB_\ii\times\BB_\oo\times\YY$ and a
$\Delta^{\Omega_\oo-\Omega_{\oo'}}_{\Tilde\YY}$-closed state $\Tilde\psi$ on $\BB_\oo\times\BB_{\oo'}\times\Tilde \YY$. We can then pair them on $\BB_\oo$, getting a $\Delta^{\Omega_\ii-\Omega_{\oo'}}_{\YY\times\Tilde\YY}$-closed state
\[
\Hat\psi =  \int_{\BB_\oo} \psi\Tilde\psi
\]
on $\BB_\ii\times\BB_{\oo'}\times\YY\times\Tilde\YY$ (which could possibly be further reduced by BV pushforwards).
If we interpret the above states $\psi$ and $\Tilde\psi$ as evolution operators corresponding to the intervals
$[t_\ii,t_\oo]$ and $[t_\oo,t_{\oo'}]$, respectively, then $\Hat\psi$ can be interpreted as the corresponding evolution operator on $[t_\ii,t_{\oo'}]$.
\end{remark}


\section{BV-BFV boundary structures for linear polarizations}\label{s:linpols}
In this section we prepare for the application of the quantum BV-BFV formalism to the case of one-dimensional systems with constraints, whose classical treatment was the focus of Sections~\ref{s:HJnondeg} to 
\ref{r:other}, with linear polarizations. We discuss here only the quantization of the boundary structure, whereas the quantization of the bulk will be discussed in Section~\ref{s:BVBFVqlin}.
We defer the case of nonlinear polarizations to Section~\ref{s:nonlinpols} and the case of infinite-dimensional targets, discussed at the classical level in Section~\ref{s:inftarg}, to \cite{CS_cyl}. In many places, we also make use of deformation quantization, for which we refer to the valuable review \cite{DS} and references therein.

\subsection{Three cases}\label{s:three}
We will work out some representative cases of linear polarizations. Namely,
 we assume that the coordinates in $\BB_\oo$ are either 
\begin{description}
\item[Case~I] $(p^\oo,c_\oo)$, or
\item[Case~II] $(p^\oo,b^\oo)$,
\end{description}
with $(q_\ii,c_\ii)$ as the coordinates in $\BB_\ii$ (see Remark~\ref{r:BVBFVghostchanging}).
For completeness, and for later applications, we will also consider: 
\begin{description}
\item[Case~III] $(q_\ii,b^\ii)$ as the coordinates in $\BB_\ii$ and $(p^\oo,c_\oo)$ as the coordinates in $\BB_\oo$.
\end{description}

For the quantization we will consider a discontinuous splitting of the bulk fields (see Example~\ref{e:2plittings}).

Finally, we confine ourselves to the case when the structure functions are constant (see Example~\ref{e:BVBFVLie}); we then have
\begin{align*}
S_\text{BFV} &= c^\alpha_\ii H_\alpha(p^\ii,q_\ii) - \frac12 f_{\alpha\beta}^\gamma\, b_\gamma^\ii c^\alpha_\ii c^\beta_\ii
-c^\alpha_\oo H_\alpha(p^\oo,q_\oo) + \frac12 f_{\alpha\beta}^\gamma\, b_\gamma^\oo c^\alpha_\oo c^\beta_\oo,\\
\begin{split}
S_\text{BV}
&= \int_{t_\ii}^{t_\oo} \bigg(
p_i\ddd q^i - e^+_\alpha\ddd c^\alpha -
e^\alpha H_\alpha(p,q) 
+ c^\alpha\left(q^+_i\frac{\dd H_\alpha}{\dd p_i}(p,q){-}p_+^i\frac{\dd H_\alpha}{\dd q^i}(p,q)
\right)\\
&\phantom{=}-
 \frac12 f_{\alpha\beta}^\gamma\, c^+_\gamma c^\alpha c^\beta +
 f_{\alpha\beta}^\gamma\, e^+_\gamma c^\alpha e^\beta
\bigg).
\end{split}
\end{align*}

\begin{remark}[Unimodularity]\label{r:unimodularity}
The BV action satisfies the classical master equation. To ensure $\Delta S_\text{BV}=0$, although at a formal level, we assume the Lie algebra $\g$, of which the $f_{\alpha\beta}^\gamma$s are the structure constants, to be unimodular: i.e., $f_{\alpha\beta}^\alpha=0$.
\end{remark}

\begin{remark}[The coboundary operator]
In all the cases we consider, the operator $\Omega$ is of the form $\Omega_\ii-\Omega_\oo$ with $\Omega_\ii$ and $\Omega_\oo$ acting on different spaces, so $\Omega^2=\Omega_\ii^2+\Omega_\oo^2$ and we have $\Omega^2=0$ if and only if $\Omega_\ii^2=0$ and $\Omega_\oo^2=0$.
\end{remark}


\subsection{The boundary structure in Case~I}
We use $\Tilde f=f=p_iq^i$, so
\[
\alpha_\text{BFV}^{\Tilde f} = \alpha_\text{BFV} + \ddd f(q_\oo,p^\oo) 
=  p_i^\ii\ddd q^i_\ii +  b_\alpha^\ii\ddd c^\alpha_\ii
+q_\oo^i\ddd p^\oo_i
 - b_\alpha^\oo\ddd c^\alpha_\oo
.
\]
We now proceed as in Section~\ref{s:IIdiscosplit}.
If we write $\alpha_\text{BFV}^{\Tilde f}=z^*_\mu\ddd z^\mu$, we get
\[
p^\ii=q_\ii^*,\quad b^\ii=c_\ii^*,\quad q_\oo=(p^\oo)^*,\quad b^\oo=-c_\oo^*.
\]
The recipe \eqref{e:Omega} for constructing $\Omega$ is to replace $z^*_\mu$ with $\II\hbar\frac\dd{\dd z^\mu}$ in
$S_\text{BFV}$. Therefore, we get
\begin{align*}
\Omega_\ii &= c^\alpha_\ii H_\alpha\left(\II\hbar\frac\dd{\dd q_\ii},q_\ii\right) - \frac{\II\hbar}2 f_{\alpha\beta}^\gamma\, c^\alpha_\ii c^\beta_\ii \frac\dd{\dd c^\gamma_\ii},\\
\Omega_\oo &= c^\alpha_\oo H_\alpha\left(p^\oo,\II\hbar\frac\dd{\dd p^\oo}\right) + \frac{\II\hbar}2 f_{\alpha\beta}^\gamma\,  c^\alpha_\oo c^\beta_\oo\frac\dd{\dd c^\gamma_\oo}.
\end{align*}
For example, if the index sets are both $\{1\}$, so $f_{\alpha\beta}^\gamma=0$, and we take $H_1(p,q)=q$, we get the $\Omega$ of Example~\ref{e:qBVBFVtoy}.

We can easily see that $\Omega^2=0$ if and only if we have, at both endpoints, the following quantization of \eqref{e:involutivity}:
\begin{equation}\label{e:HHhat}
\left[\widehat{H}_\alpha,\widehat{H}_\beta\right] = \II\hbar f^\gamma_{\alpha\beta}\widehat{H}_\gamma.
\end{equation}
Here $\widehat{H}_\alpha$ denotes the corresponding quantization of $H_\alpha$ at the given endpoint. Equivalently, we may reformulate this as
\begin{equation}\label{e:HHstar}
[H_\alpha,\stackrel{\star}, H_\beta]=\II\hbar f_{\alpha\beta}^\gamma H_\gamma,
\end{equation}
where $[\phi\stackrel{\star},\psi]:=\phi\star\psi-\psi\star\phi$ denotes the commutator with respect to either star product induced from the above quantization (if the condition is satisfied for one star product, it is automatically satisfied also for the other). Namely, at the initial endpoint we have
\begin{equation}\label{e:iistar}
\phi\star\psi (p,q)= \EE^{\II\hbar \frac{\dd^2}{\dd \Tilde q^i\dd \Tilde p_i}} \phi(\Tilde p, q)\psi(p,\Tilde q)\Big|_{\Tilde p=p,\ \Tilde q=q},
\end{equation}
whereas at the final endpoint we have
\begin{equation}\label{e:oostar}
\phi\star\psi (p,q)= \EE^{\II\hbar \frac{\dd^2}{\dd \Tilde q^i\dd \Tilde p_i}} \phi(p,\Tilde q)\psi(\Tilde p, q)\Big|_{\Tilde p=p,\ \Tilde q=q}.
\end{equation}

In principle, if \eqref{e:HHstar} is not satisfied on the nose, 
one may look for $\hbar$-deformations of the $H_\alpha$s that satisfy it. However, there are in general obstructions to do it. For the purposes of the present paper the following remark is important.

\begin{remark}
In both the linear case of Example~\ref{e:flinear} and the biaffine case of Example~\ref{exa:biaffine}, 
condition \eqref{e:HHstar} is automatically satisfied, so we have $\Omega^2=0$.
\end{remark}

\begin{digression}
Lemma~4.14 of \cite{CMR15} states that $\Omega^2=0$ in case of $BF$-like theories (like the theories at hand). The regularization was assumed there to be in terms of the
Fulton--MacPherson--Axelrod--Singer
compactified configurations spaces \cite{FM,AS} (which is also the case we will consider in the next sections). What was also implicitly assumed there, although unfortunately not made clear, was that the bulk dimension should be larger than $1$.\footnote{The proof of that lemma relied in particular on the contribution of the principal faces---i.e., collapses of two points in the bulk---of the boundary of compactified configurations spaces. In $d$ dimensions, such faces have dimension $d-1$, whereas propagators in $BF$-like theories are $(d-1)$-forms. For $d>1$, there must then be exactly one propagator between the two collapsing points; the sum over all these contributions is zero as a consequence of the classical master equation. For $d=1$, however, there is no bound on the number of propagators between two collapsing points.} For this reason that lemma does not apply to the present case.
We will briefly discuss in Appendix~\ref{s:mQME} the correct version of the lemma in the case at hand.
 Note that Lemma~4.14 of \cite{CMR15} also does not apply directly to the case of infinite-dimensional targets that will be treated in \cite{CS_cyl}, since the appropriate gauge fixing  needed to reduce to the present considerations (the axial gauge fixing) is not compatible with the aforementioned  compactification of configuration spaces. Actually, the failure of \eqref{e:HHstar} has not only the consequence that the boundary operator $\Omega$ does not square to zero, but also that the classical master equation $(S_\text{BV},S_\text{BV})=0$ is not enough for  the regularized QME to hold in the bulk. 
What will save the game in \cite{CS_cyl}, where we will consider Chern--Simons theories on cylinders with the axial gauge, is that the theories considered there are examples of the linear or of the biaffine case.
\end{digression}

\subsection{The boundary structure in Case~II}
We use $\Tilde f=p_iq^i+c^\alpha b_\alpha$, so
\[
\alpha_\text{BFV}^{\Tilde f} = \alpha_\text{BFV} + \ddd \Tilde f(q_\oo,c_\oo,p^\oo,b^\oo) 
=  p_i^\ii\ddd q^i_\ii +  b_\alpha^\ii\ddd c^\alpha_\ii
+q_\oo^i\ddd p^\oo_i - c^\alpha_\oo \ddd b_\alpha^\oo
.
\]
In this case, the prescriptions of Section~\ref{s:IIdiscosplit} yield
\[
p^\ii=q_\ii^*,\quad b^\ii=c_\ii^*,\quad q_\oo=(p^\oo)^*,\quad 
c_\oo
=-(b^\oo)^*
\]
and
\begin{align*}
\Omega_\ii &=
c^\alpha_\ii H_\alpha\left(\II\hbar\frac\dd{\dd q_\ii},q_\ii\right) - \frac{\II\hbar}2 f_{\alpha\beta}^\gamma\, c^\alpha_\ii c^\beta_\ii \frac\dd{\dd c^\gamma_\ii},\\
\Omega_\oo &= -\II\hbar H_\alpha\left(p^\oo,\II\hbar\frac\dd{\dd p^\oo}\right) \frac\dd{\dd b_\alpha^\oo}
+ \frac{\hbar^2}2 f_{\alpha\beta}^\gamma\,  b_\gamma^\oo
\frac{\dd^2}{\dd b_\alpha^\oo\dd b_\beta^\oo}.
\end{align*}
Again $\Omega^2=0$ if and only if \eqref{e:HHstar} is satisfied at both endpoints.

\subsection{The boundary structure in Case~III}
In this case have to use the one-form
\[
\alpha_\text{BFV}^{\Hat f} =  p_i^\ii\ddd q^i_\ii +  c^\alpha_\ii\ddd b_\alpha^\ii 
+q_\oo^i\ddd p^\oo_i
 - b_\alpha^\oo\ddd c^\alpha_\oo
.
\]
To achieve this we add $p_iq^i$ at the final endpoint of the BV action, like in Cases I and II, but now we subtract $c^\alpha b_\alpha$ at the initial endpoint. The function $\Hat f$ is now more precisely written as $p^\oo_iq^i_\oo-c^\alpha_\ii b_\alpha^\ii$ and we have $\alpha_\text{BFV}^{\Hat f} = \alpha_\text{BFV} + \ddd{\Hat f}$.
In this case, the prescriptions of Section~\ref{s:IIdiscosplit} yield
\[
p^\ii=q_\ii^*,\quad c_\ii =(b^\ii)^*,\quad q_\oo=(p^\oo)^*,\quad b^\oo=-c_\oo^*,
\]
and
\begin{align*}
\Omega_\ii &=
\II\hbar H_\alpha\left(\II\hbar\frac\dd{\dd q_\ii},q_\ii\right) \frac\dd{\dd b_\alpha^\ii}
+ \frac{\hbar^2}2 f_{\alpha\beta}^\gamma\, b_\gamma^\ii\frac{\dd^2}{\dd b_\alpha^\ii\dd b_\beta^\ii},\\
\Omega_\oo &= c^\alpha_\oo H_\alpha\left(p^\oo,\II\hbar\frac\dd{\dd p^\oo}\right) + \frac{\II\hbar}2 f_{\alpha\beta}^\gamma\,  c^\alpha_\oo c^\beta_\oo\frac\dd{\dd c^\gamma_\oo}.
\end{align*}
Also in this case $\Omega^2=0$ if and only if \eqref{e:HHstar} is satisfied at both endpoints.

\section{BV-BFV quantization with linear polarizations}\label{s:BVBFVqlin}
In this section we quantize the constrained one-dimensional systems with linear polarizations on the endpoints.
We refer to the three cases whose boundary structure has been discussed in Section~\ref{s:linpols} (see Section~\ref{s:three} for the list of the three case).
We postpone the more complicated Case~I and start with Cases~II and~III. 

After that, we will discuss the gluing of intervals and the composition of the corresponding states. We will conclude the section with the quantization of the systems with nontrivial evolution discussed in Section~\ref{r:nontrivev} at the classical level: in particular, classical mechanics, Section~\ref{r:SwithH}, and the free relativistic particle, Section~\ref{s:RP}.

\subsection{Quantization in Case~II}\label{s:IIq}
 First we have to subtract the final-endpoint pullback of $\Tilde f=p_iq^i+c^\alpha b_\alpha$ from the BV action getting
\[
S_\text{BV}^{\Tilde f}
= S_\text{BV}
- p_i(t_\oo) q^i(t_\oo)+c^\alpha(t_\oo) e^+_\alpha(t_\oo).
\]
Next we pick the discontinuous splitting $\Tilde\FF=\YY\oplus\BB$ of Example~\ref{e:2plittings} for all bulk fields. Namely, by decorating the elements of $\YY$ with a hat, we set
\begin{align*}
q(t)&=\begin{cases}
q_\ii & t=t_\ii\\
\Hat q(t) & t> t_\ii
\end{cases},&\quad
c(t)&=\begin{cases}
c_\ii & t=t_\ii\\
\Hat c(t) & t> t_\ii
\end{cases},\\
p(t)&=\begin{cases}
\Hat p(t) & t< t_\oo\\
p^\oo & t=t_\oo
\end{cases},&\quad
e^+(t)&=\begin{cases}
\Hat e^+(t) & t< t_\oo\\
-b^\oo & t=t_\oo
\end{cases},
\end{align*}
with boundary conditions
\[
\Hat q(t_\ii)=0,\quad \Hat c(t_\ii)=0, \quad \Hat p(t_\oo) = 0, \quad \Hat e^+(t_\oo) = 0.
\]
We decorate the remaining fields, all belonging to $\YY$, also with a hat: $\Hat p^+,\Hat q_+,\Hat e,\Hat c^+$.
We now formally integrate by parts\footnote{See Remark~\ref{rem:Discontinuous splitti}.} the terms $p_i\ddd q^i$ and  $e^+_\alpha\ddd c^\alpha$ 
in $S_\text{BV}^{\Tilde f}$,
\[
\int_{t_\ii}^{t_\oo} p_i\ddd q^i \to -\Hat p_i(t_\ii)q_\ii^i + \int_{t_\ii}^{t_\oo} \Hat p_i\ddd \Hat q^i,\qquad
\int_{t_\ii}^{t_\oo} e^+_\alpha\ddd c^\alpha \to \Hat e^+_\alpha(t_\ii)c_\ii^\alpha + \int_{t_\ii}^{t_\oo} \Hat e^+_\alpha\ddd \Hat c^\alpha,
\]
 in order to get the BV action adapted to this splitting:
\[
\begin{split}
\Tilde S_\text{BV}^{\Tilde f}
:= &- p_i^\oo \Hat q^i(t_\oo)-\Hat c^\alpha(t_\oo) b_\alpha^\oo
-\Hat p_i(t_\ii)q_\ii^i  - \Hat e^+_\alpha(t_\ii)c_\ii^\alpha
\\
&+\int_{t_\ii}^{t_\oo} \bigg(
\Hat p_i\ddd \Hat q^i - \Hat e^+_\alpha\ddd \Hat c^\alpha -
\Hat e^\alpha H_\alpha(\Hat p,\Hat q) \\ 
&+ \Hat c^\alpha\left(\Hat q^+_i\frac{\dd H_\alpha}{\dd p_i}(\Hat p,\Hat q)-\Hat p_+^i\frac{\dd H_\alpha}{\dd q^i}(\Hat p,\Hat q)
\right)\\
&-
 \frac12 f_{\alpha\beta}^\gamma\, \Hat c^+_\gamma \Hat c^\alpha \Hat c^\beta +
 f_{\alpha\beta}^\gamma\, \Hat e^+_\gamma \Hat c^\alpha \Hat e^\beta
\bigg).
\end{split}
\]
We finally pick the gauge-fixing lagrangian 
\begin{equation}\label{e:simpleL}
\LL=\{\Hat p^+=\Hat q^+=\Hat e=\Hat c^+=0\},
\end{equation}
getting
\[
\begin{split}
\Tilde S_\text{BV}^{\Tilde f}\Big|_\LL
= &- p_i^\oo \Hat q^i(t_\oo)-\Hat c^\alpha(t_\oo) b_\alpha^\oo
-\Hat p_i(t_\ii)q_\ii^i  - \Hat e^+_\alpha(t_\ii)c_\ii^\alpha
+\int_{t_\ii}^{t_\oo} \left(
\Hat p_i\ddd \Hat q^i - \Hat e^+_\alpha\ddd \Hat c^\alpha \right).
\end{split}
\]
The propagators defined by the kinetic terms and compatible with the boundary conditions are
\begin{equation}\label{e:propII}
\langle \Hat q^i(s)\,\Hat p_j(t) \rangle = \II\hbar\delta^i_j\,\theta(s-t)\quad\text{ and}\quad
\langle \Hat c^\alpha(s)\,\Hat e^+_\beta(t) \rangle = \II\hbar\delta^\alpha_\beta\,\theta(s-t), 
\end{equation}
where
\[
\theta(t)=\begin{cases}
0 & \text{for }t<0\\
1 & \text{for }t>0
\end{cases}
\]
is the Heaviside step function. We can then compute
\[
Z_\text{II} := \int_\LL \EE^{\frac\II\hbar \Tilde S_\text{BV}^{\Hat f}}
\]
as
\begin{equation}\label{e:II}
Z_\text{II}(q_\ii,c_\ii,p^\oo,b^\oo)=\EE^{-\frac\II\hbar(p^\oo_iq_\ii^i-b^\oo_\alpha c_\ii^\alpha)}.
\end{equation}
One can immediately check that $\Omega Z_\text{II}=0$.

The striking result is that $Z_\text{II}$ is extremely simple and does not depend on the constraints $H_\alpha$. Still its definition requires the whole BV-BFV machinery and, in particular, condition \eqref{e:HHstar}. The reason is that we have chosen a clever gauge fixing that simplifies the computation, but we want to be sure that the result is invariant (up to $\Omega$-coboundary terms) under deformations of the gauge-fixing lagrangian (see, in case, the following Digression).

As we discuss in the following Digression, just deforming the gauge fixing, picking $e$ different from zero, does not change the result (in $\Delta^\Omega$-cohomology). Still one can get a less trivial answer for Case~I, as we will dicuss in Section~\ref{s:gluing}. The present ``trivial'' answer is anyway useful for gluing solutions on different intervals along the lines of Section~\ref{s:compo}. We will see in Section~\ref{s:gluing} how to get a nontrivial answer.

\begin{digression}[Change of gauge fixing]\label{d:cgf}
By the general BV-BFV philosophy---proved in the case at hand in
Appendix~\ref{s:mQME}---we know that deformations of $\LL$
produce $\Omega$-cohomologous results. Instead of using this general result,
as a matter of example we explicitly consider a particular case. Let
\[
\LL_u:=\left\{\Hat p^+=\Hat q^+=\Hat c^+=0,\ \Hat e= e_0 :=\frac{u\,\ddd t}{t_\oo-t_\ii}\right\},
\]
where, for simplicity, we assume $u$ to be a constant function on $[t_\ii,t_\oo]$. We now have
\[
\begin{split}
\Tilde S_\text{BV}^{\Tilde f}\Big|_{\LL_u}
= &- p_i^\oo \Hat q^i(t_\oo)-\Hat c^\alpha(t_\oo) b_\alpha^\oo
-\Hat p_i(t_\ii)q_\ii^i  - \Hat e^+_\alpha(t_\ii)c_\ii^\alpha
\\
&+\int_{t_\ii}^{t_\oo} \left(
\Hat p_i\ddd \Hat q^i - \Hat e^+_\alpha\ddd \Hat c^\alpha
- e_0^\alpha H_\alpha(\Hat p,\Hat q) 
 +
 f_{\alpha\beta}^\gamma\, \Hat e^+_\gamma \Hat c^\alpha e_0^\beta
\right).
\end{split}
\]
Let us first consider the case when the structure constants vanish. In this case, we compute, following \cite[Section~4.4]{CMR15},
\[
Z_\text{II}^u := \int_{\LL_u} \EE^{\frac\II\hbar \Tilde S_\text{BV}^{\Hat f}}=
\EE^{-\frac\II\hbar(p^\oo_iq_\ii^i-b^\oo_\alpha c_\ii^\alpha)}\;
\EE_\star^{-\frac\II\hbar u^\alpha H_\alpha}(p^\oo,q_\ii),
\]
where $\EE_\star$ is the star exponential with respect to the initial star product \eqref{e:iistar}. We have
\[
\Omega_\ii Z_\text{II}^u = c^\alpha_\ii Z_\text{II}^u\star H_\alpha,\quad
\Omega_\oo Z_\text{II}^u = c^\alpha_\ii H_\alpha\star Z_\text{II}^u,
\]
so the mQME $\Omega Z_\text{II}^u=0$ is satisfied, provided that \eqref{e:HHstar}, which now simply reads $[H_\alpha,\stackrel{\star}, H_\beta]=0$, holds. Let us now consider a family of gauge fixings as above, parametrized by a path $u(x)$, $x$ is some interval, and check gauge-fixing independence under the assumption \eqref{e:HHstar}. We have
\[
\frac{\dd Z_\text{II}^u}{\dd x}=-\frac\II\hbar \frac{\ddd u^\alpha}{\ddd x} H_\alpha\star Z_\text{II}^u;
\]
that is,
\begin{equation}\label{e:dZII}
\frac{\dd Z_\text{II}^u}{\dd x} = \Delta^\Omega \psi
\end{equation}
with $\psi=-b_\alpha^\oo\frac{\ddd u^\alpha}{\ddd x}Z_\text{II}^u$ and $\Delta^\Omega=\frac1{\hbar^2}\Omega$ in accordance with \eqref{e:DeltaOmega}. The general case with nonvanishing structure constants is a bit more complicated. We only report the final computations. As in Section~\ref{s:fLiealgebra}, we view the $f_{\alpha\beta}^\gamma$s as structure constants of a Lie algebra $\frg$ for a certain basis; we then view
objects with an upper Greek index as taking value in $\frg$ and  objects with a lower Greek index as taking values in $\frg^*$; finally, we denote by
$\langle\ ,\ \rangle$ the pairing of $\frg^*$ with $\frg$. We then get
\[
Z_\text{II}^u =
\EE^{-\frac\II\hbar\left(p^\oo_iq_\ii^i-\langle b^\oo,\, \Ad_{\EE^u} c_\ii\rangle\right)}\;
\EE_\star^{-\frac\II\hbar  \langle H,\,u\rangle}(p^\oo,q_\ii).
\]
Therefore, we have
\begin{align*}
\Omega_\ii Z_\text{II}^u &= Z_\text{II}^u\star \langle H,\, c_\ii\rangle -\frac12\langle b^\oo,\,\Ad_{\EE^u} [c_\ii,c_\ii]\rangle\, Z_\text{II}^u,\\ 
\Omega_\oo Z_\text{II}^u &= \langle H,\, \Ad_{\EE^u} c_\ii \rangle \star Z_\text{II}^u-\frac12\langle b^\oo,\,\Ad_{\EE^u} [c_\ii,c_\ii]\rangle\, Z_\text{II}^u,\\
\intertext{and hence}
\Omega Z_\text{II}^u &= Z_\text{II}^u\star\left(\langle H,\, c_\ii\rangle
-\left\langle
\EE_\star^{\frac\II\hbar  \langle H,\,u\rangle}
\star H \star
\EE_\star^{-\frac\II\hbar  \langle H,\,u\rangle}
,\, \Ad_{\EE^u} c_\ii \right\rangle
\right).
\end{align*}
Using \eqref{e:HHstar}, one can then easily prove that
\[
\chi(s):=
\left\langle
\EE_\star^{\frac\II\hbar  s\langle H,\,u\rangle}
\star H \star
\EE_\star^{-\frac\II\hbar s \langle H,\,u\rangle}
,\, \Ad_{\EE^{su}} c_\ii \right\rangle
\]
is constant in $s$, and therefore $\Omega Z_\text{II}^u=0$. If we consider  again a family of gauge fixings parametrized by a path $u(x)$ and assume \eqref{e:HHstar}, we get again \eqref{e:dZII}, now with $\psi=-\left\langle b^\oo,\, \Ad_{\EE^u} \frac{\ddd u}{\ddd x}\right\rangle Z_\text{II}^u$.
\end{digression}

\subsection{Quantization in Case~III}
Case~III is very similar to Case~II. We proceed very quickly, just outlining the differences. First, the BV action
for the given boundary polarizations is
\[
S_\text{BV}^{\Hat f} = S_\text{BV}- p_i(t_\oo) q^i(t_\oo)-c^\alpha(t_\ii) e^+_\alpha(t_\ii).
\]
In picking the discontinuous extension, $p$ and $q$ are realized as in Case~II, with the same boundary conditions, whereas
\begin{equation}\label{e:e+cIII}
e^+(t)=\begin{cases}
-b^\ii & t=t_\ii\\
\Hat e^+(t) & t> t_\ii
\end{cases},
\qquad
c(t)=\begin{cases}
\Hat c(t) & t< t_\oo\\
c_\oo & t=t_\oo
\end{cases},
\end{equation}
with boundary conditions
\begin{equation}\label{e:e+cIIIboundary}
\Hat e^+(t_\ii)=0,  \quad \Hat c(t_\oo) = 0.
\end{equation}
The BV action adapted to this splitting is then
\begin{equation}\label{e:TildeSIII}
\begin{split}
\Tilde S_\text{BV}^{\Hat f}
= &- p_i^\oo \Hat q^i(t_\oo)-\Hat e^+_\alpha(t_\oo) c^\alpha_\oo
-\Hat p_i(t_\ii)q_\ii^i  + \Hat c^\alpha(t_\ii) b^\ii_\alpha
\\
&+\int_{t_\ii}^{t_\oo} \bigg(
\Hat p_i\ddd \Hat q^i - \Hat e^+_\alpha\ddd \Hat c^\alpha -
\Hat e^\alpha H_\alpha(\Hat p,\Hat q) \\ 
&+ \Hat c^\alpha\left(\Hat q^+_i\frac{\dd H_\alpha}{\dd p_i}(\Hat p,\Hat q)-\Hat p_+^i\frac{\dd H_\alpha}{\dd q^i}(\Hat p,\Hat q)
\right)\\
&-
 \frac12 f_{\alpha\beta}^\gamma\, \Hat c^+_\gamma \Hat c^\alpha \Hat c^\beta +
 f_{\alpha\beta}^\gamma\, \Hat e^+_\gamma \Hat c^\alpha \Hat e^\beta
\bigg).
\end{split}
\end{equation}
We pick again the gauge-fixing lagrangian \eqref{e:simpleL}. This produces the same $\Hat p\Hat q$ propagator as in Case~II, see \eqref{e:propII}, but now we have
\[
\langle \Hat e^+_\beta(s)\, \Hat c^\alpha(t) \rangle = \II\hbar\delta^\alpha_\beta\,\theta(s-t).
\]
We then get
\[
Z_\text{III}(q_\ii,b^\ii,p^\oo,c_\oo)=\EE^{-\frac\II\hbar(p^\oo_iq_\ii^i- c_\oo^\alpha b^\ii_\alpha)}.
\]
One can immediately check that $\Omega Z_\text{III}=0$.

The considerations of Digression~\ref{d:cgf} of course apply also to this case.


\subsection{Quantization in Case~I}\label{s:Iq}
Finally, we consider Case~I. Now we have to subtract the final-endpoint pullback of $f=p_iq^i$ from the BV action getting
\[
S_\text{BV}^{f}
= S_\text{BV}
- p_i(t_\oo) q^i(t_\oo).
\]

Since the boundary values $c_\ii$ and $c_\oo$ are given (as coordinates of $\BB$), the fluctuations $\Hat c$ of
the ghost fields $\Hat c$, in the discontinuous splitting, will have vanishing boundary conditions at both endpoints. On the other hand, there are no endpoint conditions on the fluctuations $\Hat e^+$. Proceeding this way, we would get, as in the previous cases, a term $\int_{t_\ii}^{t_\oo}\Hat e^+_\alpha\ddd \Hat c^\alpha$ in the action. This term is not well suited for the perturbative expansion as $\ddd$ now has a kernel on the space of the $\Hat e^+$ fields (namely, the constant ones). Since we compute the functional integral perturbatively, we have to remove this constant term, call it $T^+$. More precisely, we are going to proceed as in Section~\ref{s:BVpf}; namely, in addition to a splitting $\Tilde F = \YY\oplus\BB$, we also split $\YY=\YY_1\oplus\YY_2$, where $\YY_2$ contains the fluctuations (which we will denote with a hat) over which we integrate by perturbation theory, and $\YY_1$ is a finite-dimensional, odd symplectic space  of which the $T^+$s are half of the coordinates. The other half will have to be variables $T$ which are part of the fields $e$. 
In general, we will not integrate over $\YY_1$, but in some cases this is possible (yet one has to be careful that there are in general nonequivalent choices of gauge-fixing lagrangians: see Example~\ref{e:qBVBFVtoy} and Remark~\ref{r:redI}).
In accordance with \cite{CMR15}, we call $\YY_1$ the space of residual fields.

Explicitly, the spaces $\YY_1,\YY_2,\BB$ are defined as follows. In $\YY_1$ we have coordinates $(T^\alpha,T^+_\alpha)$ with canonical odd symplectic structure $\delta T^+_\alpha\,\delta T^\alpha$. In $\YY_2$ we have the hat fields $\Hat p,\Hat q,\Hat e,\Hat c,\Hat p_+,\Hat q^+,\Hat e^+,\Hat c^+$ that must obey the following conditions
\begin{equation}\label{e:condCaseI}
\Hat q(t_\ii)=0,\quad \Hat c(t_\ii)=0, \quad \Hat p(t_\oo) = 0, \quad \Hat c(t_\oo) = 0,
\quad \int_{t_\ii}^{t_\oo}\Hat e^+\ddd t = 0,\quad \quad \int_{t_\ii}^{t_\oo}\Hat e = 0.
\end{equation}
Finally, $\BB$ has coordinates $q_\ii,c_\ii,p^\oo,c_\oo$.
We define $\Tilde F = \YY_1\oplus\YY_2\oplus\BB$ via
\[
q(t)=\begin{cases}
q_\ii & t=t_\ii\\
\Hat q(t) & t> t_\ii
\end{cases},\qquad
p(t)=\begin{cases}
\Hat p(t) & t< t_\oo\\
p^\oo & t=t_\oo
\end{cases},\qquad
c(t)=\begin{cases}
c_\ii & t=t_\ii\\
\Hat c(t) & t_\ii<t<t_\oo\\
c_\oo & t = t_\oo
\end{cases},
\]
\[
e(t)=\frac{T\,\ddd t}{t_\oo-t_\ii}+\Hat e(t),\quad
e^+(t)=T^++\Hat e^+(t),
\]
and
\[
p_+(t)=\Hat p_+(t),\quad q^+(t)=\Hat q^+(t),\quad c^+(t)=\Hat c^+(t).
\]
To get the adapted BV action, we use the replacements
\[
\begin{split}
\int_{t_\ii}^{t_\oo} p_i\ddd q^i &\to -\Hat p_i(t_\ii)q_\ii^i + \int_{t_\ii}^{t_\oo} \Hat p_i\ddd \Hat q^i,\\
\int_{t_\ii}^{t_\oo} e^+_\alpha\ddd c^\alpha &\to -T^+_\alpha\,(c_\oo^\alpha-c_\ii^\alpha)
-\Hat e^+_\alpha(t_\oo) c_\oo^\alpha + \Hat e^+_\alpha(t_\ii)c_\ii^\alpha
+ \int_{t_\ii}^{t_\oo} \Hat e^+_\alpha\ddd \Hat c^\alpha,
\end{split}
\]
which yield
\[
\begin{split}
\Tilde S_\text{BV}^{f}
&= T^+_\alpha\,(c_\oo^\alpha-c_\ii^\alpha)
- p_i^\oo \Hat q^i(t_\oo)
+\Hat e^+_\alpha(t_\oo) c_\oo^\alpha 
-\Hat p_i(t_\ii)q_\ii^i
- \Hat e^+_\alpha(t_\ii)c_\ii^\alpha
\\
&+\int_{t_\ii}^{t_\oo} \bigg(
\Hat p_i\ddd \Hat q^i - \Hat e^+_\alpha\ddd \Hat c^\alpha 
-\frac{T^\alpha\,\ddd t}{t_\oo-t_\ii}H_\alpha(\Hat p,\Hat q)
-\Hat e^\alpha H_\alpha(\Hat p,\Hat q) \\ 
&+ \Hat c^\alpha\left(\Hat q^+_i\frac{\dd H_\alpha}{\dd p_i}(\Hat p,\Hat q)-\Hat p_+^i\frac{\dd H_\alpha}{\dd q^i}(\Hat p,\Hat q)
\right) - \frac12 f_{\alpha\beta}^\gamma\, \Hat c^+_\gamma \Hat c^\alpha \Hat c^\beta \\
&- f_{\alpha\beta}^\gamma\, T^+_\gamma \Hat c^\alpha \frac{T^\beta\,\ddd t}{t_\oo-t_\ii}
 + f_{\alpha\beta}^\gamma\, \Hat e^+_\gamma \Hat c^\alpha \frac{T^\beta\,\ddd t}{t_\oo-t_\ii}
-  f_{\alpha\beta}^\gamma\, T^+_\gamma \Hat c^\alpha \Hat e^\beta 
 + f_{\alpha\beta}^\gamma\, \Hat e^+_\gamma \Hat c^\alpha \Hat e^\beta
\bigg).
\end{split}
\]
On the gauge-fixing lagrangian \eqref{e:simpleL} we have
\[
\begin{split}
\Tilde S_\text{BV}^{f}\Big|_\LL
&=T^+_\alpha\,(c_\oo^\alpha-c_\ii^\alpha)
- p_i^\oo \Hat q^i(t_\oo)
+\Hat e^+_\alpha(t_\oo) c_\oo^\alpha 
-\Hat p_i(t_\ii)q_\ii^i
- \Hat e^+_\alpha(t_\ii)c_\ii^\alpha
\\
&+\int_{t_\ii}^{t_\oo} \bigg(
\Hat p_i\ddd \Hat q^i - \Hat e^+_\alpha\ddd \Hat c^\alpha 
-\frac{T^\alpha\,\ddd t}{t_\oo-t_\ii}H_\alpha(\Hat p,\Hat q)
 \\ 
&- f_{\alpha\beta}^\gamma\, T^+_\gamma \Hat c^\alpha \frac{T^\beta\,\ddd t}{t_\oo-t_\ii}
 + f_{\alpha\beta}^\gamma\, \Hat e^+_\gamma \Hat c^\alpha \frac{T^\beta\,\ddd t}{t_\oo-t_\ii}
\bigg).
\end{split}
\]
The $\Hat p\Hat q$ propagator is the same as in Case~II, see \eqref{e:propII}, but the $\Hat c\Hat e^+$ propagator now has to take into account the conditions \eqref{e:condCaseI} on $\Hat c$ and $\Hat e^+$. We get
\begin{equation}\label{e:Hatce+prop}
\langle \Hat c^\alpha(s)\,\Hat e^+_\beta(t) \rangle = \II\hbar\delta^\alpha_\beta\,(\theta(s-t)-\phi(s))
\end{equation}
with
\[
\phi(s) = \frac{s-t_\ii}{t_\oo-t_\ii}.
\]
Note that there are no mixed terms involving the physical fields $\Hat p\Hat q$ and the ghost fields $\Hat c\Hat e^+$
in $\Tilde S_\text{BV}^{f}\Big|_\LL$, so the partition function $Z_\text{I} := \int_\LL \EE^{\frac\II\hbar \Tilde S_\text{BV}^{f}}$ is the product of a ``physical part'' $Z_\text{I}^\text{phys}$ and a ``ghost part'' $Z_\text{I}^\text{gh}$. The physical part is computed following \cite[Section~4.4]{CMR15}:
\[
Z_\text{I}^\text{phys}(q_\ii,p^\oo;T) = \EE^{-\frac\II\hbar p^\oo_iq_\ii^i}\;
\EE_\star^{-\frac\II\hbar T^\alpha H_\alpha}(p^\oo,q_\ii).
\]

\begin{remark}\label{r:semi1}
The $\star$-exponential arises from Feynman diagrams consisting of multi-edge linear graphs of arbitrary length. 
The lowest-order terms are the linear trees. Note that each vertex (either in the bulk or at one endpoint) contributes with a factor $\frac1\hbar$, whereas each edge contributes with a factor $\hbar$, so a tree of any length yields a global factor $\frac1\hbar$. Moreover, such trees simply define the (semi)classical evolution operator, i.e., the Hamilton--Jacobi action. Therefore,
\begin{equation}\label{e:qHJgeneralLie}
Z_\text{I}^\text{phys}(q_\ii,p^\oo;T) = \EE^{\frac\II\hbar (\Hat S^f_\text{HJ}(q_\ii,p^\oo,\EE^T)+O(\hbar))}
\end{equation}
with the HJ action as in \eqref{e:HJgeneralLie}.
\end{remark}
\begin{remark}[Semiclassical expansion]\label{r:semi2}
The last formula can also be obtained as follows.
As we have seen, in the Lie algebra case and in the chosen gauge fixing, the physical part completely decouples from the ghost part. Moreover, as in Remark~\ref{r:genexp},
we may view the gauge-fixed action for the physical part as the action for a mechanical system with hamiltonian $\Tilde H={\langle H,T\rangle}$ (and time variable $s=\frac{t}{t_\oo-t_\ii}$). For such a system we may go back, via a translation in the functional integral, to the usual nondiscontinuous choice of fields, $p(t)=p_{q_\ii,p^\oo}(t)+\check p(t)$ and $q(t)=q_{q_\ii,p^\oo}(t)+\check q(t)$,
and integrate over the fluctuations $\Check p$ and $\Check q$, getting the semiclassical expansion \eqref{e:qHJgeneralLie}. Note that, on the other hand, the discontinuous extension, besides being the correct one in general for BV-BFV, is also interesting because it produces the $\star$-product representation of the partition function.\footnote{Cf.\ the Clifford-exponential presentation of the nonabelian 1D Chern--Simons partition function in \cite{1dCS}.}
\end{remark}
We may summarize the content of equation \eqref{e:qHJgeneralLie} and of the two above remarks in the following
\begin{theorem}\label{t-thm2}
The physical part of
the dominant contribution of the perturbative expansion of the BV-BFV partition function, in the polarization of Case~I, is the HJ action.
\end{theorem}

The ghost part is very easy to compute in the abelian case, $f_{\alpha\beta}^\gamma=0$, since the interaction terms vanish and the boundary source terms cannot be coupled to each other:
\begin{equation}\label{e:ZIghabe}
Z_\text{I}^\text{gh,abelian}(c_\ii,c_\oo;T^+)= \EE^{-\frac\II\hbar T^+_\alpha\,(c_\ii^\alpha-c_\oo^\alpha)}.
\end{equation}
Therefore,
\begin{equation}\label{e:ZIabe}
Z_\text{I}^\text{abelian}(q_\ii,c_\ii,p^\oo,c_\oo;T,T^+) = \EE^{-\frac\II\hbar p^\oo_iq_\ii^i}\;
\EE_\star^{-\frac\II\hbar T^\alpha H_\alpha}(p^\oo,q_\ii)\;\EE^{-\frac\II\hbar T^+_\alpha\,(c_\ii^\alpha-c_\oo^\alpha)}.
\end{equation}
It is very easy in this case to check that the mQME holds. In fact, we have
\[
\Omega_\ii Z_\text{I}^\text{abelian} = c^\alpha_\ii Z_\text{I}^\text{abelian}\star H_\alpha,\quad
\Omega_\oo Z_\text{I}^\text{abelian} = c^\alpha_\oo H_\alpha\star Z_\text{I}^\text{abelian}.
\]
Therefore, since \eqref{e:HHstar} implies in this case that the $H$s $\star$-commute,
\[
\Omega Z_\text{I}^\text{abelian} = (c^\alpha_\ii-c^\alpha_\oo) H_\alpha\star Z_\text{I}^\text{abelian}.
\]
Since the right hand side is also clearly equal to $-\hbar^2\frac{\dd^2}{\dd T^\alpha\dd T^+_\alpha}Z_\text{I}^\text{abelian}$, the mQME is satisfied.

In the nonabelian case, the ghost part $Z_\text{I}^\text{gh}$ can be computed in terms of Feynman diagrams. In Figure~\ref{fig:gprop} we introduce a notation for the ghost propagator.
\begin{figure}[h!]
\begin{tikzpicture}
\node[label=below:{$\Hat e^+$}] (e) at (-1,0) {};
\node[label=below:{$\Hat c$}] (c) at (1,0) {}
edge[draw,->,dashed] (e);
\end{tikzpicture}
\caption{Ghost propagator}\label{fig:gprop}
\end{figure}
In $\Tilde S_\text{BV}^{f}\Big|_\LL$,
we have three univalent vertices (one at each endpoint, each containing $\Hat e^+$, and one in the bulk, containing $\Hat c$) and one bulk $\Hat e^+\Hat c$-bivalent vertex;
see Figure~\ref{fig:gvert}. 
\begin{figure}[h!]
\begin{tikzpicture}

\draw[thick] (-5,-.5) -- (-5,.5);
\draw[thick] (-4,-.5) -- (-4,.5);
\node  at (-3,0) {,};

\draw[thick] (-2,-.5) -- (-2,.5);
\draw[thick] (-1,-.5) -- (-1,.5);
\node  at (-0,0) {,};

\draw[thick] (.5,-.5) -- (.5,.5);
\draw[thick] (2.5,-.5) -- (2.5,.5);
\node  at (3,0) {,};

\draw[thick] (3.5,-.5) -- (3.5,.5);
\draw[thick] (5.5,-.5) -- (5.5,.5);

\node[bv,label=left:{$c_\ii$}] (cin) at (-5,0) {};
\node (cini) at (-4.3,0) {};
\draw[dashed,->]  (cini) -- (cin);

\node[bv,label=right:{$c_\oo$}] (cout) at (-1,0) {};
\node (couti) at (-1.7,0) {};
\draw[dashed,->]  (couti) -- (cout);

\node[iv,label=below:{$T^+T$}] (TT) at (2,0) {};
\node (TTT) at (1,0) {};
\draw[dashed,->]  (TT) -- (TTT);

\node (a) at (5.2,0) {};
\node[iv,label=below:{$T$}] (b) at (4.5,0) {};
\node (c) at (3.8,0) {};
\draw[dashed,->]  (a) -- (b);
\draw[dashed,->]  (b) -- (c);

\end{tikzpicture}
\caption{Ghost vertices}\label{fig:gvert}
\end{figure}
If a univalent $\Hat c$-vertex is present, we can only build a linear tree, of any length, connecting it to one of the endpoint $\Hat e^+$-vertices via the bivalent vertices: these trees will give rise to the $F$-terms below (in which we include also the ``abelian'' term $T^+_\alpha\,(c_\ii^\alpha-c_\oo^\alpha)$); see Figure~\ref{fig:Fdiag} (where the grey vertices are integrated along the interval; the dots denote other possible vertices).
\begin{figure}[h!]
\begin{tikzpicture}

\draw[thick] (-6,-.5) -- (-6,.5);
\draw[thick] (-1,-.5) -- (-1,.5);
\node  at (0,0) {,};

\draw[thick] (1,-.5) -- (1,.5);
\draw[thick] (6,-.5) -- (6,.5);

\node[bv,label=left:{$c_\ii$}] (cin) at (-6,0) {};
\node[iv,label=below:{$T$}] (e) at (-5.4,0) {};
\node[iv,label=below:{$T$}] (d) at (-4.8,0) {};
\node[iv,label=below:{$T$}] (c) at (-4,0) {};
\node[iv,label=below:{$T$}] (b) at (-3.4,0) {};
\node[iv,label=below:{$T^+T$}] (a) at (-2.5,0) {};

\draw[dashed,->]  (a) -- (b);
\draw[dashed,->]  (b) -- (c);
\draw[thick,dotted]  (c) -- (d);
\draw[dashed,->]  (d) -- (e);
\draw[dashed,->]  (e) -- (cin);

\node[bv,label=right:{$c_\oo$}] (cout) at (6,0) {};
\node[iv,label=below:{$T$}] (e') at (5.4,0) {};
\node[iv,label=below:{$T$}] (d') at (4.8,0) {};
\node[iv,label=below:{$T$}] (c') at (4,0) {};
\node[iv,label=below:{$T$}] (b') at (3.4,0) {};
\node[iv,label=below:{$T^+T$}] (a') at (2.5,0) {};

\draw[dashed,->]  (a') -- (b');
\draw[dashed,->]  (b') -- (c');
\draw[thick,dotted]  (c') -- (d');
\draw[dashed,->]  (d') -- (e');
\draw[dashed,->]  (e') -- (cout);

\end{tikzpicture}
\caption{$F$-diagrams}\label{fig:Fdiag}
\end{figure}
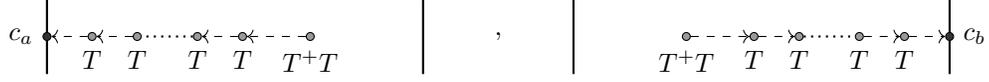
In addition, we have wheels of any length using only bivalent vertices which produce the $\mathbb{W}$-terms below 
(note that, thanks to the unimodularity assumption of Remark~\ref{r:unimodularity}, there are no tadpoles, which would otherwise have to be regularized); see Figure~\ref{fig:Wdiag} (where we use a two-dimensional representation just for pictorial reasons: the grey vertices are still integrated along the interval).
\begin{figure}[h!]
\begin{tikzpicture}

\draw[thick] (-5,-1.5) -- (-5,1.5);
\draw[thick] (5,-1.5) -- (5,1.5);

\node[iv,label=right:{$T$}] (A) at (2,0) {};
\node[iv,label=above:{$T$}] (B) at (1,1) {};
\node[iv,label=above:{$T$}] (C) at (-1,1) {};
\node[iv,label=left:{$T$}] (D) at (-2,0) {};
\node[iv,label=below:{$T$}] (E) at (-1,-1) {};
\node[iv,label=below:{$T$}] (F) at (1,-1) {};

\draw[dashed,->]  (A) -- (B);
\draw[dashed,->]  (B) -- (C);
\draw[thick,dotted]  (C) -- (D);
\draw[dashed,->]  (D) -- (E);
\draw[dashed,->]  (E) -- (F);
\draw[dashed,->]  (F) -- (A);

\end{tikzpicture}
\caption{$\mathbb{W}$-diagrams}\label{fig:Wdiag}
\end{figure}
To get a more readable expression, we use again the notations of Section~\ref{s:fLiealgebra}, and of the second part of Digression~\ref{d:cgf}, in terms of the Lie algebra $\frg$, its dual $\frg^*$ and their pairing $\langle\ ,\ \rangle$. We get
\begin{equation}
Z_\text{I}(q_\ii,c_\ii,p^\oo,c_\oo;T,T^+) = \EE^{-\frac\II\hbar p^\oo_iq_\ii^i}\;
\EE_\star^{-\frac\II\hbar \langle H, T\rangle}\;\EE^{-\frac\II\hbar \langle T^+,F_-(\ad_T)c_\oo+F_+(\ad_T)c_\ii\rangle}\;\EE^{ \mathbb{W}(T)}.
\end{equation}
The series $F_+,F_-,\mathbb{W}$ can be explicitly computed, although we do not present the details here:
\begin{align*}
F_+(x)&=\frac{x}{1-e^{-x}}=
\sum_{n\geq 0}(-1)^n\frac{B_{n}}{n!}x^n,\\
F_-(x)&=-\frac{x}{e^x-1}=
-\sum_{n\geq 0} \frac{B_n}{n!} x^n,\\
 \mathbb{W}(T)&= \sum_{n\geq 1} \frac{B_n}{n\cdot n!}\, \tr(\ad_T)^n  =
 \tr \log \frac{\sinh \frac{\ad_T}{2}}{\frac{\ad_T}{2}},
\end{align*}
where the $B_n$s are the Bernoulli numbers: $B_0=1$, $B_1=-\frac12$, $B_2=\frac16$, $B_3=0$, $B_4=-\frac{1}{30}$,\dots. The mQME can also be directly verified.

\begin{remark}[Reduction of the residual fields]\label{r:redI}
We may sometimes perform a further BV pushforward on $\YY_1$. One possibility, which leads to a less interesting result as in Case~II, consists in choosing the lagrangian $\LL^\text{triv}=\{T=0\}$. We get in this case
\begin{equation}\label{e:Itriv}
Z_\text{I}^{\LL^\text{triv}}(q_\ii,c_\ii,p^\oo,c_\oo)= \EE^{-\frac\II\hbar p^\oo_iq_\ii^i}\;\delta(c_\oo-c_\ii).
\end{equation}
More interesting choices consist in taking $\LL$ to be the conormal bundle, shifted by $-1$, of a submanifold of the Lie algebra $\g$. For example, we can take the whole $\frg$ as such a submanifold, and therefore $\LL=\{T^+=0\}$, which yields
\begin{equation}\label{e:ZL}
Z_\text{I}^{\LL}(q_\ii,p^\oo)=\int_\g  \EE^{-\frac\II\hbar p^\oo_iq_\ii^i}\;
\EE_\star^{-\frac\II\hbar \langle H, T\rangle}\;d^kT,
\end{equation}
where $k=\dim\g$ and $d^kT$ is the Lebesgue measure (of course under the assumption that the integral converges).
Note that $Z_\text{I}^{\LL}$ does no longer depend on the ghost variables.
\end{remark}

\begin{remark}[Passing to the group]\label{r:passtogroup}
In the nonabelian case, assuming the exponential map to be surjective, 
 it is convenient to make a change of variables, setting $g=\psi(T):=\EE^T$. A first advantage of this transformation is that the semiclassical term of the physical part of
 $Z_\text{I}$ is brought to the form of Section~\ref{s:fLiealgebra}. Since we are in the BV setting, we have to complete this transformation to a symplectomorphism. 
 We do it by taking the cotangent lift: 
 \[
 \Psi\colon \begin{array}[t]{ccc}
 T^*[-1]\g &\to &T^*[-1]G\\ 
 (T,T^+) &\mapsto &\left(\EE^T, 
 (\ddd_e l_{\EE^T})^{*,-1}\frac{\ad_T^*}{1-\EE^{-\ad_T^*}}T^+
 \right),
 \end{array}
 \]
where $e$ denotes the identity element of $G$ and
 $l_g$ the left multiplication by $g$:
$l_g(h)=gh$. The first term of the ghost part now simplifies as follows:
\[
\langle T^+,F_-(\ad_T)c_\oo+F_+(\ad_T)c_\ii\rangle  =
\langle (\ddd_e l_g)^*g^+,c_\ii-\Ad_g^{-1}c_\oo
\rangle.
\]
Finally, we have to observe that, more appropriately for the BV formalism \cite{Khud,Sev}, the partition function should be regarded as a half-density:
\[
\Hat Z_\text{I} = Z_\text{I} \sqrt{d^kTd^kT^+},
\]
with $k=\dim\g$, $d^k T$ the Lebesgue measure, and $d^kT^+$ the standard Berezinian measure.
We have, again under the assumption of unimodularity (see Remark~\ref{r:unimodularity}),\footnote{In the case of a nonunimodular Lie algebra, one may extend the $\Hat c\Hat e^+$ propagator of \eqref{e:Hatce+prop} to the diagonal as
$\langle\Hat c^\alpha(s)\,\Hat e^+_\beta(s) \rangle = -\II\hbar\delta^\alpha_\beta\,\phi(s)$. With this regularization, the tadpoles produce the additional factor $\EE^{-\frac12\tr\,\ad_T}$ in $Z_\text{I}$. On the other hand, without the assumption of unimodularity, one has
$\psi^*\mu_G = \EE^{ \mathbb{W}(T)}\EE^{-\frac12\tr\,\ad_T} d^kT$, so the results discussed in this remark in principle also hold without the assumption of unimodularity.}
that 
\[
\psi^*\mu_G = \det\frac{\sinh \frac{\ad_T}{2}}{\frac{\ad_T}{2}}\ d^kT = \EE^{ \mathbb{W}(T)} d^kT,
\]
where $\mu_G$ is the Haar measure.
Therefore,
\[
\Hat Z_\text{I}(q_\ii,c_\ii,p^\oo,c_\oo;g,g^+) =
\EE^{-\frac\II\hbar p^\oo_iq_\ii^i}\;
\EE_\star^{-\frac\II\hbar \langle H, \log g\rangle}\;\EE^{-\frac\II\hbar \langle (\ddd_e l_g)^*g^+,c_\ii-\Ad_g^{-1}c_\oo
\rangle}\; \mu_G,
\]
where we reinterpret the Haar measure $\mu_G$ on $G$ as a half-density on $T^*[-1]G$. We can now ``improve''
\eqref{e:ZL} by taking $\LL:=\{g^+=0\}=G$---the zero section of $T^*[-1]G$---getting
\begin{equation}\label{e:ZLG}
Z_\text{I}^{\LL}(q_\ii,p^\oo)=\int_G  \EE^{-\frac\II\hbar p^\oo_iq_\ii^i}\;
\EE_\star^{-\frac\II\hbar \langle H, \log g\rangle}\; d\mu_G,
\end{equation}
which converges for $G$ compact. 
The integral in \eqref{e:ZL} usually does not converge, as it corresponds to summing infinitely many copies of the integral over $G$ (this is an instance of what is known as Gribov's ambiguity in the physics literature). The ``improved'' formula \eqref{e:ZLG} corresponds to choosing one single domain in $\g$ diffeomorphic to $G$ (up to a measure zero subset). A similar procedure was developed in \cite{2dYM}.
\end{remark}

We conclude by discussing two important examples.

\begin{example}[Linear case]\label{e:flinearq}
Consider linear constraints as in Example~\ref{e:flinear}. There are then two simplifications. The first is that in the computation of $Z_I^\text{phys}$ we only have linear vertices in $\Hat q$ and $\Hat p$, so the only connected Feynman graphs have one single edge: connecting two boundary points, or one boundary and a bulk point, or two bulk points. See Figure~\ref{fig:gpprop} for our notation for the physical propagator and
Figure~\ref{fig:linear} for the connected Feynman graphs.
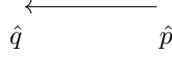
\begin{figure}[h!]
\begin{tikzpicture}
\node[label=below:{$\Hat q$}] (e) at (-1,0) {};
\node[label=below:{$\Hat p$}] (c) at (1,0) {}
edge[draw,->] (e);
\end{tikzpicture}
\caption{Physical propagator}\label{fig:gpprop}
\end{figure}
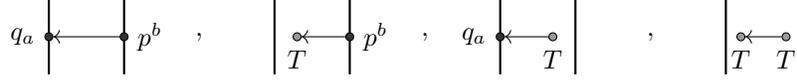
\begin{figure}[h!]
\begin{tikzpicture}

\draw[thick] (-5,-.5) -- (-5,.5);
\draw[thick] (-4,-.5) -- (-4,.5);
\node  at (-3,0) {,};

\draw[thick] (-2,-.5) -- (-2,.5);
\draw[thick] (-1,-.5) -- (-1,.5);
\node  at (-0,0) {,};

\draw[thick] (1,-.5) -- (1,.5);
\draw[thick] (2,-.5) -- (2,.5);
\node  at (3,0) {,};

\draw[thick] (4,-.5) -- (4,.5);
\draw[thick] (5,-.5) -- (5,.5);

\node[bv,label=right:{$p^\oo$}] (pout) at (-4,0) {};
\node[bv,label=left:{$q_\ii$}] (qin) at (-5,0) {};
\draw[->]  (pout) -- (qin);

\node[bv,label=right:{$p^\oo$}] (pout') at (-1,0) {};
\node[iv,label=below:{$T$}] (a) at (-1.7,0) {};
\draw[->]  (pout') -- (a);

\node[bv,label=left:{$q_\ii$}] (qin') at (1,0) {};
\node[iv,label=below:{$T$}] (b) at (1.7,0) {};
\draw[->]  (b) -- (qin');

\node[iv,label=below:{$T$}] (c) at (4.8,0) {};
\node[iv,label=below:{$T$}] (d) at (4.2,0) {};
\draw[->]  (c) -- (d);

\end{tikzpicture}
\caption{Linear case}\label{fig:linear}
\end{figure}
These graphs correspond exactly to the terms in \eqref{e:Sfablin} (this result is of course consistent with Remarks~\ref{r:semi1} and~\ref{r:semi2}.) The second simplification is that the structure functions vanish, so the ghost part is simply the one of equation \eqref{e:ZIghabe}. In conclusion, we get
\begin{equation}\label{e:ZIlin}
\begin{split}
Z_\text{I}(q_\ii,c_\ii,p^\oo,c_\oo;T,T^+) &= \EE^{-\frac\II\hbar\left( 
p^\oo_i q_\ii^i
+T^\alpha(p^\oo_i v_\alpha^i + w_{\alpha,i}q_\ii^i)+
\frac12 T^\alpha T^\beta A_{\alpha\beta}+
T^+_\alpha\,(c_\ii^\alpha-c_\oo^\alpha)\right)}\\
&= \EE^{\frac\II\hbar
\Hat S^f_\text{HJ}(q_\ii,p^\oo,T)}
\EE^{-\frac\II\hbar T^+_\alpha\,(c_\ii^\alpha-c_\oo^\alpha)},
\end{split}
\end{equation}
with $\Hat S^f_\text{HJ}$ as in \eqref{e:Sfablin}. This result is the prototype for abelian Chern--Simons theory, see \cite[Sections 5.3.1 and 6]{CS_cyl}. We may also integrate over the residual fields with the gauge fixing $T^+=0$ as in \eqref{e:ZL}. In general, there will be some linear terms in $T$, corresponding to the kernel of $A$, which produce delta functions. The remaining terms yield a Gaussian integration which produces quadratic terms in the $p,q$ variables. This is the quantum analogue of solving some of the constraints. For example, in the case of Remark~\ref{r:solveT}, we get
\[
Z_\text{I}^{\LL}(q_\ii,p^\oo)=\EE^{\frac\II\hbar\Tilde S^f(q_\ii,p^\oo)}
=\EE^{\frac\II\hbar\left(
\frac12\frac vw (p^\oo)^2 + \frac12\frac wv (q_\ii)^2\right)}.
\]
\end{example}

\begin{example}[Biaffine case]\label{exa:biaffineq}
In the case of constraints that are affine both in the $p$ and in the $q$ variables, see Example~\ref{exa:biaffine},
the physical part of the partition function can still be computed explicitly, as no multiple edges appear. We then end up having only linear trees in the computation of the $\star$-exponential; see Figure~\ref{fig:biaffine}. 
\begin{figure}[h!]
\begin{tikzpicture}

\draw[thick] (-5.5,-.5) -- (-5.5,.5);
\draw[thick] (-3.5,-.5) -- (-3.5,.5);
\node  at (-2.8,0) {,};

\draw[thick] (-2.5,-.5) -- (-2.5,.5);
\draw[thick] (-.5,-.5) -- (-.5,.5);
\node  at (.2,0) {,};

\draw[thick] (1,-.5) -- (1,.5);
\draw[thick] (3,-.5) -- (3,.5);
\node  at (3.25,0) {,};

\draw[thick] (3.5,-.5) -- (3.5,.5);
\draw[thick] (5.5,-.5) -- (5.5,.5);

\node[bv,label=right:{$p^\oo$}] (pout) at (-3.5,0) {};
\node[iv,label=below:{$T$}] (a1) at (-4,0) {};
\node[iv,label=below:{$T$}] (b1) at (-5,0) {};
\node[bv,label=left:{$q_\ii$}] (qin) at (-5.5,0) {};
\draw[->]  (pout) -- (a1);
\draw[thick,dotted] (a1) -- (b1);
\draw[->] (b1) -- (qin);

\node[bv,label=right:{$p^\oo$}] (pout') at (-.5,0) {};
\node[iv,label=below:{$T$}] (a2) at (-.9,0) {};
\node[iv,label=below:{$T$}] (b2) at (-1.6,0) {};
\node[iv,label=below:{$T$}] (c2) at (-2,0) {};
\draw[->]  (pout') -- (a2);
\draw[thick,dotted] (a2) -- (b2);
\draw[->] (b2) -- (c2);

\node[bv,label=left:{$q_\ii$}] (qin') at (1,0) {};
\node[iv,label=below:{$T$}] (c3) at (1.4,0) {};
\node[iv,label=below:{$T$}] (b3) at (2.1,0) {};
\node[iv,label=below:{$T$}] (a3) at (2.5,0) {};
\draw[->]  (a3) -- (b3);
\draw[thick,dotted] (b3) -- (c3);
\draw[->] (c3) -- (qin');

\node[iv,label=below:{$T$}] (d4) at (3.7,0) {};
\node[iv,label=below:{$T$}] (c4) at (4.1,0) {};
\node[iv,label=below:{$T$}] (b4) at (4.9,0) {};
\node[iv,label=below:{$T$}] (a4) at (5.3,0) {};
\draw[->]  (a4) -- (b4);
\draw[thick,dotted] (b4) -- (c4);
\draw[->] (c4) -- (d4);

\end{tikzpicture}
\caption{Biaffine case}\label{fig:biaffine}
\end{figure}
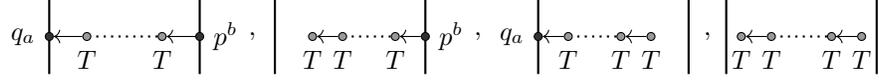
Therefore, we get \eqref{e:HJgeneralLie} 
with HJ action \eqref{e:Sbiaffine}, putting $g=\EE^T$,
and with no $O(\hbar)$ corrections. The partition half-density is then more naturally expressed in group variables as in Remark~\ref{r:passtogroup}. We get
\begin{multline}\label{e:ZIbiaff}
\Hat Z_\text{I}(q_\ii,c_\ii,p^\oo,c_\oo;g,g^+) =
\EE^{\frac\II\hbar
\Hat S^f_\text{HJ}(q_\ii,p^\oo,g)}
\;\EE^{-\frac\II\hbar \langle (\ddd_e l_g)^*g^+,c_\ii-\Ad_g^{-1}c_\oo
\rangle}\, \mu_G\\
=
\;\EE^{-\frac\II\hbar \left(
p^\oo R_{g}^{-1}q_\ii
+p^\oo R_{g}^{-1}\Phi(g)
+\Psi(g)q_\ii+\text{WZW}(g)
+\langle (\ddd_e l_g)^*g^+,c_\ii-\Ad_g^{-1}c_\oo
\rangle\right)}\, \mu_G.
\end{multline}
In the particular case of the adjoint representation, using the results of Example~\ref{exa:adjoint}, we may also write
\begin{multline*}
\Hat Z_\text{I}(q_\ii,c_\ii,\bar q^\oo,c_\oo;g,g^+) =\\=
\EE^{-\frac\II\hbar \left(
\langle \bar q^\oo,{g}^{-1}q_\ii g\rangle
+\langle \bar q^\oo,{g}^{-1}v(g)\rangle
+\langle \bar v(g)g^{-1},q_\ii\rangle
+\text{WZW}(g)
+\langle (\ddd_e l_g)^*g^+,c_\ii-\Ad_g^{-1}c_\oo
\rangle\right)}\, \mu_G
\end{multline*}
with the WZW term as in \eqref{e:WZWad}. This result is the prototype for nonabelian Chern--Simons theory, see \cite[Sections~5.3 and 5.4]{CS_cyl}. Also in this case we may integrate over the residual fields, now with the gauge fixing $g^+=0$, as in \eqref{e:ZLG}. We get
\[
Z_\text{I}^{\LL}(q_\ii,\bar q^\oo)=\int_G
\EE^{-\frac\II\hbar \left(
\langle \bar q^\oo,{g}^{-1}q_\ii g\rangle
+\langle \bar q^\oo,{g}^{-1}v(g)\rangle
+\langle \bar v(g)g^{-1},q_\ii\rangle
+\text{WZW}(g)
\rangle\right)}\, d\mu_G.
\]
\end{example}

\subsection{Gluing}\label{s:gluing}
Partition functions may be composed as described in Section~\ref{s:compo}. Here we are interested in composing
a partition function from $(q_\ii,c_\ii)$ to $(p^1,c_1)$ with 
a partition function from $(p^1,c_1)$ to $(q_2,c_2)$ with 
a partition function from $(q_2,c_2)$ to $(p^\oo,b^\oo)$, where $(p^1,c_1)$ and $(q_2,c_2)$ are some intermediate boundary fields. The first partition function belongs to Case~I and the last to Case~II. The intermediate one will be discussed below. The result will be a partition function from $(q_\ii,c_\ii)$ to $(p^\oo,b^\oo)$, so of type II. See 
Figure~\ref{fig:comp1}.

\begin{figure}[h!]
\begin{tikzpicture}
\node[bdry,label=below:{$(q_\ii,c_\ii)$}] (w1) at (-4,0) {};
\node[bdry,label=below:{$(p^\oo,b^\oo)$}] (w2) at (-2,0) {}
edge[draw] node[below] {$Z^\text{new}_\text{II}$} (w1);
\node[coordinate,label={$=$}] at (-1,-0.2) {};
\node[bdry,label=below:{$(q_\ii,c_\ii)$}] (v1) at (0,0) {};
\node[b2,label=below:{$(p^1,c_1)$}] (v2) at (2,0) {}
edge[draw] node[below] {$Z_\text{I}$} (v1);
\node[b2,label=below:{$(q_2,c_2)$}] (v3) at (4,0) {}
edge[draw,black] node[below] {$Z^{\LL^\text{triv}}_\text{I'}$} (v2);
\node[bdry,label=below:{$(p^\oo,b^\oo)$}] (v4) at (6,0) {}
edge[draw,black] node[below] {$Z_\text{II}$} (v3);
\end{tikzpicture}
\caption{Composition of the partition function $Z_\text{II}^\text{new}$}\label{fig:comp1}
\end{figure}
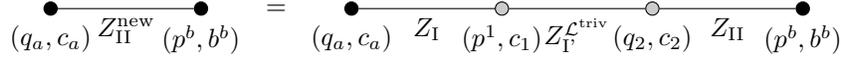

The intermediate partition function can be computed along the lines of Case~I, with the only difference that we now fix $p$ at the initial endpoint and $q$ at the final one (instead of the other way around). We are actually interested in the version where we integrate out the residual fields, as we did in Remark~\ref{r:redI} for Case~I. The analogue of \eqref{e:Itriv} is now
\begin{equation}\label{e:I'triv}
Z_\text{I'}^{\LL^\text{triv}}(p^\ii,c_\ii,q_\oo,c_\oo)= \EE^{\frac\II\hbar p^\ii_iq_\oo^i}\;\delta(c_\oo-c_\ii),
\end{equation}
and we refer to this as Case~I', in the trivial gauge fixing.

We now proceed to the computation of the ``new'' version of Case~II as
\begin{multline*}
Z_\text{II}^\text{new}(q_\ii,c_\ii,p^\oo,b^\oo;T,T^+):=\int
Z_\text{I}(q_\ii,c_\ii,p^1,c_1;T,T^+)\\
Z_\text{I'}^{\LL^\text{triv}}(p^1,c_1,q_2,c_2)\,
Z_\text{II}(q_2,c_2,p^\oo,b^\oo)\;
\frac{d^np^1\,d^nq_2}{(2\pi\hbar)^n}\ {d^kc_1\,d^kc_2}(-\II\hbar)^k,
\end{multline*}
where $n$ is the dimension of the target configuration space, $k$ is the dimension of the Lie algebra, and we have conveniently normalized the measure. A simple computation yields
\begin{multline*}
Z_\text{II}^\text{new}(q_\ii,c_\ii,p^\oo,b^\oo;T,T^+) =
\EE^{-\frac\II\hbar p^\oo_iq_\ii^i}\;
\EE_\star^{-\frac\II\hbar \langle H, T\rangle}\\
\EE^{\frac\II\hbar \langle T^+,F_+(\ad_T)c_\ii\rangle}\;\EE^{ \mathbb{W}(T)}\;
\delta(b^\oo+F_-(\ad_T)^*T^+).
\end{multline*}

If we now reduce it using the ``trivial'' gauge-fixing lagrangian $\{T=0\}$, we get exactly $Z_\text{II}$ as in
\eqref{e:II}. 

If we instead choose the gauge-fixing lagrangian $\{T^+=0\}$, possibly passing to group variables, typically we get a less trivial, nonequivalent result. 

This shows that the results of Section~\ref{s:IIq} actually correspond to choosing a somewhat trivial gauge fixing, but that more interesting results can be obtained also in this case. We saw in Digression~\ref{d:cgf} that just deforming the gauge fixing, picking $e$ different from zero, does not change the result (in $\Delta^\Omega$-cohomology). To get interesting results, we should actually split the space $\YY$ as $\YY_1\oplus\YY_2$, make (the analogue of) gauge fixing \eqref{e:simpleL}  and then choose a nontrivial gauge fixing on $\YY_1$. To get a glimpse of (the outcome of) this procedure, take the formula that defined $Z_\text{II}^\text{new}$ but now without integrating over $c_1$. We get
\begin{multline*}
\Tilde Z_\text{II}^\text{new}(q_\ii,c_\ii,p^\oo,b^\oo;T,T^+,c_1) =
\EE^{-\frac\II\hbar p^\oo_iq_\ii^i}\;
\EE_\star^{-\frac\II\hbar \langle H, T\rangle}\\
\EE^{\frac\II\hbar \langle T^+,F_-(\ad_T)c_1+F_+(\ad_T)c_\ii\rangle}\;\EE^{ \mathbb{W}(T)}\;
\EE^{\frac\II\hbar \langle b^\oo,c_1\rangle}.
\end{multline*}
We should view this as a BV pushforward, where we have set to zero the momenta $c_1^+$ of $c_1$.
The residual fields $T,T^+,c_1,c_1^+$ arise as part of the fields $e,e^+,c,c^+$.

\begin{remark}
The above discussion shows that leaving some residual fields is important not to end up with a trivial answer. It moreover teaches us that, in gluing several intervals together, it is enough to retain residual fields in the partition function associated to one interval, whereas on each of the others we can safely consider a fully reduced, ``trivial'' partition function. Similar considerations were presented in \cite[Remark~3.13]{2dYM}.
\end{remark}

\subsection{Quantum mechanics}\label{s:QM}
The system described in Section~\ref{r:SwithH} is a particular case of what we have considered above. 
The only, notational, difference is that we have an additional position variable $\ttt$ and its momentum $E$; moreover, they appear in the action with an extra minus sign. Since there is a single hamiltonian, $H-E$, we are in the abelian case of \eqref{e:ZIabe}.
This taken into account, we get
\begin{multline*}
Z_\text{I}^\text{QM}(q_\ii,\ttt_\ii,c_\ii,p^\oo,E^\oo,c_\oo;T,T^+) =\\= \EE^{\frac\II\hbar (E^\oo\ttt_\ii-p^\oo_iq_\ii^i)}\,
\EE_\star^{-\frac\II\hbar T (H-E)}\!(p^\oo,E^\oo,q_\ii,\ttt_\ii)\;
\EE^{\frac\II\hbar T^+\,(c_\oo-c_\ii)}.
\end{multline*}
Since $E$ commutes with the $(p,q)$ variables and $H$ does not depend on $\ttt$, we actually have
\[
\EE_\star^{-\frac\II\hbar T (H-E)}(p^\oo,E^\oo,q_\ii,\ttt_\ii)=
\EE^{\frac\II\hbar TE^\oo}\EE_\star^{-\frac\II\hbar TH}(p^\oo,q_\ii),
\]
so
\[
Z_\text{I}^\text{QM}(q_\ii,\ttt_\ii,c_\ii,p^\oo,E^\oo,c_\oo;T,T^+) = \EE^{-\frac\II\hbar p^\oo_iq_\ii^i}\,
\EE^{\frac\II\hbar E^\oo(\ttt_\ii+T)}\,
\EE_\star^{-\frac\II\hbar TH}\!(p^\oo,q_\ii)\;\EE^{\frac\II\hbar T^+\,(c_\oo-c_\ii)},
\]
which is the quantum version of \eqref{e:HJCMtE}. 

We are also interested in the partition function from $(q_\ii,\ttt_\ii,c_\ii)$ to $(q_\oo,\ttt_\oo,c_\oo)$. We can obtain it, following the strategy of Section~\ref{s:gluing}, by composing
a partition function from $(q_\ii,\ttt_\ii,c_\ii)$ to $(p^1,E^1,c_1)$  with 
a partition function from $(p^1,E^1,c_1)$ to  $(q_\oo,\ttt_\oo,c_\oo)$. The first is Case~I, $Z_\text{I}^\text{QM}$, with a relabeling of the endpoint variables. The second is of type I', which again we compute with trivial gauge fixing as in \eqref{e:I'triv}:
\begin{equation}\label{e:I'trivQM}
Z_\text{I'}^{\text{QM},\LL^\text{triv}}(p^\ii,E^\ii,c_\ii,q_\oo,\ttt_\oo,c_\oo)= \EE^{\frac\II\hbar p^\ii_iq_\oo^i}
\,\EE^{-\frac\II\hbar E^\ii\ttt_\oo}
\;\delta(c_\oo-c_\ii).
\end{equation}
We then get
\begin{multline*}
\Tilde Z^\text{QM}(q_\ii,\ttt_\ii,c_\ii,q_\oo,\ttt_\oo,c_\oo;T,T^+):=\int
Z_\text{I}^\text{QM}(q_\ii,\ttt_\ii,c_\ii,p^1,E^1,c_1;T,T^+)\\
Z_\text{I'}^{\text{QM},\LL^\text{triv}}(p^1,E^1,c_1,q_\oo,\ttt_\oo,c_\oo)\;
\frac{d^np^1}{(2\pi\hbar)^n}\,\frac{dE^1}{2\pi\hbar}\,dc_1,
\end{multline*}
that is,
\begin{multline*}
\Tilde Z^\text{QM}(q_\ii,\ttt_\ii,c_\ii,q_\oo,\ttt_\oo,c_\oo;T,T^+)=\int
\EE^{\frac\II\hbar p^1_i(q_\oo^i-q_\ii^i)}\,
\EE_\star^{-\frac\II\hbar TH}\!(p^1,q_\ii)\\
\EE^{\frac\II\hbar T^+\,(c_\oo-c_\ii)}\,
\delta(T+\ttt_\ii-\ttt_\oo)\;
\frac{d^np^1}{(2\pi\hbar)^n}.
\end{multline*}
In this case, using the gauge-fixing lagrangian $\LL=\{T^+=0\}$,
\[
Z^\text{QM}(q_\ii,\ttt_\ii,q_\oo,\ttt_\oo):= \int \Tilde Z^\text{QM}(q_\ii,\ttt_\ii,c_\ii,q_\oo,\ttt_\oo,c_\oo;T,0)\;dT,
\]
 yields the interesting result
 \[
 Z^\text{QM}(q_\ii,\ttt_\ii,q_\oo,\ttt_\oo)=\int
 \EE^{\frac\II\hbar p^1_i(q_\oo^i-q_\ii^i)}\,
\EE_\star^{-\frac\II\hbar (\ttt_\oo-\ttt_\ii)H}\!(p^1,q_\ii)\;
\frac{d^np^1}{(2\pi\hbar)^n},
 \]
which is the quantum version of \eqref{e:HJCMtt}. To see this more clearly, observe that we have the relation
\[
\langle p|\EE^{-\frac\II\hbar T\Hat H}|q\rangle = \EE^{-\frac\II\hbar p_iq^i}
\EE_\star^{-\frac\II\hbar TH}\!(p,q)
\]
between the operator formalism and the $\star$-product formalism. Therefore,
\[
 Z^\text{QM}(q_\ii,\ttt_\ii,q_\oo,\ttt_\oo)=\int
 \EE^{\frac\II\hbar p^1_iq_\oo^i}\,
 \langle p^1|\EE^{-\frac\II\hbar (\ttt_\oo-\ttt_\ii)\Hat H}|q_\ii\rangle\;
 \frac{d^np^1}{(2\pi\hbar)^n}=
 \langle q_\oo|\EE^{-\frac\II\hbar (\ttt_\oo-\ttt_\ii)\Hat H}|q_\ii\rangle.
 \]
Therefore, we eventually get the standard quantum mechanics evolution with hamiltonian $\Hat H$ and time lapse
$\ttt_\oo-\ttt_\ii$. In our picture, however, the original theory was parametrization invariant and $(\ttt_\oo,\ttt_\ii)$
are endpoint variables (and not fixed time endpoints).

\subsection{The quantum relativistic particle}\label{s:QRP}
We can adapt the discussion of Section~\ref{s:QM} to the case of the relativistic particle introduced in Section~\ref{s:RP}. Again we have an additional position variable $\ttt$ and its momentum $E$, and we
are in the abelian case, now with hamiltonian $\frac12(p^2+m^2-E^2)$. Therefore, we have, see \eqref{e:ZIabe},
\begin{multline*}
Z_\text{I}^\text{QRP}(q_\ii,\ttt_\ii,c_\ii,p^\oo,E^\oo,c_\oo;T,T^+) =\\= \EE^{\frac\II\hbar (E^\oo\ttt_\ii-p^\oo_iq_\ii^i)}\,
\EE_\star^{-\frac\II\hbar \frac T2 (p^2+m^2-E^2)}\!(p^\oo,E^\oo,q_\ii,\ttt_\ii)\;
\EE^{\frac\II\hbar T^+\,(c_\oo-c_\ii)}.
\end{multline*}
The simplification now is that the hamiltonian only depends on momentum variables, so its $\star$-exponential is just the usual exponential, and we get
\[
Z_\text{I}^\text{QRP}(q_\ii,\ttt_\ii,c_\ii,p^\oo,E^\oo,c_\oo;T,T^+) = \EE^{\frac\II\hbar (E^\oo\ttt_\ii-p^\oo_iq_\ii^i)}\,
\EE^{-\frac\II\hbar \frac T2 ((p^\oo)^2+m^2-(E^\oo)^2)}\;
\EE^{\frac\II\hbar T^+\,(c_\oo-c_\ii)}.
\]
This shows that the physical part of the partition function, in this polarization, is exactly the exponential of $\II/\hbar$ times the HJ action \eqref{e:RPpE} with no quantum corrections.

We can obtain  the partition function from $(q_\ii,\ttt_\ii,c_\ii)$ to $(q_\oo,\ttt_\oo,c_\oo)$ exactly as in Section~\ref{s:QM}:
\begin{multline*}
\Tilde Z^\text{QRP}(q_\ii,\ttt_\ii,c_\ii,q_\oo,\ttt_\oo,c_\oo;T,T^+):=\int
Z_\text{I}^\text{QRP}(q_\ii,\ttt_\ii,c_\ii,p^1,E^1,c_1;T,T^+)\\
Z_\text{I'}^{\text{QRP},\LL^\text{triv}}(p^1,E^1,c_1,q_\oo,\ttt_\oo,c_\oo)\;
\frac{d^np^1}{(2\pi\hbar)^n}\,\frac{dE^1}{2\pi\hbar}\,dc_1,
\end{multline*}
with $Z_\text{I'}^{\text{QRP},\LL^\text{triv}}=Z_\text{I'}^{\text{QM},\LL^\text{triv}}$, see \eqref{e:I'trivQM}. We get
\[
\Tilde Z^\text{QRP}(q_\ii,\ttt_\ii,c_\ii,q_\oo,\ttt_\oo,c_\oo;T,T^+)=
\frac{\EE^{\II\frac\pi4(1-n)}}{(2\pi\hbar T)^{\frac{n+1}2}}\,
\EE^{\frac\II\hbar\left(\frac{(\Delta q)^2-(\Delta\ttt)^2}{2T}-\frac12 m^2T\right)}\;
\EE^{\frac\II\hbar T^+\,(c_\oo-c_\ii)},
\]
with $\Delta q=q_\oo-q_\ii$ and $\Delta\ttt=t_\oo-t_\ii$.
Again we have that, apart from the prefactor $T^{-\frac{n+1}2}$, the physical part of the partition function, in this polarization, is the exponential of $\II/\hbar$ times the HJ action \eqref{e:RPqtT}.

Finally, we may integrate over the residual fields using the gauge-fixing lagrangian $\LL=\{T^+=0\}$. This yields
\[
Z^\text{QRP}(q_\ii,\ttt_\ii,q_\oo,\ttt_\oo)=\int_{-\infty}^{\infty}
\frac{\EE^{\II\frac\pi4(1-n)}}{(2\pi\hbar T)^{\frac{n+1}2}}\,
\EE^{\frac\II\hbar\left(\frac{(\Delta q)^2-(\Delta\ttt)^2}{2T}-\frac12 m^2T\right)}\;dT.
\]
To make the integral well-defined, we actually have to deform the integration contour, the real line, to avoid the singularity at $T=0$: namely, we replace the interval $(-\epsilon,\epsilon)$ with a half circle of radius $\epsilon$ centered at zero in the upper or in the lower half plane, and take the limit for $\epsilon\to0$ after integrating.
 We distinguish two cases:
\begin{enumerate}
\item In the timelike case $(\Delta\ttt)^2>(\Delta q)^2$, we have to pick the half circle in the upper half plane to tame the singularity of $\frac1T$ in the exponent. The semiclassical asymptotics of $Z^\text{QRP}$ can be computed by the saddle-point approximation around the two critical points of the exponent. Each expansion gives semiclassically the exponential of $\II/\hbar$ times one of the two HJ actions in  \eqref{e:RPqt}, but now there are also quantum corrections.
\item In the spacelike case $(\Delta q)^2>(\Delta\ttt)^2$, we have on the other hand to pick the half circle in the lower half plane to tame the singularity of $\frac1T$ in the exponent. In this case, we get $Z^\text{QRP}=0$, since we can close the contour with a half circle at infinity in the lower half plane, on which the integrand is holomorphic. (Note that we cannot use this argument in the timelike case, for such a contour  would bound $T=0$. On the other hand, we cannot pick a half circle at infinity in the upper half plane because of the $m^2T$ term in the exponent.)
\end{enumerate}
Also note that $Z^\text{QRP}$ satisfies the mQME, which, in this case, is equivalent to saying that it satisfies the Klein--Gordon equation with mass square $m^2$ both with respect to $(q_\ii,\ttt_\ii)$ and with respect to $(q_\oo,\ttt_\oo)$.



\section{BV-BFV quantization with nonlinear polarizations}\label{s:nonlinpols}
In this section we discuss quantization with a nonlinear change of polarization at the final endpoint, still with finite-dimensional target. As in Remark~\ref{r:changing}, we assume we have a generating function 
$f(q,Q)$ such that
\[
p_i\ddd q^i=P_i\ddd Q^i + \ddd f.
\]
We then assume that the coordinates in $\BB_\ii$ are $(q_\ii,c_\ii)$ and the coordinates in $\BB_\oo$ are $(Q^\oo,c_\oo)$. This is then a generalization of Case~I, according to the terminology at the beginning of Section~\ref{s:linpols}. We will refer to it as INL (NL for nonlinear). Our goal is to compute the corresponding partition function $Z_\text{INL}$. Even though this can be computed directly generalizing what we did in Section~\ref{s:Iq}, we will present a simpler computation that follows the ideas of Section~\ref{s:gluing}.

We start introducing some notation. If $H_\alpha(p,q)$ is one of the constraints (or, more generally, an arbitrary function of the $p,q$ variables), we write
\begin{equation}\label{e:Htilde}
\Tilde H_\alpha(P,Q):=H_\alpha(p(P,Q),q(P,Q)),
\end{equation}
where on the right hand side we used the symplectomorphism induced by the generating function $f$.
It turns out, as briefly announced in Section~\ref{s:IIdiscosplit} 
(and proved in Appendix~\ref{s:mQME})
that the construction of a boundary operator $\Omega$ so that the partition function satisfies the mQME is possible if we make the following assumption.
\begin{assumption}\label{a:}
We assume 
\begin{equation}\label{e:quantumHtilderel-}
H_\alpha\left(\II\hbar\frac\dd{\dd q},q\right)\EE^{-\frac\II\hbar f}
=
\Tilde H_\alpha\left(-\II\hbar\frac\dd{\dd Q},Q\right)\EE^{-\frac\II\hbar f},\quad\forall\alpha.
\end{equation}
\end{assumption}
\begin{remark}
The assumption is automatically satisfied if all the hamiltonians $H_\alpha$ are linear in the $p$ variables and
 all the hamiltonians $\Tilde H_\alpha$ are linear in the $P$ variables. Note that these are precisely the quantizability conditions in geometric quantization.
 This case is relevant for 7D Chern--Simons theory as in \cite{GS}.
\end{remark}
That this is a sufficient condition follows from a delicate analysis of the Feynman diagrams near the boundary, generalizing \cite[Section 4.2]{CMR15}. We will perform this analysis in Appendix~\ref{s:mQME}.  As for this section we will simply check that, under this assumption, the partition functions we construct indeed satisfy the mQME with $\Omega$ constructed as in Section~\ref{s:IIdiscosplit}. 


A minimalistic reading of this section, skipping the intricacies of Appendix~\ref{s:mQME}, 
can of course just be that we produce a solution of the mQME, which is enough for the applications. We stress here that Assumption~\ref{a:} is satisfied by the important example of Section~\ref{s:7d}. Therefore, what we dicuss here is a toy model for the quantization of \cite{GS}, to which we will return in \cite{CS_cyl}.

\subsection{Boundary structure}\label{s:bnlin}
In the present Case~INL, we use $f$ to change the potential $\alpha_\text{BFV}$ to
\[
\alpha_\text{BFV}^{f} = \alpha_\text{BFV} + \ddd f(q_\oo,Q_\oo) 
=  p_i^\ii\ddd q^i_\ii +  b_\alpha^\ii\ddd c^\alpha_\ii
-P^\oo_i\ddd Q_\oo^i
 - b_\alpha^\oo\ddd c^\alpha_\oo
.
\]
Following the recipe \eqref{e:Omega}, we get
\begin{align*}
\Omega_\ii &= c^\alpha_\ii H_\alpha\left(\II\hbar\frac\dd{\dd q_\ii},q_\ii\right) - \frac{\II\hbar}2 f_{\alpha\beta}^\gamma\, c^\alpha_\ii c^\beta_\ii \frac\dd{\dd c^\gamma_\ii},\\
\Omega_\oo &= c^\alpha_\oo \Tilde H_\alpha\left(-\II\hbar\frac\dd{\dd Q_\oo},Q_\oo\right) + \frac{\II\hbar}2 f_{\alpha\beta}^\gamma\,  c^\alpha_\oo c^\beta_\oo\frac\dd{\dd c^\gamma_\oo}.
\end{align*}
The first consequence of Assumption~\ref{a:} is that $\Omega_\ii$ and $\Omega_\oo$ square to zero, so
$\Omega^2=0$. In fact, since the hamiltonians are linear in the momentum variables, their $\star$-commutators, in the induced $\star$-products, are the same as $\II\hbar$ times their Poisson brackets, so \eqref{e:HHstar} is automatically satisfied (for the $H$s and for the $\Tilde H$s).

\subsection{Gluing}\label{s:gluingINL}
As announced, we will compute $Z_\text{INL}$ via a gluing procedure, similarly to what we did in Section~\ref{s:gluing}. Namely, we will compose a partition function from $(q_\ii,c_\ii)$ to $(p^1,c_1)$ with 
a partition function from $(p^1,c_1)$ to $(q_2,b^2)$ with 
a partition function from $(q_2,b^2)$ to $(Q_\oo,c_\oo)$, where $(p^1,c_1)$ and $(q_2,b^2)$ are some intermediate boundary fields. See Figure~\ref{fig:comp2}.
 The first partition function belongs to Case~I. The second, which we will denote as Case~II', can be computed along the lines of Case~II with the only difference that we now fix $p$ at the initial endpoint and $q$ at the final one (instead of the other way around). The analogue of \eqref{e:II} is
\[
Z_\text{II'}(p^\ii,c_\ii,q_\oo,b^\oo)=\EE^{\frac\II\hbar(p^\ii_iq_\oo^i+b^\oo_\alpha c_\ii^\alpha)}.
\]
The third partition function belongs to the analogue of Case~III with $Q_\oo$ instead of $p^\oo$. We will 
call it Case~IIINL and compute it in Section~\ref{s:IIINL}. Equipped with all this, we have
\begin{multline}\label{e:INL}
Z_\text{INL}(q_\ii,c_\ii,Q_\oo,c_\oo;T,T^+)=\int
Z_\text{I}(q_\ii,c_\ii,p^1,c_1;T,T^+)\\
Z_\text{II'}(p^1,c_1,q_2,b^2)\,
Z_\text{IIINL}(q_2,b^2,Q_\oo,c_\oo)\; \frac{d^np^1\,d^nq_2}{(2\pi\hbar)^n}\ {d^kc_1\,d^kb^2}(-\II\hbar)^k,
\end{multline}
where again $n$ is the dimension of the target configuration space, $k$ is the dimension of the Lie algebra, and we have conveniently normalized the measure. 

\begin{figure}[h!]
\begin{tikzpicture}
\node[bdry,label=below:{$(q_\ii,c_\ii)$}] (w1) at (-4.5,0) {};
\node[bdry,label=below:{$(Q^\oo,c^\oo)$}] (w2) at (-2,0) {}
edge[draw] node[below] {$Z_\text{INL}$} (w1);
\node[coordinate,label={$=$}] at (-1,-0.2) {};
\node[bdry,label=below:{$(q_\ii,c_\ii)$}] (v1) at (0,0) {};
\node[b2,label=below:{$(p^1,c_1)$}] (v2) at (2,0) {}
edge[draw] node[below] {$Z_\text{I}$} (v1);
\node[b2,label=below:{$(q_2,b^2)$}] (v3) at (4,0) {}
edge[draw,black] node[below] {$Z_\text{II'}$} (v2);
\node[bdry,label=below:{$(Q_\oo,c_\oo)$}] (v4) at (7,0) {}
edge[draw,black] node[below] {$Z_\text{IIINL}$} (v3);
\end{tikzpicture}
\caption{Composition of the partition function $Z_\text{INL}$}\label{fig:comp2}
\end{figure}
As a composition of partition functions that satisfy the mQME, $Z_\text{INL}$ satisfies it as well.

\subsection{The Case~IIINL}\label{s:IIINL}
In this case, we have to use the function $\Check f(b^\ii,c_\ii,p^\oo,q_\oo)=f(q_\oo,Q(p^\oo,q_\oo))-c_\ii^\alpha b^\ii_\alpha$ to modify the potential
\[
\alpha_\text{BFV}^{\Check f} = \alpha_\text{BFV} + \ddd \Check f(q_\oo,Q_\oo) 
=  p_i^\ii\ddd q^i_\ii +  c^\alpha_\ii\ddd b_\alpha^\ii
-P^\oo_i\ddd Q_\oo^i
 - b_\alpha^\oo\ddd c^\alpha_\oo
.
\]
Following again the recipe \eqref{e:Omega}, we get
\begin{subequations}
\begin{align}
\Omega_\ii &=
\II\hbar H_\alpha\left(\II\hbar\frac\dd{\dd q_\ii},q_\ii\right) \frac\dd{\dd b_\alpha^\ii}
+ \frac{\hbar^2}2 f_{\alpha\beta}^\gamma\, b_\gamma^\ii\frac{\dd^2}{\dd b_\alpha^\ii\dd b_\beta^\ii},\\
\Omega_\oo &= c^\alpha_\oo \Tilde H_\alpha\left(-\II\hbar\frac\dd{\dd Q_\oo},Q_\oo\right) + \frac{\II\hbar}2 f_{\alpha\beta}^\gamma\,  c^\alpha_\oo c^\beta_\oo\frac\dd{\dd c^\gamma_\oo},\label{sube:Omegaoutnlin}
\end{align}
\end{subequations}
and we clearly have $\Omega^2=0$ (again by Assumption~\ref{a:}).
The BV action for the given boundary polarizations is then 
\[
S_\text{BV}^{\Check f} = S_\text{BV}- f(q(t_\oo),Q(p(t_\oo),q(t_\oo))) 
-c^\alpha(t_\ii) e^+_\alpha(t_\ii).
\]

Next we choose the discontinuous splitting of the fields $e^+$, $c$ as in \eqref{e:e+cIII} with boundary conditions \eqref{e:e+cIIIboundary}. For the $p$, $q$ fields we choose instead
\[
q(t)=\begin{cases}
q_\ii & t=t_\ii\\
\Hat q(t) & t>t_\ii
\end{cases},\qquad
p(t)=\begin{cases}
\Hat p(t) & t<t_\oo\\
p^0 & t=t_\oo
\end{cases},
\]
with
\begin{equation}\label{e:pzero}
p^0 := \frac{\dd f}{\dd q}(\Hat q(t_\oo),Q_\oo)
\end{equation}
and boundary conditions
\[
\Hat q(t_\ii)=0, \quad \Hat p(t_\oo) = 0.
\]
 The BV action adapted to this splitting is then the same as \eqref{e:TildeSIII} with only the first term changed:
\[
\begin{split}
\Tilde S_\text{BV}^{\Check f}
= &- f(\Hat q(t_\oo),Q_\oo)
-\Hat e^+_\alpha(t_\oo) c^\alpha_\oo
-\Hat p_i(t_\ii)q_\ii^i  + \Hat c^\alpha(t_\ii) b^\ii_\alpha
\\
&+\int_{t_\ii}^{t_\oo} \bigg(
\Hat p_i\ddd \Hat q^i - \Hat e^+_\alpha\ddd \Hat c^\alpha -
\Hat e^\alpha H_\alpha(\Hat p,\Hat q) \\ 
&+ \Hat c^\alpha\left(\Hat q^+_i\frac{\dd H_\alpha}{\dd p_i}(\Hat p,\Hat q)-\Hat p_+^i\frac{\dd H_\alpha}{\dd q^i}(\Hat p,\Hat q)
\right)\\
&-
 \frac12 f_{\alpha\beta}^\gamma\, \Hat c^+_\gamma \Hat c^\alpha \Hat c^\beta +
 f_{\alpha\beta}^\gamma\, \Hat e^+_\gamma \Hat c^\alpha \Hat e^\beta
\bigg).
\end{split}
\]
We pick again the gauge-fixing lagrangian \eqref{e:simpleL}. This produces the same $\Hat p\Hat q$ 
and $\Hat e^+\Hat c$ propagators as in Case~III (since the hat fields  here and there have exactly the same boundary conditions), and we can then easily compute
\[
Z_\text{IIINL}(q_\ii,b^\ii,Q_\oo,c_\oo)=\EE^{-\frac\II\hbar(f(q_\ii,Q_\oo)- c_\oo^\alpha b^\ii_\alpha)}.
\]
One can immediately check that $\Omega Z_\text{IIINL}=0$. Assumption~\ref{a:} plays here a fundamental role.
%

\begin{remark}[mQME vs unitarity]\label{r:NLnonuni}
The physical part of the partition function, $\EE^{-\frac\II\hbar f(q_\ii,Q_\oo)}$, is exactly the exponential of the HJ action  \eqref{e:HJonlyf} as in \eqref{e:HatKonlyf}. This is not unexpected because, since the physical part decouples from the ghost part, in the physical part we can go back to the continuous splitting, as in Example~\ref{exa:qQfinal}. Note that the partition function we get is not unitary, as it misses the corrections discussed in Digression~\ref{d:uni}. 
In principle, we might add such corrections---e.g., as in \eqref{e:KKfhalf}---if possible, but this would spoil the mQME
unless we were able to find corrections to the $H_\alpha$s and $\Tilde H_\alpha$s that satisfy \eqref{e:quantumHtilderel-} with the corrected $f$. There is no guarantee that such corrections exist.
This is the price we have to pay in the quantum BV-BFV formalism in order to get a solution to the mQME. The mQME is 
fundamental to ensure that the result of a composition, as in the procedure of Section~\ref{s:gluingINL}, also solves the mQME. As a consequence, if there is a conflict, we prefer to have $Z_\text{INL}$ solve the mQME at the expense of not being unitary. We wil return to this in Remark~\ref{r:NLnonuniEND}.
\end{remark}

\subsection{The computation of $Z_\text{INL}$}
We can now use \eqref{e:INL} to compute
\[
\begin{split}
Z_\text{INL}(q_\ii,c_\ii,Q_\oo,c_\oo;T,T^+) &=\int
\EE^{\frac\II\hbar(p^1_i(q_2^i-q_\ii^i)-f(q_2,Q_\oo))}
\EE_\star^{-\frac\II\hbar \langle H, T\rangle}(p^1,q_\ii)
\; \frac{d^np^1\,d^nq_2}{(2\pi\hbar)^n}\\
&\phantom{=\int}
\EE^{\frac\II\hbar \langle T^+,F_-(\ad_T)c_\oo+F_+(\ad_T)c_\ii\rangle}\;\EE^{ \mathbb{W}(T)}.
\end{split}
\]

\begin{remark}
If $f(q,Q)=Q_iq^i$, then one easily sees that $Z_\text{INL}=Z_\text{I}$ with $p^\oo=Q_\oo$.
\end{remark}

\begin{remark}
At the leading order in $\hbar$, the physical part of $Z_\text{I}$ is the exponential of $\II/\hbar$ times the HJ action, see \eqref{e:qHJgeneralLie}. The integral in the $p^1,q_2$ variables then produces, also at the leading order, the composition of the generating functions $f$ and $\Hat S_\text{HJ}$. Therefore, we have
\begin{equation}\label{e:ZINL}
Z_\text{INL}^\text{phys}(q_\ii,Q_\oo;T) = \EE^{\frac\II\hbar (\Hat S^f_\text{HJ}(q_\ii,Q_\oo,\EE^T)+O(\hbar))}.
\end{equation}
\end{remark}

\begin{example}[Linear case]\label{e:flinearqNL}
In the case of linear constraints of Example~\ref{e:flinear}, there are no corrections in the physical part to the HJ action, see Example~\ref{e:flinearq}. Moreover, since the HJ action \eqref{e:Sfablin} is at most linear in the momenta, the integral in the $p^1,q_2$ variables produces exactly the composition of the generating functions $f$ and $\Hat S_\text{HJ}$. Finally, the theory is abelian in this case. Therefore, we get
\begin{equation}\label{e:ZINLlin}
\begin{split}
Z_\text{INL}(q_\ii,c_\ii,Q_\oo,c_\oo;T,T^+) &= \EE^{-\frac\II\hbar\left( 
f(q_\ii+T^\alpha v_\alpha,Q_\oo)
+T^\alpha w_{\alpha,i}q_\ii^i+
\frac12 T^\alpha T^\beta A_{\alpha\beta}+
T^+_\alpha\,(c_\ii^\alpha-c_\oo^\alpha)\right)}\\
&= \EE^{\frac\II\hbar
\Hat S^f_\text{HJ}(q_\ii,Q_\oo;T)}
\EE^{-\frac\II\hbar T^+_\alpha\,(c_\ii^\alpha-c_\oo^\alpha)},
\end{split}
\end{equation}
with $\Hat S^f_\text{HJ}$ as in \eqref{e:Sfablingen}. This result is the prototype for abelian 7D Chern--Simons theory with nonlinear polarization, see \cite[Sections 6]{CS_cyl}, i.e., for the quantization of  \cite{GS}.

\begin{remark}\label{r:NLnonuniEND}
As a consequence of Remark~\ref{r:NLnonuni}, the partition function $Z_\text{INL}$ is in general nonunitary if we consider a nonlinear change of polarization. On the other hand, it satisfies the mQME. An application of this, following Section~\ref{s:compo}, consists in composing $Z_\text{INL}$ with an $\Omega$-closed state $\psi$ in the final variables $(Q_\oo,c_\oo)$. The result is a state $\psi'$ in the initial variables $(q_\ii,c_\ii)$ and in the residual fields $(T,T^+)$ which satisfies the mQME. The state $\psi'$ will not be unitarily equivalent to the state $\psi$ in general, but the idea is that we introduce the state $\psi$ just to produce the solution $\psi'$ of the mQME. In the abelian case, for $\psi$ we may take an arbitrary function of $Q_\oo$ times a delta function of all the ghost variables. This is automatically $\Omega$ closed. As an application, in \cite{CS_cyl} we will use the state \eqref{e:KS} (times the delta function of the ghosts) as the final state $\psi$. This will show that the result of \cite{GS}, which connects the Kodaira--Spencer \cite{KS} action \cite{BCOV} to 7D Chern--Simons theory (with the nonlinear Hitchin polarization \cite{Hitchin} at the final endmanifold), survives at the quantum level with no corrections.
\end{remark}

\end{example}


\appendix

\section{Symplectic geometry and generating functions}\label{a:generatingfunctions}
In this appendix we recall some basic facts on symplectic geometry (also in the infinite-dimensional case) and on (generalized) generating functions for lagrangian submanifolds (and canonical relations). Our main reference will be \cite{BW97}. For the infinite-dimensional case, we will also follow \cite{CC}. 

We will discuss first the linear case, where actually most of the subtleties are already present, and then briefly review the extension to manifolds.

\subsection{Symplectic spaces}
A symplectic form on a vector space $V$ is a skew-symmetric bilinear form $\omega$ with the property of being nondegenerate:
\begin{equation}\label{e:weaklynondeg}
\omega(v,w)=0\ \forall w\in V \Rightarrow v=0.
\end{equation}
The pair $(V,\omega)$ is called a symplectic space.

A symplectomorphism between symplectic spaces $(V,\omega_V)$ and $(W,\omega_W)$ is an isomorphism
$\phi\colon V\to W$ such that
\begin{equation}\label{e:symplectomorphism}
\omega_W(\phi(v_1),\phi(v_2))=\omega_V(v_1,v_2),\quad \forall v_1,v_2\in V.
\end{equation}

\begin{remark}
Note that to a bilinear form $\omega$ on $V$ one can associate the linear map $\omega^\sharp\colon V\to V^*$,
\[
\omega^\sharp(v)(w):=\omega(v,w), \qquad v,w\in V.
\]
The nondegeneracy of $\omega$ is equivalent to the injectivity of $\omega^\sharp$. If $V$ is finite-dimensional, then this implies that $\omega^\sharp$ is actually a bijection, but this is not guaranteed in the infinite-dimensional case.
\end{remark}

\begin{remark}
In the infinite-dimensional case, the above condition is often referred to as weak nondegeneracy and $\omega$ is called a weak symplectic form, reserving the term symplectic form to the case when $\omega^\sharp$ is a bijection. In this paper we will only consider weak symplectic forms, and we will just call them symplectic forms.
\end{remark}

\begin{example}\label{exa:weaksymplCS}
One of the infinite-dimensional examples for this paper is the symplectic space associated to abelian Chern--Simons theory. Let $\Sigma$ be a closed, oriented $2(2k+1)$-manifold. Let $V=\Omega^{2k+1}(\Sigma)$. Then the bilinear form
\[
\omega(\alpha,\beta)=\int_\Sigma\alpha\wedge\beta,\qquad \alpha,\beta\in V,
\]
is skew-symmetric and nondegenerate, hence symplectic.
\end{example}

\subsubsection{Special subspaces}
Given a subspace $W$ of $V$, one defines its orthogonal space as
\[
W^\perp:=\{v\in V \colon \omega(v,w) = 0,\ \forall w\in W\}.
\]
Note that $V^\perp=\{0\}$.

\newcommand{\pperp}{{\perp\perp}}
\begin{lemma}\label{l:perpsubsp}
Let $W$ and $Z$ be two subspaces. Then
\begin{enumerate}
\item $W\subseteq Z \Rightarrow Z^\perp\subseteq W^\perp$;
\item $(W+Z)^\perp = W^\perp \cap Z^\perp$.
\item $W\subseteq W^\pperp$ and $W^\perp = W^{\perp\perp\perp}$.
\item If $V$ is finite-dimensional, then $\dim V = \dim W + \dim W^\perp$ and
 $W=W^\pperp$.
\end{enumerate}
\end{lemma}
\begin{proof}
For $(1)$, let $v\in Z^\perp$. By definition $\omega(v,z)=0\ \forall z\in Z$. A fortiori $\omega(v,w)=0\ \forall w\in W$. So $v\in W^\perp$.

As for $(2)$, observe that $v\in (W+Z)^\perp$ if and only if $\omega(v,w+z)=0\ \forall w\in W,\ \forall z\in Z$. In particular, taking either $w$ or $z$ equal
to zero, we see that $v$ belongs to both $W^\perp$ and $Z^\perp$.
On the other hand, if $v\in W^\perp \cap Z^\perp$, by linearity we get $\omega(v,w+z)=\omega(v,w)+\omega(v,z)=0\ \forall w\in W,\ \forall z\in Z$;
so $v\in (W+Z)^\perp$.

As for $(3)$, let $w\in W$. Then $\omega(w,v)=-\omega(v,w)=0\ \forall v\in W^\perp$. Hence, $w\in W^\pperp$.
For the second equation, taking orthogonals in the first yields $W^{\perp\perp\perp}\subseteq W^\perp$.
On the other hand, applying the first equation to the subspace $W^\perp$ yields $W^\perp\subseteq W^{\perp\perp\perp}$.

In finite dimensions, $\omega$ yields an isomorphism $V/W^\perp\to W^*$, so $\dim V-\dim W^\perp=\dim W^*=\dim W$. Moreover,
\[
\dim W^\pperp = \dim V - \dim W^\perp = \dim V - (\dim V-\dim W) = \dim W.
\]
\end{proof}

A subspace $W$ is called isotropic if $W\subseteq W^\perp$, coisotropic if $W^\perp\subseteq W$ and lagrangian
if $W=W^\perp$.  

Note that a subspace $W$ is isotropic if and only if $\omega(w_1,w_2)=0$ for all $w_1,w_2\in W$ and that a subspace is lagrangian if and only if it is at the same time isotropic and coisotropic.

In the finite-dimensional case, $L$ is lagrangian if and only if $L$ is isotropic and $2\dim L=\dim V$. Finally, an isotropic subspace $L$ is called split lagrangian if it admits an isotropic complement. (For more on lagrangian and split lagrangian subspaces, we refer to \cite{CC}).


\begin{lemma}\label{l:Lag-maxiso}
\[
\text{split lagrangian}\Longrightarrow\text{lagrangian}
\]
In the finite-dimensional case, the two concepts are equivalent.
\end{lemma}


\begin{proof}
Suppose that $L$ is split lagrangian. Then we can find an isotropic subspace $W$ such that $V=L\oplus W$. In particular, $V=L+W$. By $(2)$ in Lemma~\ref{l:perpsubsp}, we then have $L^\perp\cap W^\perp=\{0\}$. Since
$W\subseteq W^\perp$, we also have $L^\perp\cap W=\{0\}$, which implies $L^\perp\subseteq L$. Therefore, $L$ is lagrangian.

\newcommand{\calI}{\mathcal{I}}
Suppose on the other hand that $L$ is lagrangian and $V$ is finite-dimensional. Suppose $W$ is isotropic but not lagrangian and $L\cap W=\{0\}$.  Observe that by dimensional reasons, $L+W$ cannot be the whole $V$. Therefore $(L+W)^\perp$ cannot be $\{0\}$. Using $(2)$ in Lemma~\ref{l:perpsubsp}, we then have $W^\perp\cap L\not=\{0\}$.

We now want to show that there is a nonzero vector $v\in W^\perp$ that does not belong to $L$ nor to $W$.
In fact, if this were not the case, we would have $W^\perp\subseteq W+L$. By $(1)$ and $(2)$ in Lemma~\ref{l:perpsubsp}, we would then have $W\supseteq W^\perp\cap L^\perp= W^\perp\cap L$. Intersecting with $L$, we would conclude $\{0\}=W^\perp\cap L$, which is a contradiction.

Let finally $\calI$ be the set of isotropic subspaces of $V$ whose intersection with $L$ is $\{0\}$. Note that 
$\{0\}\in\calI$. If $W\in\calI$ is not lagrangian, by the above argument we can find a nonzero vector $v\in W^\perp$ that does not belong to $L$ nor to $W$. But then $W\oplus\Span{v}$ is also in $\calI$ and of dimension strictly higher than that of $W$. After finitely many steps, we end up with a $W\in\calI$ which is lagrangian. But since $\dim L=\dim W=\dim V/2$, we see that $W$ is a complement to $L$.
\end{proof}

%
%
%

\begin{remark}[Split symplectic spaces]\label{r:splitss}
For applications in field theory, we will only consider symplectic spaces that admit split lagrangian subspaces. In particular, we fix a lagrangian subspace $W$ and a lagrangian complement $W'$ and write 
\begin{equation}\label{e:Vsplit}
V=W\oplus W'.
\end{equation}
We will call this a split symplectic space.

Note that in the finite-dimensional case, $\omega^\sharp$ establishes an isomorphism between $W'$ and $W^*$, which in turns establishes a symplectomorphism between $V$ and $W\oplus W^*$ with the canonical symplectic form
\begin{equation}\label{e:omegacan}
\omega(v\oplus a,w\oplus b)=b(v)-a(w),\qquad v,w\in W,\ a,b\in W^*.
\end{equation}
\end{remark}

In the infinite-dimensional case, $\omega^\sharp|_{W'}\colon W'\to W^*$ is only injective. Still, we may regard $W\oplus W'$ as a symplectic subspace of $W\oplus W^*$, which is what we will do in the following.

\begin{example}
Consider Example~\ref{exa:weaksymplCS} with $k=0$. If we work with complex valued forms, then we have the splitting of $V=\Omega^1(\Sigma)\otimes\CC$ into $W=\Omega^{1,0}(\Sigma)$ and $W'=\Omega^{0,1}(\Sigma)$, with respect to some complex structure on $\Sigma$.
\end{example}

\subsubsection{Generating functions}
Assume we have a split symplectic space as in \eqref{e:Vsplit}. 
Consider the graph of a linear map $\phi\colon W\to W'$:
\[
L:=\{w\oplus \phi(w),\ w\in W\}\subset V.
\]
Since
\[
\omega(w_1\oplus \phi(w_1),w_2\oplus \phi(w_2))=
\omega(w_1,\phi(w_2))+\omega(\phi(w_1),w_2),
\]
we see that $L$ is isotropic, hence split lagrangian because $W'$ is an isotropic complement to it,
 if and only if 
\[
\omega(w_1,\phi(w_2))=-\omega(\phi(w_1),w_2),\quad\forall w_1,w_2\in W.
\]
In this case,  we can use $\phi$ and $\omega$ to define a quadratic form on $W$:
\[
\psi(w):=\frac12 \omega(w,\phi(w)).
\]
Note then that
\[
\delta\psi(w)=\omega(\phi(w),\delta w)=\omega^\sharp(\phi(w))(\delta w),
\]
so
\[
\delta\psi=\omega^\sharp\circ\phi\colon W\to W^*
\]
as a one-form on $W$. Since $\omega^\sharp$ is injective, it is possible to recover $\phi$ out of $\delta\psi$.

Vice versa, suppose we have a splitting $V=W\oplus W'$, already realized as a subspace of $W\oplus W^*$ with canonical symplectic structure, and a quadratic form $\psi$ on $W$. If $\delta\psi(W)\subseteq W'$, then
the graph $L_\psi$ of $\delta\psi$ is a split lagrangian subspace of $V$ and $\psi$ is called a generating function for it.

\begin{example}\label{exa:genfunctCzero}
Let $I$ be the closed interval $[-1/2,1/2]$ and $V:=C^0(I)\oplus C^0(I)$ as a real vector space.
Define
\[
\omega(f\oplus g, h\oplus k)=\int_I (fk-gh)\, \ddd x.
\]
Then $(V,\omega)$ is split symplectic with $W=W'= C^0(I)$ (more properly $W=
C^0(I)\oplus \{0\}$ and $W'=\{0\}\oplus C^0(I)$). The quadratic form 
$\psi(f)=\frac12\int_I f^2\,\ddd x$ has differential $\delta_f\psi = f$, so it is a generating function for the diagonal
$L_\psi=\{f\oplus f,\ f\in C^0(I)\}$. On the other hand, the quadratic form 
$\Tilde\psi(f)=\frac12 f(0)^2$, $f\in W$, has differential $\delta_f\Tilde\psi = f(0)\delta$, with $\delta$ the delta function at $0$, which is an element of $W^*$ but not of $W'$, so it is not an allowed generating function.
\end{example}

\newcommand{\uC}{\underline{C}}
\newcommand{\uomega}{\underline{\omega}}
\newcommand{\uL}{\underline{L}}
\subsubsection{Symplectic reduction}\label{s:symplecticreduction}
We will only consider the particular case of coisotropic reduction, which is of interest for this paper.

Let $C$ be a coisotropic subspace of $(V,\omega)$. Then the kernel of the restriction of $\omega$ to $C$ is precisely the orthogonal subspace $C^\perp$. 
If follows that
$\uC:=C/C^\perp$, the symplectic reduction of $C$, is endowed with the symplectic form
\begin{equation}\label{e:uomega}
\uomega([v],[w])=\omega(v,w),\quad v\in[v],\ w\in[w].
\end{equation}

Denote by $\pi$ the canonical projection $C\to\uC$. If $L$ is a subspace of $V$, we set $\uL:=\pi(L\cap C)\subseteq\uC$.
\begin{lemma}[Relative linear reduction]\label{l:reductionL}
If $L$ is isotropic in $V$, then $\uL$ is isotropic in $\uC$. In the finite-dimensional case, if $L$ is lagrangian in $V$, then $\uL$ is lagrangian in $\uC$.
\end{lemma}
\begin{proof}
The first statement may be easily proved observing that, if $[v]$ and $[w]$ are in $\uL$, then every $v\in[v]$ and every $w\in[w]$ lie in $L\cap C$ and hence in $L$. Therefore, $\omega(v,w)=0$ and, by \eqref{e:uomega}, we get
$\uomega([v],[w])=0$, so $\uL$ is isotropic. 

A different proof, which also leads to the second statement, is based on the following:
\begin{multline*}
\pi^{-1}(\uL^\perp)=(\pi^{-1}(\uL))^\perp\cap C = (L\cap C + C^\perp)^\perp\cap C =\\
=(L\cap C)^\perp\cap C
\supset (L^\perp + C^\perp)\cap C = L^\perp \cap C + C^\perp.
\end{multline*}
In the last equality we have used the obvious fact that $(A+B)\cap C = A\cap C + B$ for any three subspaces $A$, $B$, and $C$ with $B\subset C$.
Thus, $\uL^\perp\supset \pi(L^\perp\cap C)\supset \pi(L\cap C) = \uL$. All inclusions are equalities under the assumptions that $V$ is finite-dimensional and $L$ is lagrangian.
\end{proof}

In infinite dimensions, it may happen that $L$ is lagrangian but $\uL$ is not: 
\begin{example}\label{exa:C0}
Consider the symplectic space of Example~\ref{exa:genfunctCzero} and
\[
C:=\{f\oplus g\in V : g(0)=0\}.
\]
One has that $C^\perp=\{0\}$, so $C$ is coisotropic and $\uC=C$. Let
\[
L:=\left\{f\oplus g\in V : \int_I f\;\ddd x=0,\ g \text{\ constant}\right\}.
\]
One can easily show that $L$ is lagrangian, actually split lagrangian, in $V$, but $\uL$ is not lagrangian in $\uC$.
\end{example}

\begin{remark}[Generalized generating functions in the linear case]\label{r:gflinear}
Suppose we have a sum $V\oplus Z$ of symplectic spaces with $Z=B\oplus B'$ split.
First observe that $C:=V\oplus B\oplus\{0\}$ is coisotropic in $V\oplus Z$ with $C^\perp=\{0\}\oplus B\oplus\{0\}$, so
$\uC=V$. 
Now assume that also $V=W\oplus W'$ is split and rearrange $V\oplus Z$ as the split symplectic space
$(W\oplus B)\oplus(W'\oplus B')$.
Let $L_\psi$ be the split lagrangian subspace of $V\oplus Z$ 
with generating function a quadratic function $\psi$ on $W\oplus B$. It then follows that $\uL_\psi$ is isotropic in $V$.
If it is split lagrangian (as it automatically happens in the finite-dimensional case), then one says that $\psi$ is a generalized generating function for $\uL_\psi$, depending on extra parameters in $B$.
\end{remark}

\begin{example}
Consider $V=W\oplus W^*$ and $Z=B\oplus B^*$ with $W$ and $B$ finite-dimensional. Let $F\colon W\to B^*$ be a linear map and let $F^*\colon B\to W^*$ denote its dual. Then the quadratic form $\psi(w\oplus b)=(Fw,b)$, with
$(\ ,\ )$ denoting the pairing of a space with its dual, is a generalized generating function for
$\uL_\psi=\ker F\oplus\text{im}\, F^*=\ker F\oplus(\ker F)^0$, with $ ^0$ denoting the annihilator space.
Unless $F$ is the zero map, $\uL_\psi$ cannot be generated by a function on $W$ only, so extra parameters are indeed necessary.
\end{example}

\begin{example}
The above example may be extended to the infinite-dimensional case as follows.
Let $F\colon W\to B'\subseteq B^*$ and $F'\colon B\to  W'\subseteq W^*$
be linear maps such that $\iota_{W'}F'=(\iota_{B'}F)^*\iota$, where $\iota_{B'}\colon B'\to B^*$ and $\iota_{W'}\colon W'\to W^*$ are the given inclusions and $\iota\colon B\to B^{**}$ is the canonical inclusion. Then the quadratic form $\psi(w\oplus b)=(\iota_{B'}Fw,b)=(\iota_{W'}F'b,w)$ is a generating function for a lagrangian subspace $L_\psi$ of $V\oplus Z$.
The corresponding isotropic subspace $\uL_\psi$ of $V$ is $\ker F\oplus\text{im}\, F'$.
\end{example}

\begin{example}
A simple illustration of the last example is the following. Let $\Sigma$ be a closed, oriented two-manifold and let 
$W=W'=\Omega^1(\Sigma)$, $B=\Omega^0(\Sigma)$, $B'=\Omega^2(\Sigma)$, with symplectic forms induced by integration over $\Sigma$:
\[
\omega(\alpha\oplus\alpha',\beta\oplus\beta')=\int_\Sigma (\alpha\wedge\beta'-\beta\wedge\alpha').
\] 
The linear maps $F=\ddd\colon \Omega^1(\Sigma)\to \Omega^2(\Sigma)$ and
$F'=\ddd\colon\Omega^0(\Sigma)\to\Omega^1(\Sigma)$ satisfy the assumptions, so we have the generalized generating function $\psi(w\oplus b)=\int_\Sigma \ddd w\wedge b$, with $w\in\Omega^1(\Sigma)$ and $b\in\Omega^0(\Sigma)$. We get $\uL_\psi=\Omega^1_\text{cl}(\Sigma)\oplus\Omega^1_\text{ex}(\Sigma)$, where the subscripts denote closed and exact forms. In this example $\uL_\psi$ is actually a split lagrangian. This follows from the Hodge decomposition theorem. 
 In fact, putting a metric on $\Sigma$ one can define the Hodge $*$ operator and the codifferential $\ddd^*$. As an isotropic complement of $\uL_\psi$ one can then take
$\Omega^1_\text{coex}(\Sigma)\oplus \Omega^1_\text{cocl}(\Sigma)$.

\end{example}

\subsubsection{Canonical relations}\label{a:canrel}
Canonical relations are a generalization of symplectomorphisms. 

Recall that a linear map $\phi$ between $(V,\omega_V)$ and $(W,\omega_W)$ is called a symplectomorphism if it is an isomorphism and respects the symplectic forms as in \eqref{e:symplectomorphism}. It follows that the graph of $\phi$,
\[
L_\phi=\{v\oplus \phi(v),\ v\in V\}, 
\]
is a split lagrangian subspace of $\bar V\oplus W$, where $\bar V$ denotes $V$ with symplectic form $-\omega_V$.
In fact, we have
\[
\omega_{\bar V\oplus W}(v_1\oplus \phi(v_1),v_2\oplus \phi(v_2))=-\omega_V(v_1,v_2)+\omega_W(\phi(v_1),\phi(v_2))=0,
\]
so $L_\phi$ is isotropic. Note that $L_{-\phi}$ is also isotropic and, since $\phi$ is an isomorphism, it is a complement to $L_\phi$.

More generally, one calls any (split) lagrangian subspace  $L$ of $\bar V\oplus W$ a (split) canonical relation from $V$ to $W$. One also writes
\begin{equation}\label{e:nrightarrow}
L\colon V\nrightarrow W.
\end{equation}
If $V=V_1\oplus V_2$ and $W=W_1\oplus W_2$ are split symplectic spaces, then a lagrangian subspace $L$ of
$\bar V\oplus W$ might have a (generalized) generating function $\psi$, a quadratic form on $V_1\oplus W_1(\oplus B)$, which is then called a (generalized) generating function for the canonical relation $L$.

Here is an example of a canonical relation that is important for our applications:
\begin{lemma}\label{ex:LinCC}
Let $C$ be a coisotropic subspace of $V$. Then the subspace $L_C:=\{(v,v')\in C\oplus C\ |\ v'-v\in C^\perp\}$
is isotropic in $\bar V\oplus V$. If $V$ is finite-dimensional, then $L$ is lagrangian, hence a canonical relation from $V$ to itself.
\end{lemma}
\begin{proof}
For $v_1,v_2\in C$ and $w_1,w_2\in C^\perp$ we have $\omega(v_1+w_1,v_2+w_2)-\omega(v_1,v_2)=0$, so $L_C$ is isotropic. If $\dim V=2n$ and $\dim C=k$, then $\dim L_C=\dim C+\dim C^\perp=k+(2n-k)=2n$, so $2\dim L_C=\dim(\bar V\oplus V)$.
\end{proof}
For an example where $L_C$ is not lagrangian, consider $C$ as in Example~\ref{exa:C0}. Since $C^\perp=\{0\}$, we have that $L_C=\Delta_C$, the diagonal in $C\oplus C$. On the other hand, $L_C^\perp=\Delta_V$, the diagonal in $V\oplus V$, which properly contains $\Delta_C$.

Another relation between canonical relations and coisotropic submanifolds is the following result (observed in \cite{CMRcq} in the context of field theory).

\begin{lemma}\label{l:LtoC}
Let $L$ be a lagrangian subspace of $\bar V\oplus V$. Define $C:=\pi_1(L)\subseteq V$, with $\pi_1$ the projection to the first summand. Then $C$ is coisotropic.
\end{lemma}
Note that this result holds also in the infinite-dimensional case.
\begin{proof}
By definition we have that $v$ belongs to $C$ if{f} there is a $w\in V$ such that $(v,w)$ belongs to $L$. Now let $\Tilde v$ be in $C^\perp$, i.e., $\omega(\Tilde v,v)=0$ $\forall v\in C$. 
This implies
\[
\omega(\Tilde v,v)-\omega(0,w)=0\ \forall (v,w)\in L.
\]
Therefore, $(\Tilde v,0)\in L^\perp=L$, which implies $\Tilde v\in C$.
\end{proof}

\subsection{Symplectic manifolds}
We now extend the above results to manifolds.

\begin{definition}[Symplectic manifolds]
A symplectic form $\omega$ on a smooth manifold $M$ is a nondegenerate closed two-form on it. By nondegenerate we mean that for each $x\in M$ the bilinear form $\omega_x$ on $T_xM$ satisfies \eqref{e:weaklynondeg}. Note that, for each $x\in M$, the pair $(T_xM,\omega_x)$ is then a symplectic space.
The pair $(M,\omega)$ is called a symplectic manifold.
\end{definition}

We allow $M$ to be an infinite-dimensional manifold (for the examples in this paper a Fr\'echet manifold).  In almost all the examples of this paper the symplectic manifold $M$ is actually a vector space, but we  also  consider nonlinear maps and nonlinear lagrangian submanifolds.

\begin{definition}[Symplectomorphisms]
A symplectomorphism between symplectic manifolds $(M,\omega_M)$ and $(N,\omega_N)$ is a diffeomorphism 
$\phi\colon M\to N$ such that $\phi^*\omega_N=\omega_M$.
\end{definition}

\begin{definition}[Special subspaces]\label{d:specialsub}
A submanifold $S$ of a symplectic manifold $M$ is called isotropic / coisotropic / lagrangian / split lagrangian
if $T_xS$ is 
a subspace of this class 
in $T_xM$ for every $x\in S$. Note that, as these concepts are defined in terms of tangent spaces, they make sense both for embedded and for immersed submanifolds. (Unless we explicitly say otherwise, by submanifold we then mean an immersed submanifold.)
\end{definition}

\begin{definition}[The characteristic distribution]\label{d:chardistr}
The characteristic distribution of a coisotropic submanifold $C$ of $M$ is the subbundle $C^\perp:=\cup_{x\in C}(T_xC)^\perp$ of $TC$. It is a subbundle, hence a regular distribution, because it is the kernel of $\omega_C^\sharp\colon TC\to T^*C$, where $\omega_C$ is the restriction of the symplectic form to $C$.
\end{definition}

\begin{lemma}\label{l:chardistr}
Let $C$ be a coisotropic submanifold of $M$. Then
\begin{enumerate}
\item For every $X\in\Gamma(C^\perp)$, we have $\operatorname{L}_X\omega_C=0$, where $\operatorname{L}$ denotes the Lie derivative;
\item The characteristic distribution $C^\perp$ is involutive.
\end{enumerate}
\end{lemma}
\begin{proof}
By definition $\iota_X\omega_C=0$. Since $\ddd\omega_C=0$, we get $\operatorname{L}_X\omega_C=0$ by Cartan's magic formula.

Next suppose that $Y$ is also a section of $C^\perp$. By Cartan's calculus we have
\[
\iota_{[Y,X]}\omega_C = [\iota_Y,\operatorname{L}_X]\omega_C,
\]
where we have also used the above statement. Therefore, $[Y,X]\in\Gamma(C^\perp)$.
\end{proof}

\begin{remark}[Split symplectic manifolds]\label{r:splitsymplman}
We are now interested in the generalization of the concept of split symplectic space of Remark~\ref{r:splitss}: a split symplectic manifold. 
In the finite-dimensional case, this is the notion of a cotangent bundle $T^*M$ with its canonical symplectic form (this is the two-form which in local coordinates $q^i$ for the base and $p_i$ for the fiber reads $\ddd p_i\ddd q^i$).\footnote{More intrinsically, the canonical symplectic form 
is the differential of the canonical one-form $\theta$ with $\theta_{(p,q)}(v):=(p,\ddd_{(p,q)}\pi(v))$, with $q\in M$ and $p\in T^*_qM$, where $\pi$ is the projection $T^*M\to M$.}
 In the infinite-dimensional case, more generally, 
we consider subbundles $\Tilde T^*M$ of $T^*M$ with symplectic form given by 
the restriction of the canonical symplectic form.
\end{remark}
\begin{remark}[Polarizations]
A more flexible notion of splitting than that of Remark~\ref{r:splitsymplman}, which is used in the classical setting, is that of polarization, which is used for geometric quantization. Namely, a (real) polarization on a symplectic manifold $M$ is an integrable smooth distribution $D$ (i.e., a subbundle of $TM$) such that $D_x$ is split lagrangian in $T_xM$ for every $x\in M$. 
One also uses the notion of complex polarization, which is an integrable subbundle  $D$ of the complexification $TM\otimes\CC$ of the tangent bundle such that $D_x$ is split lagrangian in $T_xM\otimes\CC$ for every $x\in M$.
\end{remark}

\begin{remark}[Generating functions]
Let $\Tilde T^*M$ be a split symplectic manifold and $\phi$ a section (in particular $\phi$ is a one-form on $M$). Its graph $L$ is then lagrangian, actually embedded split lagrangian, if and only if $\phi$ is closed. If $\phi=\ddd\psi$, then $\psi$ is called a generating function for $L$, which we also denote as $L_\psi$. Note that a function $\psi$ on $M$ defines an embedded split lagrangian submanifold $L_\psi$ of $\Tilde T^*M$ under the condition that $\ddd\psi$ is a section of it.
\end{remark}

\begin{remark}[Symplectic reduction]\label{r:sympred}
Next we consider symplectic reduction of coisotropic submanifolds. If $C$ is coisotropic in $M$, then its
characteristic distribution
$C^\perp:=\cup_{x\in C}(T_xC)^\perp$ is a regular involutive distribution on $C$ by Lemma~\ref{l:chardistr}
 If $C$ is defined by constraints, one can verify that the characteristic distribution is spanned by the hamiltonian vector fields of the constraints (which therefore must be in involution under the Poisson bracket).\footnote{The hamiltonian vector field of a function $H$ is the unique vector field $X_H$ such that $\iota_X\omega+\ddd H=0$. (Note that in the infinite-dimensional case the existence of a hamiltonian vector field is an assumption, but its uniqueness is guaranteed.) The Poisson bracket of two functions $H$ and $G$, with hamiltonian vector fields $X_H$ and $X_G$, is the function 
$\{H,G\}:=X_G(H)=\iota_{X_H}\iota_{X_G}\omega$. Its hamiltonian vector field is $[X_H,X_G]$. 
}
If $M$ is a Banach manifold (e.g., it is finite-dimensional), then the characteristic distribution is also integrable. In the Fr\'echet case, integrability is not guaranteed. Assuming that the characteristic distribution is integrable, we denote by $\uC$ its leaf space and by $\pi\colon C\to \uC$ the canonical projection. If $\uC$ is smooth, than it has a unique symplectic structure $\uomega$ such that
$\pi^*\uomega=\iota_C^*\omega$, with $\iota_C\colon C\to M$ the inclusion map. In particular, for each $\underline x\in \uC$, the symplectic form $\uomega_{\underline x}$ is defined as in Section~\ref{s:symplecticreduction}.
\end{remark}

\begin{remark}[Relative reduction]\label{r:reductionL}
Assume $\uC$ to be smooth.
If $L$ is an isotropic submanifold of $M$ and $\uL:=\pi(L\cap C)$ is a submanifold of $\uC$, then, by Lemma~\ref{l:reductionL}, $\uL$ is an isotropic submanifold of $\uC$. If in addition $M$ is finite-dimensional and $L$ is lagrangian, then $\uL$ is also lagrangian.
\end{remark} 

\begin{remark}[Generalized generating functions]\label{r:gfnonlinear}
Our main example, which extends Remark~\ref{r:gflinear} to the nonlinear case, will be the product of two split symplectic manifolds, $N=\Tilde T^*M\times\Tilde T^*B=\Tilde T^*(M\times B)$, with
$C=\Tilde T^*M\times B$ and hence $\uC=\Tilde T^*M$. Let $\psi$ be a function on $M\times B$ such that the one-form $\ddd\psi$ is a section of $\Tilde T^*(M\times B)$ and let $L_\psi$ the split lagrangian submanifold in $N$ it generates. If $\uL_\psi$ is a split lagrangian submanifold of $\uC=\Tilde T^*M$, then we say that $\psi$ is a generalized generating function for $\uL_\psi$. In the finite-dimensional case, one only has to check that that $\uL_\psi$ is a  submanifold because it is then automatically split lagrangian. Some conditions for this to happen are presented in \cite[Section 4.3]{BW97}. There, more generally, $B$ is just assumed to be the typical fiber of a submersion $\pi\colon E\to M$ and the triple
$(E,\pi,\psi)$, such that $\uL_\psi$ is a submanifold, is called a Morse family for $\uL_\psi$.
Explicitly, if we denote by $q^i$ and $T^\alpha$ the coordinates on $M$ and $B$, respectively, and by $p_i$ and $E_\alpha$ their momenta, the lagrangian submanifold $L_\psi$ is given by the equations $p_i=\frac{\dd\psi}{\dd q^i}$ and $E_\alpha=\frac{\dd\psi}{\dd T^\alpha}$. Intersecting with $C=\Tilde T^*M\times B$ corresponds to imposing
$\frac{\dd\psi}{\dd T^\alpha}=0$. Therefore,  $\uL_\psi$ is given by the equations
\[
p_i=\frac{\dd\psi}{\dd q^i},\qquad \frac{\dd\psi}{\dd T^\alpha}=0.
\]
\end{remark}

\begin{definition}[Canonical relations]
A (split) canonical relation $L$ between two symplectic manifolds $(M_1,\omega_1)$ and $(M_2,\omega_2)$ is a (split) lagrangian submanifold of $\bar M_1\times M_2$ (i.e., $M_1\times M_2$ with symplectic form $\pi_2^*\omega_2-\pi_1^*\omega_1$ with $\pi_i\colon M_1\times M_2\to M_i$ the canonical projections). In particular, the graph of a symplectomorphism from $M_1$ to $M_2$ is an embedded split canonical relation. 
\end{definition}
Here is an example (generalizing Lemma~\ref{ex:LinCC}):
\begin{lemma}\label{ex:LinCCman}
Let $C$ be a coisotropic submanifold of $M$ and assume that its involutive characteristic distribution is actually integrable. We say that two points $(x,x')\in C\times C$ lie on the same leaf of the characteristic distribution
if there is a path in $C$ with endpoints $x$ and $x'$ and tangent vector at each point lying in the characteristic distribution.
Then, if 
\[
L_C:=\{(x,x')\in C\times C\ |\ x,x' \text{ lie on the same leaf of the characteristic distribution}\}
\]
is a submanifold, it
is isotropic in $\bar M\times M$. If $M$ is finite-dimensional, then $L$ is lagrangian, hence a canonical relation from $M$ to itself.
\end{lemma}
\begin{proof}
Let $(x,x')$ be a point in $L_C$. 
By integrability, there is a neighborhood $U$ of $x$ that is mapped to a convex open subspace of a vector space in which the distribution corresponds to a family of linear subspaces. We first consider the case where $x' \in U$.  The images $y$ and $y'$ of $x$ and $x'$ can then be actually connected by a straight line. We consider the constant vector field corresponding to this direction and its flow $\Phi$ at time $1$, a translation mapping a neighborhood of $y$ to a neighborhood of $y'$ satisfying
$\Phi(y)=y'$. Since $\Phi$ is a flow along a vector field in the kernel of the image of the restriction of the symplectic form, it leaves the latter invariant by Lemma~\ref{l:chardistr}.
Let $A$ be the differential of $\Phi$ at $y$ pulled back to $x$. Therefore, $A$ is a linear isomorphism 
$T_xC\to T_{x'}C$ satisfying
\[
\omega_x(v,w)=\omega_{x'}(Av,Aw)\ \forall v,w\in T_xC.
\]
By construction $A(T_xC^\perp)=T_{x'}C^\perp$. Therefore, $A$ induces a symplectomorphism $\underline A$ between $\underline{T_xC}:=T_xC/T_xC^\perp$ and $\underline{T_{x'}C}:=T_{x'}C/T_{x'}C^\perp$. The first consequence of this is that
\[
(v,v')\in T_{(x,x')}L_C \iff v'-Av \in T_{x'}C^\perp \iff [v']=\underline A[v].
\]
The second consequence is that $T_{(x,x')}L_C$ is isotropic. In fact, for $(v,v'),(w,w')\in T_{(x,x')}L$ we get
\[
\omega_{x'}(v',w')=\omega_{x'}(Av,Aw)=\omega_x(v,w).
\]

For a generic point $(x,x')$ in $L_C$, we have by definition a path $\gamma\colon[0,1]\to C$ such that
$\gamma(0)=x$, $\gamma(1)=x'$ and $\ddd_t\gamma\in T_{\gamma(t)}C^\perp$ $\forall t\in[0,1]$. By compactness of $[0,1]$, we may find a finite partition $0=t_0<t_1<t_2<\dots<t_n=1$ such that, for each $i=0,\dots,n-1$,
the points $\gamma(t_i)$ and $\gamma(t_{i+1})$ lie in a neighborhood $U_i$ as in the previous paragraph. Set $x_i:=\gamma(t_i)$ and proceed as above.
We then have that $(v_i,v_{i+1})\in T_{(x_i,x_{i+1})}L_C$ if{f} $[v_{i+1}]=\underline A_i[v_i]$. Therefore,
\[
(v,v')\in T_{(x,x')}L_C \iff [v']= \underline A_{n-1}\dots \underline A_1\underline A_0[v]
\iff v'- A_{n-1}\dots  A_1 A_0v \in T_{x'}C^\perp.
\]
Given $v\in T_xC$, we set $[v_{i+1}]:=\underline A_{i}\dots \underline A_1\underline A_0[v]$ and choose representatives $v_{i+1}\in[v_{i+1}]$. For $(v,v'),(w,w')\in T_{(x,x')}L$, we then have
\[
\omega_{x_{i+1}}(v_{i+1},w_{i+1})-\omega_{x_i}(v_i,w_i)=0\ \forall i =0,\dots,n-1.
\]
Summing over $i$, we get
\[
\omega_{x'}(v',w')-\omega_x(v,w)=0,
\]
which shows that $T_{(x,x')}L_C$ is isotropic.

In the finite-dimensional case, we easily check, as in the proof of Lemma~\ref{ex:LinCC}, that $\dim L=\frac12\dim(\bar M\times M)$.
%
%
\end{proof}

We also have the following straightforward generalization of Lemma~\ref{l:LtoC}.
\begin{lemma}\label{l:LtoCman}
Let $L$ be a lagrangian submanifold of $\bar M\times M$. If $C:=\pi_1(L)$, where $\pi_1$ is the projection to the first factor, is a submanifold of $M$, then $C$ is coisotropic.
\end{lemma}
\begin{proof}
By definition we have that $x$ belongs to $C$ if{f} there is a $y\in M$ such that $(x,y)$ belongs to $L$.
Moreover, $v$ belongs to $T_xC$ if{f} there is a $w\in T_yM$ such that $(v,w)$ belongs to $T_{(x,y)}L$. Now let $\Tilde v$ be in $T_xC^\perp$, i.e., $\omega_x(\Tilde v,v)=0$ $\forall v\in T_xC$. This implies
\[
\omega_x(\Tilde v,v)-\omega_y(0,w)=0\ \forall (v,w)\in T_{(x,y)}L.
\]
Therefore, $(\Tilde v,0)\in T_{(x,y)}L^\perp=T_{(x,y)}L$, which implies $\Tilde v\in T_xC$.
\end{proof}

\begin{remark}[Applications in field theory]
The above two lemmata, \ref{ex:LinCCman} and~\ref{l:LtoCman}, are important for the applications in field theory.
In the finite-dimensional case, we have assumed in Section~\ref{s:syssevcon}, see Remark~\ref{r:coisomany}, that
the constraints define a coisotropic submanifold $C$. Then the evolution relation $L$ coincides with $L_C$ and is therefore lagrangian by Lemma~\ref{ex:LinCCman}. In the infinite-dimensional case,
it is instead better to start from the evolution relation $L$, which is automatically isotropic. If we can prove that it is lagrangian (e.g., by the Hodge decomposition theorem), as we can do in all the examples in Section~\ref{s:inftarg}, we then get by Lemma~\ref{l:LtoCman} that $C$ is coisotropic. Finally, we just have to check that $L=L_C$.
\end{remark}

If $M_i=\Tilde T^* N_i$, $i=1,2$, then we may identify $\bar M_1\times M_2$ with $\Tilde T^*(N_1\times N_2)$ (changing the sign of the $\Tilde T^* N_1$ components of the canonical symplectic form). A (split) lagrangian submanifold 
$L$ of $\bar M_1\times M_2$
might have a (generalized) generating function $\psi$ on $ N_1\times N_2$, which is then called a (generalized) generating function for the (split) canonical relation $L$ from $M_1$ to $M_2$.

\begin{remark}[Generating function for a canonical relation]\label{r:gfcanrel}
Explicitly, if we denote by $q^i$ and $Q^i$ the coordinates on $N_1$ and $N_2$, respectively, and by $p_i$ and $P_i$ their momenta, the canonical relation defined by a generating function $\psi$ is given by the equations
\[
P_i=\frac{\dd\psi}{\dd Q^i},\qquad
p_i=-\frac{\dd\psi}{\dd q^i}.
\]
\end{remark}

\begin{example}[Generating function for a change of polarization]\label{r:gfchangepol}
A particular case is when $M_1=M_2$ and $L$ is the graph of the identity map, which is obviously a symplectomorphism, but we choose $N_1$ and $N_2$ to be different. If there is a generating function $\psi$ for $L$, with respect to these splittings, then the coordinates on $\Tilde T^* N_1$ and $\Tilde T^* N_2$ are related by the last equation.
\end{example}

\begin{remark}[Generalized generating function for a canonical relation]\label{r:ggfcanrel}
If we have instead a generalized generating function $\psi$ on $N_1\times N_2\times B$, with coordinates $q^i,Q^i,T^\alpha$, then the corresponding canonical relation is defined by the equations
\[
P_i=\frac{\dd\psi}{\dd Q^i},\qquad
p_i=-\frac{\dd\psi}{\dd q^i},\qquad  \frac{\dd\psi}{\dd T^\alpha}=0.
\]
\end{remark}

\begin{remark}[Composition of canonical relations]\label{r:compcanrel}
If we have a relation $L_1$ from $M_1$ to $M_2$ and a relation $L_2$ from $M_2$ to $M_3$ (i.e., $L_1\subseteq M_1\times M_2$ and $L_2\subseteq M_2\times M_3$), we may set-theoretically compose them to a relation $L_2\circ L_1$ from $M_1$ to $M_3$:
\begin{equation}\label{e:L2circL1}
L_2\circ L_1 := \{(x_1,x_3)\in M_1\times M_3\ |\ \exists x_2\in M_2 : (x_1,x_2)\in L_1 
\, \mr{and} \,
(x_2,x_3)\in L_2\}.
\end{equation}
If $M_1$, $M_2$ and $M_3$ are manifolds and $L_1$ and $L_2$ are submanifolds, there is no guarantee that $L_2\circ L_1$ is also a submanifold. Let us assume this to be the case and also that $L_1$ and $L_2$ are (split) canonical relations---i.e., $M_1$, $M_2$ and $M_3$ are symplectic manifolds and $L_1$ and $L_2$ are (split) lagrangian submanifolds of $\bar M_1\times M_2$ and $\bar M_2\times M_3$, respectively. Then the composition
\eqref{e:L2circL1} may be realized as a particular case of Remark~\ref{r:reductionL}. Namely, we consider the (split) lagrangian submanifold $L:=L_1\times L_2$ and the coisotropic submanifold $C:=M_1\times\Delta_{M_2}\times M_3$ in $\bar M_1\times M_2\times\bar M_2\times M_3$, where $\Delta_{M_2}:=\{(x_2,x_2)\in M_2\times M_2\}$ denotes the diagonal submanifold. We clearly have $\underline C=\bar M_1\times M_3$ and $\underline L=L_2\circ L_1$.
Therefore, in the finite-dimensional case $L$ is lagrangian, whereas in the infinite-dimensional case, at least it is isotropic and one has to check separately if it is also (split) lagrangian.
\end{remark}

\begin{remark}[Composition of generating functions I: mixed polarizations]\label{r:compgenfun}
Let us assume that the composition of the (split) canonical relations $L_1$ and $L_2$ of Remark~\ref{r:compcanrel}
is a (split) canonical relation. Let us also assume that $L_1$ and $L_2$ have (generalized) generating functions
$\psi_1$ and $\psi_2$ in the following way. We assume $M_i=\Tilde T^* N_i$, $i=1,3$, $M_2=U\times U'$
with $U$ and $U'$ open subsets of $V$ and $V'\subseteq V^*$, respectively, and that
$\psi_1\in C^\infty(N_1\times U'\times B)$ and $\psi_2\in C^\infty(U\times N_3\times \Tilde B)$, where $B$ and $\Tilde B$ are parameter spaces.
For notational simplicity, we work as in the finite-dimensional case and avoid writing component indices. We denote
the coordinates as follows: $(p,q)$ for $\Tilde T^* N_1$, $(\Tilde p,\Tilde q)$ for $\Tilde T^* N_3$, 
$T$ for $B$, and $\tilde T$ for $\Tilde B$. For the first copy of $M_2$ we denote the coordinates in $U\times U'$ by
$(Q,P)$ and for the second copy by $(\Hat Q,\Hat P)$. Then $\psi_1$ is a function of $q,P,T$ and $L_1$ is given by the equations
\[
Q=-\frac{\dd\psi_1}{\dd P},\quad p=-\frac{\dd\psi_1}{\dd q},\quad \frac{\dd\psi_1}{\dd T}=0,
\]
whereas $\psi_2$ is a function of $\Hat Q,\Tilde q, \Tilde T$ and $L_2$ is given by the equations
\[
\Tilde p=\frac{\dd\psi_2}{\dd \Tilde q}, \quad \Hat P=-\frac{\dd\psi_2}{\dd \Hat Q},\quad \frac{\dd\psi_2}{\dd\Tilde T}=0.
\]
Since $C$ is defined by the equations $\Hat P=P$ and $\Hat Q= Q$,
we see that $L$ has generalized generating function, on $N_1\times N_3\times B\times \Tilde B\times U'\times U$,
\[
\psi_3(q,\Tilde q; T,\tilde T, P,Q) = \psi_1(q,P;T)+\psi_2(Q,\Tilde q; \Tilde T)+PQ,
\]
where the last term is defined via the pairing. Note that this is a generalized generating function with parameter space
$B\times\Tilde B\times M_2$. The conditions with respect to $M_2$, given by
\[
\frac{\dd\psi_1}{\dd P}+Q=0\quad\text{and}\quad \frac{\dd\psi_2}{\dd Q}+P=0,
\]
are precisely the first condition for $L_1$ and the second condition for $L_2$. If we can solve these equations with respect to $(P,Q)$, we also get a (generalized) generating function $\underline\psi_3\in C^\infty(N_1\times N_3\times B\times \Tilde B)$, which is the critical value of $\psi_3$ at this point. Note that if $\psi_1$ and $\psi_2$ are not generalized (i.e., $B$ and $\Tilde B$ are one-point sets), then $\underline\psi_3$ is also not generalized.
\end{remark}

\begin{remark}[Composition of generating functions II: same polarization]\label{r:compgenfunII}
In the setting of the preceding Remark~\ref{r:compgenfun}, we now assume that
$\psi_1\in C^\infty(N_1\times U\times B)$ and $\psi_2\in C^\infty(U\times N_3\times \Tilde B)$.
Now $\psi_1$ is a function of $q,Q,T$ and $L_1$ is given by the equations
\[
P=\frac{\dd\psi_1}{\dd Q},\quad p=-\frac{\dd\psi_1}{\dd q},\quad \frac{\dd\psi_1}{\dd T}=0,
\]
whereas, as before, $\psi_2$ is a function of $\Hat Q,\Tilde q, \Tilde T$ and $L_2$ is given by the equations
\[
\Tilde p=\frac{\dd\psi_2}{\dd \Tilde q}, \quad \Hat P=-\frac{\dd\psi_2}{\dd \Hat Q},\quad \frac{\dd\psi_2}{\dd\Tilde T}=0.
\]
Setting $\Hat P=P$ and $\Hat Q= Q$, yields
\[
\frac{\dd\psi_1}{\dd Q}+\frac{\dd\psi_2}{\dd Q}=0.
\]
Therefore, the composition $L$ has generalized generating function, on $N_1\times N_3\times B\times \Tilde B\times U$,
\[
\psi_3(q,\Tilde q; T,\tilde T, Q) = \psi_1(q,Q;T)+\psi_2(Q,\Tilde q; \Tilde T).
\]
Note that now the parameter space is $B\times \Tilde B\times U$. Also in this case, it might be possible to solve the critical equation in $Q$ and get a (generalized) generating function $\underline\psi_3\in C^\infty(N_1\times N_3\times B\times \Tilde B)$, which is not generalized if $\psi_1$ and $\psi_2$ are not.
\end{remark}

\begin{remark}[Evolution relations and generalized generating functions in field theory]\label{r:evorel}
Let us conclude with the application to field theory \cite{KT,CMRcq}. Under some regularity assumptions, a lagrangian field theory in $d+1$ dimensions produces a (usually infinite-dimensional) symplectic manifold $\mathcal{F}^\partial_\Sigma$
 of boundary fields on every closed $d$-manifold $\Sigma$ (roughly speaking, this is the space of fields and some of their normal jets on which the boundary one-form obtained after integration by parts from the variation of the action functional is the potential for a symplectic form). The field theory on $I\times\Sigma$, where $I$ is an interval, produces a subset $L$, which we call the evolution relation, of
 $\mathcal{F}^\partial_\Sigma\times\mathcal{F}^\partial_\Sigma$ given by the boundary data corresponding to possible solutions to the EL equations.  
If $L$ is a submanifold, it is then automatically isotropic in $\Bar{\mathcal{F}^\partial_\Sigma}\times\mathcal{F}^\partial_\Sigma$.  In a good field theory, we require $L$ to be split lagrangian, so a canonical relation from $\mathcal{F}^\partial_\Sigma$
to itself. (In the case of a field theory with nondegenerate lagrangian, $L$ is the graph of the associated hamiltonian flow, so it is an embedded split canonical relation.)
In this case, there might be a (generalized) generating function for $L$ with respect to a split symplectic structure on $\mathcal{F}^\partial_\Sigma$. Throughout this paper we have constructed (generalized) generating functions for several examples from field theory. The discussion of the first part of this paper
shows how to construct a generalized generating function for the evolution relation in terms of the HJ action. Note that in the infinite-dimensional case
one has to show separately (e.g., by Hodge decomposition theorem) that it is actually split lagrangian.
 \end{remark}
 
\begin{remark}[Gauge-fixing fermion]
Another application of the formalism of (generalized) generating functions occurs in the gauge-fixing procedure in the BRST or BV formalism, where the symplectic manifold is actually a supermanifold and the symplectic form is odd with respect to the internal grading. In this case a generating function $\psi$ for a gauge-fixing lagrangian $L_\psi$ must be odd, and for this reason it is usually called a gauge-fixing fermion. Usually, there is no generating function on the
base $\mathcal{F}$ of the BV space of fields $\Pi\Tilde T^*\mathcal{F}$, so one really has to resort to a generalized generating function on $\mathcal{F}\times B$, where $B$ is a space of parameters. For example, in abelian gauge theories (on a trivial principal bundle) on a manifold $M$, one has $\mathcal{F}=\Omega^1(M)\times\Pi\Omega^0(M)\ni (A,c)$, where $A$ is the connection and $c$ is the ghost. We denote the corresponding momenta, with opposite parity, by $A^+$ and $c^+$.
To impose the Lorenz gauge fixing $\ddd^*A=0$ (working in the complement of cohomology) actually means defining $L$ as the conormal bundle of the submanifold of $\mathcal{F}$ defined by this equation, which clearly does not have a generating function. 
One therefore introduces extra parameters $\bar c\in B=\Pi\Omega^\text{top}(M)$ and the generalized generating function (the gauge-fixing fermion) $\psi=\int_M \bar c\, \ddd^*A$. The corresponding gauge-fixing lagrangian $L=\uL_\psi$ 
is then correctly given by the equations
\[
A^+=\ddd^*\bar c,\quad c^+=0,\quad \ddd^*A=0.
\]
In field theory, one prefers to introduce the last condition (namely, $\bar c^+=\frac{\dd\psi}{\dd\bar c}=0$) in terms of a Lagrange multiplier $\lambda\in B'=\Omega^0(M)$,
\[
\delta(\bar c^+)=\int D\lambda \EE^{\frac\II\hbar\int_M\lambda \bar c^+},
\]
and then to regard $\int_M\lambda \bar c^+$ as a new term to be added to the action. The extended BV space of fields eventually is $\Pi\Tilde T^*(\mathcal{F}\times B\times B')$. (For a slightly more general presentation of this, see \cite[Sect.\ 3.2.1]{C}.)
\end{remark}

\subsection{Generalized Hamilton--Jacobi actions with infinite-dimensional targets}\label{s:HJinf}
We now give more details on the constructions of Section~
\ref{s:inftarg}, in the case of an infinite-dimensional target. This is based on results of \cite{CMR14} and \cite{Contr13}.

The first thing we have to generalize is the target $T^*\RR^n$ and the term $p_i\ddd q^i$ in the action. We focus on the case when the target is linear (as this is the case for all the examples in this paper) and we take it to be a split symplectic space $V=W\oplus W'$ as in Remark~\ref{r:splitss}. More concretely, we assume we have it realized with $W'\subseteq W^*$ and canonical symplectic form $\omega$ induced from the pairing as in \eqref{e:omegacan}. We may also write $\omega$ as the differential of the $1$-form $\alpha$ on $V$ defined by
\[
(\alpha(v\oplus a),w\oplus b):= a(w).
\]
The map $(p,q)\colon[t_\ii,t_\oo]\to T^*\RR^n$ is now generalized to a map $\phi\colon [t_\ii,t_\oo]\to V$ and
the term in the action $\int_{t_\ii}^{t_\oo}p_i\ddd q^i$ is generalized to $\int_{t_\ii}^{t_\oo}\phi^*\alpha$. 

Next we have to generalize the Lagrange multipliers $e^\alpha$, the contraints $H_\alpha$ and the second term in the action: $\int_{t_\ii}^{t_\oo} e^\alpha H_\alpha$. For this we take a (possibly infinite-dimensional) vector space $U$ and a subspace $U'$ of its dual $U^*$ such that the canonical pairing $\langle\ ,\ \rangle$ of $U'$ with $U$ is still weakly nondegenerate. We assume we are given a (possibly nonlinear) map $H\colon V\to U'$ and introduce $e$ as a $1$-form on $[t_\ii,t_\oo]$ taking values in $U$. We can then write the action as
\[
S[\phi,e]=\int_{t_\ii}^{t_\oo} (\phi^*\alpha-\langle e,H\circ\phi\rangle).
\]
We denote by $\FF$ the space of fields $(\phi,e)$ as an infinite-dimensional Fr\'echet manifold and assume $S$ to be a smooth function on $\FF$. Integrating by parts, we may express its differential as 
\begin{equation}\label{e:deltaSinf}
\delta S = \EL + \pi_\oo^*\alpha - \pi_\ii^*\alpha,
\end{equation}
where $\EL$ is again expressed as an integral over $[t_\ii,t_\oo]$ and $\pi_\oo,\pi_\ii\colon\FF\to V$ are the surjective submersions $\pi_\oo\colon (\phi,e)\mapsto\phi(t_\oo)$ and $\pi_\ii\colon (\phi,e)\mapsto\phi(t_\ii)$.

To define the EL equations we have to assume that, for every $x\in V$, $\ddd_xH\in V^*\otimes U'$ lies in the image of $\omega^\sharp\otimes\Id\colon V\otimes U'\to V^*\otimes U'$. We then have a unique $X_x\in V\otimes U'$ such that $\omega^\sharp(X_x)=\ddd_xH$. The EL equations then read
\[
\ddd\phi = \langle e,X\rangle,\quad H\circ\phi=0.
\]
Again we call the first equation the evolution equation and the second the constraint. 

\begin{remark}
In the case of field theory, the existence of the family $X$ of hamiltonian vector fields is ensured by locality. Namely, $V$ and $U$ are (modeled on) spaces of sections of fiber bundles 
over a closed manifold $\Sigma$. The fields $(\phi,e)$ are then sections of fiber bundles over $[t_\ii,t_\oo]\times\Sigma$. Finally, $H$ and $\alpha$ are assumed to be local (over $\Sigma$).
\end{remark}

We can now define the evolution relation $L$ as the set of pairs of points in $V$ that can be connected by solutions of the EL equations.  We will assume $L$ to be a (possibly immersed) lagrangian submanifold.\footnote{For the boundary value problem to be well-defined, one should actually ask for $L$ to be split lagrangian, but this is not needed in the present section.} Note that $L$ does not depend on the choice of splitting of $V$, since the EL equations are defined only in terms of $H$ and the symplectic form.
\begin{remark}
That $L$ is isotropic is guaranteed, under some regularity assumptions. Indeed, we have $L=\pi_\ii\times\pi_\oo(\mathcal{EL})$,
where $\mathcal{EL}\subset\FF$ is the set of solutions to the EL equations. If we assume that $\mathcal{EL}$ is a (possibly immersed) submanifold, 
then \eqref{e:deltaSinf} implies that $L$ is isotropic. In fact, with $\pi:=\pi_\ii\times\pi_\oo$ and $\alpha_{\Bar V\oplus V}:=(-\alpha,\alpha)$, we have
\[
\iota_\mathcal{EL}^*\delta S=\pi^*\iota_L^*(\alpha_{\Bar V\oplus V}),
\]
with $\iota$s denoting the inclusion maps.
This implies $\pi^*\iota_L^*(\delta \alpha_{\Bar V\oplus V})=0$ and, therefore, $\iota_L^*(\delta\alpha_{\Bar V\oplus V})=0$ because $\pi$ is a submersion.
\end{remark}

Let $\pi_i\colon V\oplus V\to V$, $i=1,2$, denote the two projections and set $C_i=\pi_i(L)$. It then follows that $C_1=C_2=C:=H^{-1}(0)$. In fact, since the constraints have to be satisfied everywhere, they are satisfied in particular at the endpoints, so $C_i\subseteq C$. On the other hand, for $e=0$, we have the constant solutions $(c,c)$ with $c\in C$, so $C\subseteq C_i$. Since, by definition, $L\subset C_1\times C_2$, we then have 
$L\subset C\times C$.\footnote{Note that, as a relation on $V$, $L$ is symmetric and transitive, but in general not reflexive, and that $C$ is the largest subset of $V$ on which $L$ restricts as an equivalence relation.}
Next we assume that $C$ is a (possibly immersed) submanifold of $V$. 
\begin{lemma}
$C$ is coisotropic in $V$.
\end{lemma}
\begin{proof}
Fix $u\in C$ and $\Hat u\in \pi_1^{-1}(u)\subset L$. We have assumed that $T_{\Hat u}L$ is lagrangian, and
we want to prove that $(T_uC)^\perp\subseteq T_uC$. We have
\[
T_uC = \{v\in V\ |\ \exists w\in V : (v,w)\in T_{\Hat u}L\}.
\]
By definition, if $\Tilde v\in V$ belongs to $(T_uC)^\perp$, then $\omega(\Tilde v,v)=0$ for every $v\in T_uC$, i.e.,
$\omega(\Tilde v,v)=0$ for every $(v,w)\in T_{\Hat u}L$.  This implies $(\Tilde v,0)\in(T_{\Hat u}L)^\perp$. Since $L$ is lagrangian, we get $(\Tilde v,0)\in  T_{\Hat u}L$, so $\Tilde v\in T_uC$.
\end{proof}
Observe that, for every $t\in [t_\ii,t_\oo]$, the right hand side of the evolution equation lies in $(T_{\phi(t)}C)^\perp$. Since $C$ is coisotropic, this is contained in $T_{\phi(t)}C$. As a consequence, a solution of the evolution equation starting at some point of $C$ will never leave $C$, i.e., if the constraints are imposed at the initial (or final) endpoint, they will be satisfied at every time. (In the finite-dimensional case, this statement followed from \eqref{e:dddHalpha}.)

\begin{remark}
The above assumptions that $V$ is split and $L$ is lagrangian are satisfied in the case of Chern--Simons theory thanks to the Hodge decomposition of differential forms (actually one even proves that $L$ is split lagrangian). One can also easily check that $L$ and $C$ are submanifolds.
\end{remark}

We now move to the HJ action. For this we might want to choose a different splitting $V=Z\oplus Z'$ at the final endpoint. The HJ action as a generalized generating function for $L$ will depend on the choices of splitting, even though $L$ does not. 
We denote by $\beta$ the $1$-form corresponding to this splitting. 
We assume that we have a generating function $f$ on $W\oplus Z$ for the change of polarization, as in Example~\ref{r:gfchangepol}. In particular, $\alpha=\beta+\delta f$. 
If we define
\[
S^f:=S-\pi_\oo^*f,
\]
we get
\[
\delta S^f = \EL + \pi_\oo^*\beta - \pi_\ii^*\alpha.
\]
Finally we define $S^f_\text{HJ}$, as a function on $W\times Z\times\Omega^1([t_\ii,t_\oo],U)$, 
as $S^f$
evaluated on the assumedly unique solution $\bar\phi$ of the evolution equation with given endpoint conditions for $\phi$ in $W$ and $Z$. We then get
\[
\delta S^f_\text{HJ} = -\int_{t_\ii}^{t_\oo} \langle \delta e,H\circ\bar\phi\rangle + \pi_\oo^*\beta - \pi_\ii^*\alpha,
\]
which shows that $S^f_\text{HJ}$ is a generalized generating function for $L$. In general, it may be difficult to characterize the variations of $e$ that leave $S^f_\text{HJ}$ invariant, since we cannot use \eqref{e:dddHalpha}.
This may be however done explicitly, e.g., when $U$ is a Lie algebra and $H\colon V\to U'$ is an equivariant momentum map.  In this case, which is enough to treat Chern--Simons theory, one easily sees that
$\delta e = \ddd\gamma+[e,\gamma]$, with $\gamma$ a map $[t_\ii,t_\oo]\to U$ that vanishes at the endpoints, leaves 
$S^f_\text{HJ}$ invariant.


\section{The modified differential quantum master equation}\label{s:mQME}
Following \cite{CMremCS} and \cite[Section 4.2]{CMR15}, we explain how the mQME and the (appropriate notion of) invariance under deformations of gauge fixing arise from the perturbative treatment of the theory. We will be rather sketchy and rely on some knowledge of such references.

\subsection{Assumptions}
We focus on the case where the constraints are in involution with structure constants; i.e., the constraints $H_\alpha$ are the components of an equivariant momentum map as in \eqref{e:involutivity}:
\begin{equation}\label{e:involutivityagain}
\{H_\alpha,H_\beta\}=f_{\alpha\beta}^\gamma H_\gamma,\quad\forall\alpha,\beta.
\end{equation}
Moreover, if we change the polarization via a generating function $f(q,Q)$, we define the transformed constraints $\Tilde H_\alpha$ as in \eqref{e:Htilde}, so we have the relations
\begin{equation}\label{e:Htilderel}
H_\alpha\left(\frac{\dd f}{\dd q},q\right)=
\Tilde H_\alpha\left(-\frac{\dd f}{\dd Q},Q\right),\quad\forall\alpha.
\end{equation}
For example, if we just interchange the $p$ and $q$ variables---$Q=p$, $P=-q$, $f=qQ$---then we have the definition
$\Tilde H_\alpha(P,Q):=H_\alpha(Q,-P)$ and the relation $H_\alpha(Q,q)=\Tilde H_\alpha(-q,Q)$.

As already mentioned several times in the paper, our result is based on the following
\begin{assumption}\label{a:assu}
We assume that
the constraints $H_\alpha$ are quantized to the components of a quantum equivariant momentum map as in \eqref{e:HHhat}, i.e,
\begin{equation}\label{e:HHhatagain}
\left[\widehat{H}_\alpha,\widehat{H}_\beta\right] = \II\hbar f^\gamma_{\alpha\beta}\widehat{H}_\gamma,
\end{equation}
with
\[
\widehat{H}_\alpha=H_\alpha\left(\II\hbar\frac\dd{\dd q},q\right),
\]
and that they are compatible, at the quantum level, with the generating function $f(q,Q)$ of the change of polarization
as in \eqref{e:quantumHtilderel-}, i.e,
\begin{equation}\label{e:quantumHtilderel}
H_\alpha\left(\II\hbar\frac\dd{\dd q},q\right)\EE^{-\frac\II\hbar f}
=
\Tilde H_\alpha\left(-\II\hbar\frac\dd{\dd Q},Q\right)\EE^{-\frac\II\hbar f},\quad\forall\alpha.
\end{equation}
\end{assumption}

\begin{remark} Observe the following:
\begin{enumerate}
\item In the above formulae, the standard ordering is assumed, i.e., all the derivatives are placed to the right.
\item The assumptions \eqref{e:HHhatagain} and \eqref{e:quantumHtilderel} can equivalently be written using star products.
\item The classical relations \eqref{e:involutivityagain} automatically imply the quantum relations \eqref{e:HHhatagain} if the constraints $H_\alpha$ are linear or biaffine.
\item The classical relations \eqref{e:Htilderel} automatically imply the quantum relations \eqref{e:quantumHtilderel} if
\begin{enumerate}
\item the endpoint polarizations are linear (in particular, $f$ is bilinear in $(q,Q)$), or
\item the constraints $H_\alpha(p,q)$ are linear in $p$ and the transformed constraints $\Tilde H_\alpha(P,Q)$ are linear in $P$.
\end{enumerate}
\item If the assumptions \eqref{e:HHhatagain} and \eqref{e:quantumHtilderel} are not satisfied on the nose, one may try to deform the $H_\alpha$s and the $\Tilde H_\alpha$s with $\hbar$-corrections in order to impose the relations, but there is no guarantee that this is possible.
\item If the change of polarization is not linear, $\EE^{-\frac\II\hbar f}$ might not be the integral kernel of a unitary operator. One might try to deform $f$, if possible, with $\hbar$-corrections in order to make the operator unitary. One should then try to find
deformations of the $H_\alpha$s and the $\Tilde H_\alpha$s so as to achieve the relations \eqref{e:HHhatagain} and \eqref{e:quantumHtilderel}. Note that the classical condition of point (4b) that the nondeformed $H_\alpha$s and  $\Tilde H_\alpha$s are linear in the respective momenta is no longer enough for the relation \eqref{e:quantumHtilderel} to be satisfied on the nose, since the $\Tilde H_\alpha$ are defined using the nondeformed $f$.
\end{enumerate}
\end{remark}

\subsection{The propagator and the mdQME}
The starting point for our discussion is the propagator $\eta$, which is an inverse of (more precisely, a parametrix for) the de~Rham differential $\ddd$, compatible with the endpoint conditions. 

For example, 
in the quantization in Case~II, the propagator is $\eta(s,t)=\theta(s-t)$---the Heaviside step function---for all fields, see
\eqref{e:propII}. In Case~I, because of the different endpoint conditions (and the choice of $e$), the propagator in the ghost sector is $\eta(s,t)=\theta(s-t)-\phi(s)$, where $\phi$ is a chosen function with the property $\phi(t_\oo)-\phi(t_\ii)=1$.

In general, the propagator $\eta$ is a smooth function on the compactified configuration space, see \cite{FM,AS}, of two points on the interval $I=[t_\ii,t_\oo]$, i.e.,
\[
C_2(I)=\{(s,t)\in I\times I\ |\ s\le t\}\sqcup\{(s,t)\in I\times I\ |\ s\ge t\},
\]
that satisfies the endpoint conditions and the limit $\eta|_{s\to t^+}-\eta|_{s\to t^-}=1$ and such that $\ddd\eta=\chi$, where $\chi$ is a smooth $1$-form on $I\times I$ containing information on the cohomology and the endpoint conditions.

The endpoint conditions are fixed throughout. However, the choices of $\chi$ and $\eta$ are not unique and correspond to different gauge fixings and choices of residual fields. We can actually change $\eta$ by adding to it a smooth function $g$ on $I\times I$ that satisfies the endpoint conditions and simultaneously change $\chi$ by adding $\ddd g$ to it. 
It is convenient to
interpolate between the two propagators
using a parameter $u\in J:=[0,1]$ defining
$\eta_u:=\eta +ug$ and $\chi_u=\chi+u\ddd g$. It is even better to define $\Check\eta\in C^\infty(C_2(I)\times J)$ as
$\Check\eta(s,t,u)=\eta_u(s,t)$ and $\Check\chi\in\Omega^1(C_2(I)\times J)$ via\footnote{In our theory we use two propagators---the propagator in the physical sector $(p,q,p^+,q^+)$ and the propagator in the ghost sector $(e^+,c,c^+,e)$---and we may want to take different choices for them. This means that we will have a pair $(\Check\eta,\Check\chi)$ for the physical sector and a possibly different one for the ghost sector.}
\begin{equation}\label{e:detacheck}
\ddd\Check\eta=\Check\chi,
\end{equation}
where now $\ddd$ denotes the total differential on $\Omega^\bullet(C_2(I)\times J)$.

As the choice of $\chi_u$ is related to the choice of the residual fields, we have to deform them accordingly. We next compute the partition ``form'' for these choices and we denote it as $\Check Z\in\Omega^\bullet(J)$. Note that the zero-form component is the usual partition function, depending on the parameter $u$. 

We claim that \eqref{e:detacheck} and Assumption~\ref{a:assu} imply that the partition form satisfies the differential mQME (mdQME)
\[
(\Omega+\II\hbar\ddd_J+\hbar^2\Delta)\Check Z=0,
\]
where $\ddd_J=\ddd u\frac\dd{\dd u}$ is the de~Rham differential on $J$. The zero-form part 
of the mdQME is the mQME for the partition function (for every value of $u$), whereas the one-form part implies that the partition function changes by an $(\Omega+\hbar^2\Delta)$-exact term under a deformation of the data. These are the two statements we wanted to prove. Therefore, we are just left with showing that the mdQME holds.

\subsection{Proof of the mdQME}
To prove this we simply apply $\II\hbar\ddd_J$ to $\Check Z$ and observe that, for every integral in the Feynman diagram expansion, we can use the generalized Stokes's formula
\[
\ddd_J\int\ = \pm \int\ddd\ \ \pm\int_\dd\ ,
\]
where the signs depend on dimensions (number of vertices) and on the choices of orientation. 

The first contribution on the right hand side replaces the propagators, one by one, by $\Hat\chi$. The sum of all these diagrams 
corresponds to the application of $\II\hbar\Delta$.

The second contribution is related to $\Omega$, as we are going to see. First observe that we are considering integrals over boundaries of compactified configuration spaces $C_n(I)$
of $n$ points on the interval $I=[t_\ii,t_\oo]$. Such boundaries are of two types.
\begin{description}
\item[Internal boundaries] In this case $k>1$ points collapse together in the interior of $I$. Such a boundary component fibers over $C_{n-k+1}(I)$ with a $(k-2)$-dimensional fiber.
\item[Endpoint boundaries] In this case $k\ge1$ points collapse together at one endpoint of $I$.
Such a boundary component fibers over $C_{n-k}(I)$ with a $(k-1)$-dimensional fiber.
\end{description}

The crucial point is that the residual fields and the forms $\Hat\chi$, which have no discontinuities, are basic in the above fibrations. Along the fibers we only have to integrate the product of the propagators $\Hat\eta$ whose both vertices are in the fiber. As the propagators are zero-forms, we may get a nonzero integral only if
\begin{enumerate}
\item we are at an internal boundary component with $k=2$, or
\item we are at an endpoint boundary component with $k=1$.
\end{enumerate}

In the first case, we have three kinds of collapsing diagrams (we call physical vertices those containing the constraint hamiltonians---and their derivatives---and ghost vertices those containing the structure constants):
\begin{enumerate}[i.]
\item Two physical vertices collapse together: these diagrams---which may have any positive number of physical propagators---produce a contribution of the form $[H_\alpha,\stackrel{\star}, H_\beta]$.
\item One physical vertex collapses with a ghost vertex: these diagrams produce a contribution of the form $f_{\alpha\beta}^\gamma H_\gamma$.
\item Two ghost vertices collapse together: these diagrams sum up to zero thanks to the Jacobi identity.
\end{enumerate}
The contributions in i.\ and ii.\ cancel each other thanks to \eqref{e:HHhatagain} in Assumption~\ref{a:assu}.

Therefore, we are only left with the case of a single bulk vertex collapsing at one endpoint. We consider the following cases:
\begin{enumerate}[(a)]
\item Consider first the case of a ghost vertex collapsing, say, at the initial endpoint. It is a very simple situatation, as the fiber contains one or two ghost propagators, connecting the vertex to initial endpoint insertions. This corresponds to the application of 
$f_{\alpha\beta}^\gamma\, c^\alpha_\ii c^\beta_\ii \frac\dd{\dd c^\gamma_\ii}$ or
$f_{\alpha\beta}^\gamma\, b_\gamma^\ii\frac{\dd^2}{\dd b_\alpha^\ii\dd b_\beta^\ii}$, depending on the choice of endpoint condition. The case of the final endpoint is similar.
\item Next consider the case of a physical vertex collapsing at the initial endpoint, where, by convention, we always choose a linear polarization. Depending on the ghost initial endpoint conditions, the ghost propagator stemming from the vertex
is either connected to an initial endpoint 
insertion---corresponding to the insertion of a $c_\ii$---or leaves the collapsing diagram---corresponding to the application of a derivative with respect to a $b^\ii$. Besides, there is an arbitrary number of physical propagators stemming from the collapsing vertex, some connecting it to an initial endpoint insertion and some leaving the collapsing diagram. 
Depending on the physical initial endpoint conditions, {summing over all such diagrams} corresponds to the application of 
$H_\alpha\left(\II\hbar\frac\dd{\dd q_\ii},q_\ii\right)$ or $H_\alpha\left(p^\ii,-\II\hbar\frac\dd{\dd p^\ii}\right)$.
This together with the result of (a) for the initial endpoint corresponds to the application of $\Omega_\ii$.
\item Finally consider the case of a physical vertex collapsing at the final endpoint. If the polarization is linear, this is very similar to case (b) and,  together with the result of (a) for the final endpoint, corresponds to the application of $\Omega_\oo$. 
In general, we have instead several propagators connecting the collapsing vertex to several $f$ insertions.
This sums up to a final endpoint insertion of $H_\alpha\left(\II\hbar\frac\dd{\dd q},q\right)\EE^{-\frac\II\hbar f}|_{q=\Hat q(t_\oo),Q=Q_\oo}$.\footnote{If the constraint hamiltonians $H_\alpha(p,q)$ are linear in $p$,
we have at most one physical propagator connecting the collapsing vertex to a final endpoint insertion of $f$. 
This then simply sums up to a final endpoint insertion of $H_\alpha(p^0,\Hat q)$, with $p^0=\frac{\dd f}{\dd q}(\Hat q(t_\oo),Q_\oo)$ as in \eqref{e:pzero}, times the insertion of $\EE^{-\frac\II\hbar f(\Hat q(t_\oo),Q_\oo)}$.}
By the assumed relation \eqref{e:quantumHtilderel}, we have\footnote{If the constraint hamiltonians are linear in the respetive momenta, this is simply
\[
\begin{split}
\Tilde H_\alpha\left(-\II\hbar\frac\dd{\dd Q_\oo},Q_\oo\right) \EE^{-\frac\II\hbar f(\Hat q(t_\oo),Q_\oo)}&=
\Tilde H_\alpha\left(P(p^0,\Hat q(t_\oo)),Q_\oo\right)  \EE^{-\frac\II\hbar f(\Hat q(t_\oo),Q_\oo)}\\
&=
H_\alpha(p^0,\Hat q(t_\oo))\,  \EE^{-\frac\II\hbar f(\Hat q(t_\oo),Q_\oo)},
\end{split}
\]
where we have used the classical relation \eqref{e:Htilderel} in the second equality.}
\[
\Tilde H_\alpha\left(-\II\hbar\frac\dd{\dd Q_\oo},Q_\oo\right) \EE^{-\frac\II\hbar f(\Hat q(t_\oo),Q_\oo)}=
H_\alpha\left(\II\hbar\frac\dd{\dd q},q\right)\EE^{-\frac\II\hbar f}|_{q=\Hat q(t_\oo),Q=Q_\oo},
\]
so this contribution, together with the result of (a) for the final endpoint, corresponds to the application of $\Omega_\oo$ as in \eqref{sube:Omegaoutnlin}.
\end{enumerate}

\subsection{Historical remarks}
The first instance of the method described in this appendix for tracking the changes under deformation of the gauge fixing goes back to \cite{AS} in the case of Chern--Simons theory without boundary, without residual fields and with gauge fixings parametrized by a Riemannian metric. The method was extended in \cite{BC} to more general propagators. In \cite{CMremCS} the case with residual fields was first considered, and in  \cite[Section 4.2]{CMR15} the general method delineated here---for general AKSZ theories, in presence of boundary and with residual fields---was introduced. By this method the boundary-related operator $\Omega$ is constructed in terms of boundary terms as explained above. The main difference with the present study is that there the case of dimension larger than one was considered, where some vanishing theorems simplify the computations. A differential, compatible with the BV structure, was explicitly introduced in \cite{BCM} to keep track of changes in the choice of background, leading to the differential QME (dQME).  The incorporation of the boundary, along the lines of \cite{CMR15}, leading to the mdQME was finally presented in \cite{CMW}.

\end{document}